\documentclass[10pt]{article}
\usepackage{amssymb}
\usepackage{tikz}
\usetikzlibrary{arrows, calc, decorations.pathmorphing, backgrounds, positioning, fit, petri, automata}
\definecolor{grey1}{rgb}{0.5,0.5,0.5}
\usepackage[top=2.54cm, bottom=2.54cm, left=2.2cm, right=2.2cm]{geometry}
\usepackage{algorithmicx}
\usepackage{algcompatible}
\usepackage{algorithm}
\usepackage{algpseudocode}
\makeatletter
\newenvironment{breakablealgorithm}
  {
   \begin{center}
     \refstepcounter{algorithm}
     \hrule height.8pt depth0pt \kern2pt
     \renewcommand{\caption}[2][\relax]{
       {\raggedright\textbf{\fname@algorithm~\thealgorithm} ##2\par}%
       \ifx\relax##1\relax 
         \addcontentsline{loa}{algorithm}{\protect\numberline{\thealgorithm}##2}%
       \else 
         \addcontentsline{loa}{algorithm}{\protect\numberline{\thealgorithm}##1}%
       \fi
       \kern2pt\hrule\kern2pt
     }
  }{
     \kern2pt\hrule\relax
   \end{center}
  }
\makeatother

\usepackage{changebar}

\usepackage{booktabs}

\usepackage{hyperref}
\hypersetup{ 
  colorlinks,
  citecolor=blue,
  linkcolor=red,
  urlcolor=magenta}

\usepackage{mathtools}
\usepackage{graphicx}
\usepackage{subfigure}
\usepackage{amsmath,amsthm,bm,bbm}
\usepackage{rotating}
\usepackage{mathrsfs}
\usepackage{chngcntr}
\usepackage{apptools}
\usepackage[titletoc,title]{appendix}
\usepackage{comment}

\usepackage{xcolor}         
\definecolor{grau}{rgb}{0.8,0.8,0.8}

\newcommand{\chen}[1]{\color{orange}}
\usepackage{setspace,lscape,longtable}
\usepackage{amsmath,amsthm,bm}
\usepackage{fullpage}  
\usepackage[shortlabels]{enumitem}
\usepackage{titlesec}
\usepackage[english]{babel}
\usepackage[square,numbers]{natbib}

\numberwithin{equation}{section}
\numberwithin{equation}{section}

\newtheorem{theorem}{Theorem}[section]

\newtheorem{lemma}[theorem]{Lemma}
\newtheorem{proposition}{Proposition}[section]
\theoremstyle{remark}
\newtheorem{remark}{Remark}
\newtheorem{definition}[theorem]{Definition}
\newtheorem{assumption}{Assumption}

\newtheorem{result}{Result}[section]

\DeclareMathOperator*{\argmax}{arg\,max}

\newcommand{\prob}{{\mathbb{P}}}
\newcommand{\var}{{\mathrm{var}}}
\newcommand{\expect}{\mathbb{E}}

\newcommand{\Optilde}{{\widetilde{O}_{p}}}

\newcommand{\transpose}{^{\mathrm{T}}}

\newcommand{\bdelta}{{\bm{\delta}}}

\newcommand{\calA}{{\mathcal{A}}}
\newcommand{\calB}{{\mathcal{B}}}

\newcommand{\calD}{{\mathcal{D}}}
\newcommand{\calE}{{\mathcal{E}}}

\newcommand{\calH}{{\mathcal{H}}}
\newcommand{\calI}{{\mathcal{I}}}

\newcommand{\calK}{{\mathcal{K}}}
\newcommand{\calL}{{\mathcal{L}}}
\newcommand{\calM}{{\mathcal{M}}}
\newcommand{\calN}{{\mathcal{N}}}

\newcommand{\calQ}{{\mathcal{Q}}}

\newcommand{\calS}{{\mathcal{S}}}
\newcommand{\calT}{{\mathcal{T}}}

\newcommand{\balpha}{{\boldsymbol{\alpha}}}

\newcommand{\bU}{{\mathbf{U}}}
\newcommand{\bP}{{\mathbf{P}}}

\newcommand{\bY}{{\mathbf{Y}}}
\newcommand{\bv}{{\mathbf{v}}}
\newcommand{\bx}{{\mathbf{x}}}

\newcommand{\bA}{{\mathbf{A}}}
\newcommand{\bB}{{\mathbf{B}}}
\newcommand{\bC}{{\mathbf{C}}}
\newcommand{\bD}{{\mathbf{D}}}
\newcommand{\bE}{{\mathbf{E}}}
\newcommand{\bF}{{\mathbf{F}}}
\newcommand{\bG}{{\mathbf{G}}}
\newcommand{\bH}{{\mathbf{H}}}
\newcommand{\bM}{{\mathbf{M}}}

\newcommand{\bK}{{\mathbf{K}}}
\newcommand{\bL}{{\mathbf{L}}}
\newcommand{\bJ}{{\mathbf{J}}}
\newcommand{\bR}{{\mathbf{R}}}
\newcommand{\bQ}{{\mathbf{Q}}}
\newcommand{\bS}{{\mathbf{S}}}
\newcommand{\bW}{{\mathbf{W}}}
\newcommand{\bZ}{{\mathbf{Z}}}
\newcommand{\ba}{{\mathbf{a}}}
\newcommand{\bw}{{\mathbf{w}}}

\newcommand{\bh}{{\mathbf{h}}}

\newcommand{\br}{{\mathbf{r}}}

\newcommand{\bt}{{\mathbf{t}}}
\newcommand{\bu}{{\mathbf{u}}}

\newcommand{\by}{{\mathbf{y}}}
\newcommand{\bz}{{\mathbf{z}}}

\newcommand{\be}{{\mathbf{e}}}
\newcommand{\bb}{{\mathbf{b}}}

\newcommand{\bV}{{\mathbf{V}}}
\newcommand{\bkappa}{{\bm{\kappa}}}
\newcommand{\bgamma}{{\bm{\gamma}}}
\newcommand{\bGamma}{{\bm{\Gamma}}}
\newcommand{\bbeta}{{\bm{\beta}}}
\newcommand{\bDelta}{{\bm{\Delta}}}

\newcommand{\bomega}{{\bm{\omega}}}
\newcommand{\bsigma}{{\bm{\sigma}}}
\newcommand{\bSigma}{{\bm{\Sigma}}}
\newcommand{\bXi}{{\bm{\Xi}}}
\newcommand{\bxi}{{\boldsymbol{\xi}}}
\newcommand{\beps}{{\bm{\varepsilon}}}
\newcommand{\eye}{{\mathbf{I}}}
\newcommand{\one}{{\mathbf{1}}}
\newcommand{\bTheta}{{\bm{\Theta}}}
\newcommand{\btheta}{{\bm{\theta}}}
\newcommand{\bzeta}{{\bm{\zeta}}}
\newcommand{\bmu}{{\bm{\mu}}}
\newcommand{\bnu}{{\bm{\nu}}}

\newcommand{\zero}{{\bm{0}}}
\newcommand{\eps}{\epsilon}

\begin{document}

{
  \title{\bf Uncertainty quantification for mixed membership in multilayer networks with degree heterogeneity using Gaussian variational inference}
  \author{Fangzheng Xie \\
    Department of Statistics\\
    Indiana University\\
    and \\
    Hsin-Hsiung Huang\thanks{
    Correspondence should be addressed to Hsin-Hsiung Huang (Hsin-Hsiung.Huang@ucf.edu)}
    \hspace{.2cm}\\
    School of Data, Mathematical
    and Statistical Sciences\\ University of Central Florida}
  \maketitle
}


\maketitle

\begin{abstract}
Analyzing multilayer networks is central to understanding complex relational measurements collected across multiple conditions or over time. A pivotal task in this setting is to quantify uncertainty in community structure while appropriately pooling information across layers and accommodating layer-specific heterogeneity. Building on the multilayer degree-corrected mixed-membership (ML-DCMM) model, which captures both stable community membership profiles and layer-specific vertex activity levels, we propose a Bayesian inference framework based on a spectral-assisted likelihood. We then develop a computationally efficient Gaussian variational inference algorithm implemented via stochastic gradient descent. Our theoretical analysis establishes a variational Bernstein--von Mises theorem, which provides a frequentist guarantee for using the variational posterior to construct confidence sets for mixed memberships. We demonstrate the utility of the method on a U.S. airport longitudinal network, where the procedure yields robust estimates, natural uncertainty quantification, and competitive performance relative to state-of-the-art methods.
\end{abstract}

\noindent%

\newpage

\setcounter{tocdepth}{1}
\tableofcontents

\section{Introduction}
\label{sec:introduction}

Networks are a class of non-Euclidean data that characterize relationships among multiple entities rather than their individual attributes. Network data arise in a wide variety of application domains, including social science \citep{doi:10.1126/science.1236498,wasserman1994social}, neuroscience \citep{eichler2017complete,10.1001/jamapsychiatry.2022.0020}, and microbiome studies \citep{doi:10.1073/pnas.1809349115}, among others. Statistical network analysis typically views an observed network as a random realization from a latent network distribution, where vertices are treated as fixed and edges, which encode relational features, are modeled as random variables. Within this framework, numerous probabilistic network models have been proposed, such as the stochastic block model \citep{holland_stochastic_1983} and its degree-corrected and mixed membership variants \citep{airoldi_mixed_2008,JIN2023,PhysRevE.83.016107}, generalized random dot product graphs \citep{https://doi.org/10.1111/rssb.12509,10.1007/978-3-540-77004-6_11}, latent space models \citep{Hoff01122002}, and exchangeable random graphs \citep{10.1111/rssb.12233}. We refer readers to \cite{JMLR:v18:16-480,JMLR:v18:17-448,gao_minimax_2021} for reviews.

Multilayer networks \citep{10.1093/comnet/cnu016}, in which multiple network layers are observed over a common set of vertices, have become increasingly prevalent in contemporary applications such as protein--protein interaction networks \citep{10.1371/journal.pcbi.1000807} and public transportation systems \citep{bts-style-manual-2025}. Such multilayer networks naturally arise when relational data are collected across different sources or contexts, for example social interactions across distinct platforms, brain connectivity under various cognitive states, or international trade flows over a sequence of years or across product categories. These data exhibit intrinsic heterogeneity across layers and often involve large numbers of vertices and edges. This combination calls for network models and inferential methods that are expressive enough to capture salient structural features while remaining computationally tractable at scale, a requirement that goes well beyond their single-layer counterparts.

Community detection \citep{JMLR:v18:16-480} plays a central role in statistical network analysis, with the stochastic block model \citep{holland_stochastic_1983} providing a foundational framework for this task. Multilayer extensions, such as the multilayer stochastic block model and its degree-corrected variants, have been widely investigated for community detection in multilayer networks \citep{Agterberg28072025,bhattacharyya2018spectral,bhattacharyya2020consistent,7990044,6265414,PhysRevX.6.031005,10.1093/biomet/asz068,Lei02102023,doi:10.1073/pnas.1718449115,10023524,10.1111/rssb.12200,paul2016consistent,paul2020spectral,5360349,xie2024biascorrected,6758385,doi:10.1080/01621459.2016.1260465}. Nevertheless, these models typically assume that vertices are partitioned into disjoint communities and therefore do not accommodate mixed memberships or the uncertainty associated with community allocation.

Uncertainty quantification for heterogeneous multilayer networks has been studied in \cite{ArroyoEtAlCOSIE,xie2024biascorrected,zheng2022limit} under the common subspace independent edge (COSIE) model. In this model, each layer has a low-rank expected adjacency matrix that may vary across layers, but all layers share a common principal subspace. This framework effectively captures layer-wise eigenvalue heterogeneity, but does not account for layer-wise degree heterogeneity, which is essential in many longitudinal and multilayer networks \citep{he2023semiparametric}. Moreover, existing methods in this line of work are primarily based on spectral decompositions of adjacency matrices rather than likelihood-based modeling, which limits their ability to provide fully probabilistic uncertainty quantification.

Bayesian methods have recently been developed for multilayer network models to facilitate uncertainty quantification \citep{loyal2024fast,loyal2025generalized,loyal2023eigenmodel,zhao2024structured}. These approaches typically adopt a dynamic perspective, in which latent feature vectors evolve over time according to a Markov or random-walk process, and inference is carried out via Markov chain Monte Carlo or structured mean-field variational inference. Convergence rates for the corresponding posteriors have been established in several settings. However, for these dynamic models the validity of the resulting uncertainty quantification, such as the coverage properties of posterior credible sets, has not been rigorously justified.

In this paper, we focus on the \emph{multilayer degree-corrected mixed membership} (ML-DCMM) network model, which extends the degree-corrected mixed membership (DCMM) model of \cite{JIN2023} to the multilayer setting. In the single-layer DCMM model with $n$ vertices labeled by $\{1,\ldots,n\}$ and $d$ communities, each vertex $i$ is associated with a $d$-dimensional probability vector $\bz_i$ that encodes its membership allocation across communities, together with a degree correction parameter $\theta_i$ that accounts for degree heterogeneity. The edge probability between vertices $i$ and $j$ is given by $\theta_i\theta_j\bz_i\transpose\bB\bz_j$, where $\bB$ is a $d\times d$ block probability matrix. In the ML-DCMM network model, each layer is a DCMM network on the same vertex set. The community membership profile $\bz_i$ of each vertex is shared across layers, capturing an intrinsic and stable structural role, whereas the degree correction parameters and block probability matrices are allowed to vary by layer, capturing dynamic or context-specific activity patterns. This structure is particularly relevant for applications such as international trade, where a country's fundamental economic profile evolves slowly, but its trade volumes and partner relationships fluctuate over time in response to policy changes or economic shocks.

The overarching goal of this work is to develop a statistically principled and computationally efficient framework for uncertainty quantification of the mixed membership structure in heterogeneous multilayer networks generated under the ML-DCMM network model. Our main contributions are summarized as follows:
\begin{enumerate}[(1), noitemsep, topsep = 0ex]
  \item We construct a Bayesian inference framework for the ML-DCMM network model by introducing a \emph{spectral-assisted likelihood} that combines spectral information from the adjacency matrices with the likelihood structure. Building on this formulation, we develop a Gaussian variational inference (VI) algorithm that is scalable to multiple networks with a large number of layers. A key ingredient is a block Cholesky factorization of the joint Fisher information matrix of the membership profile and degree correction parameters for each vertex. This yields a structured class of covariance matrices that substantially reduces computational cost relative to a fully unstructured Gaussian VI while still allowing for nontrivial dependence between membership profiles and degree parameters.

  \item We establish the theoretical validity of using the variational posterior for uncertainty quantification by proving a variational Bernstein--von Mises theorem. This result clarifies a trade-off between computational efficiency and the fidelity of uncertainty quantification. A fully unstructured Gaussian VI can, in principle, deliver valid uncertainty quantification but is computationally burdensome, whereas a structured mean-field VI is computationally attractive but generally fails to capture the correct uncertainty. In contrast, we show that the proposed Gaussian VI, with a carefully structured covariance matrix, achieves a practical compromise: it is both computationally efficient and asymptotically valid for uncertainty quantification.
\end{enumerate}

The remainder of the paper is organized as follows. Section \ref{sec:MLDCMM} introduces the ML-DCMM model and presents a preliminary estimation algorithm based on layer-wise mixed-SCORE and aggregation. Section \ref{sec:gaussian_variational_inference} develops the proposed Gaussian VI method via the spectral-assisted likelihood and exploits the Cholesky factorization of the per-vertex Fisher information matrix. Section \ref{sec:theoretical_properties} presents the main theoretical results, including Bernstein--von Mises theorems for both the exact posterior and the variational posterior. Section \ref{sec:real_world_data_analysis} illustrates the methodology on a U.S. airport transportation network dataset. Section \ref{sec:conclusion} concludes with a discussion.

\noindent
\textbf{Notations:} Given a positive integer $n\in\mathbb{N}_+$, let $[n] = \{1,2,\ldots,n\}$ denote the set of consecutive positive integers from $1$ through $n$. For a symmetric matrix $\bA\in\mathbb{R}^{n\times n}$, let $\lambda_k(\bA)$ denote its $k$th largest eigenvalue, ordered so that $\lambda_1(\bA)\geq\ldots\geq\lambda_n(\bA)$. For any matrix $\bA\in\mathbb{R}^{m\times n}$, let $\sigma_k(\bA)$ denote its $k$th largest singular value, ordered so that $\sigma_1(\bA)\geq\ldots\geq\sigma_{\min(m,n)}(\bA)$. Given a vector $\ba = [a_1,\ldots,a_n]\transpose\in\mathbb{R}^n$, define $\textsf{diag}(\ba) = \textsf{diag}(a_1,\ldots,a_n)$ as the diagonal matrix whose $k$th diagonal element is $a_k$. For a vector $\bu = [u_1,\ldots,u_n]\transpose\in\mathbb{R}^n$, we write $\|\bu\|_2 = (\sum_{i = 1}^nu_i^2)^{1/2}$ for its Euclidean norm. When the dimension is clear from context, $\be_i$ denotes the $i$th standard basis vector, whose $i$th element is $1$ and all other elements are zero. We use $\one_d$ for the $d$-dimensional vector of ones and $\eye_d$ for the $d\times d$ identity matrix.

For two nonnegative sequences $(a_n)_{n = 1}^\infty$ and $(b_n)_{n = 1}^\infty$, we write $a_n \gg b_n$ if $\lim_{n\to\infty}a_n/b_n = \infty$, $a_n \ll b_n$ if $b_n \gg a_n$, $a_n \gtrsim b_n$ if there exists a constant $C > 0$ such that $a_n\geq Cb_n$ for all $n\in\mathbb{N}_+$, $a_n\lesssim b_n$ if $b_n\gtrsim a_n$, and $a_n\asymp b_n$ if both $a_n\gtrsim b_n$ and $a_n\lesssim b_n$ hold. For a sequence of random variables $(X_n)_{n = 1}^\infty$ and a nonnegative deterministic sequence $(\epsilon_n)_{n = 1}^\infty$, we write $X_n = \Optilde(\epsilon_n)$ if for any $c > 0$ there exist constants $K_c,M_c > 0$ such that $\prob(|X_n| > K_c\epsilon_n)\leq M_c n^{-c}$. Similarly, a sequence of events $(\calE_n)_{n = 1}^\infty$ is said to occur with high probability (w.h.p.) if for any $c > 0$ there exists a constant $K_c > 0$ with $\prob(\calE_n)\geq 1 - K_c n^{-c}$. For two probability measures $\mu_1(\mathrm{d}\bx)$ and $\mu_2(\mathrm{d}\bx)$ on $\mathbb{R}^d$, the total variation distance is defined as
\[
d_{\mathrm{TV}}(\mu_1, \mu_2) = \sup_{\calA} \bigg|\int_\calA \mu_1(\mathrm{d}\bx) - \int_\calA \mu_2(\mathrm{d}\bx)\bigg|.
\]
If $\mu_1$ and $\mu_2$ are mutually absolutely continuous, we denote their Kullback--Leibler (KL) divergence by
\[
D_{\mathrm{KL}}(\mu_1\|\mu_2) = \expect_{\mu_1}\bigg[\log\bigg(\frac{\mathrm{d}\mu_1(\bx)}{\mathrm{d}\mu_2}\bigg)\bigg].
\]
Finally, we write $\calN(\bx\mid\bmu, \bSigma)$ to indicate that $\bx\sim\mathrm{N}(\bmu, \bSigma)$.

\section{Preliminaries}
\label{sec:MLDCMM}



We begin by reviewing the degree-corrected mixed-membership (DCMM) network model introduced in \cite{JIN2023} for single-layer network data. For generality, we assume networks are weighted and the edge weight densities are in the following exponential family
\begin{align}
\label{eqn:exponential_family}
\mathcal{F} = \{f(x; \mu) = h(x)\exp\{\eta(\mu)x - B(\mu)\}:\mu\in\mathcal{I} \subset (0, +\infty)\}
\end{align} 
with respect to some underlying $\sigma$-finite measure, where $B,\eta:\calI\to\mathbb{R}$ are continuously differentiable functions such that $B'(\mu) =  \mu\eta'(\mu)$. By construction, if $X$ has density $f(x; \mu)$, then $\expect X = \mu$. Throughout, we assume that each (random) weighted edge in the networks follows an exponential family distribution of the form \eqref{eqn:exponential_family}.For the theoretical analysis we assume that $\calI$ is a compact interval contained in $(0,+\infty)$; see Assumption~\ref{assumption:likelihood}(a).

\begin{definition}[DCMM]
\label{def:DCMM}
Consider a network over the set of vertices $[n]$. For each $i\in[n]$, let $\bz_i = [z_{i1},\ldots,z_{id}]\transpose$ be the membership profile explaining the percentage of community assignment for vertex $i$, and it is subject to the constraints $\sum_{k = 1}^dz_{ik} = 1$ and $z_{ik}\in[0, 1]$ for all $k\in[d]$. Denote by $\bZ = [\bz_1,\ldots,\bz_n]\transpose$ and let $\bTheta = \textsf{diag}(\theta_1,\ldots,\theta_n)$, where $\theta_i\in\calI$ is the degree-correction parameter for vertex $i$ capturing the degree heterogeneity. Suppose $\bB = [B_{kl}]_{d\times d}\in(0, 1)^{d\times d}$ is a block mean matrix such that $B_{kk} = 1$ for all $k\in [d]$. Then, we say that $\bA = [A_{ij}]_{n\times n}$ is the adjacency matrix of a (weighted) degree-corrected mixed-membership network, denoted by $\bA\sim\textsf{DCMM}(\bB, \bZ, \bTheta)$, if $A_{ij}\sim f(\cdot; P_{ij})$ independently for all $i,j\in[n]$, $i\leq j$, and $A_{ij} = A_{ji}$ if $i > j$, where
$P_{ij} = \theta_i\bz_i\transpose\bB\bz_j\theta_j$. 
\end{definition}
As mentioned in \cite{JIN2023}, DCMM networks form a rich class of models that encompass several popular examples in network science. For instance, the stochastic block model \citep{holland_stochastic_1983} coincides with the DCMM network model when there is no degree heterogeneity ($\theta_i = 1$ for all $i\in[n]$) and each vertex can only belong to one community ($\bz_i = \be_k$ for some $k\in[d]$ for all $i\in[n]$). Similarly, the mixed-membership stochastic block model \citep{airoldi_mixed_2008} is a special case of the DCMM network model without degree heterogeneity. Another perspective on the DCMM network model is through the signal-plus-noise formulation: if $\bA\sim\textsf{DCMM}(\bB, \bZ, \bTheta)$, then $\bA = \bP + \bE$, where $\bP = \bTheta\bZ\bB\bZ\transpose\bTheta$, 
and $\bE = [E_{ij}]_{n\times n}$ is a symmetric noise matrix whose upper triangular entries are independent mean-zero random variables. 

Now we are in a position to introduce the multilayer degree-corrected mixed membership (ML-DCMM) network model. Roughly speaking, the ML-DCMM network model is a collection of independent DCMM networks over a common vertex set, and they share a consensus membership profile matrix across layers. Formally, we define the ML-DCMM network model as follows. 
\begin{definition}[ML-DCMM]
\label{def:ML_DCMM}
Consider $m$ networks over a common set of vertices $[n]$. Let $\bZ = [\bz_1,\ldots,\bz_n]\transpose$ be a consensus membership profile matrix, where $\bz_i = [z_{i1},\ldots,z_{id}]\transpose$ is the membership profile for vertex $i$ such that $\sum_{k = 1}^dz_{ik} = 1$ and $z_{ik}\geq 0$ for all $k\in[d]$. For each $t\in[m]$, suppose $\bTheta^{(t)} = \textsf{diag}(\theta_1^{(t)},\ldots,\theta_n^{(t)})$ is the degree correction matrix for layer $t$, and $\bB^{(t)} = [B_{kl}^{(t)}]_{d\times d}$ is the block mean matrix for layer $t$, where $\theta_i^{(t)}\in\calI$ and $B_{kl}^{(t)}\in(0, 1)$ hold for all $i\in[n]$, $k,l\in[d]$. Then, we say that $(\bA^{(t)})_{t = 1}^m = ([A_{ij}^{(t)}]_{n\times n})_{t = 1}^m$ are the adjacency matrices of a collection of (weighted) multilayer degree-corrected mixed-membership networks, denoted by $(\bA^{(1)},\ldots,\bA^{(m)})\sim\textsf{ML-DCMM}(\bZ, (\bB^{(t)})_{t = 1}^m, (\bTheta^{(t)})_{t = 1}^m)$, if $\bA^{(t)}\sim\textsf{DCMM}(\bZ, \bB^{(t)}, \bTheta^{(t)})$ independently for all $t\in[m]$. 
\end{definition}


Estimating the consensus membership profile matrix $\bZ$ in the ML-DCMM network model is a nontrivial task. A natural heuristic is to apply the spectral-decomposition-based mixed-SCORE algorithm \citep{JIN2023} for each layer separately and then aggregate by averaging. We describe this procedure in Algorithm \ref{alg:aggregated_mixed_SCORE} below. 

\begin{breakablealgorithm}
\caption{Aggregated mixed-SCORE}
\label{alg:aggregated_mixed_SCORE}
\begin{algorithmic}[1]
\State \textbf{Input:} Adjacency matrices $(\bA^{(t)})_{t = 1}^m$ and number of communities $d$.
\For{$t = 1,2,\ldots,m$}
    \State Apply mixed-SCORE to obtain estimators $(\overline{\bZ}^{(t)}, \widetilde{\bTheta}^{(t)}, \overline{\bB}^{(t)})$ from $\bA^{(t)}$. 
    \State Find a permutation matrix $\bC^{(t)}\in\{0, 1\}^{d\times d}$ that minimizes $\|\overline{\bZ}^{(t)}\bC^{(t)} - \overline{\bZ}^{(1)}\|_{\mathrm{F}}^2$. 
    \State Compute $\widetilde{\bB}^{(t)} = (\bC^{(t)})\transpose\overline{\bB}^{(t)}(\bC^{(t)})$. 
\EndFor
\State \textbf{Output: } Aggregated estimators $\widetilde{\bZ} = (1/m)\sum_{t = 1}^m\overline{\bZ}^{(t)}\bC^{(t)}$, $(\widetilde{\bTheta}^{(t)})_{t = 1}^m$, and $(\widetilde{\bB}^{(t)})_{t = 1}^m$. 
\end{algorithmic}
\end{breakablealgorithm}
Although the aggregated mixed-SCORE estimator given by Algorithm \ref{alg:aggregated_mixed_SCORE} is straightforward and intuitive, two issues remain. First, Algorithm \ref{alg:aggregated_mixed_SCORE} aggregates the layer-wise mixed-SCORE estimator via simple averaging, which implicitly assigns equal weight to each layer. Nonetheless, for any weight vector $\bomega = [\omega_1,\ldots,\omega_m]\transpose$ with $\sum_{t = 1}^m\omega_t = 1$ and $\omega_t\geq 0$, the weighted estimator $\widetilde{\bZ}_\bomega = \sum_{t = 1}^m\omega_t\overline{\bZ}^{(t)}\bC^{(t)}$ also pools the shared membership profiles across layers. Nevertheless, it is unclear how to choose an optimal, possibly data-dependent weighting scheme. Second, the mixed-SCORE estimators rely exclusively on the truncated spectral information of the adjacency matrices and do not fully exploit the likelihood structure. These limitations call for a statistically principled framework for uncertainty quantification of the membership profiles in ML-DCMM networks. 

\section{Spectral-Assisted Gaussian Variational Inference}
\label{sec:gaussian_variational_inference}

This section elaborates on the proposed spectral-assisted Gaussian variational inference for the ML-DCMM network model. Following the spirit of \cite{wuxie2025,wu2022statistical,XieWu2024,doi:10.1080/01621459.2021.1948419}, we begin by considering the following oracle problem for a single vertex $i$ across layers $t\in[m]$: the parameters of inferential interest are the membership profile $\bz_i$ and the degree correction parameter $\btheta_i = [\theta_i^{(1)},\ldots,\theta_i^{(m)}]\transpose$, while the remaining parameters, \emph{i.e.}, $(\bz_{0j}, \btheta_{0j})_{j\in[n]\backslash\{i\}}$, and $(\bB_0^{(t)})_{t = 1}^m$, are treated as known at their true values. Here, the subscript $0$ denotes the true value of a parameter. For each $i\in[n]$, the oracle problem can be described by the model
\begin{align}\label{eqn:oracle_problem}
A_{ij}^{(t)} \sim f(\cdot; P_{ij}^{(t)}),\quad P_{ij}^{(t)} = \theta_i^{(t)}\bz_i\transpose\bB_0^{(t)}\bz_{0j}\theta_{0j}^{(t)},\quad j\in[n]\backslash\{i\},\quad t\in[m].
\end{align}
The oracle problem \eqref{eqn:oracle_problem} is not practically implementable because it requires knowledge of the true values of unknown parameters. In the context of single-layer random dot product graphs, prior work \citep{wuxie2025,wu2022statistical,doi:10.1080/01621459.2021.1948419} addressed this challenge by replacing the unknown parameters with the adjacency spectral embedding (\emph{i.e.}, the scaled eigenvectors corresponding to the top-$d$ eigenvalues of the adjacency matrix). However, this strategy is not entirely applicable to the ML-DCMM network model due to (i) the more intricate relationship between the spectral decomposition of the adjacency matrices and the membership profiles, and (ii) the heterogeneity introduced by the layer-specific degree correction parameters. In what follows, we first introduce the spectral-assisted likelihood and then present the corresponding Gaussian variational inference method. 

\subsection{Spectral-Assisted Likelihood}
\label{sub:spectral_assisted_likelihood}
The main preliminary for spectral-assisted likelihood is to construct suitable spectral-based estimators for $(\bz_{0j}, \btheta_{0j})_{j\in[n]\backslash\{i\}}$ and $(\bB_0^{(t)})_{t = 1}^m$. To address the two challenges noted above, we estimate $(\bz_{0j})_{j\in[n]\backslash\{i\}}$, $(\btheta_{0j})_{j\in[n]\backslash\{i\}}$, $(\bB_0^{(t)})_{t = 1}^m$ using the aggregated mixed-SCORE estimators $(\widetilde{\bz}_j)_{j\in[n]\backslash\{i\}}$, $(\widetilde{\btheta}_j)_{j\in[n]\backslash\{i\}}$, $(\widetilde{\bB}^{(t)})_{t = 1}^m$, where $\widetilde{\bz}_j$ denotes the $j$th row of $\widetilde{\bZ}$, $\widetilde{\btheta}_j = [\widetilde{\theta}_j^{(1)},\ldots,\widetilde{\theta}_j^{(m)}]\transpose$, and $\widetilde{\theta}_j^{(t)}$ is the $j$th diagonal entry of $\widetilde{\bTheta}^{(t)}$. 

Now denote by
\begin{align*}
\Delta^{d - 1} &= \bigg\{\bz = [z_1,\ldots,z_d]\transpose\in\mathbb{R}^d:\sum_{k = 1}^dz_k = 1,z_k\geq 0\text{ for all }k\in[d]\bigg\},\\
\calS^{d - 1} &= \bigg\{\bz^* = [z_1^*,\ldots,z_{d - 1}^*]\transpose\in\mathbb{R}^{d - 1}:\sum_{k = 1}^{d - 1}z_k^* \leq 1,z_k^*\geq 0\text{ for all }k\in[d - 1]\bigg\},
\end{align*}
and define the one-to-one function $\bkappa:\calS^{d - 1}\to\Delta^{d - 1}$ by $\kappa(\bz^*) = [z_1^*,\ldots,z_{d - 1}^*,1 - \sum_{k = 1}^{d - 1}z_k^*]\transpose$, where $\bz^* = [z_1^*,\ldots,z_{d - 1}^*]\transpose$. Here, $\Delta^{d - 1}$ is the $(d - 1)$-dimensional unit simplex of all probability vectors in $\mathbb{R}^d$, $\calS^{d - 1}$ is a compact $(d - 1)$-dimensional reparameterization domain for $\Delta^d$, and $\bz^*$ corresponds to the first $(d - 1)$ entries of $\bkappa(\bz^*)$. Writing $\bJ = [\eye_{d - 1}, -\one_{d - 1}]\transpose$, it is immediate that $\bkappa(\bz_i^*) = \bJ\bz_i^* + \be_d$, and the inverse map is given by $\bkappa^{-1}(\bz_i) = (\bJ\transpose\bJ)^{-1}\bJ\transpose(\bz_i - \be_d)$ for any $\bz_i\in\Delta^{d - 1}$. 
Given the preliminary estimators $(\widetilde{\bZ}, (\widetilde{\btheta}_j)_{j = 1}^n, (\widetilde{\bB}^{(t)})_{t = 1}^m)$, we define the following spectral-assisted log-likelihood function:
\begin{align}\label{eqn:spectral_assisted_loglik}
\widetilde{\ell}_{in}(\bz_i^*, \btheta_i) & = \sum_{t = 1}^m\sum_{j = 1}^n\log f(A_{ij}^{(t)}, \theta_i^{(t)}\bkappa(\bz_i^*)\transpose\widetilde{\bB}^{(t)}\widetilde{\bz}_j\widetilde{\theta}_{j}^{(t)}), 
\end{align}
Note that \eqref{eqn:spectral_assisted_loglik} and \eqref{eqn:oracle_problem} differ by a term corresponding to $A_{ii}^{(t)}\sim f(\cdot;P_{ij}^{(t)})$, but such a distinction is immaterial and does not affect the theoretical properties in Section \ref{sec:theoretical_properties}. 
\begin{remark}[Separability]
The proposed spectral-assisted log-likelihood \eqref{eqn:spectral_assisted_loglik} enjoys the following convenient advantage:
Once the preliminary estimators $(\widetilde{\bZ}, (\widetilde{\btheta}_j)_{j = 1}^n, (\widetilde{\bB}^{(t)})_{t = 1}^m)$ are obtained, the spectral-assisted log-likelihood $\widetilde{\ell}_{in}$ is only a function of $(\bz_i^*, \btheta_i)$ and no longer depends on the remaining unknown parameters.
This decoupling arises because replacing the unknown $\bz_j$, $\btheta_j$, $\bB^{(t)}$ with their preliminary estimators decomposes the interaction between the vertex-$i$-specific parameters $(\bz_i, \btheta_i)$ and the vertex-$j$-specific parameters $(\bz_j, \btheta_j)$ for $j\neq i$. 
Such separability is beneficial for both theoretical analysis and practical computation. 
\end{remark}

\subsection{Gaussian Variational Inference}
\label{sub:gaussian_variational_inference}

To describe VI, we first formulate the inference problem associated with \eqref{eqn:spectral_assisted_loglik} in a Bayesian framework. Let $\pi_\theta(\cdot)$ be a prior density for the degree correction parameters $(\theta_i^{(t)}:i\in[n],t\in[m])$, and let $\pi_{\bz^*}(\cdot)$ be a prior density function over $\calS^{d - 1}$. For any $i\in[n]$, we consider the following joint posterior distribution of $(\bz_i^*, \btheta_i)$ given by
\begin{align}\label{eqn:posterior}
\pi_{(\bz_i^*, \btheta_i)}(\bz_i^*, \btheta_i\mid\mathbb{A}) = \frac{\exp\{\widetilde{\ell}_{in}(\bz_i^*, \btheta_i)\}\pi_{\bz^*}(\bz_i^*)\prod_{t = 1}^m\pi_\theta(\theta_i^{(t)})}{\iint_{\calS^{d - 1}\times\calI^m}\exp\{\widetilde{\ell}_{in}(\bz_i^*, \btheta_i)\}\pi_{\bz^*}(\bz_i^*)\prod_{t = 1}^m\pi_\theta(\theta_i^{(t)})\mathrm{d}\bz_i^*\mathrm{d}\btheta_i},
\end{align}
where we denote by $\mathbb{A} = (\bA^{(t)})_{t = 1}^m$. The above posterior distribution provides a natural ground for uncertainty quantification. Classical Bayesian inference proceeds by computing this posterior using an MCMC sampler that generates Markov chains whose stationary distribution coincides with $\pi_{(\bz_i^*, \btheta_i)}(\bz_i^*, \btheta_i\mid\mathbb{A})$. In contrast, this section focuses on VI as a computationally efficient alternative to MCMC for Bayesian inference. 

In general, VI is a class of algorithms that seeks a more tractable distribution, referred to as the \emph{variational posterior distribution}, within a pre-specified class of distributions $\calQ$, to approximate the \emph{exact} posterior distribution $\pi_{(\bz_i^*, \btheta_i)}(\bz_i^*, \btheta_i\mid\mathbb{A})$ defined in \eqref{eqn:posterior}. This is done by solving the minimization problem
\begin{align*}
\min_{q_i(\bz_i^*, \btheta_i)\in\calQ}D_{\mathsf{KL}}(q_i(\bz_i^*, \btheta_i) \| \pi_{(\bz_i^*, \btheta_i)}(\bz_i^*, \btheta_i\mid\mathbb{A})).
\end{align*}
In Gaussian VI, the candidate class $\calQ$ is taken as the class of all multivariate Gaussian distributions with non-degenerate covariance matrices. Note, however, that this is not directly implementable in our setting because the support of $\bz_i^*$ in $\pi_{(\bz_i^*, \btheta_i)}(\bz_i^*, \btheta_i\mid\mathbb{A})$ is the bounded set $\calS^{d - 1}$, while Gaussian VI requires $q_i(\bz_i^*, \btheta_i)$ to be supported over $\mathbb{R}^{d - 1}\times\mathbb{R}^m$. 

To address this, we introduce continuously differentiable one-to-one transformations
$\calT_\bz:\mathbb{R}^{d - 1}\to\calS^{d - 1}$ and $\calT_\btheta:\mathbb{R}^m\to\calI^m$:
\begin{align}\label{eqn:Tz_transformation}
\calT_\bz(\bx) = \left\{
\begin{aligned}
&\frac{1}{1 + e^{x_1}},&\quad&\text{if }d = 2,\\
&\begin{bmatrix}
1 - \frac{1}{1 + e^{-x_{1}}} \\ 
\frac{1}{1 + e^{-x_{1}}}\left(1 - \frac{1}{1 + e^{-x_2}}\right) \\
\vdots \\
\prod_{k = 1}^{d - 2}\frac{1}{1 + e^{-x_k}}\left(1 - \frac{1}{1 + e^{-x_{d - 1}}}\right)
\end{bmatrix},&\quad&\text{if }d\geq 3.
\end{aligned}\right.
\end{align}
and
\begin{align}\label{eqn:Ttheta_transformation}
\calT_\btheta(\bnu) = 
&\begin{bmatrix}
    \calT_\btheta^{(1)}(\nu_1)
    \\
    \vdots
    \\
    \calT_\btheta^{(m)}(\nu_m)
\end{bmatrix},\quad\text{where } \calT_\btheta^{(t)}(\nu_{t}) = \inf\calI + \frac{\sup\calI - \inf\calI}{1 + e^{-\nu_t}},\quad t\in[m].
\end{align}
Let $\partial\calT_\bz(\bx)/\partial\bx\transpose$ and $\partial\calT_\btheta(\bnu)/\partial\bnu\transpose$ be the Jacobians of $\calT_\bz$ and $\calT_\btheta$, respectively. Denote by
\begin{align}
\label{eqn:prior_change_of_variable}
\pi_\bx(\bx) = \pi_{\bz^*}(\calT_\bz(\bx))\bigg|\det\bigg(\frac{\partial\calT_\bz}{\partial\bx\transpose}(\bx)\bigg)\bigg|,\; \pi_{\bnu}(\bnu) = \prod_{t = 1}^m\pi_\nu(\nu_t)
\end{align}
the prior distributions of $\bx$ and $\bnu$ induced from \eqref{eqn:Tz_transformation}--\eqref{eqn:Ttheta_transformation} and \eqref{eqn:posterior}, respectively, where
\[
\pi_\nu(\nu_{t}) = \pi_\theta(\calT_\btheta^{(t)}(\nu_{t}))\bigg|\frac{\mathrm{d}\calT_\btheta^{(t)}(\nu_{t})}{\mathrm{d}\nu_{t}}\bigg|.
\]
Correspondingly, let  
\begin{align}
\label{eqn:posterior_change_of_variable}
\pi_{(\bx_i, \bnu_i)}(\bx_i, \bnu_i\mid\mathbb{A})
& = \pi_{(\bz_i^*, \btheta_i)}(\calT_\bz(\bx_i), \calT_\btheta(\bnu_i)\mid\mathbb{A})\bigg|\det\bigg(\frac{\partial\calT_\bz}{\partial\bx\transpose}(\bx_i)\bigg)
\det\bigg(\frac{\partial\calT_\btheta}{\partial\bnu\transpose}(\bnu_i)\bigg)
\bigg|
\end{align}
be the joint posterior distribution of $(\bx_i, \bnu_i)$ induced from $\bz_i^* = \calT_\bz(\bx_i)$, $\btheta_i = \calT_\btheta(\bnu_i)$, and \eqref{eqn:posterior}. 

Rather than minimizing the KL divergence between the variational posterior distribution and the exact posterior distribution $\pi_{(\bz_i^*, \btheta_i)}(\bz_i^*, \btheta_i\mid\mathbb{A})$, we consider the equivalent problem
\begin{align}
\label{eqn:VI}
\min_{q_i(\bx_i, \bnu_i)\in\calQ}D_{\mathsf{KL}}(q_i(\bx_i, \bnu_i) \| \pi_{(\bx_i, \bnu_i)}(\bx_i, \bnu_i\mid\mathbb{A})),
\end{align}
where $\calQ$ is a class of multivariate Gaussian distributions on $\mathbb{R}^{d - 1}\times\mathbb{R}^m$. 
Directly taking $\calQ$ to be all $(d-1+m)$-dimensional Gaussian distributions, nevertheless, can still be computationally expensive when the number of layers $m$ is even moderately large. Meanwhile, an overly simplified class of Gaussian distributions (\emph{e.g.}, the class of all multivariate Gaussians with diagonal covariance matrices) may lead to unreliable uncertainty quantification. Therefore, it is crucial to identify a class of multivariate Gaussian distributions that remains computationally manageable while delivering valid uncertainty quantification. The proposition below sheds light on this goal. 
\begin{proposition}[Fisher Information Matrix]
\label{prop:Fisher_information_matrix}
Under oracle model \eqref{eqn:oracle_problem}, for a fixed vertex $i\in[n]$, the inverse Fisher information matrix of $(\bz_i^*, \btheta_i)$ has the following Cholesky factorization:
\begin{align*}
&\left[\frac{1}{mn}\expect\left\{
\begin{bmatrix}
\nabla_{\bz_i^*}\ell_{0in}(\bz_i^*, \btheta_i) \nabla_{\bz_i^*}\ell_{0in}(\bz_i^*, \btheta_i)\transpose & \nabla_{\bz_i^*}\ell_{0in}(\bz_i^*, \btheta_i)\nabla_{\btheta_i}\ell_{0in}(\bz_i^*, \btheta_i)\transpose\\
\nabla_{\btheta_i}\ell_{0in}(\bz_i^*, \btheta_i) \nabla_{\bz_i^*}\ell_{0in}(\bz_i^*, \btheta_i)\transpose & \nabla_{\btheta_i}\ell_{0in}(\bz_i^*, \btheta_i)\nabla_{\btheta_i}\ell_{0in}(\bz_i^*, \btheta_i)\transpose
\end{bmatrix}
\right\}\right]^{-1}\\
&\quad = \begin{bmatrix}
\bL_i & \\
-\bM_i\bL_i & \bD_i
\end{bmatrix}
\begin{bmatrix}
\bL_i & \\
-\bM_i\bL_i & \bD_i
\end{bmatrix}\transpose,
\end{align*}
where $\ell_{0in}(\bz_i^*, \btheta_i) = \sum_{t = 1}^m\sum_{j = 1}^n\log f(A_{ij}^{(t)}; \theta_i^{(t)}\bkappa(\bz_i^*)\transpose\bB_0^{(t)}\bz_{0j}\theta_{0j}^{(t)})$, $\bL_i\in\mathbb{R}^{(d - 1)\times(d - 1)}$ is a lower triangular matrix, $\bM_i\in\mathbb{R}^{m\times (d - 1)}$, and $\bD_i\in\mathbb{R}^{m\times m}$ is a diagonal matrix with positive diagonal elements. 
\end{proposition}
It is well known that, under suitable regularity conditions, the Bernstein-von Mises (BvM) theorem implies that the exact posterior distribution of the parameter of interest is asymptotically normal with covariance matrix proportional to the inverse Fisher information matrix. This result will be stated rigorously in Section \ref{sec:theoretical_properties}. The key implication of the BvM theorem, together with Proposition \ref{prop:Fisher_information_matrix}, is that one can take  $\calQ$ to be the following class of multivariate Gaussian distributions with structured covariance matrices:
\begin{align}
\calQ& = \Bigg\{\calN\bigg(
    \begin{bmatrix} \bx_i \\ \bnu_i \end{bmatrix} \bigg|
    \bmu_i,
    \frac{1}{mn}\bSigma_i\bigg):\bmu_i = \begin{bmatrix}\bmu_{i1}\\\bmu_{i2}\end{bmatrix}, \bSigma_i = \begin{bmatrix}
    \bL_i & \\
    -\bM_i\bL_i & \bD_i
    \end{bmatrix}
    \begin{bmatrix}
    \bL_i & \\
    -\bM_i\bL_i & \bD_i
    \end{bmatrix}\transpose,\nonumber\\
    &\quad\quad \bmu_{i1}\in\mathbb{R}^{d - 1},\bmu_{i2}\in\mathbb{R}^m,\bL_i\in\mathbb{R}^{(d - 1)\times (d - 1)}\text{ is lower triangular with positive diagonals}, \nonumber\\
    \label{eqn:Gaussian_VI_distribution_class}
    &\quad\quad \bM_i\in\mathbb{R}^{m\times(d - 1)}, \bD_i = \textsf{diag}(\sigma_{i1},\ldots,\sigma_{im}),\sigma_{i1},\ldots,\sigma_{im} > 0\Bigg\}.
\end{align}
The above Cholesky factorization of $\bSigma$ allows us to reparameterize any $q_i(\bx_i, \bnu_i)\in\calQ$ using
\[
\begin{bmatrix}
\bx_i\\\bnu_i
\end{bmatrix} = \begin{bmatrix}\bmu_{i1}\\\bmu_{i2}\end{bmatrix} + \frac{1}{\sqrt{mn}}\begin{bmatrix}\bL_i & \\ -\bM_i\bL_i & \textsf{diag}(\bsigma_i)\end{bmatrix}\beps_i,
\]
where $\beps_i = [\beps_{i1}\transpose, \beps_{i2}\transpose]\transpose$, $\beps_{i1}\sim\mathrm{N}(\zero_{d - 1}, \eye_{d - 1})$ and $\beps_{i2}\sim\mathrm{N}(\zero_m, \eye_m)$ are independent, and $\bsigma_i = [\sigma_{i1},\ldots,\sigma_{im}]\transpose$ with $\sigma_{it} > 0$, $t\in[m]$. 
Now denote by $\calL_{in}(\bx_i, \bnu_i) = \widetilde{\ell}_{in}(\calT_\bz(\bx_i), \calT_\btheta(\bnu_i))$. 
Then, a simple algebra (see, for example, \cite{wuxie2025,xu2022computational}) shows that the minimization problem \eqref{eqn:VI} is equivalent to
\begin{equation}
\label{eqn:Gaussian_VI}
\begin{aligned}
&\min_{\bmu_i,\bL_i,\bM_i,\bsigma_i}F_{in}(\bmu_i, \bL_i, \bM_i, \bsigma_i)\\
&\text{s.t.} \;\bL_i\text{ lower triangular with positive diagonals},\;
\sigma_{i1},\ldots,\sigma_{im} > 0,
\end{aligned}
\end{equation}
where the objective function is given by
\begin{equation}\label{eqn:Gaussian_VB_objective}
\begin{aligned}
&F_{in}(\bmu_i, \bL_i, \bM_i, \bsigma_i)\\
&\quad = -\log\det(\bL_i) - \sum_{t = 1}^m\log\sigma_{it}\\ 
&\qquad- \expect_{\beps_i}\bigg\{\calL_{in}\bigg(\bmu_{i1} + \frac{\bL_i\beps_{i1}}{\sqrt{mn}},\bmu_{i2} - \frac{\bM_i\bL_i\beps_{i1}}{\sqrt{mn}} + \frac{\textsf{diag}(\bsigma_i)\beps_{i2}}{\sqrt{mn}}\bigg)\bigg\}\\ 
&\qquad- \expect_{\beps_{i}}\bigg\{\log \pi_{\bx_i}\bigg(\bmu_{i1} + \frac{\bL_i\beps_{i1}}{\sqrt{mn}}\bigg)\bigg\}
 - \expect_{\beps_{i}}\bigg\{\log \pi_{\bnu_i}\bigg(\bmu_{i2} - \frac{\bM_i\bL_i\beps_{i1}}{\sqrt{mn}} + \frac{\textsf{diag}(\bsigma_i)\beps_{i2}}{\sqrt{mn}}\bigg)\bigg\}.
\end{aligned}
\end{equation}
The closed-form formula for the objective function $F_{in}$ above is not directly available. However, $F_{in}$ is the expected value of a noisy function whose gradient can be computed in a closed form, thereby enabling the use of stochastic gradient descent. We provide the detailed algorithm in the Supplementary Material. 

\subsection{Comparison With Prior Work}
\label{sub:comparison_with_prior_work}

We briefly compare our work with a selective set of closely related prior studies on the statistical analysis of multilayer networks with latent space structures. See Table \ref{tab:comparison} for a sketch of the comparison. 
\begin{table*}[t]
\caption{Comparison between our work and prior work on multilayer network inference. The column ``UQ'' represents theoretical guarantee for uncertainty quantification, layer-wise means the uncertainty is assessed with respect to a layer-specific parameter, and vertex-wise means that the uncertainty is assessed with respect to a vertex-specific parameter. 
}
\label{tab:comparison}
\centering
\begin{tabular}{c | c | c | c}
\hline
Method & Model & Degree correction & UQ \\
\hline
MASE \citep{ArroyoEtAlCOSIE} & COSIE & No & Layer-wise \\
MASE \citep{zheng2022limit} & COSIE & No & Vertex-wise \\
BCJSE \citep{xie2024biascorrected} & COSIE & No & Vertex-wise\\
OSE \citep{he2023semiparametric} &Longitudinal LSM & Yes & No\\
DC-MASE \citep{Agterberg28072025} & ML-DCSBM & Yes & No\\
SMF-VI \citep{zhao2024structured} & Dynamic LSM & Yes & Not justified\\
SMF-VI \citep{loyal2023eigenmodel} &Eigenmodel & Yes & Not justified\\
SMF-VI \citep{loyal2024fast} &Dynamic LSM & Yes & Not justified\\
MCMC \citep{loyal2025generalized} &Dynamic GRDPG & Yes & Not justified\\
Our work & ML-DCMM & Yes & Vertex-wise\\
\hline
\end{tabular}
\end{table*}

Uncertainty quantification in multilayer network analysis has previously been explored and theoretically justified in \cite{ArroyoEtAlCOSIE,xie2024biascorrected,zheng2022limit} under the common subspace independent edge (COSIE) model, where each network layer has its own low-rank edge probability matrix, but all layers share a common column subspace. The authors of \cite{ArroyoEtAlCOSIE} proposed a multiple adjacency spectral embedding (MASE) method, which is similar to Algorithm \ref{alg:aggregated_mixed_SCORE} in spirit, and established the asymptotic normality of the layer-wise score matrix estimator for uncertainty quantification. Vertex-wise uncertainty quantification for MASE was later explored by \cite{zheng2022limit}, in which the author derived a row-wise limit theorem for the MASE estimator of the common subspace. To address the challenging regime where each network layer is extremely sparse, the author of \cite{xie2024biascorrected} developed a bias-corrected joint spectral embedding (BCJSE) algorithm that aggregates signals across layers and established a row-wise central limit theorem for vertex-wise uncertainty quantification. In contrast to our work, this line of research is not likelihood-based and does not accommodate layer-wise degree heterogeneity. 

Multilayer network analysis methods that incorporate layer-wise degree heterogeneity have been explored by \cite{Agterberg28072025,he2023semiparametric}. In \cite{he2023semiparametric}, the authors modeled multilayer weighted networks using a Poisson longitudinal latent space model (LSM) with the exponential link function, where the degree correction parameters capture the time-varying latent structure, and proposed a one-step estimator (OSE) for efficient estimation of the time-invariant latent feature vectors. The authors of \cite{Agterberg28072025} studied community detection in the multilayer degree-corrected stochastic block model (ML-DCSBM) using a degree-corrected multiple adjacency spectral embedding (DC-MASE) method. Among existing approaches, ML-DCSBM in \cite{Agterberg28072025} is perhaps the most related to the ML-DCMM network model. However, the inferential goals differ. \cite{Agterberg28072025} focused on community detection, whereas our work centers on uncertainty quantification for membership profiles. Neither \cite{he2023semiparametric} nor \cite{Agterberg28072025} discussed uncertainty quantification, which is precisely one of the features of our framework.

From the perspective of Bayesian inference, the authors of \cite{zhao2024structured} and \cite{loyal2024fast} proposed structured mean-field variational inference (SMF-VI) for dynamic LSM networks, and the authors of \cite{loyal2023eigenmodel} further extended SMF-VI to a dynamic multilayer network model called eigenmodel. In \cite{loyal2025generalized}, the author introduced another dynamic network modeling approach based on random dot product graphs (RDPGs) \citep{10.1007/978-3-540-77004-6_11} and a Gaussian random walk prior and developed a generalized Bayesian inference procedure in which the likelihood function is replaced by the exponentiated negative least squares loss function. There, the authors further proposed a computationally efficient MCMC algorithm for practical computation. These Bayesian or generalized Bayesian approaches provide natural environments for uncertainty quantification via full (variational) posterior distributions. Nonetheless, the theoretical results established in these studies primarily focused on convergence rates, while BvM-type results, along with rigorous guarantees for the validity of posterior-based uncertainty quantification, remain unaddressed. By contrast, our work is supported by the BvM theorem (see Section \ref{sec:theoretical_properties}), thereby justifying the validity of Bayesian inference using the variational posterior distribution. 

\section{Theoretical Properties}
\label{sec:theoretical_properties}

This section establishes the main theoretical results of the proposed spectral-assisted Gaussian VI. To prepare, we first impose several assumptions that are standard in the literature. The first assumption is a requirement for the parameters of the ML-DCMM network model. 
\begin{assumption}[DCMM network model]
\label{assumption:identifiability}
The following conditions are satisfied:
\begin{enumerate}[(a)]
  \item The diagonal entries of $(\bB_0^{(t)})_{t = 1}^m$ are equal to $1$, $(\bB_0^{(t)})_{t = 1}^m$ are non-singular matrices, and for each $k\in [d]$, there exists some $i_k\in[n]$ such that $\bz_{i_k} = \be_k$, namely, each community has at least one pure node. 

  \item Let $\calN_k = \{i\in[n]:\bz_i = \be_k\}$ be the set of pure nodes in the $k$th community and $\calM = [n]\backslash(\bigcup_{k = 1}^{d}\calN_k)$ be the set of all mixed nodes. Then, $\min(\min_{k\in[d]}|\calN_k|,|\calM|) \asymp n$, $\min\calN_1 < \ldots < \min\calN_d$, and there exist an integer $L_0 \geq 1$, a partition of $\calM = \bigcup_{l = 1}^{L_0}\calM_l$, a set of probability vectors $\bgamma_1,\ldots,\bgamma_{L_0}$, and constants $c_1,c_2 > 0$, such that
  \begin{align*}
  \min_{1\leq j\leq l\leq L_0}\|\bgamma_j - \bgamma_l\|_2\geq c_1,\;
  \min_{j\in[L_0], k\in[d]}\|\bgamma_l - \be_k\|_2 \geq c_2,\;
  \max_{l\in[L_0]}\max_{i\in\calM_l}\|\bz_i - \bgamma_l\|\leq 1/(\log n).
  \end{align*}

  \item There exists constants $c_3, c_4 > 0$, such that $c_3\leq \theta_i^{(t)}\leq c_4$ for all $i\in[n]$, $t\in[m]$, and the elements of $\bP_0^{(t)}\bZ_0\transpose\bB_0^{(t)}\bZ_0$ are bounded between $c_3$ and $c_4$. 
\end{enumerate}
\end{assumption}
Assumption \ref{assumption:identifiability} is an adaptation of those in \cite{JIN2023} for the single-layer DCMM network model. Among them, item (a) ensures that each layer is identifiable, \emph{i.e.}, $\bZ_0$ can be recovered up to a permutation of columns, and item (b) corresponds to setting 2 in \cite{JIN2023}, which assumes that the mixed node membership profiles are coarsely partitioned into $L_0$ clusters centered around well-separated vectors $(\bgamma_l)_{l = 1}^{L_0}$. The requirement that $\min\calN_1 < \ldots < \min\calN_d$ in item (b) is a labeling convention: the first pure node in the $k$th community always appears before the first pure node in the $(k + 1)$th community for $k\in\{1,\ldots,d - 1\}$. This convention allows us to implement the following label alignment algorithm.


The second assumption below is related to the spectral properties of the population matrices of the ML-DCMM network model. 
\begin{assumption}[Eigenvalues]
\label{assumption:condition_number}
There exist constants $C_1, C_2, c > 0$, such that:
\begin{enumerate}[(a)]
   \item $C_1\leq \sigma_d(\bB_0^{(t)})\leq \sigma_1(\bB_0^{(t)})\leq C_2$ for all $t\in[m]$;

   \item $C_1n\leq \lambda_d(\bZ_0\transpose(\bTheta_0^{(t)})^2\bZ_0)\leq \lambda_1(\bZ_0\transpose(\bTheta_0^{(t)})^2\bZ_0)\leq C_2n$ for all $t\in[m]$;

   \item $|\lambda_2(\bP_0^{(t)})|\leq (1 - c)\lambda_1(\bP_0^{(t)})$ for all $t\in[m]$;

   \item For each $t\in[m]$, there exists a positive integer $d_+^{(t)}$ and a nonnegative integer $d_-^{(t)}$, such that $d_+^{(t)} + d_-^{(t)} = d$, $\lambda_k(\bP_0^{(t)}) = 0$ for all $t \in \{d_+^{(t)} + 1,\ldots, n - d_-^{(t)}\}$, and
    \[
    \lambda_1(\bP_0^{(t)})\geq\ldots\geq\lambda_{d^{(t)}_+}(\bP_0^{(t)}) > 0 > \lambda_{n - d_-^{(t)} + 1}(\bP_0^{(t)})\geq\ldots\geq \lambda_n(\bP_0^{(t)}),
    \] 
  and $|\lambda_k(\bP_0^{(t)})| \asymp n$ for any $k\in\{1,\ldots,d_+^{(t)}\}\cup\{n - d_-^{(t)} + 1,\ldots,n\}$. 
 \end{enumerate}
\end{assumption}
Assumption \ref{assumption:condition_number} requires that the eigenvalues of $(\bB_0^{(t)})_{t = 1}^m$, $(\bZ_0\transpose(\bTheta_0^{(t)})^2\bZ_0)_{t = 1}^m$, and $(\bP_0)_{t = 1}^m$ are well behaved and is standard for the single-layer DCMM network model. In particular, item (b) above corresponds to Assumption 2 in \cite{JIN2023}, and items (c) and (d) correspond to Assumption 3 there. 

In order to describe the last assumption on the likelihood function, we need to introduce additional notations. Denote by 
\begin{align*}
\bG_{0in}^{(t)}
& = \sum_{j = 1}^n\frac{(\theta_{0j}^{(t)})^2\bB_0^{(t)}\bz_{0j}\bz_{0j}\transpose\bB_0^{(t)}}{n\var(A_{ij}^{(t)})},\;
\bDelta_{in}^{(t)}
 = \bG_{0in}^{(t)} - \frac{\bG_{0in}^{(t)}\bz_{0i}\bz_{0i}\transpose\bB_0^{(t)}}{\bz_{0i}\transpose\bG_{0in}^{(t)}\bz_{0i}},\\
\bGamma_{in}& = \begin{bmatrix}\bJ\transpose & \\ & \eye_m\end{bmatrix}
\begin{bmatrix}
\frac{1}{m}\sum_{t = 1}^m\theta_{0i}^{2(t)}\bG_{0in}^{(t)} & \frac{1}{m}\bG_{0in}^{(1)}\bz_{0i}\theta_{0i}^{(1)} & \ldots & \frac{1}{m}\bG_{0in}^{(m)}\bz_{0i}\theta_{0i}^{(m)}\\ 
\frac{1}{m}\theta_{0i}^{(1)}\bz_{0i}\transpose\bG_{0in}^{(1)} & \frac{1}{m}\bz_{0i}\transpose\bG_{0in}^{(1)}\bz_{0i} & \ldots & 0\\
\vdots & \vdots & \ddots & \vdots \\
\frac{1}{m}\theta_{0i}^{(m)}\bz_{0i}\transpose\bG_{0in}^{(m)} & 0 & \ldots & \frac{1}{m}\bz_{0i}\transpose\bG_{0in}^{(m)}\bz_{0i}
\end{bmatrix}
 \begin{bmatrix}\bJ & \\ & \eye_m\end{bmatrix}. 
\end{align*}

\begin{assumption}[Likelihood and Noise Distribution]
\label{assumption:likelihood}
The following conditions hold:
\begin{enumerate}[(a)]
  \item $\calI = [a, b]\subset(0, +\infty)$ is an interval not depending on $n$ with $0 < a < b < \infty$, and $P_{0ij}^{(t)}\in\calI,\bz_i\transpose\bB_0^{(t)}\bz_j\in\calI$ for all $i,j\in[n]$, $\bz_i,\bz_j\in\Delta^{d - 1}$, and $t\in[m]$, where $P_{0ij}^{(t)}$ denote the $(i, j)$th element of $\bP_0^{(t)}$. 

  \item $\eta'(\theta) > 0$ for all $\theta\in\calI$ abd $\eta(\cdot)$ is three times continuously differentiable over $\calI$.

  \item $\lambda_k(\bJ\transpose\bDelta_{in}^{(t)}\bJ)\asymp 1$ for all $k\in[d]$, $t\in[m]$, and $i\in[n]$.

  \item There exists some constant $C > 0$, such that $\|E_{ij}^{(t)}\|_{\psi_1} := \sup_{p\geq 1}p^{-1}(\expect|E_{ij}^{(t)}|^p)^{1/p}\leq C$ for all $i,j\in[n],t\in[m]$, where $E_{ij}^{(t)}:=A_{ij}^{(t)} - P_{ij}^{(t)}$.   
\end{enumerate}
\end{assumption}
Assumption \ref{assumption:likelihood} imposes mild regularity conditions on the likelihood function and the noise distribution. Item (b) is on the differentiability of $\eta(\cdot)$, item (c) ensures that the $\bGamma_{in}^{(t)}$ is invertible, and item (d) requires that the noise distributions have sub-exponential tails. Note that sub-exponential noise is quite flexible and includes unweighted binary networks, Poisson networks, and Gamma networks as special cases. 

Our first main result is summarized in Theorem \ref{thm:MLE} below, which guarantees the existence of a local maximum spectral-assisted likelihood estimator and its asymptotic normality. 
\begin{theorem}\label{thm:MLE}
Suppose Assumptions \ref{assumption:identifiability}--\ref{assumption:likelihood} hold. Further assume that there exist some $\xi > 1$, such that $m \ll n/(\log n)^{6\xi}$.
Let $i\in[n]$ be a vertex such that $\min_{k\in[d]}z_{0ik} \geq c$ for some constant $c > 0$, and denote by $\bz_{0i}^* = \bkappa^{-1}(\bz_{0i})$. Then, w.h.p.,
there exists some $\widehat{\bz}_i^*$ inside the interior of $\calS^{d - 1}$, $\widehat{\btheta}_i = [\widehat{\theta}_i^{(1)},\ldots,\widehat{\theta}_i^{(m)}]\transpose\in\calI^m$ satisfying the following likelihood equation 
    \[
        \frac{\partial\widetilde{\ell}_{in}}{\partial\bz_i^*}(\bz_i^*, \btheta_i)\mathrel{\Bigg|_{\bz_i^* = \widehat{\bz}_i^*, \btheta_i = \widehat{\btheta}_i}} = \zero_{d - 1},\quad
        \frac{\partial\widetilde{\ell}_{in}}{\partial\btheta_i}(\bz_i^*, \btheta_i)\mathrel{\Bigg|_{\bz_i^* = \widehat{\bz}_i^*, \btheta_i = \widehat{\btheta}_i}} = \zero_m.
    \]
    Furthermore, we have the following asymptotic normality of $\widehat{\bz}_i^*$:
    \begin{align*}
    \bigg(\frac{1}{m}\sum_{t = 1}^m\theta_{0i}^{2(t)}\bJ\transpose\bDelta_{in}^{(t)}\bJ\bigg)^{1/2}\sqrt{mn}(\widehat{\bz}_i^* - \bz_{0i}^*) \overset{\calL}{\to}\mathrm{N}(\zero_{d - 1}, \eye_{d - 1})\quad\text{as }n\to\infty.
    \end{align*}
\end{theorem}

Although Theorem \ref{thm:MLE} establishes the existence of a solution to the likelihood equation w.h.p. and its asymptotic normality, it does not ensure the uniqueness, and hence, it is not guaranteed that any solution to the likelihood equation or any local maximizer of the spectral-assisted log-likelihood function satisfies the desired asymptotic normality. The next result
shows that the Bayes estimator not only has the desired asymptotic normality but also can be computed or approximated in a natural way, such as via MCMC or VI. To state this result, we need an additional assumption regarding the prior densities. 

\begin{assumption}\label{assumption:prior}
The prior densities $\pi_{\bz^*}(\cdot)$ and $\pi_\theta(\cdot)$ are supported over $\calS^{d - 1}$ and $\calI$, respectively, with $\pi_{\bz^*}(\bz_{0i}^*) \geq c$ and $\pi_\theta(\theta_{0i}^{(t)}) > c$ for all $i\in[n]$ and $t\in[m]$ for some constant $c > 0$, where $\bz_{0i}^* = \bkappa^{-1}(\bz_{0i})$
\end{assumption}

\begin{theorem}[BvM theorem]
\label{thm:BvM}
Suppose the conditions of Theorem \ref{thm:MLE} and Assumption \ref{assumption:prior} hold. 
Denote by $\bt_i = (\bt_{\bz_i}, \bt_{\btheta_i}) = (\sqrt{mn}(\bz_i^* - \widehat{\bz}_i^*), \sqrt{mn}(\btheta_i - \widehat{\btheta}_i))$ and $\pi_{(\bt_{\bz_i}, \bt_{\btheta_i})}(\bt_{\bz_i}, \bt_{\btheta_i}\mid\mathbb{A})$ be the density function of $\bt_i$ induced from $\pi_{(\bz_i^*, \btheta_i)}(\bz_i^*, \btheta_i\mid\mathbb{A})$. If vertex $i\in[n]$ satisfies $\min_{k\in[d]}z_{0ik} \geq c$ for some constant $c > 0$, then, for any $\alpha > 0$, 
\begin{align*}
\iint_{\mathbb{R}^{d + m - 1}}(1 + \|\bt_{\bz_i}\|_2^\alpha)\bigg|
\pi_{(\bt_{\bz_i}, \bt_{\btheta_i})}(\bt_{\bz_i}, \bt_{\btheta_i}\mid\mathbb{A}) - \frac{e^{-\bt_i\transpose\bGamma_{in}\bt_i/2}}{\det(2\pi\bGamma_{in}^{-1})^{1/2}}
\bigg|\mathrm{d}\bt_{\bz_i}\mathrm{d}\bt_{\btheta_i} = \Optilde\bigg\{\frac{\sqrt{m}(\log n)^{3\xi}}{\sqrt{n}}\bigg\}.
\end{align*}
In particular, if we denote by $\widehat{\bz}_i^{*(\mathrm{BE})} = \iint_{\mathbb{R}^{d + m - 1}}\bz_i^*\pi_{(\bz_i^*, \btheta_i)}(\bz_i^*, \btheta_i\mid\mathbf{A})\mathrm{d}\bz_i^*\mathrm{d}\btheta_i$, then 
\begin{align*}
  \bigg(\frac{1}{m}\sum_{t = 1}^m\theta_{0i}^{2(t)}\bJ\transpose\bDelta_{in}^{(t)}\bJ\bigg)^{1/2}\sqrt{mn}(\widehat{\bz}_i^{*(\mathrm{BE})} - \bz_{0i}^*) \overset{\calL}{\to}\mathrm{N}(\zero_{d - 1}, \eye_{d - 1})\quad\text{as }n\to\infty.
\end{align*}
\end{theorem}

Our last result is the BvM theorem for the variational posterior distribution. We need an additional assumption on the differentiability of the prior densities \eqref{eqn:prior_change_of_variable}.
\begin{assumption}\label{assumption:prior_transform}
There exists a constant $C > 0$, such that:
\begin{enumerate}[(a)]
  \item $\|\partial\log\pi_\bx(\bx_{0i})/\partial\bx_i\|_2 + |\mathrm{d}\log\pi_\nu(\nu_{0it})/\mathrm{d}\nu_{it}|\leq C$ for all $i\in[n]$, $t\in[m]$. 

  \item $\log\pi_\bx(\cdot)$ and $\log\pi_\nu(\cdot)$ are concave and 
  \[
    \sup_{\bx\in\mathbb{R}^{d - 1}}\bigg\|\frac{\partial^2}{\partial\bx\partial\bx\transpose}\log\pi_\bx(\bx)\bigg\|_2 + \sup_{\nu\in\mathbb{R}}\bigg|\frac{\mathrm{d}^2}{\mathrm{d}\nu^2}\log\pi_\nu(\nu)\bigg|\leq C.
  \]
\end{enumerate}
\end{assumption}

\begin{theorem}[Variational BvM theorem]
\label{thm:variational_BvM}
Suppose the conditions of Theorem \ref{thm:BvM} and Assumption \ref{assumption:prior_transform} hold. Denote by $\widehat{\bx}_i = \calT_\bz^{-1}(\widehat{\bz}_i^*)$ and $\widehat{\bnu}_i = \calT_\btheta^{-1}(\widehat{\btheta}_i)$. Let $(\bmu_i^\star,{\bL}_i^\star,\bM_i^*,\bsigma_i^*)$ solve \eqref{eqn:Gaussian_VI}, and denote by
\begin{align*}
\bGamma_{in}^* &= \begin{bmatrix}
\frac{\partial\calT_\bz}{\partial\bx\transpose}(\bx_{0i}) & \\ 
& \frac{\partial\calT_\btheta}{\partial\bnu\transpose}(\bnu_{0i})
\end{bmatrix}\transpose\bGamma_{in}
\begin{bmatrix}
\frac{\partial\calT_\bz}{\partial\bx\transpose}(\bx_{0i}) & \\ 
& \frac{\partial\calT_\btheta}{\partial\bnu\transpose}(\bnu_{0i})
\end{bmatrix},\\
\bSigma_{in}^\star &= \begin{bmatrix}
\bL_i^\star & \\
-\bM_i\bL_i^\star & \textsf{diag}(\sigma_{i1}^{\star},\ldots,\sigma_{im}^\star)
\end{bmatrix}\begin{bmatrix}
\bL_i^\star & \\
-\bM_i\bL_i^\star & \textsf{diag}(\sigma_{i1}^{\star},\ldots,\sigma_{im}^\star)
\end{bmatrix}\transpose,
\end{align*}
where $\bx_{0i} = \calT_\bz^{-1}(\bz_{0i}^*)$ and $\bnu_{0i} = \calT_\btheta^{-1}(\btheta_{0i})$. 
If vertex $i\in[n]$ satisfies $\min_{k\in[d]}z_{0ik} \geq c$ for some constant $c > 0$, then, for any $\alpha > 0$, 
\begin{align*}
d_{\mathrm{TV}}\bigg(
\calN\bigg(\begin{bmatrix}\bx_i\\\bnu_i\end{bmatrix}\mathrel{\bigg|}\begin{bmatrix}\widehat{\bx}_i\\\widehat{\bnu}_i\end{bmatrix}, \frac{1}{{mn}}(\bGamma_{in}^*)^{-1}\bigg),
  \calN\bigg(\begin{bmatrix}\bx_i\\\bnu_i\end{bmatrix}\mathrel{\bigg|}
  \bmu_i^\star, \bSigma_{in}^\star\bigg)\bigg) = \Optilde\bigg\{\frac{m^{1/4}(\log n)^{3\xi/2}}{n^{1/4}}\bigg\}.
\end{align*}
Furthermore, if $\widehat{\bx}_i^{(\mathrm{VI})}$ is the variational posterior mean of $\bx_i$, namely, $\bmu_{i1}^\star$, where $\bmu_i^\star = (\bmu_{i1}^\star,\bmu_{i2}^\star)$ and $\bmu_{i1}^\star\in\mathbb{R}^{d - 1}$, $\bmu_{i2}^\star\in\mathbb{R}^m$, then 
\[
\bigg(\frac{1}{m}\sum_{t = 1}^m\theta_{0i}^{2(t)}\bJ\transpose\bDelta_{in}^{(t)}\bJ\bigg)^{-1/2}\sqrt{mn}(\calT_\bz(\widehat{\bx}_i^{(\mathrm{VI})}) - \bz_{0i}^*)\overset{\calL}{\to}\mathrm{N}(\zero_{d - 1}, \eye_{d - 1}).
\]
\end{theorem}

\begin{remark}[Necessity of Joint VI]
Although the inferential parameter of interest is the membership profile $\bz_i$, the VI framework posits a joint Gaussian distribution over $\bz_i$ and the nuisance degree correction parameter $\btheta_i$ with a structured covariance matrix specified by \eqref{eqn:Gaussian_VI_distribution_class}. In contrast to the (structured) mean-field VI methods \citep{loyal2024fast,loyal2023eigenmodel,zhao2024structured} where the latent feature vectors and the degree correction parameters are modeled independently under the class of candidate distributions $\calQ$, the proposed Gaussian VI accounts for the correlation between $\bz_i$ and $\btheta_i$. This joint VI approach is necessary to ensure valid uncertainty quantification. Indeed, one can show \citep{han2019statistical,doi:10.1080/01621459.2018.1473776} that for a structured mean-field VI, under mild regularity conditions, 
\[
  d_{\mathrm{TV}}\bigg(q_i(\bx_i, \bnu_i), \mathrm{N}\bigg(\begin{bmatrix}\widehat{\bx}_i\\\widehat{\bnu}_i\end{bmatrix}, \frac{1}{{mn}}(\textsf{block-diag}(\bGamma_{in}^*))^{-1}\bigg)\bigg) \overset{\prob_0}{\to}0\quad\text{as }n\to\infty, 
\]
where
\begin{align*}
  \textsf{block-diag}(\bGamma_{in}^*)
& = \begin{bmatrix}
\frac{\partial\calT_\bz}{\partial\bx\transpose}(\bx_{0i}) & \\ 
& \frac{\partial\calT_\btheta}{\partial\bnu\transpose}(\bnu_{0i})
\end{bmatrix}\transpose
\begin{bmatrix}
\sum_{t = 1}^m\frac{\theta_{0i}^{2(t)}\bG_{0in}^{(t)}}{m} &  & & \\
& \frac{\bz_{0i}\transpose\bG_{0in}^{(1)}\bz_{0i}}{m} & &\\
& & \ddots & \\
& & & \frac{\bz_{0i}\transpose\bG_{0in}^{(m)}\bz_{0i}}{m}
\end{bmatrix}
\\
&\quad\times \begin{bmatrix}
\frac{\partial\calT_\bz}{\partial\bx\transpose}(\bx_{0i}) & \\ 
& \frac{\partial\calT_\btheta}{\partial\bnu\transpose}(\bnu_{0i})
\end{bmatrix}.
\end{align*}
Such a discrepancy between the BvM limit covariance matrix and the asymptotic covariance matrix of $(\widehat{\bx}_i, \widehat{\bnu}_i)$ can lead to miscalibrated uncertainty quantification. 
In contrast, the structured covariance matrix in \eqref{eqn:Gaussian_VI_distribution_class} allows the proposed Gaussian VI to produce credible sets with the desired coverage probability, while still being more computationally efficient than a Gaussian VI with a fully unstructured covariance matrix. 
\end{remark}

\section{U.S. airport transportation network data analysis}
\label{sec:real_world_data_analysis}

We illustrate the proposed methodology using a longitudinal network dataset of U.S. domestic air transportation from January 2016 to September 2021, which is publicly available from the U.S. Bureau of Transportation Statistics \citep{jacobini2020bureau}. Each vertex represents a commercial airport in the 48 contiguous states, and each layer corresponds to one calendar month. The edge weight $A_{ij}^{(t)}$ records the number of flights from airport $i$ to airport $j$ in month $t$. 

Following \cite{Agterberg28072025}, we restrict attention to the intersection of the largest connected components across all months, which yields $n = 343$ airports observed in every layer. After preprocessing, the dataset consists of $69$ vertex-aligned networks. To study the impact of the COVID-19 pandemic, we split the data into two regimes and analyze them separately:
\begin{itemize}[nosep]
    \item \emph{Pre-COVID period:} January 2016 to December 2019 ($m = 48$ layers).
    \item \emph{During-COVID period:} January 2020 to September 2021 ($m = 21$ layers).
\end{itemize}
This split allows the multilayer model to capture relatively stable structural features within each regime while accommodating changes in community structure and traffic intensity across regimes.


\textbf{Model specification and estimation.} Flight counts in this dataset are visibly overdispersed relative to a Poisson model. To accommodate this, we adopt a negative binomial edge model within the ML-DCMM network framework.  Conditional on membership profiles $\bz_i$, degree correction parameters $\theta_i^{(t)}$, and block mean matrices $\bB^{(t)}$, we use
$A_{ij}^{(t)} \sim \mathrm{NB}(\rho^{(t)}, \rho^{(t)} / (\mu_{ij}^{(t)} + \rho^{(t)}))$, $i,j\in[n]$, $t\in[m]$,  as the working model, where the mean is $\mu_{ij}^{(t)} = \theta_i^{(t)}\bz_i\transpose\bB^{(t)}\bz_j\theta_j^{(t)}$ and $\rho^{(t)} > 0$ is a layer-specific dispersion parameter. The probability mass function is
\begin{align*}
f\!\left(A_{ij}^{(t)};\mu_{ij}^{(t)},\rho^{(t)}\right)
&=
\frac{\Gamma(A_{ij}^{(t)} + \rho^{(t)})}{\Gamma(\rho^{(t)})\Gamma(A_{ij}^{(t)} + 1)}
\frac{(\rho^{(t)})^{\rho^{(t)}} \left(\mu_{ij}^{(t)}\right)^{A_{ij}^{(t)}}}
     {\left(\mu_{ij}^{(t)} + \rho^{(t)}\right)^{A_{ij}^{(t)} + \rho^{(t)}}}.
\end{align*}

The dispersion parameters $(\rho^{(t)})_{t=1}^m$ are treated as layer-specific nuisance parameters. For each regime (pre-COVID and during-COVID), we first compute preliminary estimators $\widetilde{\bZ}$, $(\widetilde{\btheta}_i)_{i=1}^n$, and $(\widetilde{\bB}^{(t)})_{t=1}^m$ using Algorithm~\ref{alg:aggregated_mixed_SCORE} with $d = 4$, guided by \citet{Agterberg28072025}. We then estimate the dispersion parameters by maximizing a profile likelihood in which the mean parameters are frozen at these preliminary values:
\[
\widehat{\rho}^{(t)}
=
\argmax_{\rho^{(t)} > 0}
\sum_{1\leq i\leq j\leq n}
\log f\!\left(A_{ij}^{(t)};\,
\widetilde{\theta}_i^{(t)}\widetilde{\bz}_i\transpose\widetilde{\bB}^{(t)}\widetilde{\bz}_j\widetilde{\theta}_j^{(t)},\,\rho^{(t)}\right),
\quad t\in[m].
\]
With $(\widehat{\rho}^{(t)})_{t=1}^m$ in hand, we form the spectral-assisted likelihood \eqref{eqn:spectral_assisted_loglik} and solves the Gaussian VI optimization problem \eqref{eqn:Gaussian_VI} separately for the pre-COVID and during-COVID networks. For airport $i$, the variational posterior for the transformed parameters $(\bx_i,\bnu_i)$ is Gaussian with mean $\bmu_i = (\bmu_{i1},\bmu_{i2})$ and structured covariance as in \eqref{eqn:Gaussian_VI_distribution_class}. We obtain point estimates of the membership profiles via
$\widehat{\bz}_i = \calT_\bz(\bmu_{i1})$, $i\in[n]$, and credible sets by back-transforming the corresponding Gaussian variational posterior.


\textbf{Estimated community structure.} To visualize the inferred network structure, we hard-cluster airports according to their dominant membership: $\widehat{g}_i = \argmax_{k\in[d]}\widehat{z}_{ik}$, $i\in[n]$, and plot the resulting community labels on a map of the contiguous United States. Figure~\ref{fig:USAirport_community_assignment} displays the estimated communities in the pre-COVID (left) and during-COVID (right) regimes, with community labels for the during-COVID period aligned to pre-COVID communities for ease of comparison. Before the onset of COVID-19, the estimated communities exhibit a clear east versus west segmentation. One community is concentrated on major western hubs, another on the eastern seaboard, and the remaining communities primarily capture central and cross-country connectivity. This pattern is consistent with the geography of U.S. airline networks and with the findings of \cite{Agterberg28072025}. During the COVID-19 period, the community pattern has changed remarkably. Several western airports that formed a distinct group pre-COVID have become more tightly integrated with eastern and central hubs, resulting in a large mixed community spanning most of the country. These shifts in community assignments suggest that pandemic-era schedule adjustments and demand shocks altered the effective structural roles of many airports in the national network, even after accounting for degree heterogeneity through the layer-specific parameters $\theta_i^{(t)}$.
\begin{figure}[htbp]
  \centering
  \includegraphics[width=\textwidth]{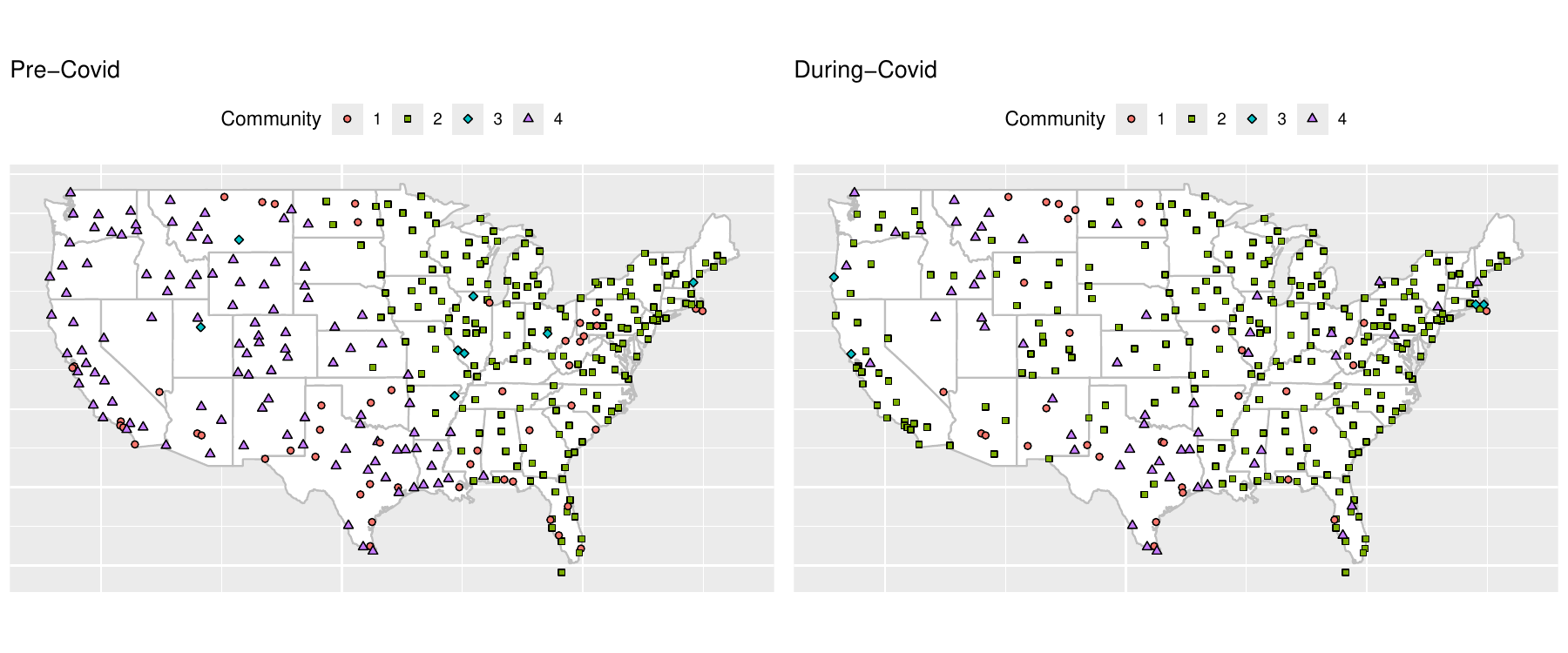}
  \caption{Community detection for U.S. airports in Section~\ref{sec:real_world_data_analysis}: Pre-COVID (left) versus during-COVID (right). Labels for during-COVID are aligned to the pre-COVID anchor for comparison.}
  \label{fig:USAirport_community_assignment}
\end{figure}


\textbf{Membership uncertainty for major hubs.} A distinctive advantage of the proposed approach is that it delivers airport-specific uncertainty quantification for the mixed memberships. Figure~\ref{fig:USAirport_membership_large_hubs} compares the $90\%$ variational posterior credible ellipses for several large hubs (ATL, CLT, DEN, DFW, JFK, LAS, LAX, ORD) between the pre-COVID and the during-COVID networks, plotted in the $(d-1)$-dimensional simplex, together with those for the remaining airports.

In the pre-COVID period (Figure~\ref{fig:USAirport_membership_large_hubs}, left panel), the ellipses for most large hubs are concentrated near the vertices of the simplex. This indicates that these hubs behave as nearly pure members of specific communities, with relatively small uncertainty in their roles. For instance, major western hubs such as LAX and LAS are strongly associated with one community, and ATL and CLT are tightly associated with an eastern community.

In contrast, during the COVID-19 period (Figure~\ref{fig:USAirport_membership_large_hubs}, right panel), the inferred membership profiles of several large hubs shift toward the interior of the simplex and their credible ellipses enlarge. Hubs such as DEN, ATL, and DFW exhibit posterior mass that spreads across multiple communities, indicating that they play more mixed roles in the network. This shift is consistent with pandemic-era route reconfiguration and temporary capacity reductions, which change the relative prominence of long-haul versus regional connections and alter how these hubs mediate traffic between communities.
\begin{figure}[t]
  \centering
  \includegraphics[width=1\textwidth]{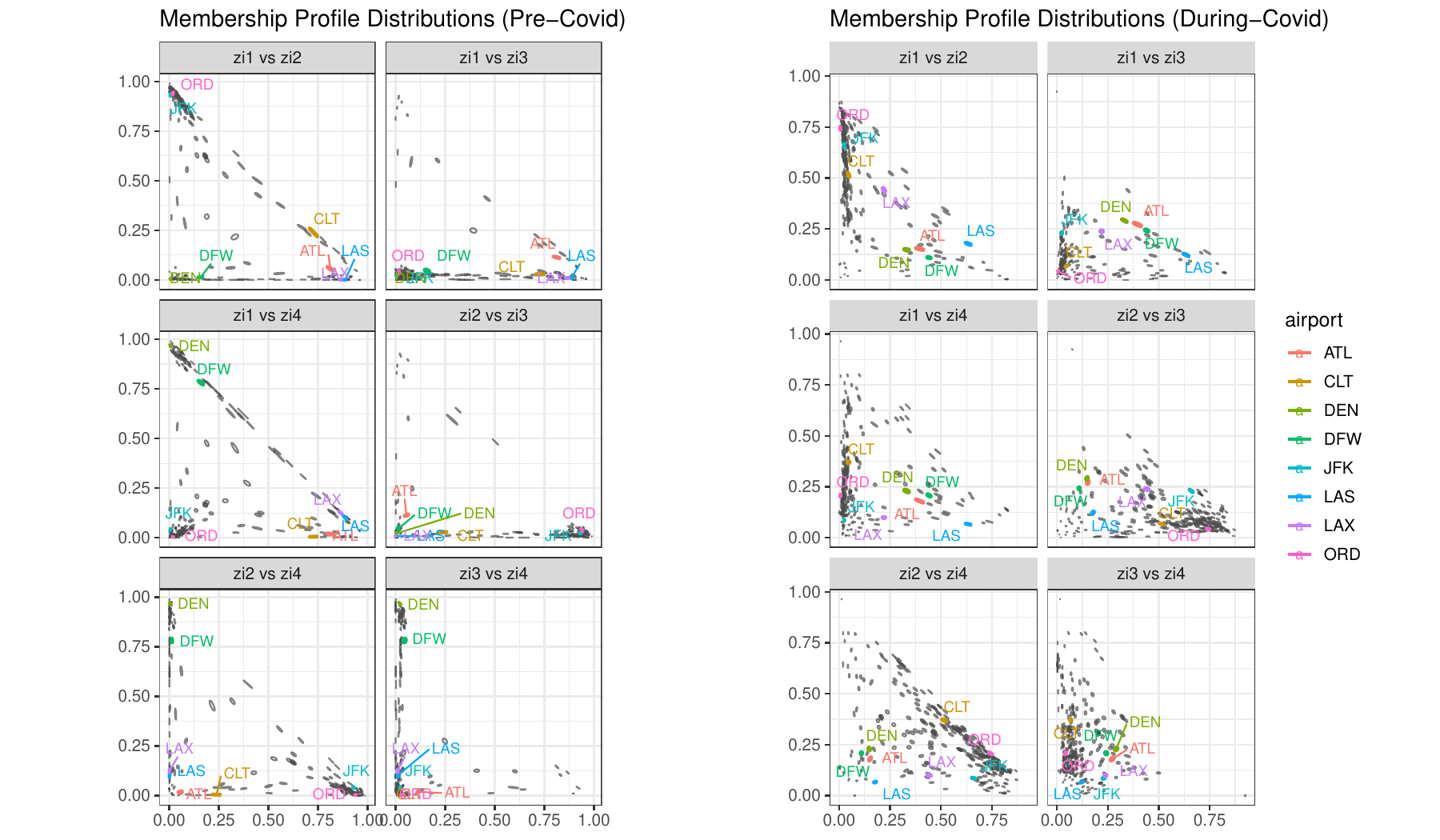}
  \caption{Variational posterior distributions of the membership profiles of large hubs (ATL, CLT, DEN, DFW, JFK, LAS, LAX, ORD) versus other airports for the pre-COVID versus during-COVID networks. Ellipses represent $90\%$ credible regions in the simplex.}
  \label{fig:USAirport_membership_large_hubs}
\end{figure}


\textbf{Degree correction and temporal patterns.} The degree correction parameters $(\theta_i^{(t)})_{i\in[n],t\in[m]}$ summarize airport-specific activity levels after controlling for the block structure. To make these parameters interpretable across layers, we normalize $(\theta_i^{(t)})_{i\in[n],t\in[m]}$ and $(\bB^{(t)})_{t = 1}^m$ such that $\max_{k,l\in[d]} B_{kl}^{(t)} = 1$ for each $t$ (Note that $\theta_i^{(t)}$ and $\bB^{(t)}$ are generic values under the variational posterior distribution). Under this normalization, $\theta_i^{(t)}$ can be interpreted as a relative multiplicative factor for the overall traffic intensity of airport $i$ in month $t$.

Figure~\ref{fig:USAirport_degree_correction_large_hubs} plots the posterior means of $\theta_i^{(t)}$ over time, together with $95\%$ time-wise variational credible intervals, for the large hubs and for an aggregate of the remaining airports. Two features are apparent. First, all groups experience a sharp decline in degree correction parameters around the start of 2020, consistent with the abrupt collapse in U.S. air traffic following the onset of the COVID-19 pandemic. The drop is pronounced and falls well outside the pre-COVID credible bands, providing clear evidence that the observed changes cannot be explained by typical month-to-month variability. Second, the magnitude and pace of recovery differ across hubs. Some hubs exhibit relatively rapid rebounds in $\theta_i^{(t)}$, while others remain substantially below their pre-COVID levels throughout 2020 and 2021. The proposed model quantifies these differences while jointly accounting for global community structure and edge overdispersion.


Taken together, the U.S. airlines analysis shows that the spectral-assisted Gaussian VI method recovers an interpretable multilayer community structure, reveals meaningful shifts in airport roles across regimes, and delivers airport-specific uncertainty quantification, all at a computational cost that is practical for networks of this size.
\begin{figure}[htbp]
  \centering
  \includegraphics[width=.7\textwidth]{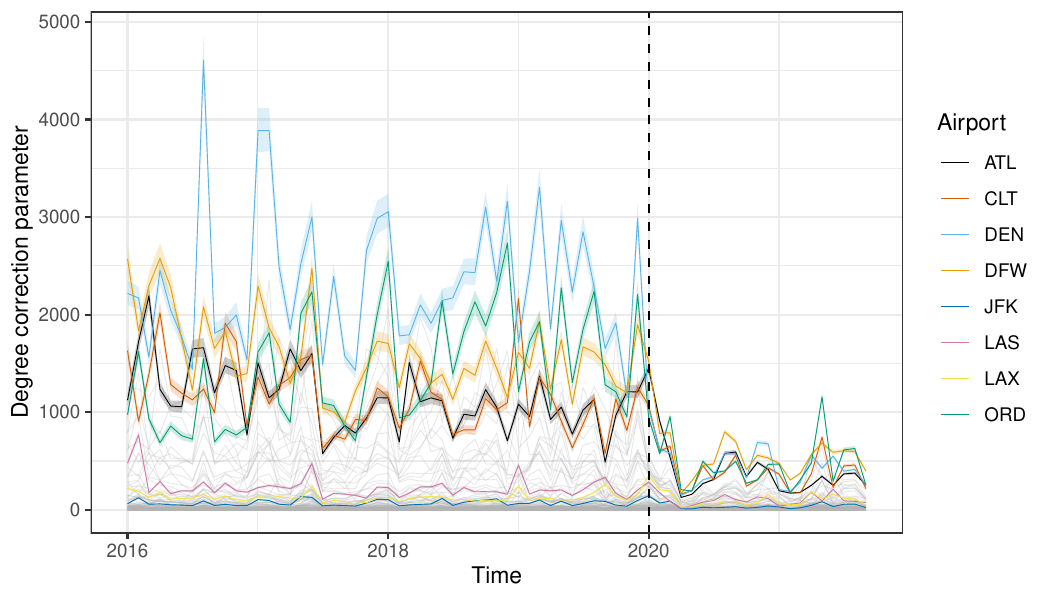}
  \caption{Degree correction parameters for large hubs (ATL, CLT, DEN, DFW, JFK, LAS, LAX, ORD) and for other airports versus time. Shaded regions indicate $95\%$ variational posterior credible intervals for $(\theta_i^{(t)}: i\in[n], t\in[m])$.}
  \label{fig:USAirport_degree_correction_large_hubs}
\end{figure}

\section{Conclusion}
\label{sec:conclusion}

We propose a Gaussian variational inference framework for uncertainty quantification in the multilayer degree-corrected mixed-membership (ML-DCMM) network model that avoids the computational burden of full MCMC. The key ingredient is a spectral-assisted likelihood that combines layer-wise spectral estimators with exponential family likelihood contributions. This construction leads to vertex-wise inference problems that are separable across vertices yet retain the essential dependence between mixed memberships and degree correction parameters. Leveraging this likelihood, we propose a computationally efficient Gaussian VI with a structured covariance matrix derived from a block Cholesky factorization of the Fisher information matrix. 
From a theoretical standpoint, we establish the Bernstein--von Mises theorems for both the exact posterior and the variational posterior. The variational Bernstein--von Mises result shows that the proposed Gaussian variational inference achieves asymptotically correct uncertainty quantification for mixed membership profiles, provided that the covariance structure captures the dependence between membership and degree parameters.

Our analysis clarifies a trade-off between computational efficiency and the fidelity of uncertainty quantification. A fully unstructured Gaussian variational family would, in principle, match the exact posterior covariance, but is computationally more expensive in moderately large multilayer settings. In contrast, block-diagonal or structured mean-field variational families ignore important correlations and typically yield credible sets with poor frequentist coverage. The structured covariance class specified by \eqref{eqn:Gaussian_VI_distribution_class} strikes a balance between these extremes: it preserves the relevant dependence structure needed for valid uncertainty quantification while remaining tractable to optimize with large network datasets. Our simulation studies confirm that this choice delivers both accurate point estimation and empirical coverage close to the nominal level, whereas simpler covariance structures, such as block-diagonal covariance (corresponding to a structured mean-field family), do not.

Several extensions remain open. First, our asymptotic guarantees are derived for vertices whose membership proportions are strictly positive in all communities. In applications where interest centers on pure or nearly pure nodes, it would be natural to consider spike-and-slab priors and variational distributions for the membership profiles. Establishing Bernstein--von Mises type results in such sparsified settings would require new techniques that accommodate boundary parameters and model selection effects. Second, the current methodology follows an ``estimate-then-aggregate'' strategy, in which layer-specific spectral information is first distilled and then combined. In regimes where individual layers are extremely sparse, an ``aggregate-then-estimate'' strategy in the spirit of \cite{Lei02102023,xie2024biascorrected} may extract more signal, but adapting such ideas to multilayer models with degree heterogeneity is nontrivial.

Beyond these directions, it would be interesting to integrate covariates at the vertex or layer level, to extend the framework to directed or bipartite multilayer networks, and to explore fully dynamic versions of the ML-DCMM network model in which membership profiles evolve over time. We expect the combination of spectral-assisted likelihoods with carefully structured variational families to remain a useful paradigm for these extensions, providing a principled and scalable route to uncertainty quantification in increasingly complex network models.

\begin{appendices}
\counterwithin{theorem}{section}
\counterwithin{result}{section}
\counterwithin{proposition}{section}
\counterwithin{corollary}{section}
\counterwithin{remark}{section}

\appendix

\allowdisplaybreaks

\begin{center}
\LARGE\bf Supplementary Material
\end{center}

\begin{abstract}
This supplementary material provides the details of the related algorithms, simulation studies, and the proofs of the results in Section \ref{sec:theoretical_properties} of the manuscript. Section \ref{sec:additional_notations} introduces additional notations. Section \ref{sec:detailed_algorithms} provides additional details of the related algorithms. Section \ref{sec:simulation_studies} reports empirical results via simulation studies. Section \ref{sec:auxiliary_results} provides several auxiliary results that are used throughout the proofs. Section \ref{sec:proof_of_theorem_mle} establishes the first-order stochastic expansions of the preliminary estimators. Sections \ref{sec:proof_of_theorem_mle}, \ref{sec:proof_of_theorem_ref_bvm}, and \ref{sec:proof_of_variational_BvM_theorem} elaborate the proofs of Theorems \ref{thm:MLE}, \ref{thm:BvM}, and \ref{thm:variational_BvM}, respectively. 
\end{abstract}

\section{Additional Notations}
\label{sec:additional_notations}

We first provide additional mathematical notations that are used throughout the proofs below. For any positive integers $n\geq d$, we denote by $\mathbb{O}(n, d)$ the collection of all $n\times d$ Stiefel matrices $\mathbb{O}(n, d) = \{\bU\in\mathbb{R}^{n\times d}:\bU\transpose\bU = \eye_d\}$ and write simply $\mathbb{O}(d) = \mathbb{O}(d, d)$. The matrix sign $\mathsf{sgn}(\bH)$ of a square matrix $\bH\in\mathbb{R}^{d\times d}$ is defined as follows: Let $\bH$ yield singular value decomposition $\bW_1\textsf{diag}(\sigma_1(\bH),\ldots,\sigma_d(\bH))\bW_2\transpose$, where $\bW_1,\bW_2\in\mathbb{O}(d)$. Then, $\textsf{sgn}(\bH) = \bW_1\bW_2\transpose$. For a matrix $\bU = [u_{ik}]_{n\times d}\in\mathbb{R}^{n\times d}$, we denote by $\|\bU\|_{\mathrm{F}} = (\sum_{i = 1}^n\sum_{k = 1}^du_{ik}^2)^{1/2}$ the Frobenius norm of $\bU$, $\|\bU\|_2 = \sigma_1(\bU)$ the spectral norm of $\bU$, $\|\bU\|_\infty = \max_{i\in[n]}\sum_{k = 1}^d|u_{ik}|$ the infinity norm of $\bU$, and $\|\bU\|_{2\to\infty} = \max_{i\in[n]}(\sum_{k = 1}^du_{ik}^2)^{1/2}$ the two-to-infinity norm of $\bU$. 
For a vector $\bu = [u_1,\ldots,u_n]\transpose\in\mathbb{R}^n$, we denote by $\|\bu\|_\infty = \max_{i\in[n]}|u_i|$ the infinity norm of $\bu$. We use $c,c_1,c_2,C,C_1,C_2,\ldots$ to denote generic constants that may vary from line to line but are fixed and independent of $n$ throughout. For a set $\calN$ with finitely many elements, we use $|\calN|$ to record its cardinality.

Next, we introduce several new notations that are related for the ML-DCMM network model, our proposed method, and the proof techniques employed in this supplementary material. Let $\bP_0^{(t)}$ denote the true value of $\bP^{(t)}$ and $\bU_0^{(t)}\bS_0^{(t)}(\bU_0^{(t)})\transpose$ denote the spectral decomposition of $\bP_0^{(t)}$, where $\bU^{(t)}_0\in\mathbb{O}(n, d)$ is the matrix of orthonormal eigenvectors and $\bS_0^{(t)} = \textsf{diag}(\lambda_1(\bP_0^{(t)}),\ldots,\lambda_d(\bP_0^{(t)}))$. For simplicity, we also denote by $\lambda_k^{(t)} = \lambda_k(\bP_0^{(t)})$. Similarly, we denote by $\bZ_0,\bTheta^{(t)}_0,\bB^{(t)}_0$ the true values of $\bZ,\bTheta^{(t)},\bB^{(t)}$, respectively. Define $\bXi^{(t)}_0 = (\bZ_0\transpose\bTheta_0^{2(t)}\bZ_0)^{-1}\bU_0^{(t)}$ and let $\xi_{0lk}^{(t)}$ be the $(l, k)$th element of $\bXi_0^{(t)}$. Denote by $\bv_{0k}^{(t)} = [\xi_{02k}^{(t)}/\xi_{01k}^{(t)},\ldots,\xi_{0dk}^{(t)}/\xi_{01k}^{(t)}]\transpose\in\mathbb{R}^{d - 1}$ for each $k\in[d]$ and let $\bV^{(t)} = [\bv_{01}^{(t)},\ldots,\bv_{0d}^{(t)}]\in\mathbb{R}^{(d - 1)\times d}$. By Lemma C.1 in \cite{JIN2023}, one can select the sign of $\bu_{01}^{(t)}$ such that every element of it is positive, and for this choice of $\bu_{01}^{(t)}$, we have $\xi_{01k}^{(t)} > 0$ for all $k\in [d]$. Let $\bxi_{01}^{(t)} = [\xi_{011}^{(t)},\ldots,\xi_{0d1}^{(t)}]\transpose$ and $\bw_{0i}^{(t)} = [\one_d,\bV^{(t)\mathrm{T}}]^{-\mathrm{T}}[1,\br_{0i}^{(t)\mathrm{T}}]\transpose$. Let $\by_{0j}^{(t)} = \bB_0^{(t)}\bz_{0j}\theta_{0j}^{(t)}$ and $\widetilde{\by}_{j}^{(t)} = \widetilde{\bB}^{(t)}\widetilde{\bz}_j\widetilde{\theta}_j^{(t)}$, where $\widetilde{\bB}^{(t)}$, $\widetilde{\bz}_j$, and $\widetilde{\theta}_j^{(t)}$ are computed using Algorithm \ref{alg:aggregated_mixed_SCORE}. Denote by $\widehat{\bU}^{(t)} = [\widehat{\bu}_1^{(t)},\ldots,\widehat{\bu}_d^{(t)}]$ and $\bU_0^{(t)} = [\bu_{01}^{(t)},\ldots,\bu_{0d}^{(t)}]$. 
Now denote by 
$\widehat{\bS}^{(t)} = \textsf{diag}(\widehat{\lambda}_1^{(t)},\ldots,\widehat{\lambda}_d^{(t)})$
and $\bE^{(t)} = \bA^{(t)} - \bP_0^{(t)} = [E_{ij}^{(t)}]_{n\times n}$. For simplicity, we also denote by $\bS_0^{-(t)} = (\bS_0^{(t)})^{-1}$, $\widehat{\bS}^{-(t)} = (\widehat{\bS}^{(t)})^{-1}$, $\bS_0^{-m(t)} = (\bS_0^{(t)})^{-m}$, $\bE^{m(t)} = (\bE^{(t)})^m$, $\widehat{\bS}^{-m(t)} = (\widehat{\bS}^{(t)})^{-m}$, $\bU^{(\mathrm{t})\mathrm{T}} = (\bU^{(t)})\transpose$, and $\widehat{\bU}^{(t)\mathrm{T}} = (\widehat{\bU}^{(t)})\transpose$. Let $u_{0ik}^{(t)},\widehat{u}_{ik}^{(t)}$ denote the $(i, k)$th element of $\bU_0^{(t)},\widehat{\bU}^{(t)}$, $\br_{0i}^{(t)},\widehat{\br}_i^{(t)}$ denote the $i$th row of $\bR_0^{(t)},\widehat{\bR}^{(t)}$, and $r_{0ik}^{(t)},\widehat{r}_{ik}^{(t)}$ denote the $(i, k)$ element of $\bR_0^{(t)},\widehat{\bR}^{(t)}$, respectively, where we define $r_{0ik}^{(t)} = u_{0i(k + 1)}^{(t)}/u_{0i1}^{(t)}$, $k\in[d - 1]$. 

\section{Detailed Algorithms}
\label{sec:detailed_algorithms}

This section provides additional details on the algorithms related to the manuscript. 

\subsection{Mixed-SCORE Algorithm}
\label{sub:mixed_score_algorithm}

For ease of reference, we review the mixed-SCORE method \citep{JIN2023} in Algorithm \ref{alg:mixed_SCORE} below.
\begin{breakablealgorithm}
\caption{\texttt{Mixed-SCORE}}
\label{alg:mixed_SCORE}
\begin{algorithmic}[1]
\State \textbf{Input:} Adjacency matrices $\bA$ and number of communities $d$.
    \State Compute the spectral decomposition $\bA = \sum_{k = 1}^n\widehat{\lambda}_k\widehat{\bu}_k\widehat{\bu}_k\transpose$, where $(\widehat{\bu}_k)_{k = 1}^n$ are orthonormal eigenvectors and the eigenvalues $|\widehat{\lambda}_1|\geq\ldots\geq|\widehat{\lambda}_n|$ are sorted in nonincreasing order in magnitude. 

    \State Let $\widehat{d}_+$ be the number of positive eigenvalues in $\{\widehat{\lambda}_1,\ldots,\widehat{\lambda}_d\}$. Rearrange the indices such that $\widehat{\lambda}_1\geq\ldots\geq\widehat{\lambda}_{\widehat{d}_+} > 0 > \widehat{\lambda}_{\widehat{d}_+ + 1}\geq\ldots\geq\widehat{\lambda}_d$ and let $\widehat{\bu}_k$ be the corresponding eigenvector. Let $\bU_{\bA} = [\widehat{u}_{ik}]_{n\times d} = [\widehat{\bu}_1, \ldots, \widehat{\bu}_d]$. Pick $\widehat{\bu}_1$ such that $\sum_{i = 1}^n\mathbbm{1}(\widehat{u}_{i1} > 0) \geq 2n/3$. 

    \State Compute the ratios of eigenvectors $\widehat{\bR} = [\widehat{R}_{ik}] = [\widehat{\br}_1,\ldots,\widehat{\br}_n]\in \mathbb{R}^{n\times (d - 1)}$
    \[
    \widehat{R}_{ik} = \textsf{sgn}\bigg(\frac{\widehat{u}_{i(k + 1)}}{\widehat{u}_{i1}}\bigg)\min\bigg(
    \bigg|\frac{\widehat{u}_{i(k + 1)}}{\widehat{u}_{i1}}\bigg|,\log n\bigg),\quad i\in[n],\;k\in[d],\;t\in[m].
    \]

    \State Compute row indices $i_1,\ldots,i_d$ corresponding to the sketched vertex search (SVS) \citep{JIN2023}. Denote by $\widehat{\bV} = [\widehat{\br}_{i_1},\ldots,\widehat{\br}_{i_d}] = [\widehat{\bv}_1,\ldots,\widehat{\bv}_d]\in\mathbb{R}^{(d - 1)\times d}$.

    \State Compute vector $\widehat{\bb}_1 = [\widehat{b}_{11},\ldots,\widehat{b}_{1d}]\transpose$, where 
    \[ 
    \widehat{b}_{1k} = \left\{\widehat{\lambda}_1 + \widehat{\bv}_k\transpose\textsf{diag}\left(\widehat{\lambda}_2,\ldots,\widehat{\lambda}_d\right)\widehat{\bv}_k\right\}^{-1/2},\quad \widehat{\bw}_i = [\widehat{w}_{i1},\ldots,\widehat{w}_{id}]\transpose = \begin{bmatrix}\one_d\transpose\\\widehat{\bV}\end{bmatrix}^{-1}\begin{bmatrix}1 \\ \widehat{\br}_i\end{bmatrix}.
    \]
    \State Estimate the membership profiles $\overline{\bZ} = [\overline{z}_{ik}]_{n\times d}$ by 
    \[
    \overline{z}_{ik} = \frac{\max(0, \widehat{w}_{ik}/\widehat{b}_{1k})}{\sum_{l = 1}^d\max(0, \widehat{w}_{il}/\widehat{b}_{1l})},\quad i \in[n],\; k\in[d].
    \]
\State \textbf{Output: } $\overline{\bZ}$, $\widehat{\bV}$, $\widehat{\bb}_1$, $(\widehat{\bu}_k, \widehat{\lambda}_k)_{k = 1}^d$. 
\end{algorithmic}
\end{breakablealgorithm}

\subsection{Preliminary Estimators for the ML-DCMM Network Model}
\label{sub:preliminary_estimators_ML_DCMM}

Next, we provide additional details on the construction of the preliminary estimators $\widetilde{\bZ}$, $(\widetilde{\btheta}_j)_{j\in[n]}$, and $(\widetilde{\bB}^{(t)})_{t = 1}^m$. The detailed
step-by-step estimation method is described in Algorithm \ref{alg:preliminary_estimators_ML_DCMM} below. Note that the preliminary estimators $(\widetilde{\btheta}_j)_{j\in[n]}$ and $(\widetilde{\bB}^{(t)})_{t = 1}^m$ are slightly different here: We re-estimate these parameters by plugging in the aggregated estimator $\widetilde{\bZ}$ instead of the layer-specific initial estimators $(\bar{\bZ}^{(t)})_{t = 1}^m$ from the mixed-SCORE algorithm. 
\begin{breakablealgorithm}
\caption{Preliminary Estimators for $(\btheta_j)_{j = 1}^n$ and $(\bB^{(t)})_{t = 1}^m$}
\label{alg:preliminary_estimators_ML_DCMM}
\begin{algorithmic}[1]
\State \textbf{Input:} Adjacency matrices $(\bA^{(t)})_{t = 1}^m$ and number of communities $d$. 
\For{$t = 1,2,\ldots,m$}
    \State Run Algorithm \ref{alg:mixed_SCORE} to obtain $(\overline{\bZ}^{(t)}, \widehat{\bV}^{(t)}, \widehat{\bb}_1^{(t)}, (\widehat{\bu}_k^{(t)}, \widehat{\lambda}_k^{(t)})_{k = 1}^d) = \texttt{mixed-SCORE}(\bA^{(t)}, d)$. 
    \State Find a permutation matrix $\bC^{(t)}\in\{0, 1\}^{d\times d}$ that minimizes $\|\overline{\bZ}^{(t)}\bC^{(t)} - \overline{\bZ}^{(1)}\|_{\mathrm{F}}^2$. 
\EndFor
\State Compute the aggregated mixed-SCORE estimator $\widetilde{\bZ} = [\widetilde{\bz}_1,\ldots,\widetilde{\bz}_n]\transpose = (1/m)\sum_{t = 1}^m\overline{\bZ}^{(t)}\bC^{(t)}$. 
\For{$t = 1,2,\ldots,m$}
    \For {$i = 1,2,\ldots,n$}
        \State Compute degree correction parameter estimator $\widetilde{\theta}_i^{(t)} = \widehat{u}_{i1}^{(t)}/(\widetilde{\bz}_i\transpose\bC^{(t)}\widehat{\bb}_1^{(t)})$
    \EndFor 
    \State Compute block mean matrix estimator
    \begin{align*}    
    \widetilde{\bB}^{(t)} = \bC^{(t)}\textsf{diag}(\widehat{\bb}_1^{(t)})[\one_d,\widehat{\bV}^{(t)\top}]\textsf{diag}(\widehat{\lambda}_1^{(t)},\ldots,\widehat{\lambda}_d^{(t)})
    \begin{bmatrix}
    \one_d^\top\\\widehat{\bV}^{(t)}
    \end{bmatrix}\textsf{diag}(\widehat{\bb}_1^{(t)})\bC^{(t)\mathrm{T}}.
    \end{align*}
\EndFor
\State \textbf{Output: } Preliminary estimators $\widetilde{\bZ}$, $(\widetilde{\btheta}_i)_{i = 1}^n = ([\widetilde{\theta}_1^{(t)},\ldots,\widetilde{\theta}_n^{(t)}]\transpose)_{i = 1}^n$, and $(\widetilde{\bB}^{(t)})_{t = 1}^m$. 
\end{algorithmic}
\end{breakablealgorithm}

\subsection{Detailed Spectral-Assisted Gaussian VI Algorithm}
\label{sub:detailed_algorithm}
We now describe the algorithm for solving the Gaussian VI problem \eqref{eqn:Gaussian_VI}. Let $\widehat{F}_{in}(\bmu_i, \bL_i, \bM_i, \bsigma_i, \beps_i)$ denote the noisy version of ${F}_{in}$ given by
\begin{align}
\widehat{F}_{in}(\bmu_i, \bL_i, \bM_i, \bsigma_i, \beps_i)
& = -\log\det(\bL_i) - \sum_{t = 1}^m\log\sigma_{it}
\nonumber
\\
\label{eqn:Gaussian_VB_objective_noisy}
&\quad - \calL_{in}\bigg(\bmu_{i1} + \frac{\bL_i\beps_{i1}}{\sqrt{mn}},\bmu_{i2} - \frac{\bM_i\bL_i\beps_{i1}}{\sqrt{mn}} + \frac{\textsf{diag}(\bsigma_i)\beps_{i2}}{\sqrt{mn}}\bigg)
\\ &\quad
 - \log\pi_{\bx}\bigg(\bmu_{i1} + \frac{\bL_i\beps_{i1}}{\sqrt{mn}}\bigg)
 - \log\pi_{\bnu}\bigg(\bmu_{i2} - \frac{\bM_i\bL_i\beps_{i1}}{\sqrt{mn}} + \frac{\textsf{diag}(\bsigma_i)\beps_{i2}}{\sqrt{mn}}\bigg).
 \nonumber
\end{align}
Clearly, $F_{in}(\bmu_i, \bL_i, \bM_i, \bsigma_i) = \expect_\beps \widehat{F}_{in}(\bmu_i, \bL_i, \bM_i, \bsigma_i, \beps_i)$. The gradient of the noisy function $\widehat{F}_{in}$ is much easier to compute than that of the noise-free objective function $F_{in}$, allowing us to use a stochastic gradient descent (SGD) algorithm for solving \eqref{eqn:Gaussian_VI}. Our SGD algorithm is based on the Adam optimizer \citep{kingma-ba-2017-adam} with the following adjustments on $\nabla_{\bL_i}\widehat{F}_{in}$ and $\nabla_{\bsigma_i}\widehat{F}_{in}$. Observe that $\nabla_{\bL_i}\log\det(\bL_i) = \textsf{diag}(L_{i11},\ldots,L_{i(d - 1)(d - 1)})^{-1}$, $\nabla_{\bsigma_i}\log\det\textsf{diag}(\bsigma_i) = [\sigma_{i1}^{-1},\ldots,\sigma_{im}^{-1}]\transpose$, where $L_{ikk}$ is the $k$th diagonal element of $\bL_i$. Both gradients become unbounded as $L_{ikk}$ or $\sigma_{it}$ approaches $0$ from the right. Such a non-Lipschitz property can cause numerical instability in practical implementation \citep{xu2022computational}. To mitigate this issue, we adopt the scaled gradient approach of \cite{xu2022computational,wuxie2025} and introduce the function
\[
\widetilde{h}_n(x) = \frac{c_n}{c_nx + 1}\mathbbm{1}(x > 0) + (c_n - c_n^2x)\mathbbm{1}(x\leq 0),
\]
where $c_n > 0$ is an $n$-dependent tuning parameter. This yields a globally Lipschitz continuous function $\widetilde{h}(\cdot)$ that closely mimics the behavior of the function $x\mapsto 1/x$ when $x > 0$. In practice, we find that choosing $c_n = 10\sqrt{mn}$ leads to satisfactory numerical performance. Denote by 
\[
\widetilde{\bH}_n(\bL_i) = \textsf{diag}(\widetilde{h}_n(L_{i11}),\ldots,\widetilde{h}_n(L_{i(d - 1)(d - 1)})),\quad\widetilde{\bh}_n(\bsigma_i) = [\widetilde{h}_n(\sigma_{i1}),\ldots,\widetilde{h}_n(\sigma_{im})]\transpose.
\]
We then define the adjusted gradients $\widetilde{\nabla}_{\bL_i}\widehat{F}_{in}$ and $\widetilde{\nabla}_{\bsigma_i}\widehat{F}_{in}$ by
\begin{align*}
\widetilde{\nabla}_{\bL_i}\widehat{F}_{in}(\bmu_i, \bL_i, \bM_i, \bsigma_i, \beps_i)
&
 = - \widetilde{\bH}_n(\bL_i)-\nabla_{\bL_i}\calL_i\bigg(\bmu_{i1} + \frac{\bL_i\beps_{i1}}{\sqrt{mn}},\bmu_{i2} - \frac{\bM_i\bL_i\beps_{i1}}{\sqrt{mn}} + \frac{\textsf{diag}(\bsigma_i)\beps_{i2}}{\sqrt{mn}}\bigg)
\\ &
\quad
 - \nabla_{\bL_i}\log\pi_{\bx}\bigg(\bmu_{i1} + \frac{\bL_i\beps_{i1}}{\sqrt{mn}}\bigg)
 \\&\quad
 - \nabla_{\bL_i}\log\pi_{\bnu}\bigg(\bmu_{i2} - \frac{\bM_i\bL_i\beps_{i1}}{\sqrt{mn}} + \frac{\textsf{diag}(\bsigma_i)\beps_{i2}}{\sqrt{mn}}\bigg),\\
\widetilde{\nabla}_{\bsigma_i}\widehat{F}_{in}(\bmu_i, \bL_i, \bM_i, \bsigma_i, \beps_i) 
& = -\widetilde{\bh}_n(\bsigma_i) - \nabla_{\bsigma_i}\calL_i\bigg(\bmu_{i1} + \frac{\bL_i\beps_{i1}}{\sqrt{mn}},\bmu_{i2} - \frac{\bM_i\bL_i\beps_{i1}}{\sqrt{mn}} + \frac{\textsf{diag}(\bsigma_i)\beps_{i2}}{\sqrt{mn}}\bigg)\\
&\quad - \nabla_{\bsigma_i}\log\pi_{\bnu}\bigg(\bmu_{i2} - \frac{\bM_i\bL_i\beps_{i1}}{\sqrt{mn}} + \frac{\textsf{diag}(\bsigma_i)\beps_{i2}}{\sqrt{mn}}\bigg). 
\end{align*}
The remaining gradients $\nabla_{\bmu_i}\widehat{F}_{in}$, $\nabla_{\bM_i}\widehat{F}_{in}$ can be computed directly. The complete procedure is detailed in Algorithm \ref{alg:SGD_Gaussian_VI}. 
\begin{breakablealgorithm}
\caption{Stochastic gradient descent for Gaussian VI \eqref{eqn:Gaussian_VI}}
\label{alg:SGD_Gaussian_VI}
\begin{algorithmic}[1]
\State \textbf{Input:} Adjacency matrices $(\bA^{(t)})_{t = 1}^m$, number of communities $d$, and preliminary estimators $\widetilde{\bZ}$, $(\widetilde{\bB}^{(t)})_{t = 1}^m$ $(\widetilde{\btheta}_i)_{i = 1}^n$.

\State \textbf{Set:} $\tau\in(0,\frac{1}{2})$, batch size $1\leq{}s\leq{}n$, step size $\alpha_0>0$, exponential decay rates for the moments of gradients $\beta_1,\beta_2\in[0,1)$, constant $\epsilon_0=10^{-8}$.

\For {$i = 1,2,\ldots,n$}

\State Compute 
\[
\widetilde{\bG}_{in}^{(t)} = \frac{1}{n}\sum_{j=1}^n \eta'(\widetilde{\theta}_i^{(t)}\widetilde{\bz}_i\transpose\widetilde{\bz}_j\widetilde{\theta}_j^{(t)})(\widetilde{\theta}_j^{(t)})^2\widetilde{\bz}_j\widetilde{\bz}_j\transpose,\quad 
\widetilde{\bDelta}_{in}^{(t)} = \widetilde{\bG}_{in}^{(t)} - \frac{\widetilde{\bG}_{in}^{(t)}\widetilde{\bz}_i\widetilde{\bz}_i\transpose\widetilde{\bG}_{in}^{(t)}}{\widetilde{\bz}_i\transpose\widetilde{\bG}_{in}^{(t)}\widetilde{\bz}_i}.
\]
\State  Compute the Cholesky factorization $(\bJ\transpose\sum_{t = 1}^m(\widetilde{\theta}_i^{(t)})^2\widetilde{\bDelta}_{in}^{(t)}\bJ/m)^{-1} =\widetilde{\bL}_{i}\widetilde{\bL}_{i}\transpose$.

\State  Set the iteration counter $\texttt{iter = 0}$.

\State  Initialize (entrywise) gradient moments 
\begin{align*}
\mathbf{m}_{\bmu_i,1}&=\mathbf{m}_{\bmu_i,2} = \bm{0}_{d + m - 1},
\mathbf{m}_{\bL_i,1} = \mathbf{m}_{\bL_i,2}=\bm{0}_{(d -1 )\times(d - 1)},\\
\mathbf{m}_{\bM_i,1} &= \mathbf{m}_{\bM_i, 2}= \zero_{m\times(d - 1)},
\mathbf{m}_{\bsigma_i,1} = \mathbf{m}_{\bsigma_i,2} = \zero_m
\end{align*}

\State Initialize $\bmu_i = (\calT_\bz^{-1}(\bkappa^{-1}(\widetilde{\bz}_i)),\calT_\btheta^{-1}(\widetilde{\btheta}_i))$, $\bL_i = \widetilde{\bL}_{i}$, $\bM_i = \zero_{m\times(d - 1)}$, $\bsigma_i = \one_m$.

\While{not converging}

\State Set $\texttt{iter} \longleftarrow \texttt{iter + 1}$.

\State Generate independent  $\beps_{i1}, \ldots, \beps_{is}\sim\mathrm{N}(\zero_{d + m - 1}, \eye_{d + m - 1})$.

\State Compute
\begin{align*}
\mathbf{m}_{\bmu_i,1} &= \beta_1 \mathbf{m}_{\bmu_i,1} + (1-\beta_1) \frac{1}{s\sqrt{mn}}\sum_{r = 1}^s\nabla_{\bmu_i}\widehat{F}_{in}(\bmu_i, \bL_i, \bM_i, \bsigma, \beps_{ir}),\\
\mathbf{m}_{\bL_i,1} &= \beta_1 \mathbf{m}_{\bL_i,1} + (1-\beta_1) \frac{1}{s\sqrt{mn}}\sum_{r = 1}^s\widetilde{\nabla}_{\bL_i}\widehat{F}_{in}(\bmu_i, \bL_i, \bM_i, \bsigma_i, \beps_{ir}),\\
\mathbf{m}_{\bM_i,1} &= \beta_1 \mathbf{m}_{\bL_i,1} + (1-\beta_1) \frac{1}{s\sqrt{mn}}\sum_{r = 1}^s\nabla_{\bM_i}\widehat{F}_{in}(\bmu_i, \bL_i, \bM_i, \bsigma_i, \beps_{ir}),\\
\mathbf{m}_{\bsigma_i,1} &= \beta_1 \mathbf{m}_{\bsigma_i,1} + (1-\beta_1) \frac{1}{s\sqrt{mn}}\sum_{r = 1}^s\widetilde{\nabla}_{\bsigma_i}\widehat{F}_{in}(\bmu_i, \bL_i, \bM_i, \bsigma_i, \beps_{ir}),\\
\mathbf{m}_{\bmu_i,2} &= \beta_2 \mathbf{m}_{\bmu_i,2} + (1-\beta_2) \frac{1}{s\sqrt{mn}}\sum_{r = 1}^s\{\nabla_{\bmu_i}\widehat{F}_{in}(\bmu_i, \bL_i, \bM_i, \bsigma_i, \beps_{is})\}^{\odot2},\\
\mathbf{m}_{\bL_i,2} &= \beta_2 \mathbf{m}_{\bL_i,2} + (1-\beta_2) \frac{1}{s\sqrt{mn}}\sum_{r = 1}^s\{\widetilde{\nabla}_{\bL_i}\widehat{F}_{in}(\bmu_i, \bL_i, \bM_i, \bsigma_i, \beps_{is})\}^{\odot2},\\
\mathbf{m}_{\bM_i,2} &= \beta_2 \mathbf{m}_{\bM_i,2} + (1-\beta_2) \frac{1}{s\sqrt{mn}}\sum_{r = 1}^s\{\nabla_{\bM_i}\widehat{F}_{in}(\bmu_i, \bL_i, \bM_i, \bsigma_i, \beps_{is})\}^{\odot2},\\
\mathbf{m}_{\bsigma_i,2} &= \beta_2 \mathbf{m}_{\bsigma_i,2} + (1-\beta_2) \frac{1}{s\sqrt{mn}}\sum_{r = 1}^s\{\widetilde{\nabla}_{\bsigma_i}\widehat{F}_{in}(\bmu_i, \bL_i, \bM_i, \bsigma_i, \beps_{is})\}^{\odot2},
\end{align*}
\quad\qquad where $\odot2$ denotes the entry-wise square of a vector or a matrix.


\State Update
\begin{align*}
\bmu_i &\longleftarrow \bmu_i - \alpha_0\cdot \frac{\mathbf{m}_{\bmu_i,1} / (1-\beta_1^{\texttt{iter}})}{\sqrt{\mathbf{m}_{\bmu_i,2} / (1-\beta_2^{\texttt{iter}})} + \epsilon_0},\\
\bL_i &\longleftarrow \bL_i - \alpha_0\cdot \frac{\mathbf{m}_{\bL_i,1} / (1-\beta_1^{\texttt{iter}})}{\sqrt{\mathbf{m}_{\bL_i,2} / (1-\beta_2^{\texttt{iter}})} + \epsilon_0},\\
\bM_i &\longleftarrow \bM_i - \alpha_0\cdot \frac{\mathbf{m}_{\bM_i,1} / (1-\beta_1^{\texttt{iter}})}{\sqrt{\mathbf{m}_{\bM_i,2} / (1-\beta_2^{\texttt{iter}})} + \epsilon_0},\\
\bsigma_i &\longleftarrow \bsigma_i - \alpha_0\cdot \frac{\mathbf{m}_{\bsigma_i,1} / (1-\beta_1^{\texttt{iter}})}{\sqrt{\mathbf{m}_{\bsigma_i,2} / (1-\beta_2^{\texttt{iter}})} + \epsilon_0},
\end{align*}
\quad\qquad where the division and square root is computed entry-wise for vectors or matrices.

\EndWhile

\EndFor

\State \textbf{Output: } $(\bmu_i, \bL_i, \bM_i, \bsigma_i)_{i = 1}^n$.

\end{algorithmic}
\end{breakablealgorithm}

\section{Simulation Studies}
\label{sec:simulation_studies}

In this section, we evaluate the statistical accuracy and computational efficiency of the proposed spectral-assisted VI for the ML-DCMM network model using synthetic datasets. We consider the three experimental setups described below. For all setups, we let the number of communities be $d = 3$,  the number of layers be $m = 10$, and the number of vertices be $n = 500$. 
\begin{itemize}[noitemsep, topsep = 0.5ex]
    \item \textbf{Experiment 1:} We set the number of pure nodes for each community to be $n_0 = 80$. For the mixed nodes, we consider the following $4$ distinct membership profiles: $[0.4, 0.4, 0.2]\transpose$, $[0.4, 0.2, 0.4]\transpose$, $[0.2, 0.4, 0.4]\transpose$, and $[1/3, 1/3, 1/3]\transpose$, each assigned to $(n - 3n_0)/4$ vertices. The degree correction parameters $\theta_i^{(t)}$ are drawn from the $1 / \textsf{Unif}(1, 5)$ distribution independently over $i\in[n]$ and $t\in[m]$. For the block mean matrices $(\bB^{(t)})_{t = 1}^m$, the diagonal entries are set to $1$ according to the identifiability requirement in Assumption \ref{assumption:identifiability} (a), and the upper-triangular elements $B_{kl}^{(t)}$ are generated from $\max(0.01,\mathrm{N}(0.08, 0.1^2))$ independently over $t\in[m]$ and $r,s\in[d]$, $r < s$. We take the edge weight distribution to be Poisson. A similar simulation design for the single-layer unweighted DCMM network model has been adopted by \cite{JIN2023}. 

    \item \textbf{Experiment 2:} We set the number of pure nodes in each community to be $n_0 = 40$. The membership profiles $(\bz_i)_{i = 1}^n = ([z_{i1},z_{i2},z_{i3}]\transpose)_{i = 1}^n$ are generated as follows: $z_{i1},z_{i2}$ are drawn from $\textsf{Unif}(1/6, 1/2)$ independently over $i\in[n]$, and $z_{i3} = 1 - z_{i1} - z_{i2}$. The degree correction parameters $\theta_i^{(t)}$ are drawn from $\textsf{Unif}(1, 2)$ independently over $i\in[n]$ and $t\in[m]$. The block mean matrices $(\bB^{(t)})_{t = 1}^m$ are generated as Experiment 1 above, and we again use Poisson as the edge weight distribution. This setup can be viewed as a multilayer extension of Experiment 6 in \cite{JIN2023}.

    \item \textbf{Experiment 3:} In this setup, we consider unweighted networks by taking the edge weights to be Bernoulli distributed. All remaining parameters are the same as in Experiment~1.
\end{itemize}
For each of the experimental setups above, we consider the following methods: the aggregated mixed-SCORE (computed via Algorithm \ref{alg:aggregated_mixed_SCORE}), our proposed Gaussian VI with a structured covariance matrix specified by  \eqref{eqn:Gaussian_VI_distribution_class} (computed via Algorithm \ref{alg:SGD_Gaussian_VI}), and Gaussian VI with a block-diagonal covariance matrix (computed via Algorithm \ref{alg:SGD_Gaussian_VI}, except that $\bM_i$ is fixed to $\zero_{m\times(d - 1)}$ and is not updated throughout). Each experiment is replicated $100$ times independently. 

We first measure estimation accuracy for the membership profile matrix $\bZ$ defined using the squared Frobenius norm $\|\widehat{\bZ} - \bZ_0\|_{\mathrm{F}}^2$ as the estimation error, where $\widehat{\bZ}$ denotes a generic estimator. The resulting errors across the $100$ Monte Carlo replicates are summarized via side-by-side boxplots in Figure \ref{fig:Simulation_Error_Z}. 
\begin{figure}[htbp]
  \centering{
  \includegraphics[width=\textwidth]{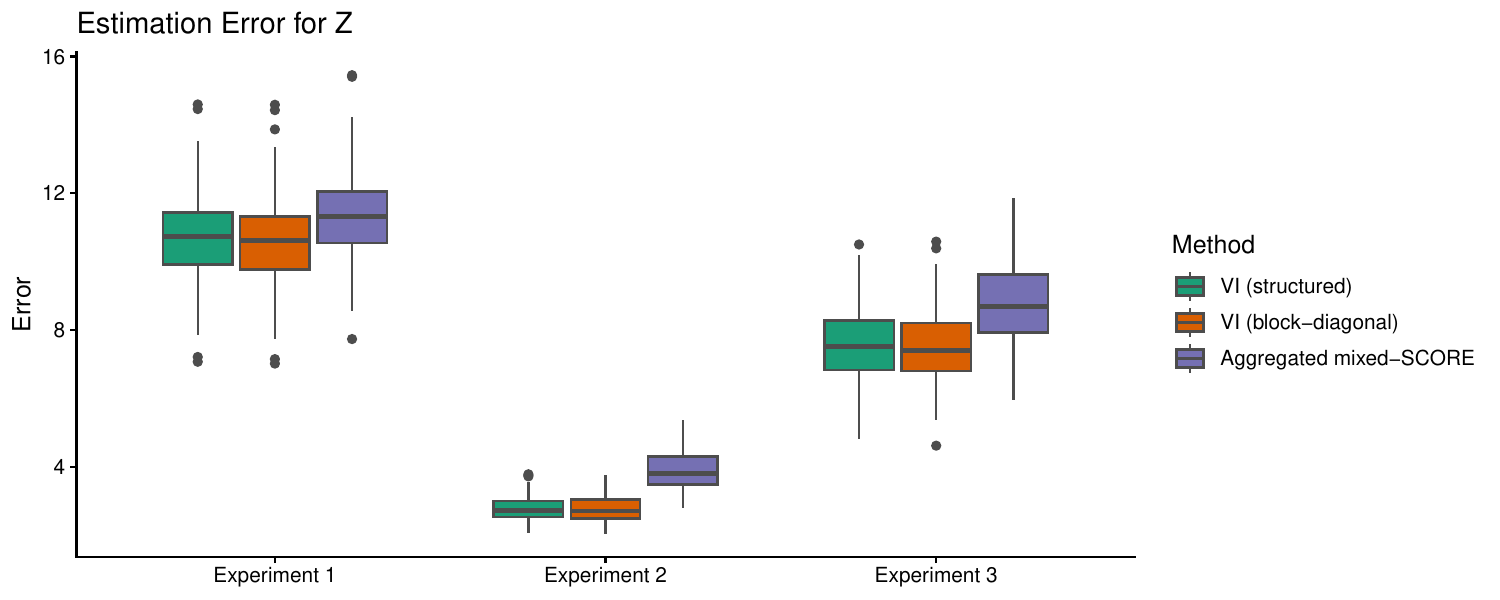}
  }
  \caption{\label{fig:Simulation_Error_Z}Simulation results: Estimation errors for the membership profile matrix $\bZ$ over $100$ repeated Monte Carlo replicates.}
\end{figure}
For the VI methods, we use the point estimates by setting $\widehat{\bz}_i = \calT_\bz(\bmu_{i1})$ for $i\in[n]$, where $\bmu_{i1}$ is the variational posterior mean of $\bx_i$. We observe that the aggregated mixed-SCORE estimates yield larger estimation errors than the point estimates of the VI methods. The paired two-sample $t$-tests between the estimation errors of the aggregated mixed-SCORE estimates and those of the VI estimates produce $p$-values less than $2.2\times 10^{-6}$, indicating that these differences are statistically significant. Note also that the two VI estimates are comparable in terms of point estimation accuracy for $\bZ$. 

To further distinguish their performance in terms of uncertainty quantification, we compute the empirical coverage rates of vertex-wise $95\%$ variational posterior credible sets for the mixed nodes by averaging over the $100$ Monte Carlo replicates. The empirical coverage rates are displayed in Figure \ref{fig:Simulation_CR}. 
\begin{figure}[htbp]
  \centering
  \includegraphics[width=\textwidth]{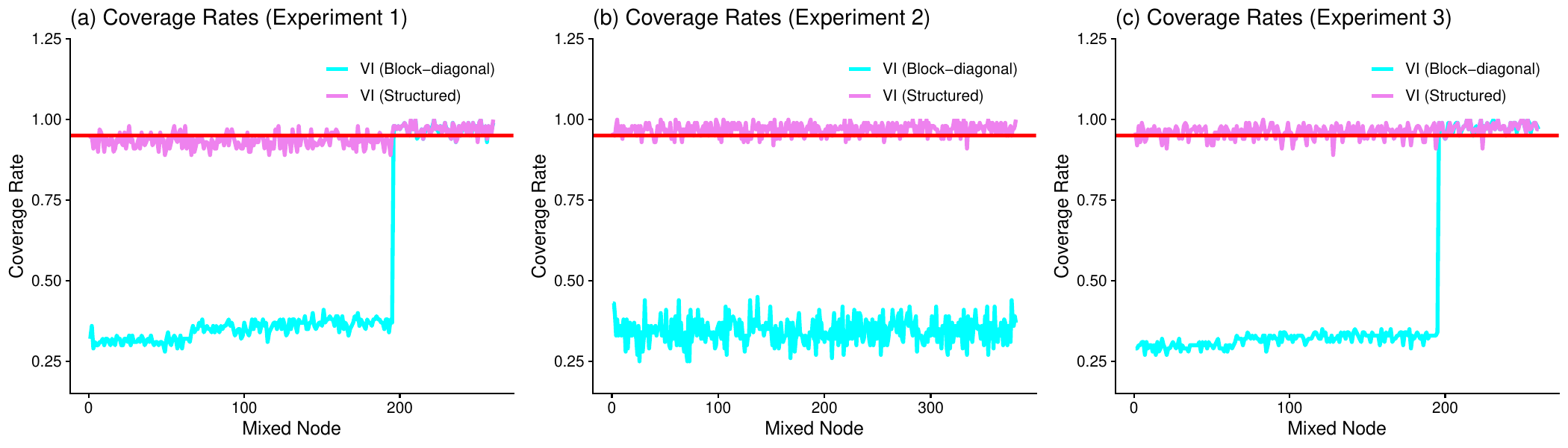}
  \caption{Simulation results: Empirical coverage rates of vertex-wise variational posterior credible sets for membership profiles over $100$ repeated Monte Carlo replicates.}
  \label{fig:Simulation_CR}
\end{figure}
The empirical coverage rates for the variational posterior credible sets are close to the nominal $95\%$ level when the covariance matrix is structured according to \eqref{eqn:Gaussian_VI_distribution_class}, thereby supporting the validity of uncertainty quantification guaranteed by Theorem \ref{thm:variational_BvM}. In contrast, when the covariance matrix is restricted to be block-diagonal, the empirical coverage rates deviate substantially from the nominal level due to the ignorance of the correlation between $\bz_i$ and $\btheta_i$ under the exact posterior distribution, leading to untrustworthy uncertainty quantification.

\section{Auxiliary Results}
\label{sec:auxiliary_results}

This section collects several auxiliary results that prepare the first-order stochastic expansion for the aggregated mixed-SCORE estimator computed using Algorithm \ref{alg:aggregated_mixed_SCORE}. The results obtained here critically rely on some properties of leading eigenvectors and eigenvalues from \cite{cape_signal-plus-noise_2019} and their adaptations, some properties of the mixed-SCORE estimator from \cite{JIN2023}, and some properties on the noise distributions originated from \cite{10.1214/11-AOP734} in the form of \cite{xie2025higher}. 

\begin{result}[Results C.2, C.3, C.4, and Lemma C.1 in \citealp{xie2025higher}]
\label{result:preliminary}
Suppose Assumptions \ref{assumption:identifiability}--\ref{assumption:likelihood} hold. Then for any $\bU_1\in\mathbb{O}(n, d_1)$, $\bU_2\in\mathbb{O}(n, d_2)$ with fixed $d_1,d_2$, and for any fixed positive integer $k$ and $i\in[n]$, we have
\begin{align*}
&\|\bE^{(t)}\|_2 = \Optilde(\sqrt{n}),\; \|\bU_1\transpose\bE^{(t)}\bU_2\|_2 = \Optilde\{n(\log n)^\xi\|\bU_1\|_{2\to\infty}\|\bU_2\|_{2\to\infty}\},\\
&\|\be_i\transpose\bE^{k(t)}\bU_1\|_{2\to\infty} = \Optilde\{\|\bU_1\|_{2\to\infty}n^{k/2}(\log n)^{k\xi}\},\;\|\bS_0^{(t)}\bU_0^{(t)\mathrm{T}}\widehat{\bU}^{(t)} - \bU_0^{(t)\mathrm{T}}\widehat{\bU}^{(t)}\widehat{\bS}^{(t)}\|_2 = \Optilde\{(\log n)^\xi\}.
\end{align*}
\end{result}

\begin{result}[Eigenvector expansion]
\label{result:eigenvector_expansion}
Suppose Assumptions \ref{assumption:identifiability}--\ref{assumption:likelihood} hold. Then for each fixed $t\in[m]$, there exists a matrix $\bQ_{d - 1}^{(t)}\in\mathbb{O}(d - 1)$, such that
\begin{align*}
&\|\bU^{(t)}_0\|_{2\to\infty}\lesssim \frac{1}{\sqrt{n}},\quad\|\widehat{\bU}^{(t)}\|_{2\to\infty} = \Optilde\bigg(\frac{1}{\sqrt{n}}\bigg),
\\
&\|\widehat{\bU}^{(t)\mathrm{T}}\bU^{(t)}_0 - \bQ^{(t)}\|_2 = \Optilde\bigg(\frac{1}{n}\bigg),\quad \|\widehat{\bU}^{(t)}\bQ^{(t)} - \bU^{(t)}_0\|_{2\to\infty} = \Optilde\bigg\{\frac{(\log n)^\xi}{n}\bigg\}\\
&\widehat{\bU}^{(t)}\bQ^{(t)} - \bU^{(t)}_0 = \bE^{(t)}\bU^{(t)}_0\bS^{-(t)}_0 + \bR_{\bU^{(t)}},\quad \|\bR_{\bU^{(t)}}\|_{2\to\infty} = \Optilde\bigg\{\frac{(\log n)^{2\xi}}{n^{3/2}}\bigg\},
\end{align*}
where $\bQ^{(t)} = \textsf{diag}(\mathsf{sgn}(\widehat{\bu}_1^{(t)\mathrm{T}}\bu_{01}^{(t)}),\bQ^{(t)}_{d - 1})$. In particular, the matrix $\bQ_{d - 1}^{(t)}$ can be taken as 
\[
\bQ_{d - 1}^{(t)} = \mathsf{sgn}\left(\begin{bmatrix}\widehat{\bu}_2^{(t)\mathrm{T}}\\\vdots\\\widehat{\bu}_d^{(t)\mathrm{T}}\end{bmatrix}\begin{bmatrix}\bu_{02}^{(t)}&\ldots&\bu_{0d}\transpose\end{bmatrix}\right)
\]
\end{result}
\begin{proof}[\bf Proof]
For the first equality, we obtain
\begin{align*}
\|\bU_0^{(t)}\|_{2\to\infty}
& = \|\bTheta^{(t)}_0\bZ_0(\bZ_0\transpose(\bTheta_0^{(t)})^2\bZ_0)^{-1/2}\|_{2\to\infty}
\leq \|\bTheta_0^{(t)}\|_{\infty}\|\bZ_0\|_{2\to\infty}\|(\bZ_0\transpose(\bTheta_0^{(t)})^2\bZ_0)^{-1/2}\|_{2\to\infty} \lesssim\frac{1}{\sqrt{n}}.
\end{align*}
The third equality follows directly from Lemma C.2 in \cite{xie2025higher}, and the Davis-Kahan theorem. The second and the fourth equality follow from the last equality, the first equality, and . Therefore, it is sufficient to show the last one. We modify the proof of Theorem 2 in \cite{cape_signal-plus-noise_2019}. By the matrix Sylvester equation $\widehat{\bU}^{(t)}\widehat{\bS}^{(t)} - \bE^{(t)}\widehat{\bU}^{(t)} = \bP^{(t)}_0\widehat{\bU}^{(t)}$, we obtain the matrix series expansion for $\widehat{\bU}^{(t)}$ by \cite{bhatia2013matrix}:
\begin{equation}\label{eqn:Eigenvector_series_expansion}
\begin{aligned}
\widehat{\bU}^{(t)}\bQ^{(t)} - \bU^{(t)}_0& = \bE^{(t)}\bU^{(t)}_0\bS_0^{-(t)} + \bE^{(t)}\bU_0^{(t)}\bS_0^{-(t)}(\bU_0^{(t)\mathrm{T}}\widehat{\bU}^{(t)} - \bQ^{(t)\mathrm{T}})\bQ^{(t)}\\
&\quad + \bE^{(t)}\bU_0^{(t)}\bS_0^{(t)}(\bU_0^{(t)\mathrm{T}}\widehat{\bU}\widehat{\bS}^{-2(t)} - \bS_0^{-2(t)}\bU_0^{(t)\mathrm{T}}\widehat{\bU}^{(t)})\bQ^{(t)}\\ 
&\quad + \sum_{m = 2}^\infty\bE^{m(t)}\bU_0^{(t)}\bS_0^{(t)}\bU_0^{(t)\mathrm{T}}\widehat{\bU}^{(t)}\widehat{\bS}^{-(m + 1)(t)}\bQ^{(t)}_0
\\
&\quad + \bU_0^{(t)}(\bU_0^{(t)\mathrm{T}}\widehat{\bU}^{(t)} - \bQ^{(t)\mathrm{T}})\bQ^{(t)}\\
&\quad + \bU_0^{(t)}(\bS_0^{(t)}\bU_0^{(t)\mathrm{T}}\widehat{\bU}^{(t)} - \bU_0^{(t)\mathrm{T}}\widehat{\bU}^{(t)}\widehat{\bS}^{(t)})\widehat{\bS}^{-(t)}\bQ^{(t)}.
\end{aligned}
\end{equation}
The two-to-infinity norm of the second term is $\Optilde(n^{-3/2})$
by Result \ref{result:preliminary}. 
The two-to-infinity norms of the third term, the fifth term, and the last term are $\Optilde\{(\log n)^\xi/n^{3/2}\}$
by Result \ref{result:preliminary} and the proof of Lemma C.1 in \cite{xie2025higher}. For the fourth term, by Lemma C.1 in \cite{xie2025higher}, by taking $M_n = 2$ we derive
\begin{align*}
&\left\|\sum_{m = 2}^\infty \bE^{m(t)}\bU_0^{(t)}\bS_0^{(t)}\bU_0^{(t)\mathrm{T}}\widehat{\bU}^{(t)}\widehat{\bS}^{-(m + 1)(t)}\right\|_{2\to\infty}\\
&\quad\leq \sum_{m = 2}^{M_n}\|\bE^{(t)m}\bU_0^{(t)}\|_{2\to\infty}\|\bS_0^{(t)}\|_2\|\widehat{\bS}^{-(m + 1)(t)}\|_2
 + \sum_{m = M_n + 1}^\infty\|\bE^{(t)}\|_2^m\|\bS^{(t)}_0\|_2\|\widehat{\bS}^{-(m + 1)(t)}\|_2\\
&\quad\leq \|\bU_0^{(t)}\|_{2\to\infty}
\sum_{m = 2}^{M_n}\left[\widetilde{O}_{\prob}\left\{\frac{c(\log n)^{2\xi}}{n}\right\}\right]^{m/2} + 
\sum_{m = M_n + 1}^\infty\left\{\widetilde{O}_{\prob}\left(\frac{1}{n}\right)\right\}^{m/2}
 = \Optilde\bigg\{\frac{(\log n)^{2\xi}}{n^{3/2}}\bigg\}.
\end{align*}
The proof is thus completed. 
\end{proof}

\begin{result}\label{result:leading_eigenvector}
Suppose Assumptions \ref{assumption:identifiability}--\ref{assumption:likelihood} hold. Then, one can pick the sign of $\bu_{01}^{(t)}$, such that every element of $\bu_{01}^{(t)}$ are positive, $\min_{i\in[n]}u_{0i1}^{(t)}\asymp n^{-1/2}$, and $\min_{i\in[n]}\widehat{u}_{i1}^{(t)}\mathsf{sgn}(\widehat{\bu}_1^{(t)\mathrm{T}}\bu_{01}^{(t)})\asymp n^{-1/2}$ w.h.p..
\end{result}
\begin{proof}[\bf Proof]
The result follows directly from Lemma C.3 in \cite{JIN2023} and Result \ref{result:eigenvector_expansion}. 
\end{proof}

\section{First-Order Expansions of Mixed-SCORE}
\label{sec:analysis_of_mixed_score_estimator}

This section establishes the first-order expansions of the estimators obtained by applying the mixed-SCORE algorithm (Algorithm \ref{alg:mixed_SCORE}) directly to each layer as well as the preliminary estimators from Algorithm \ref{alg:preliminary_estimators_ML_DCMM}. These results go beyond the perturbation bounds obtained in \cite{JIN2023} and they crucially rely on the first-order expansion result of eigenvectors of $(\bA^{(t)})_{t = 1}^m$ from Result \ref{result:eigenvector_expansion} above. We begin by the expansion of the ratio-of-eigenvector statistic $\widehat{\bR}^{(t)}$. 
 
\begin{lemma}\label{lemma:R_expansion}
Suppose Assumptions \ref{assumption:identifiability}--\ref{assumption:likelihood} hold. Then for each fixed $t\in[m]$, there exists an orthogonal matrix $\bH^{(t)}\in\mathbb{O}(d - 1)$, such that
\begin{align*}
\bH^{(t)}\widehat{\br}_i^{(t)} - \br_{0i}^{(t)} & = 
\sum_{j = 1}^nE_{ij}^{(t)}\balpha_{ij}^{(t)} + \Optilde\bigg\{\frac{(\log n)^{2\xi}}{n}\bigg\},
\end{align*}
where
\begin{align}\label{eqn:alpha_ijt}
\balpha_{ij}^{(t)} = \frac{1}{u_{0i1}^{(t)}}\left(\begin{bmatrix}u_{0j2}^{(t)}/\lambda_2^{(t)}\\\vdots\\u_{0jd}^{(t)}/\lambda_d^{(t)} \end{bmatrix} - \frac{\br_{0i}^{(t)}u_{0j1}^{(t)}}{\lambda_1^{(t)}}\right).
\end{align}
In particular, we have $\|\bH^{(t)}\widehat{\br}_i^{(t)} - \br_{0i}^{(t)}\|_2 = \Optilde\{(\log n)^\xi/\sqrt{n}\}$. 
\end{lemma}

\begin{proof}[\bf Proof]
By Result \ref{result:eigenvector_expansion}, we have
\[
\widehat{u}_{i1}^{(t)}w_1^{(t)} - u_{0i1}^{(t)} = \sum_{j = 1}^n\frac{E_{ij}^{(t)}u_{0jk}^{(t)}}{\lambda_1^{(t)}} + r_{u_{0i1}^{(t)}},\quad |r_{u_{i1}^{(t)}}| = \Optilde\bigg\{\frac{(\log n)^{2\xi}}{n^{3/2}}\bigg\},
\]
where we denote by $w_1^{(t)} = \mathsf{sgn}(\widehat{\bu}_1^{(t)\mathrm{T}}\bu_{01}^{(t)})$. Let $\bQ^{(t)}_{d - 1}$ be the orthogonal matrix given by Result \ref{result:eigenvector_expansion} and let $\bH^{(t)} = \bQ^{(t)\mathrm{T}}_{d - 1}/w_1^{(t)}$. Then, 
\begin{align*}
\bH^{(t)}\widehat{\br}_i^{(t)} - \br_{0i}^{(t)}
& = \frac{1}{\widehat{u}_{i1}^{(t)}w_1^{(t)}}\bQ^{(t)\mathrm{T}}_{d - 1}\begin{bmatrix}
\widehat{u}_{i2}^{(t)}\\\vdots\\\widehat{u}_{id}^{(t)}
\end{bmatrix}
 - \frac{1}{u_{0i1}^{(t)}}\begin{bmatrix}u_{i2}^{(t)}\\\vdots\\u_{id}^{(t)}\end{bmatrix}\\
& = \bQ^{(t)\mathrm{T}}_{d - 1}\begin{bmatrix}
\widehat{u}_{i2}^{(t)}\\\vdots\\\widehat{u}_{id}^{(t)}
\end{bmatrix}\bigg(\frac{u_{0i1}^{(t)} - \widehat{u}_{i1}^{(t)}w_1^{(t)}}{\widehat{u}_{i1}^{(t)}w_1^{(t)}u_{0i1}^{(t)}}\bigg) + \frac{1}{u_{0i1}}\sum_{j = 1}^nE_{ij}^{(t)}\begin{bmatrix}u_{0j2}^{(t)}/\lambda_2^{(t)}\\\vdots\\u_{0jd}^{(t)}/\lambda_d^{(t)} \end{bmatrix} + \Optilde\bigg\{\frac{(\log n)^{2\xi}}{n}\bigg\}\\
& = \left(\bQ^{(t)\mathrm{T}}_{d - 1}\begin{bmatrix}
\widehat{u}_{i2}^{(t)}\\\vdots\\\widehat{u}_{id}^{(t)}
\end{bmatrix} - \begin{bmatrix}u_{0i2}^{(t)}\\\vdots\\u_{0id}^{(t)}\end{bmatrix}\right)\bigg(\frac{u_{0i1}^{(t)} - \widehat{u}_{i1}^{(t)}w_1^{(t)}}{\widehat{u}_{i1}^{(t)}w_1^{(t)}u_{0i1}^{(t)}}\bigg)
 + \begin{bmatrix}u_{0i2}^{(t)}\\\vdots\\u_{0id}^{(t)}\end{bmatrix}\bigg(\frac{u_{0i1}^{(t)} - \widehat{u}_{i1}^{(t)}w_1^{(t)}}{\widehat{u}_{i1}^{(t)}w_1^{(t)}u_{0i1}^{(t)}}\bigg)\\ 
&\quad + \frac{1}{u_{0i1}}\sum_{j = 1}^nE_{ij}^{(t)}\begin{bmatrix}u_{0j2}^{(t)}/\lambda_2^{(t)}\\\vdots\\u_{0jd}^{(t)}/\lambda_d^{(t)} \end{bmatrix} + \Optilde\bigg\{\frac{(\log n)^{2\xi}}{n}\bigg\}.
\end{align*}
By Results \ref{result:eigenvector_expansion} and \ref{result:leading_eigenvector}, the first term is $\Optilde\{(\log n)^{2\xi}/n\}$. For the second term, by Result \ref{result:eigenvector_expansion}, we have
\begin{align*}
\begin{bmatrix}u_{0i2}^{(t)}\\\vdots\\u_{0id}^{(t)}\end{bmatrix}\bigg(\frac{u_{0i1}^{(t)} - \widehat{u}_{i1}^{(t)}w_1^{(t)}}{\widehat{u}_{i1}^{(t)}w_1^{(t)}u_{0i1}^{(t)}}\bigg)
& = \begin{bmatrix}u_{0i2}^{(t)}/u_{0i1}^{(t)}\\\vdots\\u_{0id}^{(t)}/u_{0i1}^{(t)}\end{bmatrix}\bigg(1 + \frac{u_{0i1}^{(t)} - \widehat{u}_{i1}^{(t)}w_1^{(t)}}{\widehat{u}_{i1}^{(t)}w_1^{(t)}}\bigg)\bigg(\frac{u_{0i1}^{(t)} - \widehat{u}_{i1}^{(t)}w_1^{(t)}}{u_{0i1}^{(t)}}\bigg)\\
& = -\br_{0i}^{(t)}\sum_{j = 1}^n\frac{E_{ij}^{(t)}u_{0j1}^{(t)}}{\lambda_1^{(t)}u_{0i1}^{(t)}} + \Optilde\bigg\{\frac{(\log n)^{2\xi}}{n}\bigg\}.
\end{align*}
Therefore, we conclude that
\begin{align*}
\bH^{(t)}\widehat{\br}_i^{(t)} - \br_{0i}^{(t)}
& = \sum_{j = 1}^n\frac{E_{ij}^{(t)}}{u_{0i1}^{(t)}}\left(\begin{bmatrix}u_{0j2}^{(t)}/\lambda_2^{(t)}\\\vdots\\u_{0jd}^{(t)}/\lambda_d^{(t)} \end{bmatrix} - \frac{\br_{0i}^{(t)}u_{0j1}^{(t)}}{\lambda_1^{(t)}}\right) + \Optilde\bigg\{\frac{(\log n)^{2\xi}}{n}\bigg\}.
\end{align*}
\end{proof}
Next, we also need the expansion of the vector $\widehat{\bb}_1^{(t)}$ obtained from applying the mixed-SCORE algorithm to the $t$th layer $\bA^{(t)}$. 
\begin{lemma}\label{lemma:w_expansion}
Suppose Assumptions \ref{assumption:identifiability}--\ref{assumption:likelihood} hold. Then, there exists a permutation matrix $\bK^{(t)}\in\mathbb{O}(d)$, such that
\begin{align*}
\bK^{(t)}\textsf{diag}(\widehat{\bb}_1^{(t)})^{-1}\widehat{\bw}_i^{(t)} - \textsf{diag}(\bxi_{01}^{(t)})^{-1}\bw_{0i}^{(t)}
& = \sum_{j = 1}^nE_{ij}\textsf{diag}(\bxi_{01}^{(t)})^{-1}\begin{bmatrix}\one_d\transpose\\\bV_0^{(t)}\end{bmatrix}^{-1}\begin{bmatrix} 0\\\balpha_{ij}^{(t)}\end{bmatrix} + \Optilde\bigg\{\frac{(\log n)^{2\xi}}{\sqrt{n}}\bigg\}, 
\end{align*}
where $\balpha_{ij}^{(t)}$ is given by \eqref{eqn:alpha_ijt}. In particular, we have 
\[
\|\bK^{(t)}\textsf{diag}(\widehat{\bb}_1^{(t)})^{-1}\widehat{\bw}_i^{(t)} - \textsf{diag}(\bxi_{01}^{(t)})^{-1}\bw_{0i}^{(t)}\|_2 = \Optilde\{(\log n)^\xi\}\quad\text{and}\quad
\textsf{diag}(\widehat{\bb}_1^{(t)})^{-1}\widehat{\bw}_i^{(t)} \asymp \sqrt{n}\quad\text{w.h.p..}
\]
\end{lemma}

\begin{proof}[\bf Proof]
By (E.84) and the proof of Lemma E.4 in \cite{JIN2023}, we know that there exists a permutation $\sigma:[d]\to[d]$, such that $\widehat{\bv}_{\sigma(k)}^{(t)} = (1/|\calN_k|)\sum_{i\in\calN_k}\widehat{\br}_i^{(t)}$. Note here that although Lemma E.4 in \cite{JIN2023} is under the unweighted network model, the proof carries over to our setup because $\max_{i\in[n]}\|\bH^{(t)}\widehat{\br}_i^{(t)} - \br_{0i}^{(t)}\|_2 = \Optilde\{(\log n)^\xi/\sqrt{n}\}$ by Lemma \ref{lemma:R_expansion} and the condition that $|\calM|,|\calN_k|\gg (\log n)^{2\xi}/n$. By Lemma \ref{lemma:R_expansion} and Lemma 3.8 in \cite{10.1214/11-AOP734}, we have
\begin{align*}
\|\bH^{(t)}\widehat{\bv}_{\sigma(k)}^{(t)} - \bv_{0k}^{(t)}\|_2
& = \bigg\|\frac{1}{|\calN_k|}\sum_{i\in\calN_k}(\bH^{(t)}\widehat{\br}_i^{(t)} - \br_{0i}^{(t)})\bigg\|_2
  = \frac{1}{|\calN_k|}\bigg\|\sum_{i\in\calN_k}\sum_{j = 1}^nE_{ij}^{(t)}\balpha_{ij}^{(t)}\bigg\|_2 + \Optilde\bigg\{\frac{(\log n)^{2\xi}}{n}\bigg\}\\ 
& = \Optilde\bigg\{\frac{n(\log n)^{\xi}}{|\calN_k|}\max_{i,j\in[n]}\|\balpha_{ij}^{(t)}\|_2 + \frac{(\log n)^{2\xi}}{n}\bigg\} = \Optilde\bigg\{\frac{(\log n)^{2\xi}}{n}\bigg\}. 
\end{align*}
This implies that 
\begin{align}\label{eqn:Vt_concentration}
\|\bH^{(t)}\widehat{\bV}^{(t)}\bK^{(t)\mathrm{T}} - \bV_0^{(t)}\|_2 = \Optilde\bigg\{\frac{(\log n)^{2\xi}}{n}\bigg\}
\end{align}
for some permutation matrix $\bK^{(t)}\in\mathbb{O}(d)$. Also, by (C.26) in \cite{JIN2023}, we have
\begin{align*}
&\left\|\begin{bmatrix}\one_d\transpose\\\bV_0^{(t)}\end{bmatrix}\right\|_2 = O(1),
\quad
\left\|\begin{bmatrix}\one_d\transpose\\\bV_0^{(t)}\end{bmatrix}^{-1}\right\|_2 = O(1),\quad
\left\|\begin{bmatrix}\one_d\transpose\\\widehat{\bV}^{(t)}\end{bmatrix}\right\|_2 = \Optilde(1),
\quad
\left\|\begin{bmatrix}\one_d\transpose\\\widehat{\bV}^{(t)}\end{bmatrix}^{-1}\right\|_2 = \Optilde(1),\\
&\left\|\begin{bmatrix}\one_d\transpose\\\bH^{(t)}\widehat{\bV}^{(t)}\bK^{(t)\mathrm{T}}\end{bmatrix}^{-1} - \begin{bmatrix}\one_d\transpose\\\bV_0^{(t)}\end{bmatrix}^{-1}\right\|_2
 = \Optilde\bigg\{\frac{(\log n)^{2\xi}}{n}\bigg\}.
\end{align*}
We then obtain
\begin{align*}
\bK^{(t)}\widehat{\bw}_i^{(t)} - \bw_{0i}^{(t)}
& = \begin{bmatrix}\one_d\transpose\\\bH^{(t)}\widehat{\bV}^{(t)}\bK^{(t)\mathrm{T}}\end{bmatrix}^{-1}\begin{bmatrix}1\\\bH^{(t)}\widehat{\br}_i^{(t)}\end{bmatrix} - 
\begin{bmatrix}\one_d\transpose\\\bV_0^{(t)}\end{bmatrix}^{-1}\begin{bmatrix}1\\\br_{0i}^{(t)}\end{bmatrix}\\
& = \left(
    \begin{bmatrix}\one_d\transpose\\\bH^{(t)}\widehat{\bV}^{(t)}\bK^{(t)\mathrm{T}}\end{bmatrix}^{-1} - \begin{bmatrix}\one_d\transpose\\\bV_0^{(t)}\end{bmatrix}^{-1}
\right)\begin{bmatrix}1\\\bH^{(t)}\widehat{\br}_i^{(t)}\end{bmatrix} + 
\begin{bmatrix}\one_d\transpose\\\bV_0^{(t)}\end{bmatrix}^{-1}\begin{bmatrix}0\\\bH^{(t)}\widehat{\br}_i^{(t)} - \br_{0i}^{(t)}\end{bmatrix}\\
& = \sum_{j = 1}^nE_{ij}^{(t)}\begin{bmatrix}\one_d\transpose\\\bV_0^{(t)}\end{bmatrix}^{-1}\begin{bmatrix}0\\\balpha_{ij}^{(t)}\end{bmatrix} + \Optilde\bigg\{\frac{(\log n)^{2\xi}}{n}\bigg\},
\end{align*}
where $\balpha_{ij}^{(t)}$ is given by \eqref{eqn:alpha_ijt}. In particular, we have $\|\bK^{(t)}\widehat{\bw}_i^{(t)} - \bw_{0i}^{(t)}\|_2 = \Optilde\{(\log n)^{\xi}/\sqrt{n}\}$ and $\|\widehat{\bw}_i^{(t)}\|_2 = \Optilde(1)$. 
Let $\widehat{\bS}_{2:d}^{(t)} = \textsf{diag}(\widehat{\lambda}_2^{(t)},\ldots,\widehat{\lambda}_{d}^{(t)})$ and
 $\bS_{2:d}^{(t)} = \textsf{diag}(\lambda_2^{(t)},\ldots,\lambda_d^{(t)})$. 
By Result \ref{result:preliminary} and the fact that $\bA^{(t)}\widehat{\bU}^{(t)} = \widehat{\bU}^{(t)}\widehat{\bS}^{(t)}$ and $\bU_0^{(t)\mathrm{T}}\bP_0^{(t)} = \bS^{(t)}_0\bU_0^{(t)\mathrm{T}}$, 
\begin{equation}
\label{eqn:HS_exchange}
\begin{aligned}
&\left\|
\begin{bmatrix}1 & \\ & \bH^{(t)}\end{bmatrix}\widehat{\bS}^{(t)} - \bS_0^{(t)}\begin{bmatrix}1 & \\ & \bH^{(t)}\end{bmatrix}
\right\|_2\\
&\quad = 
\left\|
\begin{bmatrix}w_1^{(t)} & \\ & \bQ_{d - 1}^{(t)\mathrm{T}}\end{bmatrix}\widehat{\bS}^{(t)} - \bS_0^{(t)}\begin{bmatrix}w_1^{(t)} & \\ & \bQ_{d - 1}^{(t)\mathrm{T}}\end{bmatrix}
\right\|_2 = \|\bQ^{(t)\mathrm{T}}\widehat{\bS}^{(t)} - \bS_0^{(t)}\bQ^{(t)\mathrm{T}}\|_2\\ 
&\quad\leq \|\bQ^{(t)\mathrm{T}} - \bU_0^{(t)\mathrm{T}}\widehat{\bU}^{(t)}\|_2\|\widehat{\bS}^{(t)}\|_2 + \|\bU_0^{(t)\mathrm{T}}\bE^{(t)}(\widehat{\bU}^{(t)} - \bU_0^{(t)}\bQ^{(t)\mathrm{T}})\|_2\\ 
&\quad\quad + \|\bU_0^{(t)}\bE^{(t)}\bU_0^{(t)}\|_2 + \|\bS_0\|_2\|\bQ^{(t)\mathrm{T}} - \bU_0^{(t)\mathrm{T}}\widehat{\bU}^{(t)}\|_2\\ 
&\quad = \Optilde\{(\log n)^\xi\}.
\end{aligned}
\end{equation}
In particular, we have $|\widehat{\lambda}_1^{(t)} - \lambda_1^{(t)}| + \|\bH^{(t)}\widehat{\bS}_{2:d}^{(t)} - \bS_{2:d}^{(t)}\bH^{(t)}\|_2 = \Optilde\{(\log n)^\xi\}$. 
By the definition of $\widehat{\bb}_1^{(t)}$ and $\bxi_{01}^{(t)}$ and Lemma 2.1 in \cite{JIN2023}, we follow the proof of Theorem 3.2 in \cite{JIN2023} and obtain
\begin{align*}
\left|\widehat{b}_{1\sigma(k)}^{-2(t)} - \xi_{01k}^{-2(t)}\right|
& \leq 
\left|\widehat{\lambda}_1^{(t)} - \lambda_1^{(t)}\right| + \left|(\widehat{\bv}_k^{(t)\mathrm{T}}\bH^{(t)\mathrm{T}} - \bv_{0k}^{(t)\mathrm{T}})\bH^{(t)}\widehat{\bS}^{(t)}\widehat{\bv}_k^{(t)}\right|\\ 
&\quad + \left|\bv_{0k}^{(t)\mathrm{T}}(\bH^{(t)}\widehat{\bS}_{2:d}^{(t)} - \bS_{2:d}^{(t)}\bH^{(t)})\widehat{\bv}_k^{(t)}\right| + \left|\bv_{0k}^{(t)\mathrm{T}}\bS_{2:d}^{(t)}(\bH^{(t)}\widehat{\bv}_k^{(t)} - \bv_{0k}^{(t)})\right|\\
& = \Optilde\{(\log n)^{2\xi}\}. 
\end{align*}
By the proof of Theorem 3.2 in \cite{JIN2023}, $\xi_{01k}^{-2(t)} \asymp n$, implying that $\widehat{b}_{1k}^{-2(t)} \asymp n$ w.h.p., and
\begin{align}
\label{eqn:b_concentration}
\left|\widehat{b}_{1\sigma(k)}^{-(t)} - \xi_{01k}^{-(t)}\right|
&\leq \left|\widehat{b}_{1\sigma(k)}^{-(t)} + \xi_{01k}^{-(t)}\right|^{-1}\left|\widehat{b}_{1\sigma(k)}^{-2(t)} - \xi_{01k}^{-2(t)}\right| = \Optilde\bigg\{\frac{(\log n)^{2\xi}}{\sqrt{n}}\bigg\}. 
\end{align}
Namely, $\|\bK^{(t)}\textsf{diag}(\widehat{\bb}_1^{(t)})^{-1}\bK^{(t)\mathrm{T}} - \textsf{diag}(\bxi_{01}^{(t)})^{-1}\|_2 = \Optilde\{(\log n)^{2\xi}/\sqrt{n}\}$, and 
\begin{align*}
\bK^{(t)}\textsf{diag}(\widehat{\bb}_1^{(t)})^{-1}\widehat{\bw}_i^{(t)} - \textsf{diag}(\bxi_{01}^{(t)})^{-1}\bw_{0i}^{(t)}
& = \bK^{(t)}\textsf{diag}(\widehat{\bb}_k^{(t)})^{-1}\bK^{(t)\mathrm{T}}\bK^{(t)}\widehat{\bw}_i^{(t)} - \textsf{diag}(\bxi_{01}^{(t)})^{-1}\bw_{0i}^{(t)}\\ 
& = \left\{\bK^{(t)}\textsf{diag}(\widehat{\bb}_k^{(t)})^{-1}\bK^{(t)\mathrm{T}} - \textsf{diag}(\bxi_{01}^{(t)})^{-1}\right\}\bK^{(t)}\widehat{\bw}_i^{(t)}\\ 
&\quad + \textsf{diag}(\bxi_{01}^{(t)})^{-1}(\bK^{(t)}\widehat{\bw}_i^{(t)} - \bw_{0i}^{(t)})\\ 
& = \sum_{j = 1}^nE_{ij}\textsf{diag}(\bxi_{01}^{(t)})^{-1}\begin{bmatrix}\one_d\transpose\\\bV_0^{(t)}\end{bmatrix}^{-1}\begin{bmatrix} 0\\\balpha_{ij}^{(t)}\end{bmatrix} + \Optilde\bigg\{\frac{(\log n)^{2\xi}}{\sqrt{n}}\bigg\}. 
\end{align*}
The proof is therefore completed.
\end{proof}

Now we are able to establish the first-order expansion result of the layer-wise mixed-SCORE estimator $\bar{\bZ}^{(t)}$ for the membership profile matrix by leveraging Lemma \ref{lemma:R_expansion} and Lemma \ref{lemma:w_expansion}. 
\begin{lemma}\label{lemma:Z_expansion}
Suppose Assumptions \ref{assumption:identifiability}--\ref{assumption:likelihood} hold. Then, there exists a permutation matrix $\bK^{(t)}\in\mathbb{O}(d)$, such that
\begin{align*}
\bK^{(t)}\overline{\bz}_i^{(t)} - \bz_{0i}
& = \sum_{j = 1}^nE_{ij}^{(t)}\bbeta_{ij}^{(t)} + \Optilde\bigg\{\frac{(\log n)^{2\xi}}{n}\bigg\},
\end{align*}
where 
\begin{align}\label{eqn:beta_ijt}
\bbeta_{ij}^{(t)} = \left\{\frac{\eye_d}{\one_d\transpose\textsf{diag}(\bxi_{01}^{(t)})^{-1}\bw_{0i}^{(t)}} - \frac{\textsf{diag}(\bxi_{01}^{(t)})^{-1}\bw_{0i}^{(t)}\one_d\transpose}{(\one_d\transpose\textsf{diag}(\bxi_{01}^{(t)})^{-1}\bw_{0i}^{(t)})^2}\right\}\textsf{diag}(\bxi_{01}^{(t)})^{-1}\begin{bmatrix}\one_d\transpose\\\bV_0^{(t)}\end{bmatrix}\begin{bmatrix}0\\\balpha_{ij}^{(t)}\end{bmatrix}
\end{align}
and $\balpha_{ij}^{(t)}$ is given by \eqref{eqn:alpha_ijt}.
\end{lemma}

\begin{proof}[\bf Proof]
By definition, we have $\bw_{0i}^{(t)} = \textsf{diag}(\bxi_{01}^{(t)})\bz_{0i}/\|\textsf{diag}(\bxi_{01}^{(t)})\bz_{0i}\|_2$. By (C.22) in \cite{JIN2023}, we have $\xi_{01k}(t) \asymp n^{-1/2}$. This implies that $\one_d\transpose\textsf{diag}(\bxi_{01}^{(t)})^{-1}\bw_{0i}^{(t)} = \|\textsf{diag}(\bxi_{01}^{(t)})\bz_{0i}\|_2^{-1} \asymp \sqrt{n}$. By Lemma \ref{lemma:w_expansion}, we obtain
\begin{align*}
\bK^{(t)}\overline{\bz}_i^{(t)} - \bz_{0i}
& = \frac{\bK^{(t)}\textsf{diag}(\widehat{\bb}_1^{(t)})^{-1}\widehat{\bw}_i^{(t)}}{\one_d\transpose\textsf{diag}(\widehat{\bb}_1^{(t)})^{-1}\widehat{\bw}_i^{(t)}} - \frac{\textsf{diag}(\bxi_{01}^{(t)})^{-1}\bw_{0i}^{(t)}}{\one_d\transpose\textsf{diag}(\bxi_{01}^{(t)})^{-1}\bw_{0i}^{(t)}}\\ 
& = \frac{\bK^{(t)}\textsf{diag}(\widehat{\bb}_1^{(t)})^{-1}\widehat{\bw}_i^{(t)}}{\one_d\transpose\textsf{diag}(\widehat{\bb}_1^{(t)})^{-1}\widehat{\bw}_i^{(t)}\one_d\transpose\textsf{diag}(\bxi_{01}^{(t)})^{-1}\bw_{0i}^{(t)}}\left(\one_d\transpose\textsf{diag}(\bxi_{01}^{(t)})^{-1}\bw_{0i}^{(t)} - \one_d\transpose\textsf{diag}(\widehat{\bb}_1^{(t)})^{-1}\widehat{\bw}_i^{(t)}\right)\\ 
&\quad + \frac{\bK^{(t)}\textsf{diag}(\widehat{\bb}_1^{(t)})^{-1}\widehat{\bw}_i^{(t)} - \textsf{diag}(\bxi_{01}^{(t)})^{-1}\bw_{0i}^{(t)}}{\one_d\transpose\textsf{diag}(\bxi_{01}^{(t)})^{-1}\bw_{0i}^{(t)}}\\
& = \frac{\textsf{diag}(\bxi_{01}^{(t)})^{-1}\bw_{0i}^{(t)}}{\one_d\transpose\textsf{diag}(\widehat{\bb}_1^{(t)})^{-1}\widehat{\bw}_i^{(t)}\one_d\transpose\textsf{diag}(\bxi_{01}^{(t)})^{-1}\bw_{0i}^{(t)}}\left(\one_d\transpose\textsf{diag}(\bxi_{01}^{(t)})^{-1}\bw_{0i}^{(t)} - \one_d\transpose\textsf{diag}(\widehat{\bb}_1^{(t)})^{-1}\widehat{\bw}_i^{(t)}\right)\\ 
&\quad + \frac{\bK^{(t)}\textsf{diag}(\widehat{\bb}_1^{(t)})^{-1}\widehat{\bw}_i^{(t)} - \textsf{diag}(\bxi_{01}^{(t)})^{-1}\bw_{0i}^{(t)}}{\one_d\transpose\textsf{diag}(\bxi_{01}^{(t)})^{-1}\bw_{0i}^{(t)}} + \Optilde\bigg\{\frac{(\log n)^{2\xi}}{n}\bigg\}\\ 
& = \frac{\textsf{diag}(\bxi_{01}^{(t)})^{-1}\bw_{0i}^{(t)}}{(\one_d\transpose\textsf{diag}(\bxi_{01}^{(t)})^{-1}\bw_{0i}^{(t)})^2}\one_d\transpose\left(\textsf{diag}(\bxi_{01}^{(t)})^{-1}\bw_{0i}^{(t)} - \bK^{(t)}\textsf{diag}(\widehat{\bb}_1^{(t)})^{-1}\widehat{\bw}_i^{(t)}\right)\\
&\quad\times\left\{1 + \frac{\one_d\transpose\textsf{diag}(\bxi_{01}^{(t)})^{-1}\bw_{0i}^{(t)} - \one_d\transpose\textsf{diag}(\widehat{\bb}_{1}^{(t)})^{-1}\widehat{\bw}_{i}^{(t)}}{\one_d\transpose\textsf{diag}(\widehat{\bb}_{1}^{(t)})^{-1}\widehat{\bw}_{i}^{(t)}}\right\}\\ 
&\quad + \frac{\bK^{(t)}\textsf{diag}(\widehat{\bb}_1^{(t)})^{-1}\widehat{\bw}_i^{(t)} - \textsf{diag}(\bxi_{01}^{(t)})^{-1}\bw_{0i}^{(t)}}{\one_d\transpose\textsf{diag}(\bxi_{01}^{(t)})^{-1}\bw_{0i}^{(t)}} + \Optilde\bigg\{\frac{(\log n)^{2\xi}}{n}\bigg\}\\ 
& = \left\{\frac{\eye_d}{\one_d\transpose\textsf{diag}(\bxi_{01}^{(t)})^{-1}\bw_{0i}^{(t)}} - \frac{\textsf{diag}(\bxi_{01}^{(t)})^{-1}\bw_{0i}^{(t)}\one_d\transpose}{(\one_d\transpose\textsf{diag}(\bxi_{01}^{(t)})^{-1}\bw_{0i}^{(t)})^2}\right\}\\
&\quad\times \left(\bK^{(t)}\textsf{diag}(\widehat{\bb}_1^{(t)})^{-1}\widehat{\bw}_i^{(t)} - \textsf{diag}(\bxi_{01}^{(t)})^{-1}\bw_{0i}^{(t)}\right)
 + \Optilde\bigg\{\frac{(\log n)^{2\xi}}{n}\bigg\}\\ 
& = \sum_{j = 1}^nE_{ij}^{(t)}\left\{\frac{\eye_d}{\one_d\transpose\textsf{diag}(\bxi_{01}^{(t)})^{-1}\bw_{0i}^{(t)}} - \frac{\textsf{diag}(\bxi_{01}^{(t)})^{-1}\bw_{0i}^{(t)}\one_d\transpose}{(\one_d\transpose\textsf{diag}(\bxi_{01}^{(t)})^{-1}\bw_{0i}^{(t)})^2}\right\}\textsf{diag}(\bxi_{01}^{(t)})^{-1}\begin{bmatrix}\one_d\transpose\\\bV_0^{(t)}\end{bmatrix}\begin{bmatrix}0\\\balpha_{ij}^{(t)}\end{bmatrix}\\ 
&\quad + \Optilde\bigg\{\frac{(\log n)^{2\xi}}{n}\bigg\}. 
\end{align*}
The proof is thereby completed. 
\end{proof}

The final expansion result is with regard to the rows of  $\bY^{(t)} := \textsf{diag}(\theta_1^{(t)},\ldots,\theta_n^{(t)})\bZ\bB^{(t)}$. The true value is given by $\bY_0^{(t)} = \textsf{diag}(\theta_{01}^{(t)},\ldots,\theta_{0n}^{(t)})\bZ_0\bB_0^{(t)}$. Lemma \ref{lemma:Y_expansion} below establishes the row-wise expansion of the preliminary estimator $\widetilde{\bY}^{(t)} = \textsf{diag}(\widetilde{\theta}_1^{(t)},\ldots,\widetilde{\theta}_n^{(t)})\widetilde{\bZ}\widetilde{\bB}^{(t)}$, where $\widetilde{\bZ}$, $(\widetilde{\theta}_i^{(t)}:i\in[n],t\in[m])$, $(\widetilde{\bB}^{(t)})_{t = 1}^m$ are computed by Algorithm \ref{alg:aggregated_mixed_SCORE}. 

\begin{lemma}\label{lemma:Y_expansion}
Suppose Assumptions \ref{assumption:identifiability}--\ref{assumption:likelihood} hold. Then there exists a permutation matrix $\bK^{(t)}\in\mathbb{O}(d)$ and deterministic vectors $(\bgamma_{ija}^{(t)}:i,j,a\in[n],t\in[m])\subset\mathbb{R}^d$, $(\bdelta_{ija}^{(s,t)}:i,j,a\in[n],s,t\in[m])\subset\mathbb{R}^d$ satisfying $\max_{j,a\in[n],s,t\in[m]}(\|\bgamma_{ja}^{(t)}\|_2 + \|\bdelta_{ja}^{(s,t)}\|_2)\lesssim 1/n$, such that
\begin{align*}
\bK^{(t)}\widetilde{\by}_{j}^{(t)} - \by_{0j}^{(t)} & = \sum_{a = 1}^nE_{ja}^{(t)}\bgamma_{ja}^{(t)} + \frac{1}{m}\sum_{s = 1}^m\sum_{a = 1}^nE_{ja}^{(s)}\bdelta_{ja}^{(s, t)} + \Optilde\bigg\{\frac{(\log n)^{2\xi}}{n}\bigg\}.
\end{align*}
In particular, we have $\|\bK^{(t)}\widetilde{\by}_{j}^{(t)} - \by_{0j}^{(t)}\|_2 = \Optilde\{(\log n)^\xi/\sqrt{n}\}$. 
\end{lemma}

\begin{proof}[\bf Proof]
By the definition of $\bC^{(t)}$, we know that $\bC^{(t)} = \bK^{(1)\mathrm{T}}\bK^{(t)}$ w.h.p.. Indeed, if not, then there exists some $k\in [d]$, such that $\|\bZ_0 - \bZ_0\bK^{(t)}\bC^{(t)\mathrm{T}}\bK^{(1)\mathrm{T}}\|_{\mathrm{F}}\gtrsim |\calN_k|^{1/2}$ w.h.p.. By Lemma \ref{lemma:Z_expansion}, we know that
\begin{align*}
\frac{1}{n}\|\overline{\bZ}^{(t)}\bK^{(t)\mathrm{T}} - \bZ_0\|_{\mathrm{F}}^2 = \frac{1}{n}\sum_{i = 1}^n\|\bK^{(t)}\overline{\bz}_i^{(t)} - \bz_{0i}\|_2^2 = \Optilde\bigg\{\frac{(\log n)^{2\xi}}{n}\bigg\}, 
\end{align*}
so that
\begin{align*}
\frac{1}{n}\|\overline{\bZ}^{(t)}\bK^{(t)\mathrm{T}}\bK^{(1)} - \overline{\bZ}^{(1)}\|_{\mathrm{F}}^2
&\leq \frac{2}{n}\|\overline{\bZ}^{(t)}\bK^{(t)\mathrm{T}} - \bZ_0\|_{\mathrm{F}}^2 + \frac{2}{n}\|\bZ_0 - \overline{\bZ}^{(1)}\bK^{(1)\mathrm{T}}\|_{\mathrm{F}}^2 = \Optilde\bigg\{\frac{(\log n)^{2\xi}}{n}\bigg\}. 
\end{align*}
On the other hand, we also have
\begin{align*}
\frac{1}{\sqrt{n}}\|\overline{\bZ}^{(t)}\bC^{(t)\mathrm{T}} - \overline{\bZ}^{(1)}\|_{\mathrm{F}}
& = \frac{1}{\sqrt{n}}\|\overline{\bZ}^{(t)}\bC^{(t)\mathrm{T}}\bK^{(1)\mathrm{T}} - \overline{\bZ}^{(1)}\bK^{(1)\mathrm{T}}\|_{\mathrm{F}}\\
&\geq \frac{1}{\sqrt{n}}\|\bZ_0\bK^{(t)}\bC^{(t)}\bK^{(1)\mathrm{T}} - \bZ_0\|_{\mathrm{F}} - \frac{1}{\sqrt{n}} \|\overline{\bZ}^{(t)}\bK^{(t)\mathrm{T}} - \bZ_0\|_{\mathrm{F}}\\
&\quad - \frac{1}{\sqrt{n}} \|\bZ_0 - \overline{\bZ}^{(1)}\bK^{(1)\mathrm{T}}\|_{\mathrm{F}}\\ 
&\gtrsim \frac{|\calN_k|}{\sqrt{n}} - \bigg|\Optilde\bigg\{\frac{(\log n)^\xi}{\sqrt{n\rho_n}}\bigg\}\bigg|\gtrsim \frac{|\calN_k|}{\sqrt{n}}\quad\text{w.h.p..}
\end{align*}
Namely, we know that $\|\overline{\bZ}^{(t)}\bC^{(t)\mathrm{T}} - \overline{\bZ}^{(1)}\|_{\mathrm{F}} < \|\overline{\bZ}^{(t)}\bK^{(t)\mathrm{T}}\bK^{(1)} - \overline{\bZ}^{(1)}\|_{\mathrm{F}}$ w.h.p.. This contradicts with the assumption that $\bC^{(t)}$ is the minimizer and $\bC^{(t)}\neq \bK^{(1)\mathrm{T}}\bK^{(t)}$.

By Lemma \ref{lemma:Z_expansion} and the construction of $\widetilde{\bz}_i = (1/m)\sum_{t = 1}^m\bC^{(t)}\bar{\bz}_i^{(t)} = (1/m)\sum_{t = 1}^m\bK^{(1)\mathrm{T}}\bK^{(t)}\overline{\bz}_i^{(t)}$, we can take the permutation matrix $\bK$ to be $\bK^{(1)}$, such that
\begin{align*}
\bK\widetilde{\bz}_i - \bz_{0i} & = \frac{1}{m}\sum_{t = 1}^m\sum_{j = 1}^nE_{ij}^{(t)}\bbeta_{ij}^{(t)} + \Optilde\bigg\{\frac{(\log n)^{2\xi}}{n}\bigg\},
\end{align*}
where $\bbeta_{ij}^{(t)}$ is given by \eqref{eqn:beta_ijt}. 
By \eqref{eqn:Vt_concentration}, \eqref{eqn:b_concentration}, and (C.22) in \cite{JIN2023}, we have
\begin{align*}
&\left\|\bK^{(t)}\textsf{diag}(\widehat{\bb}_1^{(t)})
\begin{bmatrix}\one_d & \widehat{\bV}^{(t)\mathrm{T}}\end{bmatrix}\begin{bmatrix}1 & \\ & \bH^{(t)\mathrm{T}}\end{bmatrix} - \textsf{diag}(\bxi_{01}^{(t)})\begin{bmatrix}\one_d & \bV_0^{(t)\mathrm{T}}\end{bmatrix}\right\|_2\\ 
&\quad = \left\|\bK^{(t)}\textsf{diag}(\widehat{\bb}_1^{(t)})\bK^{(t)\mathrm{T}}
\begin{bmatrix}\one_d & \bK^{(t)}\widehat{\bV}^{(t)\mathrm{T}}\bH^{(t)\mathrm{T}}\end{bmatrix} - \textsf{diag}(\bxi_{01}^{(t)})\begin{bmatrix}\one_d & \bV_0^{(t)}\end{bmatrix}\right\|_2\\ 
&\quad\leq \left\|\bK^{(t)}\textsf{diag}(\widehat{\bb}_1^{(t)})\bK^{(t)\mathrm{T}} - \textsf{diag}(\bxi_{01}^{(t)})\right\|_2\left\|
\begin{bmatrix}\one_d & \bK^{(t)}\widehat{\bV}^{(t)\mathrm{T}}\bH^{(t)\mathrm{T}}\end{bmatrix}\right\|_2
\\ &\quad\quad
 + \left\|\textsf{diag}(\bxi_{01}^{(t)})\right\|_2
\left\|\begin{bmatrix}\zero_d & \bK^{(t)}\widehat{\bV}^{(t)\mathrm{T}}\bH^{(t)\mathrm{T}} - \bV_0^{(t)}\end{bmatrix}\right\|_2\\ 
&\quad\leq \left\|\bK^{(t)}\textsf{diag}(\widehat{\bb}_1^{(t)})\bK^{(t)\mathrm{T}}\right\|_2\left\|(\bK^{(t)}\textsf{diag}(\widehat{\bb}_1^{(t)})\bK^{(t)\mathrm{T}})^{-1} - \textsf{diag}(\bxi_{01}^{(t)})^{-1}\right\|_2
\left\| \textsf{diag}(\bxi_{01}^{(t)})\right\|_2
\\ 
&\quad\quad\times 
\left\|
\begin{bmatrix}\one_d & \bK^{(t)}\widehat{\bV}^{(t)\mathrm{T}}\bH^{(t)\mathrm{T}}\end{bmatrix}\right\|_2
+  \left\|\textsf{diag}(\bxi_{01}^{(t)})\right\|_2
\left\|\begin{bmatrix}\zero_d & \bK^{(t)}\widehat{\bV}^{(t)\mathrm{T}}\bH^{(t)\mathrm{T}} - \bV_0^{(t)}\end{bmatrix}\right\|_2\\ 
&\quad = 
\Optilde\bigg\{\frac{(\log n)^{2\xi}}{n^{3/2}}\bigg\}.
\end{align*}
Let
\[
\overline{\bB}^{(t)} = \textsf{diag}(\widehat{\bb}_1^{(t)})
\begin{bmatrix}\one_d & \widehat{\bV}^{(t)\mathrm{T}}\end{bmatrix}\widehat{\bS}^{(t)}\begin{bmatrix}\one_d\transpose\\\widehat{\bV}^{(t)}\end{bmatrix}\textsf{diag}(\widehat{\bb}_{1}^{(t)}).
\]
By the definition $\bB_0^{(t)}$ and \eqref{eqn:HS_exchange}, we have
\begin{align*}
&\bK^{(t)}\overline{\bB}^{(t)}\bK^{(t)\mathrm{T}} - \bB_0^{(t)}\\
&\quad = \bK^{(t)}\textsf{diag}(\widehat{\bb}_1^{(t)})
\begin{bmatrix}\one_d & \widehat{\bV}^{(t)\mathrm{T}}\end{bmatrix}\widehat{\bS}^{(t)}\begin{bmatrix}\one_d\transpose\\\widehat{\bV}^{(t)}\end{bmatrix}\textsf{diag}(\widehat{\bb}_{1}^{(t)})\bK^{(t)\mathrm{T}} - \bB_0^{(t)}\\ 
&\quad = \left(\bK^{(t)}\textsf{diag}(\widehat{\bb}_1^{(t)})
\begin{bmatrix}\one_d & \widehat{\bV}^{(t)\mathrm{T}}\end{bmatrix}
\begin{bmatrix}1 & \\ & \bH^{(t)\mathrm{T}} \end{bmatrix} - \textsf{diag}(\bxi_{01}^{(t)})\begin{bmatrix}\one_d & \bV_0^{(t)\mathrm{T}}\end{bmatrix}\right)\\
&\quad\quad\times
\begin{bmatrix}1 & \\ & \bH^{(t)} \end{bmatrix}
\widehat{\bS}^{(t)}\begin{bmatrix}\one_d\transpose\\\widehat{\bV}^{(t)}\end{bmatrix}\textsf{diag}(\widehat{\bb}_{1}^{(t)})\bK^{(t)\mathrm{T}}\\ 
&\quad\quad + \textsf{diag}(\bxi_{01}^{(t)})\begin{bmatrix}\one_d & \bV_{0}^{(t)\mathrm{T}}\end{bmatrix}\left(\begin{bmatrix} 1 & \\ & \bH^{(t)}\end{bmatrix}\widehat{\bS}^{(t)} - \bS_0^{(t)}\begin{bmatrix} 1 & \\ & \bH^{(t)}\end{bmatrix}\right)\begin{bmatrix}\one_d\transpose \\ \widehat{\bV}^{(t)}\end{bmatrix}\textsf{diag}(\widehat{\bb}_1^{(t)})\bK^{(t)\mathrm{T}}\\ 
&\quad\quad + \textsf{diag}(\bxi_{01}^{(t)})\begin{bmatrix}\one_d & \bV_0^{(t)\mathrm{T}}\end{bmatrix}\bS_0^{(t)}\left(\begin{bmatrix} 1 & \\ & \bH^{(t)}\end{bmatrix}\begin{bmatrix}\one_d\transpose\\\widehat{\bV}^{(t)}\end{bmatrix}\textsf{diag}(\widehat{\bb}_1^{(t)})\bK^{(t)\mathrm{T}} - \begin{bmatrix}\one_d\transpose\\\bV_0^{(t)}\end{bmatrix}\textsf{diag}(\bxi_{01}^{(t)})\right)\\ 
&\quad = \Optilde\bigg\{\frac{(\log n)^{2\xi}}{n}\bigg\}.
\end{align*}
Therefore, 
\begin{align*}
\bK\widetilde{\bB}^{(t)}\widetilde{\bz}_i - \bB_0^{(t)}\bz_{0i}
& = \bK\bC^{(t)}\overline{\bB}^{(t)}\bC^{(t)\mathrm{T}}\widetilde{\bz}_i - \bB_0^{(t)}\bz_{0i}\\ 
& = \bK^{(t)}\overline{\bB}^{(t)}\bK^{(t)\mathrm{T}}\bK^{(1)}\widetilde{\bz}_i - \bB_0^{(t)}\bz_{0i}\\ 
& = \left(\bK^{(t)}\overline{\bB}^{(t)}\bK^{(t)\mathrm{T}} - \bB_0^{(t)}\right)\bK^{(1)}\widetilde{\bz}_i + \bB_0^{(t)}(\bK^{(1)}\widetilde{\bz}_i - \bz_{0i})\\
& = \frac{1}{m}\sum_{s = 1}^m\sum_{a = 1}^nE_{ia}^{(s)}\bB_0^{(t)}\bbeta_{ia}^{(s)} + \Optilde\bigg\{\frac{(\log n)^{2\xi}}{n}\bigg\} = \Optilde\bigg\{\frac{(\log n)^\xi}{\sqrt{n}}\bigg\}.
\end{align*}
For $\widetilde{\theta}_i^{(t)}$, by definition, Result \ref{result:eigenvector_expansion}, and Lemma \ref{lemma:Z_expansion}, we have
\begin{align*}
\widetilde{\bz}_i\transpose\bC^{(t)}\widehat{\bb}_1^{(t)} - \bz_{0i}\transpose\bxi_{01}^{(t)}
& = \widetilde{\bz}_i\transpose\bK^{(1)\mathrm{T}}\bK^{(t)}\widehat{\bb}_1^{(t)} - \bz_{0i}\transpose\bxi_{01}^{(t)} \\ 
& = \left(\widetilde{\bz}_i\transpose\bK^{(1)\mathrm{T}} - \bz_{0i}\transpose\right)\left(\bK^{(t)}\widehat{\bb}_1^{(t)} - \bxi_{01}^{(t)}\right) + \left(\widetilde{\bz}_i\transpose\bK^{(1)\mathrm{T}} - \bz_{0i}\transpose\right)\bxi_{01}^{(t)}\\
&\quad + \bz_{0i}\transpose\left(\bK^{(t)}\widehat{\bb}_1^{(t)} - \bxi_{01}^{(t)}\right)\\
& = \left(\widetilde{\bz}_i\transpose\bK^{(1)\mathrm{T}} - \bz_{0i}\transpose\right)\bxi_{01}^{(t)} + \Optilde\bigg\{\frac{(\log n)^{\xi}}{n^{3/2}}\bigg\} = \Optilde\bigg\{\frac{(\log n)^\xi}{n}\bigg\},\\
\frac{1}{\widetilde{\bz}_i\transpose\bC^{(t)}\widehat{\bb}_1^{(t)}} - \frac{1}{\bz_{0i}\transpose\bxi_{01}^{(t)}}
& = \frac{\bz_{0i}\transpose\bxi_{01}^{(t)} - \widetilde{\bz}_i\transpose\bC^{(t)}\widehat{\bb}_1^{(t)}}{(\bz_{0i}\transpose\bxi_{01}^{(t)})^2}\left(1 + \frac{\bz_{0i}\transpose\bxi_{01}^{(t)} - \widetilde{\bz}_i\transpose\bC^{(t)}\widehat{\bb}_1^{(t)}}{\widetilde{\bz}_i\transpose\bC^{(t)}\widehat{\bb}_1^{(t)}}\right)\\
& = \frac{\bz_{0i}\transpose\bxi_{01}^{(t)} - \widetilde{\bz}_i\transpose\bC^{(t)}\widehat{\bb}_1^{(t)}}{(\bz_{0i}\transpose\bxi_{01}^{(t)})^2} + \Optilde\bigg\{\frac{(\log n)^{2\xi}}{\sqrt{n}}\bigg\}\\
& = \frac{(\bK^{(1)}\widetilde{\bz}_i - \bz_{0i})\transpose\bxi_{01}^{(t)}}{(\bz_{0i}\transpose\bxi_{01}^{(t)})^2} + \Optilde\bigg\{\frac{(\log n)^{2\xi}}{\sqrt{n}}\bigg\} = \Optilde\{(\log n)^\xi\}.
\end{align*}
It follows that
\begin{align*}
\widetilde{\theta}_i^{(t)} - \theta_{0i}^{(t)}
& = \frac{\widehat{u}_{i1}^{(t)}}{\widetilde{\bz}_i\transpose\bC^{(t)}\widehat{\bb}_1^{(t)}} - \frac{u_{0i1}^{(t)}}{\bz_{0i}\transpose\bxi_{01}^{(t)}}\\
& = \frac{\widehat{u}_{i1}^{(t)} - u_{0i1}^{(t)}}{\bz_{0i}\transpose\bxi_{01}^{(t)}} + u_{0i1}^{(t)}\left(\frac{1}{\widetilde{\bz}_i\transpose\bC^{(t)}\widehat{\bb}_1^{(t)}} - \frac{1}{\bz_{0i}\transpose\bxi_{01}^{(t)}}\right) + (\widehat{u}_{i1}^{(t)} - u_{0i1}^{(t)})\left(\frac{1}{\widetilde{\bz}_i\transpose\bC^{(t)}\widehat{\bb}_1^{(t)}} - \frac{1}{\bz_{0i}\transpose\bxi_{01}^{(t)}}\right)\\ 
& = \frac{\widehat{u}_{i1}^{(t)} - u_{0i1}^{(t)}}{\bz_{0i}\transpose\bxi_{01}^{(t)}}
 + \frac{u_{0i1}^{(t)}(\bK^{(1)}\widetilde{\bz}_i - \bz_{0i})\transpose\bxi_{01}^{(t)}}{(\bz_{0i}\transpose\bxi_{01}^{(t)})^2} + \Optilde\bigg\{\frac{(\log n)^{2\xi}}{n\rho_n^{1/2}}\bigg\} = \Optilde\bigg\{\frac{(\log n)^\xi}{\sqrt{n}}\bigg\}.
\end{align*}
Therefore, 
\begin{align*}
\bK\widetilde{\by}_{j}^{(t)} - \by_{0j}^{(t)}
& = \widetilde{\theta}_j^{(t)}\bK\widetilde{\bB}^{(t)}\widetilde{\bz}_j - \theta_{0j}^{(t)}\bB_0^{(t)}\bz_{0j}\\ 
& = (\widetilde{\theta}_j^{(t)} - \theta_{0j}^{(t)})(\bK\widetilde{\bB}^{(t)}\widetilde{\bz}_j - \bB_0^{(t)}\bz_{0j})
 + (\widetilde{\theta}_j^{(t)} - \theta_{0j}^{(t)})\bB_0^{(t)}\bz_{0j}
 + \theta_{0j}^{(t)}(\bK\widetilde{\bB}^{(t)}\widetilde{\bz}_i - \bB_0^{(t)}\bz_{0j})\\ 
& = (\widetilde{\theta}_j^{(t)} - \theta_{0j}^{(t)})\bB_0^{(t)}\bz_{0j} + \theta_{0j}^{(t)}(\bK\widetilde{\bB}^{(t)}\widetilde{\bz}_j - \bB_0^{(t)}\bz_{0j}) + \Optilde\bigg\{\frac{(\log n)^{2\xi}}{n}\bigg\}\\ 
& = \sum_{a = 1}^nE_{ja}^{(t)}\frac{u_{0a1}^{(t)}}{\lambda_1^{(t)}}\bB_0^{(t)}\bz_{0j}
 + \frac{1}{m}\frac{\theta_{0j}^{(t)}u_{0j1}^{(t)}}{(\bz_{0j}\transpose\bxi_{01}^{(t)})^2}\sum_{s = 1}^m\sum_{a = 1}^nE_{ja}^{(s)}\bbeta_{ja}^{(s)\mathrm{T}}\bxi_{01}^{(t)}\bB_0^{(t)}\bz_{0j}\\ 
&\quad + \frac{1}{m}\sum_{s = 1}^m\sum_{a = 1}^nE_{ja}^{(s)}\theta_{0j}^{(t)}\bB_0^{(t)}\bbeta_{ja}^{(s)} + 
\Optilde\bigg\{\frac{(\log n)^{2\xi}}{n}\bigg\}\\ 
& = \sum_{a = 1}^nE_{ja}^{(t)}\bgamma_{ja}^{(t)} + \frac{1}{m}\sum_{s = 1}^m\sum_{a = 1}^nE_{ja}^{(s)}\bdelta_{ja}^{(s, t)} + \Optilde\bigg\{\frac{(\log n)^{2\xi}}{n}\bigg\},
\end{align*}
where
\begin{equation}\label{eqn:gamma_delta}
\begin{aligned}
\bgamma_{ja}^{(t)} & = \frac{u_{0a1}^{(t)}}{\lambda_1^{(t)}}\bB_0^{(t)}\bz_{0j},\quad
\bdelta_{ja}^{(s, t)}  = \frac{\theta_{0j}^{(t)}u_{0j1}^{(t)}}{(\bz_{0j}\transpose\bxi_{01}^{(t)})^2}\bbeta_{ja}^{(s)\mathrm{T}}\bxi_{01}^{(t)}\bB_0^{(t)}\bz_{0j} + \theta_{0j}^{(t)}\bB_0^{(t)}\bbeta_{ja}^{(s)}.
\end{aligned}
\end{equation}
The proof is thus completed.
\end{proof}

\section{Proof of Theorem \ref{thm:MLE}}
\label{sec:proof_of_theorem_mle}

This section proves the asymptotic normality of a local maximum likelihood estimator associated with the spectral-assisted likelihood \eqref{eqn:spectral_assisted_loglik}. In preparation for the proof, we first introduce two technical lemmas. The first lemma (Lemma \ref{lemma:LLN}) is a law-of-large-numbers-type result on the spectral-assisted log-likelihood function \eqref{eqn:spectral_assisted_loglik} at the true value $(\bkappa^{-1}(\bz_{0i}), \btheta_{0i})$, and the second lemma (Lemma \ref{lemma:CLT}) states the first-order expansion of the score function associated with \eqref{eqn:spectral_assisted_loglik}. 

\begin{lemma}[Law of Large Numbers]
\label{lemma:LLN}
Suppose Assumptions \ref{assumption:identifiability}--\ref{assumption:likelihood} hold. Further assume that $\psi(\cdot)$ is a continuous function defined on $\calI$. Then, for the permutation matrix $\bK\in\mathbb{O}(d)$ in Lemma \ref{lemma:Y_expansion}, for each $t\in[m]$ and $i\in[n]$, 
\begin{align*}
&\bigg\|\frac{1}{n}\sum_{j = 1}^n\psi(\theta_{0i}^{(t)}\bz_{0i}\transpose\bK\widetilde{\by}_{j}^{(t)})(A_{ij}^{(t)} - \theta_{0i}^{(t)}\bz_{0i}\transpose\bK\widetilde{\by}_{j}^{(t)})\bK\widetilde{\by}_{j}^{(t)}\bigg\|_2\\
&\quad + \bigg\|\frac{1}{n}\sum_{j = 1}^n\{\psi(\theta_{0i}^{(t)}\bz_{0i}\transpose\bK\transpose\widetilde{\by}_{j}^{(t)})(A_{ij}^{(t)} - \theta_{0i}^{(t)}\bz_{0i}\transpose\bK\widetilde{\by}_{j}^{(t)}) - \eta'(\theta_{0i}^{(t)}\bz_{0i}\transpose\bK\widetilde{\by}_{j}^{(t)})\}\bK\widetilde{\by}_{j}^{(t)}\widetilde{\by}_{j}^{(t)\mathrm{T}}\bK\transpose + \bG_{0in}^{(t)}\bigg\|_2\\
&\quad = \Optilde\bigg\{\frac{(\log n)^{\xi}}{\sqrt{n}}\bigg\}.
\end{align*}
\end{lemma}

\begin{proof}[\bf Proof]
By Lemma \ref{lemma:Y_expansion}, there exists a permutation matrix $\bK\in\mathbb{O}(d)$, such that 
\begin{align}\label{eqn:Pij_concentration}
\|\bx_{0i}^{(t)\mathrm{T}}\bK\widetilde{\by}_{j}^{(t)} - P_{0ij}^{(t)}| = |\bx_{0i}^{(t)\mathrm{T}}\bK\widetilde{\by}_{j}^{(t)} - \bx_{0i}^{(t)\mathrm{T}}\by_{0ij}^{(t)}|\leq \|\bx_{0i}^{(t)}\|_2\|\bK\widetilde{\by}_{j}^{(t)} - \by_{0j}^{(t)}\|_2 = \Optilde\bigg\{\frac{(\log n)^\xi}{\sqrt{n}}\bigg\}.
\end{align}
By Assumption \ref{assumption:likelihood} and the mean-value theorem, for any continuously differentiable function $\psi(\cdot)$ over $\calI$, we have
\begin{align}
\label{eqn:psi_concentration}
\psi(\bx_{0i}^{(t)\mathrm{T}}\bK\widetilde{\by}_{j}^{(t)}) - \psi(P_{0ij}^{(t)}) = \Optilde\bigg\{\frac{(\log n)^{\xi}}{\sqrt{n}}\bigg\}. 
\end{align}
This implies that
\begin{equation}
\label{eqn:psi_y_concentration}
\begin{aligned}
\left\|\psi(\bx_{0i}^{(t)\mathrm{T}}\bK\widetilde{\by}_{j}^{(t)})\bK\widetilde{\by}_{j}^{(t)} - \psi(P_{0ij})\by_{0j}^{(t)}\right\|_2
&\leq |\psi(\bx_{0i}^{(t)\mathrm{T}}\bK\widetilde{\by}_{j}^{(t)}) - \psi(P_{0ij})|\|\widetilde{\by}_{j}^{(t)}\|_2 + \psi(P_{0ij})\|\bK\widetilde{\by}_{j}^{(t)} - \by_{0j}^{(t)}\|_2\\ 
& = \Optilde\bigg\{\frac{(\log n)^\xi}{\sqrt{n}}\bigg\}. 
\end{aligned}
\end{equation}
By Lemma 3.7 in \cite{10.1214/11-AOP734} applied to $(|E_{ij}^{(t)}| - \expect|E_{ij}^{(t)}|)_{j = 1}^n$, we know that 
\begin{align}
\label{eqn:noise_absolute_bound}
\frac{1}{n}\sum_{j = 1}^n|E_{ij}^{(t)}| = \frac{1}{n}\sum_{j = 1}^n\expect|E_{ij}^{(t)}| + \frac{1}{n}\sum_{j = 1}^n(|E_{ij}^{(t)}| - \expect|E_{ij}^{(t)}|) = \Optilde(1).
\end{align}
It follows that
\begin{align*}
&
\bigg\|\frac{1}{n}\sum_{j = 1}^n\psi(\theta_{0i}^{(t)}\bz_{0i}\transpose\bK\widetilde{\by}_{j}^{(t)})(A_{ij}^{(t)} - \theta_{0i}^{(t)}\bz_{0i}\transpose\bK\widetilde{\by}_{j}^{(t)})\bK\widetilde{\by}_{j}^{(t)}\bigg\|_2\\ 
&\quad\leq \bigg\|\frac{1}{n}\sum_{j = 1}^n\{\psi(\theta_{0i}^{(t)}\bz_{0i}\transpose\bK\widetilde{\by}_{j}^{(t)})\bK\widetilde{\by}_{j}^{(t)} - \psi(P_{0ij})\by_{0j}^{(t)}\}E_{ij}^{(t)}\bigg\|_2 + \bigg\|\frac{1}{n}\sum_{j = 1}^nE_{ij}^{(t)}\psi(P_{0ij}^{(t)})\by_{0j}^{(t)}\bigg\|_2\\ 
&\qquad + \bigg\|\frac{1}{n}\sum_{j = 1}^n\psi(\theta_{0i}^{(t)}\bz_{0i}\transpose\bK\widetilde{\by}_{j}^{(t)})\bz_{0i}\transpose(\by_{0j}^{(t)} - \bK\widetilde{\by}_{j}^{(t)})\bK\widetilde{\by}_{j}^{(t)}\bigg\|_2\\
&\quad\leq \bigg(\frac{1}{n}\sum_{j = 1}^n|E_{ij}^{(t)}|\bigg)\max_{i,j\in[n]}\left\|\psi(\theta_{0i}^{(t)}\bz_{0i}\transpose\bK\widetilde{\by}_{j}^{(t)})\bK\widetilde{\by}_{j}^{(t)} - \psi(P_{0ij})\by_{0j}^{(t)}\right\|_2 + \Optilde\bigg\{\frac{(\log n)^\xi}{\sqrt{n}}\bigg\}\\ 
&\qquad + \max_{i,j\in[n]}\bigg[\psi(P_{0ij}) + \Optilde\bigg\{\frac{(\log n)^\xi}{\sqrt{n}}\bigg\}\bigg]\|\theta_{0i}^{(t)}\bz_{0i}\|_2\|\bK\widetilde{\by}_{j}^{(t)} - \by_{0j}^{(t)}\|_2\|\widetilde{\by}_{j}^{(t)}\|_2 = \Optilde\bigg\{\frac{(\log n)^{\xi}}{\sqrt{n}}\bigg\}. 
\end{align*}
For the second assertion, by \eqref{eqn:psi_concentration}, \eqref{eqn:psi_y_concentration}, \eqref{eqn:noise_absolute_bound}, and Lemma \ref{lemma:Y_expansion}, we have
\begin{align*}
&\bigg\|\frac{1}{n}\sum_{j = 1}^n\{\psi(\theta_{0i}^{(t)}\bz_{0i}\transpose\bK\transpose\widetilde{\by}_{j}^{(t)})(A_{ij}^{(t)} - \theta_{0i}^{(t)}\bz_{0i}\transpose\bK\widetilde{\by}_{j}^{(t)}) - \eta'(\theta_{0i}^{(t)}\bz_{0i}\transpose\bK\widetilde{\by}_{j}^{(t)})\}\bK\widetilde{\by}_{j}^{(t)}\widetilde{\by}_{j}^{(t)\mathrm{T}}\bK\transpose + \bG_{0in}^{(t)}\bigg\|_2\\
&\quad\leq \bigg\|\frac{1}{n}\sum_{j = 1}^n\psi(\theta_{0i}^{(t)}\bz_{0i}\transpose\bK\transpose\widetilde{\by}_{j}^{(t)})(A_{ij}^{(t)} - \theta_{0i}^{(t)}\bz_{0i}\transpose\bK\widetilde{\by}_{j}^{(t)})\widetilde{\by}_{j}^{(t)}\widetilde{\by}_{j}^{(t)\mathrm{T}}\bigg\|_2\\ 
&\qquad + \bigg\|\frac{1}{n}\sum_{j = 1}^n\eta'(\theta_{0i}^{(t)}\bz_{0i}\transpose\bK\widetilde{\by}_{j}^{(t)})\bK\widetilde{\by}_{j}^{(t)}\widetilde{\by}_{j}^{(t)\mathrm{T}}\bK\transpose + \bG_{0in}^{(t)}\bigg\|_2\\
&\quad\leq 
\frac{1}{n}\sum_{j = 1}^n\|\psi(\theta_{0i}^{(t)}\bz_{0i}\transpose\bK\widetilde{\by}_{j}^{(t)})\bK\widetilde{\by}_{j}^{(t)} - \psi(P_{0ij})\by_{0j}^{(t)}\|_2\|\widetilde{\by}_{j}^{(t)\mathrm{T}}\bK\transpose\|_2 | E_{ij}^{(t)}|
\\
&\qquad + \frac{1}{n}\sum_{j = 1}^n|\psi(P_{0ij})| \|\by_{0j}^{(t)}\|_2\|\bK\widetilde{\by}_{j}^{(t)} - \by_{0j}^{(t)}\|_2|E_{ij}^{(t)}|
+ \bigg\|\frac{1}{n}\sum_{j = 1}^nE_{ij}^{(t)}\psi(P_{0ij})\by_{0j}^{(t)}\by_{0j}^{(t)\mathrm{T}}\bigg\|_2
\\
&\qquad + \bigg\|\frac{1}{n}\sum_{j = 1}^n\psi(\theta_{0i}^{(t)}\bz_{0i}\transpose\bK\widetilde{\by}_{j}^{(t)})\theta_{0i}^{(t)}\bz_{0i}\transpose(\by_{0j}^{(t)} - \bK\widetilde{\by}_{j}^{(t)})\widetilde{\by}_{j}^{(t)}\widetilde{\by}_{j}^{(t)\mathrm{T}}\bigg\|_2
\\
&\qquad + \frac{1}{n}\sum_{j = 1}^n|\eta'(\theta_{0i}^{(t)}\bz_{0i}\transpose\bK\widetilde{\by}_{j}^{(t)}) - \eta'(P_{0ij})|\bigg\|\widetilde{\by}_{j}^{(t)}\widetilde{\by}_{j}^{(t)\mathrm{T}}\bigg\|_2 + \frac{1}{n}\sum_{j = 1}^n|\eta'(P_{0ij})|\bigg\|(\bK\widetilde{\by}_{j}^{(t)} - \by_{0j}^{(t)})\widetilde{\by}_{j}^{(t)\mathrm{T}}\bK\transpose\bigg\|_2\\ 
&\qquad + \frac{1}{n}\sum_{j = 1}^n|\eta'(P_{0ij}^{(t)})|\bigg\|
\by_{0j}^{(t)}(\bK\widetilde{\by}_{j}^{(t)} - \by_{0j}^{(t)})\transpose
\bigg\|_2\\ 
&\quad = \Optilde\bigg\{\frac{(\log n)^{\xi}}{\sqrt{n}}\bigg\}. 
\end{align*}
The proof is thus completed. 
\end{proof}

\begin{lemma}[Central Limit Theorem]
\label{lemma:CLT}
Suppose Assumptions \ref{assumption:identifiability}--\ref{assumption:likelihood} hold. Further assume $m\lesssim n^\alpha$ for some $\alpha > 0$. Then, for the permutation matrix $\bK$ in Lemma \ref{lemma:Y_expansion}, for each $t\in[m]$ and $i\in[n]$, we have
\begin{align*}
\bigg\|\frac{1}{n}\sum_{j = 1}^n\eta'(\theta_{0i}^{(t)}\bz_{0i}\transpose\bK\widetilde{\by}_{j}^{(t)})(A_{ij}^{(t)} - \theta_{0i}^{(t)}\bz_{0i}\transpose\bK\widetilde{\by}_{j}^{(t)})\bK\widetilde{\by}_{j}^{(t)} - \frac{1}{n}\sum_{j = 1}^nE_{ij}^{(t)}\eta'(P_{0ij}^{(t)})\by_{0j}^{(t)}\bigg\|_2& = \Optilde\bigg\{\frac{(\log n)^{2\xi}}{n}\bigg\}.
\end{align*}
\end{lemma}

\begin{proof}[\bf Proof]
By definition, we have
\begin{align*}
&\eta'(\theta_{0i}^{(t)}\bz_{0i}\transpose\bK\widetilde{\by}_{j}^{(t)})(A_{ij}^{(t)} - \theta_{0i}^{(t)}\bz_{0i}\transpose\bK\widetilde{\by}_{j}^{(t)})\bK\widetilde{\by}_{j}^{(t)} \\
&\quad = \eta'(\theta_{0i}^{(t)}\bz_{0i}\transpose\bK\widetilde{\by}_{j}^{(t)})(E_{ij}^{(t)} + \theta_{0i}^{(t)}\bz_{0i}\transpose\by_{0j}^{(t)} - \theta_{0i}^{(t)}\bz_{0i}\transpose\bK\widetilde{\by}_{j}^{(t)})\bK\widetilde{\by}_{j}^{(t)}\\ 
&\quad = E_{ij}^{(t)}\eta'(P_{0ij})\by_{0j}^{(t)} + E_{ij}^{(t)}\{\eta'(\theta_{0i}^{(t)}\bz_{0i}\transpose\bK\widetilde{\by}_{j}^{(t)})\bK\widetilde{\by}_{j}^{(t)} - \eta'(P_{0ij}^{(t)})\by_{0j}^{(t)}\}\\ 
&\qquad + \theta_{0i}^{(t)}\bz_{0i}\transpose(\by_{0j}^{(t)} - \bK\widetilde{\by}_{j}^{(t)})\eta'(P_{0ij}^{(t)})\by_{0j}^{(t)}\\ 
&\qquad + \theta_{0i}^{(t)}\bz_{0i}\transpose(\by_{0j}^{(t)} - \bK\widetilde{\by}_{j}^{(t)})
\{\eta'(\theta_{0i}^{(t)}\bz_{0i}\transpose\bK\widetilde{\by}_{j}^{(t)})\bK\widetilde{\by}_{j}^{(t)} - \eta'(P_{0ij}^{(t)})\by_{0j}^{(t)}\}.
\end{align*}
Namely,
\begin{align}
&\bigg\|\frac{1}{n}\sum_{j = 1}^n\eta'(\theta_{0i}^{(t)}\bz_{0i}\transpose\bK\widetilde{\by}_{j}^{(t)})(A_{ij}^{(t)} - \theta_{0i}^{(t)}\bz_{0i}\transpose\bK\widetilde{\by}_{j}^{(t)})\bK\widetilde{\by}_{j}^{(t)} - \frac{1}{n}\sum_{j = 1}^nE_{ij}^{(t)}\eta'(P_{0ij}^{(t)})\by_{0j}^{(t)}\bigg\|_2\nonumber\\ 
&\label{eqn:CLT_term1}
\quad 
\leq \bigg\|\frac{1}{n}\sum_{j = 1}^nE_{ij}^{(t)}\{\eta'(\theta_{0i}^{(t)}\bz_{0i}\transpose\bK\widetilde{\by}_{j}^{(t)})\bK\widetilde{\by}_{j}^{(t)} - \eta'(P_{0ij}^{(t)})\by_{0j}^{(t)}\}\bigg\|_2\\ 
&\label{eqn:CLT_term2}
\qquad + \bigg\|\frac{1}{n}\sum_{j = 1}^n\eta'(P_{0ij}^{(t)})\by_{0j}^{(t)}\theta_{0i}^{(t)}\bz_{0i}\transpose(\bK\widetilde{\by}_{j}^{(t)} - \by_{0j}^{(t)})\bigg\|_2\\ 
&\label{eqn:CLT_term3}
\qquad + \bigg\|\frac{1}{n}\sum_{j = 1}^n\theta_{0i}^{(t)}\bz_{0i}\transpose(\by_{0j}^{(t)} - \bK\widetilde{\by}_{j}^{(t)})
\{\eta'(\theta_{0i}^{(t)}\bz_{0i}\transpose\bK\widetilde{\by}_{j}^{(t)})\bK\widetilde{\by}_{j}^{(t)} - \eta'(P_{0ij}^{(t)})\by_{0j}^{(t)}\}\bigg\|_2.
\end{align}
For the term on line \eqref{eqn:CLT_term3}, by Lemma \ref{lemma:Y_expansion} and \eqref{eqn:psi_y_concentration}, we have
\begin{align*}
&\bigg\|\frac{1}{n}\sum_{j = 1}^n\theta_{0i}^{(t)}\bz_{0i}\transpose(\by_{0j} - \bK\widetilde{\by}_{j}^{(t)})
\{\eta'(\theta_{0i}^{(t)}\bz_{0i}\transpose\bK\widetilde{\by}_{j}^{(t)})\bK\widetilde{\by}_{j}^{(t)} - \eta'(P_{0ij}^{(t)})\by_{0j}^{(t)}\}\bigg\|_2\\ 
&\quad\leq \max_{i,j\in[n]}\|\theta_{0i}^{(t)}\bz_{0i}\|_2\|\bK\widetilde{\by}_{j}^{(t)} - \by_{0j}^{(t)}\|_2\|\eta'(\theta_{0i}^{(t)}\bz_{0i}\transpose\bK\widetilde{\by}_{j}^{(t)})\bK\widetilde{\by}_{j}^{(t)} - \eta'(P_{0ij}^{(t)})\by_{0j}^{(t)}\|_2 = \Optilde\bigg\{\frac{(\log n)^{2\xi}}{n}\bigg\}. 
\end{align*}
For the term on line \eqref{eqn:CLT_term2}, by Lemma \ref{lemma:Y_expansion} and Lemma 3.7 in \cite{10.1214/11-AOP734}, we have
\begin{align*}
&\bigg\|\frac{1}{n}\sum_{j = 1}^n\eta'(P_{0ij}^{(t)})\by_{0j}^{(t)}\theta_{0i}^{(t)}\bz_{0i}\transpose(\bK\widetilde{\by}_{j}^{(t)} - \by_{0j}^{(t)})\bigg\|_2\\ 
&\quad\leq \bigg\|\frac{1}{n}\sum_{j = 1}^n\eta'(P_{0ij}^{(t)})\by_{0j}^{(t)}\theta_{0i}^{(t)}\bz_{0i}\transpose\sum_{a = 1}^nE_{ja}^{(t)}\bgamma_{ja}^{(t)}\bigg\|_2 
 + \bigg\|\frac{1}{mn}\sum_{j = 1}^n\eta'(P_{0ij}^{(t)})\by_{0j}^{(t)}\theta_{0i}^{(t)}\bz_{0i}\transpose\sum_{s = 1}^m\sum_{a = 1}^nE_{ja}^{(s)}\bdelta_{ja}^{(s, t)}\bigg\|_2\\ 
&\qquad + \Optilde\bigg\{\frac{(\log n)^{2\xi}}{n}\bigg\}\\ 
&\quad = \Optilde\bigg\{n(\log n)^\xi\max_{j,a\in[n]}\frac{1}{n}\|\eta'(P_{0ij}^{(t)})\by_{0j}^{(t)}\theta_{0i}^{(t)}\bz_{0i}\transpose\bgamma_{ja}^{(t)}\|_2\bigg\}\\ 
&\qquad + \Optilde\bigg\{(mn^2)^{1/2}(\log n)^\xi\max_{j,a\in[n],s\in[m]}\frac{1}{mn}\|\eta'(P_{0ij}^{(t)})\by_{0j}^{(t)}\theta_{0i}^{(t)}\bz_{0i}\transpose\bdelta_{ja}^{(s, t)}\|_2\bigg\} + \Optilde\bigg\{\frac{(\log n)^{2\xi}}{n}\bigg\}\\ 
&\quad = \Optilde\bigg\{\frac{(\log n)^\xi}{n}\bigg\}. 
\end{align*}
We now focus on the term on line \eqref{eqn:CLT_term1}. 
By \eqref{eqn:Pij_concentration} and Taylor's theorem, there exists some $\bar{P}_{0ij}^{(t)}$ between $\theta_{0i}^{(t)}\bz_{0i}\transpose\bK\widetilde{\by}_{j}^{(t)}$ and $P_{0ij}^{(t)}$, such that $|\eta'''(\bar{P}_{0ij})|\lesssim 1$ and
\begin{align*}
&\eta'(\theta_{0i}^{(t)}\bz_{0i}\transpose\bK\widetilde{\by}_{j}^{(t)})\bK\widetilde{\by}_{j}^{(t)} - \eta'(P_{0ij}^{(t)})\by_{0j}^{(t)}\\
&\quad = \{\eta'(\theta_{0i}^{(t)}\bz_{0i}\transpose\bK\widetilde{\by}_{j}^{(t)}) - \eta'(P_{0ij}^{(t)})\}\by_{0j}^{(t)}
 + \{\eta'(\theta_{0i}^{(t)}\bz_{0i}\transpose\bK\widetilde{\by}_{j}^{(t)}) - \eta'(P_{0ij}^{(t)})\}(\bK\widetilde{\by}_{j}^{(t)} - \by_{0j}^{(t)})
 + \eta'(P_{0ij}^{(t)})(\bK\widetilde{\by}_{j}^{(t)} - \by_{0j}^{(t)})\\ 
&\quad = \eta''(P_{0ij}^{(t)})\theta_{0i}^{(t)}\bz_{0i}\transpose(\bK\widetilde{\by}_{j}^{(t)} - \by_{0j}^{(t)})\by_{0j}^{(t)} + \eta'''(\bar{P}_{0ij}^{(t)})(\bK\widetilde{\by}_{j}^{(t)} - \by_{0j}^{(t)})\transpose\theta_{0i}^{2(t)}\bz_{0i}\bz_{0i}\transpose(\bK\widetilde{\by}_{j}^{(t)} - \by_{0j}^{(t)})\by_{0j}^{(t)}
\\ &\qquad
 + \Optilde\bigg\{\frac{(\log n)^{2\xi}}{n}\bigg\} + \eta'(P_{0ij}^{(t)})(\bK\widetilde{\by}_{j}^{(t)} - \by_{0j}^{(t)})\\ 
&\quad = \{\eta''(P_{0ij})\by_{0j}^{(t)}\theta_{0i}^{(t)}\bz_{0i}\transpose + \eta'(P_{0ij}^{(t)})\eye_d\}(\bK\widetilde{\by}_{j}^{(t)} - \by_{0j}^{(t)}) + \Optilde\bigg\{\frac{(\log n)^{2\xi}}{n}\bigg\}\\ 
&\quad = \{\eta''(P_{0ij})\by_{0j}^{(t)}\theta_{0i}^{(t)}\bz_{0i}\transpose + \eta'(P_{0ij}^{(t)})\eye_d\}\bigg(\sum_{a = 1}^nE_{ja}^{(t)}\bgamma_{ja}^{(t)} + \frac{1}{m}\sum_{s = 1}^m\sum_{a = 1}^nE_{ja}^{(s)}\bdelta_{ja}^{(s, t)}\bigg) + \Optilde\bigg\{\frac{(\log n)^{2\xi}}{n}\bigg\}. 
\end{align*}
Then, by Lemma 3.7 in \cite{10.1214/11-AOP734} and Result 3 in \cite{XieWu2024}, the term on line \eqref{eqn:CLT_term1} satisfies
\begin{align*}
&\bigg\|\frac{1}{n}\sum_{j = 1}^nE_{ij}^{(t)}\{\eta'(\theta_{0i}^{(t)}\bz_{0i}\transpose\bK\widetilde{\by}_{j}^{(t)})\bK\widetilde{\by}_{j}^{(t)} - \eta'(P_{0ij}^{(t)})\by_{0j}^{(t)}\}\bigg\|_2\\ 
&\quad\leq \bigg\|\frac{1}{n}\sum_{j = 1}^n\sum_{a = 1}^nE_{ij}^{(t)}E_{ja}^{(t)}\{\eta''(P_{0ij})\by_{0j}^{(t)}\theta_{0i}^{(t)}\bz_{0i}\transpose + \eta'(P_{0ij}^{(t)})\eye_d\}\bigg(\bgamma_{ja}^{(t)} + \frac{1}{m}\bdelta_{ja}^{(t, t)}\bigg)\bigg\|_2\\ 
&\qquad + \bigg\|\frac{1}{n}\sum_{j = 1}^nE_{ij}^{(t)}\sum_{s\in[m]\backslash\{t\}}\sum_{a = 1}^n\{\eta''(P_{0ij})\by_{0j}^{(t)}\theta_{0i}^{(t)}\bz_{0i}\transpose + \eta'(P_{0ij}^{(t)})\eye_d\}\bdelta_{ja}^{(s, t)}E_{ja}^{(s)}\bigg\|_2\\ 
&\quad = \Optilde\bigg\{(\log n)^{2\xi}\max_{j,a\in[n]}\bigg\|\{\eta''(P_{0ij})\by_{0j}^{(t)}\theta_{0i}^{(t)}\bz_{0i}\transpose + \eta'(P_{0ij}^{(t)})\eye_d\}\bigg(\bgamma_{ja}^{(t)} + \frac{1}{m}\bdelta_{ja}^{(t, t)}\bigg)\bigg\|_2\bigg\}\\ 
&\qquad + \Optilde\bigg\{\frac{(\log n)^\xi}{\sqrt{n}}\max_{j\in[n]}\bigg\|\frac{1}{m}\sum_{s\in[m]\backslash\{t\}}\sum_{a = 1}^n\{\eta''(P_{0ij})\by_{0j}^{(t)}\theta_{0i}^{(t)}\bz_{0i}\transpose + \eta'(P_{0ij}^{(t)})\eye_d\}\bdelta_{ja}^{(s, t)}E_{ja}^{(s)}\bigg\|_2\bigg\}\\ 
&\quad = \Optilde\bigg\{\frac{(\log n)^{2\xi}}{n}\bigg\} + \Optilde\bigg\{\frac{(\log n)^\xi}{\sqrt{n}}\times\frac{(\log n)^\xi}{\sqrt{mn}}\bigg\} = \Optilde\bigg\{\frac{(\log n)^{2\xi}}{n}\bigg\}.
\end{align*}
The proof is thus completed.
\end{proof}
We are now in a position to prove Theorem \ref{thm:MLE}. Recall that in Assumption \ref{assumption:identifiability}, the requirement that $\min\calN_1 < \ldots < \min\calN_d$ in item (b) is a labeling convention. In practice, this convention allows us to implement the following label alignment algorithm.
\begin{breakablealgorithm}
\caption{Community Label Alignment}
\label{alg:label_alignment}
\begin{algorithmic}[1]
\State \textbf{Input:} Preliminary estimators $\widetilde{\bZ} = [\widetilde{\bz}_1,\ldots,\widetilde{\bz}_n]\transpose$, $(\widetilde{\bB}^{(t)})_{t = 1}^m$.

\State Set tuning parameter $\tau_n$ (default $\tau_n = n^{-1/3}$).
\For{$k = 1,2,\ldots,d$}
    \State Compute $\widetilde{\calN}_k = \{i\in[n]:\|\widetilde{\bz}_i - \be_k\|_2 < \tau_n\}$. 
    \State Compute $i_k = \min\calN_k$
\EndFor
\State Sort $\{i_1,\ldots,i_d\}$ such that $i_{k_1} < \ldots < i_{k_d}$ and $\{i_1,\ldots,i_d\} = \{i_{k_1},\ldots,i_{k_d}\}$.
\State Compute $\widetilde{\bZ} \longleftarrow \widetilde{\bZ}[\be_{i_{k_1}},\ldots,\be_{i_{k_d}}]$.
\For{$t = 1,2,\ldots,m$}
    \State Compute $\widetilde{\bB}^{(t)} \longleftarrow [\be_{i_{k_1}},\ldots,\be_{i_{k_d}}]\transpose\widetilde{\bB}^{(t)}[\be_{i_{k_1}},\ldots,\be_{i_{k_d}}]$. 
\EndFor
\State \textbf{Output: } Permutation aligned preliminary estimators $\widetilde{\bZ}$, $(\widetilde{\bB}^{(t)})_{t = 1}^m$
\end{algorithmic}
\end{breakablealgorithm}
Note that when the preliminary estimators are permutation aligned according to Algorithm \ref{alg:label_alignment}, one can directly take $\bK$ in Lemma \ref{lemma:Z_expansion}, Lemma \ref{lemma:Y_expansion}, Lemma \ref{lemma:LLN}, and Lemma \ref{lemma:CLT} as the identity matrix. For the sake of technical preparation for proofs of Theorems \ref{thm:BvM} and \ref{thm:variational_BvM}, we state and proof the following lemma that automatically proves Theorem \ref{thm:MLE}. 
\begin{lemma}\label{lemma:aggregated_MLE_theory}
Under the conditions of Theorem \ref{thm:MLE}, we have:
\begin{enumerate}[(i)]
    \item $\widehat{\bz}_i$ is in the interior of $\calS^{d - 1}$ w.h.p. and satisfies the likelihood equation 
    \[
        \frac{\partial\widetilde{\ell}_{in}}{\partial\bz_i^*}(\bz_i^*, \btheta_i)\mathrel{\Bigg|_{\bz_i^* = \widehat{\bz}_i^*, \btheta_i = \widehat{\btheta}_i}} = \zero_{d - 1},\quad
        \frac{\partial\widetilde{\ell}_{in}}{\partial\btheta_i}(\bz_i^*, \btheta_i)\mathrel{\Bigg|_{\bz_i^* = \widehat{\bz}_i^*, \btheta_i = \widehat{\btheta}_i}} = \zero_m.
    \]

    \item $\widehat{\bz}_i^*$ and $\widehat{\btheta}_i$ satisfies the following joint expansion:
    \begin{align*}
    \begin{bmatrix}
    \widehat{\bz}_i^* - \bz_{0i}^*\\
    \widehat{\btheta}_i - \btheta_{0i}
    \end{bmatrix}
    & = \frac{1}{mn}\bGamma_{in}^{-1}
    \begin{bmatrix}
    \sum_{t = 1}^m\sum_{j = 1}^nE_{ij}^{(t)}\theta_{0i}^{2(t)}\eta'(P_{0ij}^{(t)})\bJ\transpose\by_{0j}^{(t)}\\
    \sum_{j = 1}^nE_{ij}^{(1)}\eta'(P_{0ij}^{(1)})\by_{0j}^{(1)}\\
    \vdots \\
    \sum_{j = 1}^nE_{ij}^{(m)}\eta'(P_{0ij}^{(m)})\by_{0j}^{(m)}
    \end{bmatrix}
     + \begin{bmatrix}\Optilde\left\{\frac{(\log n)^{3\xi}}{n}\right\}\\ 
    \Optilde\left\{\frac{(\log n)^{3\xi}}{n}\right\}\\
    \vdots \\
    \Optilde\left\{\frac{(\log n)^{3\xi}}{n}\right\}
    \end{bmatrix}.
    \end{align*}
    In particular, we have 
    \[
    \|\widehat{\btheta}_i - \btheta_{0i}\|_\infty + \|\widehat{\bz}_i^* - \bz_{0i}^*\|_2 = \Optilde\bigg\{\frac{(\log n)^{\xi}}{\sqrt{mn}} + \frac{(\log n)^{3\xi}}{n}\bigg\}.
    \] 
    and
    \begin{align*}
    \widehat{\bz}_i^* - \bz_{0i}^*
    & = \bigg(\frac{1}{m}\sum_{t = 1}^m\theta_{0i}^{2(t)}\bJ\transpose\bDelta_{in}^{(t)}\bJ\bigg)^{-1}\frac{1}{mn}\sum_{t = 1}^m\sum_{j = 1}^nE_{ij}^{(t)}\theta_{0i}^{(t)}\eta'(P_{0ij}^{(t)})\bJ\transpose\bigg(\eye_d - \frac{\bG_{0in}^{(t)}\bz_{0i}\bz_{0i}\transpose}{\bz_{0i}\transpose\bG_{0in}^{(t)}\bz_{0i}}\bigg)\by_{0j}\\
    &\quad + \Optilde\bigg\{\frac{(\log n)^{3\xi}}{n}\bigg\}.
    \end{align*}
\end{enumerate}
\end{lemma}

\begin{proof}[\bf Proof]
Since the eigenvalues of $\bJ\transpose\bDelta_{in}^{(t)}\bJ:=\bJ\transpose(\bG_{0in}^{(t)} - \bG_{0in}^{(t)}\bz_{0i}\bz_{0i}\transpose\bG_{0in}^{(t)}/\bz_{0i}\transpose\bG_{0in}^{(t)}\bz_{0i})\bJ$ are bounded away from $0$ and $\infty$, by the block matrix inverse formula, the matrix
\begin{align*}
\bGamma_{in}^{-1}
& =\begin{bmatrix}
\frac{1}{m}\sum_{t = 1}^m\theta_{0i}^{2(t)}\bJ\transpose \bG_{0in}^{(t)}\bJ & \frac{1}{m}\bJ\transpose\bG_{0in}^{(1)}\bz_{0i}\theta_{0i}^{(1)} & \ldots & \frac{1}{m}\bJ\transpose\bG_{0in}^{(m)}\bz_{0i}\theta_{0i}^{(m)}\\ 
\frac{1}{m}\theta_{0i}^{(1)}\bz_{0i}\transpose\bG_{0in}^{(1)}\bJ & \frac{1}{m}\bz_{0i}\transpose\bG_{0in}^{(1)}\bz_{0i} & \ldots & 0\\
\vdots & \vdots & \ddots & \vdots \\
\frac{1}{m}\theta_{0i}^{(m)}\bz_{0i}\transpose\bG_{0in}^{(m)}\bJ & 0 & \ldots & \frac{1}{m}\bz_{0i}\transpose\bG_{0in}^{(m)}\bz_{0i}
\end{bmatrix}^{-1}\\ 
& = \begin{bmatrix}
(\frac{1}{m}\sum_{t = 1}^m\theta_{0i}^{2(t)}\bJ\transpose\bDelta_{in}^{(t)}\bJ)^{-1} & 
-(\frac{1}{m}\sum_{t = 1}^m\theta_{0i}^{2(t)}\bJ\transpose\bDelta_{in}^{(t)}\bJ)^{-1}\bJ\transpose\bM_{in}\transpose \\
-\bM_{in}\bJ(\frac{1}{m}\sum_{t = 1}^m\theta_{0i}^{2(t)}\bJ\transpose\bDelta_{in}^{(t)}\bJ)^{-1}
&
\bL_{in}^{-1} + \bM_{in}\bJ(\frac{1}{m}\sum_{t = 1}^m\theta_{0i}^{2(t)}\bJ\transpose\bDelta_{in}^{(t)}\bJ)^{-1}\bJ\transpose
\bM_{in}\transpose
\end{bmatrix}
\end{align*}
also has eigenvalues bounded away from $0$ and $\infty$, where
\[
  \bM_{in} = \begin{bmatrix} \frac{\bG_{0in}^{(1)}\bz_{0i}\theta_{0i}^{(1)}}{\bz_{0i}\transpose\bG_{0in}^{(1)}\bz_{0i}} & \ldots & \frac{\bG_{0in}^{(m)}\bz_{0i}\theta_{0i}^{(m)}}{\bz_{0i}\transpose\bG_{0in}^{(m)}\bz_{0i}}\end{bmatrix}\transpose,\quad\bL_{in} = \textsf{diag}\bigg(\frac{1}{m}\bz_{0i}\transpose\bG_{0in}^{(t)}\bz_{0i}:t\in[m]\bigg).
\]
Denote by $\bdelta_{\bz_i} = \bz_i^* - \bz_{0i}^*$, $\bdelta_{\btheta_i} = \btheta_i - \btheta_{0i}$, and let $\delta_{z_{ik}}$ be the $k$th element of $\bdelta_{\bz_i}$, $\delta_{\theta_i^{(t)}}$ be the $t$th element of $\bdelta_{\btheta_i}$. Let $(\eps_n)_{n = 1}^\infty$ be a sequence to be determined later that converges to zero as $n\to\infty$ and let 
\[
\calQ_n = \{(\bz^*, \btheta_i)\in\calS^{d - 1}\times\calI^m:\|\bz^* - \bz_{0i}^*\|_2 = \eps_n, |\theta_i^{(t)} - \theta_{0i}^{(t)}| = \eps_n\text{ for all }t\in[m]\}.
\] 
Since $\eta$ is three times continuously differentiable over $\theta\in\calI$, for any $k\in[d]$ and any $(\bar{\bz}_i, \bar{\btheta}_i)$ between $(\bz_{0i}, \btheta_{0i})$ and $(\bkappa(\bz_i^*), \btheta_i)$ where $\bz_i^*\in\calQ_n$, there exists some $\omega\in[0, 1]$, such that (with $\bar{\btheta}_i = [\bar{\theta}_i^{(1)},\ldots,\bar{\theta}_i^{(m)}]\transpose$)
\begin{align*}
&|\bar{\theta}_i^{(t)}\bar{\bz}_i\transpose\widetilde{\by}_{j}^{(t)} - P_{0ij}^{(t)}|\\
&\quad = |\{\omega\theta_{0i}^{(t)} + (1 - \omega)\theta_i^{(t)}\}\{\omega\bz_{0i} + (1 - \omega)\bkappa(\bz_i^*)\}\transpose\widetilde{\by}_j^{(t)} - \theta_{0i}^{(t)}\bz_{0i}\transpose\by_{0j}^{(t)}|\\
&\quad = |\omega^2\theta_{0i}^{(t)}\bz_{0i}\transpose\widetilde{\by}_j^{(t)} + \omega(1 - \omega)\theta_i^{(t)}\bz_{0i}\transpose\widetilde{\by}_j^{(t)} + \omega(1 - \omega)\theta_{0i}^{(t)}\bkappa(\bz_i^*)\transpose\widetilde{\by}_j^{(t)}
 + (1 - \omega)^2\bkappa(\bz_i^*)\transpose\widetilde{\by}_j^{(t)}   - \theta_{0i}^{(t)}\bz_{0i}\transpose\by_{0j}^{(t)}|\\
&\quad\leq \omega^2\theta_{0i}^{(t)}\|\bz_{0i}\transpose(\widetilde{\by}_{ij}^{(t)} - \by_{0ij}^{(t)})\|_2
 + \omega(1 - \omega)(|\theta_i^{(t)} - \theta_{0i}^{(t)}|\bz_{0i}\transpose\widetilde{\by}_j^{(t)} + \theta_{0i}^{(t)}\|\bz_{0i}\|_2\|\widetilde{\by}_j - \by_{0j}\|_2)\\ 
&\qquad + \omega(1 - \omega)\theta_{0i}(\|\bkappa(\bz_i^*) - \bz_{0i}\|_2\|\widetilde{\by}_j^{(t)}\|_2 + \|\bz_{0i}\|_2\|\widetilde{\by}_j^{(t)} - \by_{0j}^{(t)}\|_2)\\
&\qquad + (1 - \omega)^2(|\theta_i^{(t)} - \theta_{0i}^{(t)}||\bkappa(\bz_i^*)\transpose\widetilde{\by}_j| + \theta_{0i}\|\bkappa(\bz_i^*) - \bz_{0i}\|_2\|\widetilde{\by}_j^{(t)}\|_2 + \theta_{0i}^{(t)}\|\bz_{0i}\|_2\|\widetilde{\by}_j^{(t)} - \by_{0j}^{(t)}\|_2)\\ 
&\quad = \Optilde\bigg\{\frac{(\log n)^\xi}{\sqrt{n}} + \eps_n\bigg\},
\end{align*}
where we have applied the fact that $\bkappa$ is an affine transformation and Lemma \ref{lemma:Y_expansion}. 
This implies that $|\eta'''(\bar{\theta}_i^{(t)}\bar{\bz}_i\transpose\widetilde{\by}_{j}^{(t)})| = \Optilde(1)$. 

By definition, we have
\begin{align*}
&\frac{\partial^2}{\partial\bz_i^*\partial\bz_i^{*\mathrm{T}}}\sum_{t = 1}^m\sum_{j = 1}^n\log f(A_{ij}^{(t)}, \theta_i^{(t)}\bkappa(\bz_i^*)\transpose\widetilde{\by}_j^{(t)})\\
&\quad = \sum_{t = 1}^m\sum_{j = 1}^n\theta_i^{2(t)}\{\eta''(\theta_i^{(t)}\bkappa(\bz_i^*)\transpose\widetilde{\by}_j^{(t)})(A_{ij}^{(t)} - \theta_i^{(t)}\bkappa(\bz_i^*)\transpose\widetilde{\by}_j^{(t)}) - \eta'(\theta_i^{(t)}\bkappa(\bz_i^*)\transpose\widetilde{\by}_j^{(t)})\}\bJ\transpose\widetilde{\by}_j^{(t)}\widetilde{\by}_j^{(t)\mathrm{T}}\bJ,\\
&\frac{\partial^2}{\partial\btheta_i\partial\btheta_i\transpose}\sum_{t = 1}^m\sum_{j = 1}^n\log f(A_{ij}^{(t)}, \theta_i^{(t)}\bkappa(\bz_i^*)\transpose\widetilde{\by}_j^{(t)})\\
&\quad = \textsf{diag}\bigg(\sum_{j = 1}^n \{\eta''(\theta_i^{(t)}\bkappa(\bz_i^*)\transpose\widetilde{\by}_j^{(t)})(A_{ij}^{(t)} - \theta_i^{(t)}\bkappa(\bz_i^*)\transpose\widetilde{\by}_j^{(t)}) - \eta'(\theta_i^{(t)}\bkappa(\bz_i^*)\transpose\widetilde{\by}_j^{(t)})\}(\bkappa(\bz_i^*)\transpose\widetilde{\by}_j^{(t)})^2: t\in[m]\bigg),\\
&\frac{\partial^2}{\partial\bz_i^*\partial\btheta_i\transpose}\sum_{t = 1}^m\sum_{j = 1}^n\log f(A_{ij}^{(t)}, \theta_i^{(t)}\bkappa(\bz_i^*)\transpose\widetilde{\by}_j^{(t)})\\
&\quad = \bigg[\sum_{j = 1}^n \theta_i^{(t)}\{\eta''(\theta_i^{(t)}\bkappa(\bz_i^*)\transpose\widetilde{\by}_j^{(t)})(A_{ij}^{(t)} - \theta_i^{(t)}\bkappa(\bz_i^*)\transpose\widetilde{\by}_j^{(t)}) - \eta'(\theta_i^{(t)}\bkappa(\bz_i^*)\transpose\widetilde{\by}_j^{(t)})\}\bJ\transpose\widetilde{\by}_j^{(t)}\widetilde{\by}_j^{(t)\mathrm{T}}\bkappa(\bz_i^*)\\ 
&\qquad\quad - \sum_{j = 1}^n\eta'(\theta_i^{(t)}\bkappa(\bz_i^*)\transpose\widetilde{\by}_j)(A_{ij}^{(t)} - \theta_i^{(t)}\bkappa(\bz_i^*)\transpose\widetilde{\by}_j^{(t)})\bJ\transpose\widetilde{\by}_j^{(t)} : t\in[m]\bigg].
\end{align*}
Furthermore, for any $k\in[d - 1]$ and $s\in[m]$, we have
\begin{align*}
&\frac{\partial^3}{\partial z_{ik}^*\partial\bz_i^*\partial\bz_i^{*\mathrm{T}}}\sum_{t = 1}^m\sum_{j = 1}^n\log f(A_{ij}^{(t)}, \theta_i^{(t)}\bkappa(\bz_i^*)\transpose\widetilde{\by}_j^{(t)})\\
&\quad = \sum_{t = 1}^m\sum_{j = 1}^n\theta_i^{3(t)}\{\eta'''(\theta_i^{(t)}\bkappa(\bz_i^*)\transpose\widetilde{\by}_j^{(t)})(A_{ij}^{(t)} - \theta_i^{(t)}\bkappa(\bz_i^*)\transpose\widetilde{\by}_j^{(t)}) - 2\eta''(\theta_i^{(t)}\bkappa(\bz_i^*)\transpose\widetilde{\by}_j^{(t)})\}\bJ\transpose\widetilde{\by}_j^{(t)}\widetilde{\by}_j^{(t)\mathrm{T}}\bJ\\
&\qquad \times\frac{\partial}{\partial z_{ik}^*}\bkappa(\bz_i^*)\transpose\widetilde{\by}_j,\\
&\frac{\partial^3}{\partial \theta_{i}^{(s)}\partial\bz_i^*\partial\bz_i^{*\mathrm{T}}}\sum_{t = 1}^m\sum_{j = 1}^n\log f(A_{ij}^{(t)}, \theta_i^{(t)}\bkappa(\bz_i^*)\transpose\widetilde{\by}_j^{(t)})\\
&\quad = \sum_{j = 1}^n\theta_i^{2(s)}\{\eta'''(\theta_i^{(s)}\bkappa(\bz_i^*)\transpose\widetilde{\by}_j^{(s)})(A_{ij}^{(s)} - \theta_i^{(s)}\bkappa(\bz_i^*)\transpose\widetilde{\by}_j^{(s)}) - 2\eta''(\theta_i^{(s)}\bkappa(\bz_i^*)\transpose\widetilde{\by}_j^{(s)})\}\bJ\transpose\widetilde{\by}_j^{(s)}\widetilde{\by}_j^{(s)\mathrm{T}}\bJ\\
&\qquad \times \bkappa(\bz_i^*)\transpose\widetilde{\by}_j\\ 
&\qquad + \sum_{j = 1}^n2\theta_i^{(s)}\{\eta''(\theta_i^{(s)}\bkappa(\bz_i^*)\transpose\widetilde{\by}_j^{(s)})(A_{ij}^{(s)} - \theta_i^{(s)}\bkappa(\bz_i^*)\transpose\widetilde{\by}_j^{(s)}) - \eta'(\theta_i^{(s)}\bkappa(\bz_i^*)\transpose\widetilde{\by}_j^{(s)})\}\bJ\transpose\widetilde{\by}_j^{(s)}\widetilde{\by}_j^{(s)\mathrm{T}}\bJ
,\\
&\frac{\partial^3}{\partial\theta_i^{(s)}\partial\btheta_i\partial\btheta_i\transpose}\sum_{t = 1}^m\sum_{j = 1}^n\log f(A_{ij}^{(t)}, \theta_i^{(t)}\bkappa(\bz_i^*)\transpose\widetilde{\by}_j^{(t)})\\
&\quad = \sum_{j = 1}^n \{\eta'''(\theta_i^{(s)}\bkappa(\bz_i^*)\transpose\widetilde{\by}_j^{(s)})(A_{ij}^{(s)} - \theta_i^{(s)}\bkappa(\bz_i^*)\transpose\widetilde{\by}_j^{(s)}) - 2\eta''(\theta_i^{(s)}\bkappa(\bz_i^*)\transpose\widetilde{\by}_j^{(s)})\}(\bkappa(\bz_i^*)\transpose\widetilde{\by}_j^{(s)})^3\times \be_s\be_s\transpose,\\
&\frac{\partial^3}{\partial z_{ik}^*\partial\btheta_i\partial\btheta_i\transpose}\sum_{t = 1}^m\sum_{j = 1}^n\log f(A_{ij}^{(t)}, \theta_i^{(t)}\bkappa(\bz_i^*)\transpose\widetilde{\by}_j^{(t)})\\
&\quad = \textsf{diag}\bigg(\sum_{j = 1}^n \theta_i^{(t)}\{\eta'''(\theta_i^{(t)}\bkappa(\bz_i^*)\transpose\widetilde{\by}_j^{(t)})(A_{ij}^{(t)} - \theta_i^{(t)}\bkappa(\bz_i^*)\transpose\widetilde{\by}_j^{(t)}) - 2\eta''(\theta_i^{(t)}\bkappa(\bz_i^*)\transpose\widetilde{\by}_j^{(t)})\}(\bkappa(\bz_i^*)\transpose\widetilde{\by}_j^{(t)})^2\\
&\qquad\qquad\quad\times\frac{\partial}{\partial z_{ik}}\bkappa(\bz_i^*)\transpose\widetilde{\by}_j^{(t)}\\
&\qquad\qquad\quad + \sum_{j = 1}^n\{\eta''(\theta_i^{(t)}\bkappa(\bz_i^*)\transpose\widetilde{\by}_j^{(t)})(A_{ij}^{(t)} - \theta_i^{(t)}\bkappa(\bz_i^*)\transpose\widetilde{\by}_j^{(t)}) - \eta'(\theta_i^{(t)}\bkappa(\bz_i^*)\transpose\widetilde{\by}_j^{(t)})\}\\
&\qquad\qquad\quad\times \frac{\partial}{\partial z_{ik}^*}(\bkappa(\bz_i^*)\transpose\widetilde{\by}_j^{(t)})^2: t\in[m]\bigg).
\end{align*}
By Lemma \ref{lemma:LLN}, for any $s\in[m]$, we know that
\begin{align*}
&\bigg\|\frac{1}{mn}\sum_{t = 1}^m\sum_{j = 1}^n\frac{\partial}{\partial\bz_i^*}\log f(A_{ij}^{(t)}, \theta_{0i}^{(t)}\bkappa(\bz_{0i}^*)\transpose\widetilde{\by}_j^{(t)})\bigg\|_2 = \Optilde\bigg\{\frac{(\log n)^\xi}{\sqrt{n}}\bigg\},\\
&\frac{1}{mn}\sum_{t = 1}^m\sum_{j = 1}^n\frac{\partial}{\partial\theta_i^{(s)}}\log f(A_{ij}^{(t)}, \theta_{0i}^{(t)}\bkappa(\bz_{0i}^*)\transpose\widetilde{\by}_j^{(t)}) = \Optilde\bigg\{\frac{(\log n)^\xi}{m\sqrt{n}}\bigg\}\\
&\bigg\|\frac{1}{mn}\frac{\partial^2}{\partial\bz_i^*\partial\bz_i^{*\mathrm{T}}}\sum_{t = 1}^m\sum_{j = 1}^n\log f(A_{ij}^{(t)}, \theta_{0i}^{(t)}\bkappa(\bz_{0i}^*)\transpose\widetilde{\by}_j^{(t)}) + \frac{1}{m}\sum_{t = 1}^m\theta_{0i}^{2(t)}\bJ\transpose\bG_{0in}^{(t)}\bJ\bigg\|_2
 = \Optilde\bigg\{\frac{(\log n)^\xi}{\sqrt{n}}\bigg\},\\
&\bigg\|\frac{1}{mn}\frac{\partial^2}{\partial\btheta_i\partial\btheta_i\transpose}\sum_{t = 1}^m\sum_{j = 1}^n\log f(A_{ij}^{(t)}, \theta_{0i}^{(t)}\bkappa(\bz_{0i}^*)\transpose\widetilde{\by}_j^{(t)}) + \frac{1}{m}\textsf{diag}(\bz_{0i}\transpose\bG_{0in}^{(t)}\bz_{0i}:t\in[m])\bigg\|_2 = \Optilde\bigg\{\frac{(\log n)^\xi}{m\sqrt{n}}\bigg\},\\
&\frac{1}{mn}\frac{\partial^2}{\partial\bz_i^*\partial\theta_i^{(s)}}\sum_{t = 1}^m\sum_{j = 1}^n\log f(A_{ij}^{(t)}, \theta_{0i}^{(t)}\bkappa(\bz_{0i}^*)\transpose\widetilde{\by}_j^{(t)}) + \frac{1}{m}\bJ\transpose\bG_{0in}^{(t)}\bz_{0i}\theta_{0i}^{(t)} = \Optilde\bigg\{\frac{(\log n)^\xi}{m\sqrt{n}}\bigg\}.
\end{align*}
By \eqref{eqn:noise_absolute_bound}, we can compute
\begin{align*}
&\frac{1}{mn}\bigg\|\sum_{t = 1}^m\sum_{j = 1}^n\frac{\partial^3}{\partial z_{ik}^*\partial \bz_i^*\partial\bz_i^{*\mathrm{T}}}\log f(A_{ij}^{(t)}; \bar{\theta}_i^{(t)}\bkappa(\bar{\bz}_i^*)\transpose\widetilde{\by}_{j}^{(t)})\bigg\|_2\\
&\quad\leq \frac{1}{mn}\sum_{t = 1}^m\sum_{j = 1}^n|E_{ij}^{(t)}|\max_{t\in[m],j\in[n]}|\eta'''(\bar{\theta}_i^{(t)}\bkappa(\bar{\bz}_i^*)\transpose\widetilde{\by}_{j}^{(t)})
\\
&\qquad + \frac{1}{mn}\sum_{t = 1}^m\sum_{j = 1}^n(|P_{0ij}^{(t)} + |\bar{\theta}_i^{(t)}\bkappa(\bar{\bz}_i^*)\transpose\widetilde{\by}_{j}^{(t)}| + 2|\eta''(\bar{\theta}_i^{(t)}\bkappa(\bar{\bz}_i^*)\transpose\widetilde{\by}_{j}^{(t)})|)\|\bJ\transpose\widetilde{\by}_{j}^{(t)}\|_2^3 = \Optilde(1),
\end{align*}
and similarly,
\begin{align*}
\frac{1}{mn}\bigg\|\frac{\partial^3}{\partial \theta_{i}^{(s)}\partial\bz_i^*\partial\bz_i^{*\mathrm{T}}}\sum_{t = 1}^m\sum_{j = 1}^n\log f(A_{ij}^{(t)}, \bar{\theta}_i^{(t)}\bkappa(\bar{\bz}_i^*)\transpose\widetilde{\by}_j^{(t)})\bigg\|_2
& = \Optilde\bigg(\frac{1}{m}\bigg),\\
\frac{1}{mn}\frac{\partial^3}{\partial \theta_{i}^{(s)}\partial\btheta_i\partial\btheta_i^{\mathrm{T}}}\sum_{t = 1}^m\sum_{j = 1}^n\log f(A_{ij}^{(t)}, \bar{\theta}_i^{(t)}\bkappa(\bar{\bz}_i^*)\transpose\widetilde{\by}_j^{(t)})
& = \Optilde\bigg(\frac{1}{m}\bigg)\times\be_s\be_s\transpose,\\
\frac{1}{mn}\bigg\|\frac{\partial^3}{\partial z_{ik}^*\partial\btheta_i\partial\btheta_i\transpose}\sum_{t = 1}^m\sum_{j = 1}^n\log f(A_{ij}^{(t)}, \bar{\theta}_i^{(t)}\bkappa(\bar{\bz}_i^*)\transpose\widetilde{\by}_j^{(t)})\bigg\|_2
& = \Optilde\bigg(\frac{1}{m}\bigg),\\
\frac{1}{mn}\bigg\|\frac{\partial^3}{\partial z_{ik}^*\partial\bz_i^*\partial\btheta_i\transpose}\sum_{t = 1}^m\sum_{j = 1}^n\log f(A_{ij}^{(t)}, \bar{\theta}_i^{(t)}\bkappa(\bar{\bz}_i^*)\transpose\widetilde{\by}_j^{(t)})\bigg\|_2
& = \Optilde(1).
\end{align*}
By Taylor's expansion, there exists some $\bar{\bz}_i^*$ between $\bz_{0i}^*$ and $\bz_i^*$ and $\bar{\theta}_i^{(t)}$ between $\theta_i^{(t)}$ and $\theta_{0i}^{(t)}$, such that, with $\eps_n = (\log n)^{2\xi}/\sqrt{n}$, 
\begin{align*}
&\frac{1}{mn}\widetilde{\ell}_{in}(\bz_i^*, \btheta_i) - \frac{1}{mn}\widetilde{\ell}_{in}(\bz_{0i}^*, \btheta_{0i})\\
&\quad = \frac{1}{mn}\sum_{t = 1}^m\sum_{j = 1}^n\frac{\partial}{\partial\bz_i^{*\mathrm{T}}}\log f(A_{ij}^{(t)}, \theta_{0i}^{(t)}\bkappa(\bz_{0i}^*)\transpose\widetilde{\by}_{j}^{(t)})\bdelta_{\bz_i}\\
&\qquad + \frac{1}{mn}\sum_{t = 1}^m\sum_{j = 1}^n\frac{\partial}{\partial\btheta_i\transpose}\log f(A_{ij}^{(t)}, \theta_{0i}^{(t)}\bkappa(\bz_{0i}^*)\transpose\widetilde{\by}_{j}^{(t)})\bdelta_{\btheta_i}
\\ &\qquad
 + \frac{1}{2}\bdelta_{\bz_i}\transpose\frac{1}{mn}\sum_{t = 1}^m\sum_{j = 1}^n\frac{\partial^2}{\partial\bz_i^*\partial\bz_i^{*\mathrm{T}}}\log f(A_{ij}^{(t)}, \theta_{0i}^{(t)}\bkappa(\bz_{0i}^*)\transpose\widetilde{\by}_{j}^{(t)})\bdelta_{\bz_i}\\ 
&\qquad + \frac{1}{2}\bdelta_{\btheta_i}\transpose\frac{1}{mn}\sum_{t = 1}^m\sum_{j = 1}^n\frac{\partial^2}{\partial\btheta_i\partial\btheta_i^{\mathrm{T}}}\log f(A_{ij}^{(t)}, \theta_{0i}^{(t)}\bkappa(\bz_{0i}^*)\transpose\widetilde{\by}_{j}^{(t)})\bdelta_{\btheta_i}\\
&\qquad + \bdelta_{\bz_i}\transpose\frac{1}{mn}\sum_{t = 1}^m\sum_{j = 1}^n\frac{\partial^2}{\partial\bz_i^*\partial\btheta_i^{\mathrm{T}}}\log f(A_{ij}^{(t)}, \theta_{0i}^{(t)}\bkappa(\bz_{0i}^*)\transpose\widetilde{\by}_{j}^{(t)})\bdelta_{\btheta_i}\\
&\qquad + \frac{1}{6}\bdelta_{\bz_i}\transpose\frac{1}{mn}\sum_{k = 1}^d\sum_{t = 1}^m\sum_{j = 1}^n\frac{\partial^3}{\partial z_{ik}^*\partial\bz_i^*\partial\bz_i^{*\mathrm{T}}}\log f(A_{ij}^{(t)}, \bar{\theta}_i^{(t)}\bkappa(\bar{\bz}_{i}^{*})\transpose\widetilde{\by}_{j}^{(t)})\bdelta_{\bz_i}\delta_{z_{ik}}\\ 
&\qquad + \frac{1}{6}\bdelta_{\btheta_i}\transpose\frac{1}{mn}\sum_{s = 1}^m\sum_{t = 1}^m\sum_{j = 1}^n\frac{\partial^3}{\partial \theta_{i}^{(s)}\partial\btheta_i\partial\btheta_i^{\mathrm{T}}}\log f(A_{ij}^{(t)}, \bar{\theta}_i^{(t)}\bkappa(\bar{\bz}_{i}^{*})\transpose\widetilde{\by}_{j}^{(t)})\bdelta_{\btheta_i}\delta_{\theta_{i}^{(s)}}\\ 
&\qquad + \frac{1}{3}\bdelta_{\btheta_i}\frac{1}{mn}\sum_{k = 1}^d\sum_{t = 1}^m\sum_{j = 1}^n\frac{\partial^3}{\partial z_{ik}^*\partial\btheta_i\partial\btheta_i\transpose}\log f(A_{ij}^{(t)}, \bar{\theta}_i^{(t)}\bkappa(\bar{\bz}_{i}^{*})\transpose\widetilde{\by}_{j}^{(t)})\bdelta_{\btheta_i}\delta_{z_{ik}}\\
&\qquad + \frac{1}{3}\bdelta_{\bz_i}\frac{1}{mn}\sum_{s = 1}^m\sum_{t = 1}^m\sum_{j = 1}^n\frac{\partial^3}{\partial \theta_{i}^{(s)}\partial\bz_i^*\partial\bz_i^{*\mathrm{T}}}\log f(A_{ij}^{(t)}, \bar{\theta}_i^{(t)}\bkappa(\bar{\bz}_{i}^{*})\transpose\widetilde{\by}_{j}^{(t)})\bdelta_{\bz_i}\delta_{\theta_{i}^{(t)}}\\
&\quad = \Optilde\bigg\{(\|\bdelta_{\bz_i}\|_2 + \frac{1}{\sqrt{m}}\|\bdelta_{\btheta_i}\|_2)\frac{(\log n)^{\xi}}{\sqrt{n}}\bigg\} - \frac{1}{2m}\bdelta_{\bz_i}\transpose\sum_{t = 1}^m\bJ\transpose\bG_{0in}^{(t)}\bJ\bdelta_{\bz_i} - \frac{1}{2m}\sum_{t = 1}^m\delta_{\theta_i^{(s)}}^2\bz_{0i}\transpose\bG_{0in}^{(t)}\bz_{0i}\\ 
&\qquad -\frac{1}{m}\sum_{t = 1}^m\bdelta_{\bz_i}\transpose\bJ\transpose\bG_{0in}^{(t)}\bz_{0i}\delta_{\theta_i^{(t)}} + \Optilde\bigg\{\|\bdelta_{z_i}\|_2^2\frac{(\log n)^{\xi}}{\sqrt{n}} + \|\bdelta_{\btheta_i}\|_2^2\frac{(\log n)^\xi}{m\sqrt{n}} + \|\bdelta_{\bz_i}\|_2\|\bdelta_{\btheta_i}\|_2\frac{(\log n)^\xi}{\sqrt{mn}}\bigg\}\\ 
&\qquad + \Optilde\bigg(\|\bdelta_{\bz_i}\|_2^3 + \sum_{s = 1}^m\frac{|\delta_{\theta_i^{(s)}}|^3}{m}+ \frac{1}{\sqrt{m}}\|\bdelta_{\btheta_i}\|_2\|\bdelta_{\bz_i}\|_2^2 + \frac{1}{m}\|\bdelta_{\btheta_i}\|_2^2\|\bdelta_{\bz_i}\|_2\bigg)\\ 
&\quad = - \frac{1}{2}\begin{bmatrix}\bdelta_{\bz_i}\transpose & \bdelta_{\btheta_i}\end{bmatrix}
\begin{bmatrix}
\frac{1}{m}\sum_{t = 1}^m\theta_{0i}^{2(t)}\bJ\transpose\bG_{0in}^{(t)}\bJ & \frac{1}{m}\bJ\transpose\bG_{0in}^{(1)}\bz_{0i}\theta_{0i}^{(t)} & \ldots & \frac{1}{m}\bJ\transpose\bG_{0in}^{(m)}\bz_{0i}\theta_{0i}^{(t)}\\ 
\frac{1}{m}\theta_{0i}^{(t)}\bz_{0i}\transpose\bG_{0in}^{(1)}\bJ & \frac{1}{m}\bz_{0i}\transpose\bG_{0in}^{(1)}\bz_{0i} & \ldots & 0\\
\vdots & \vdots & \ddots & \vdots \\
\frac{1}{m}\theta_{0i}^{(t)}\bz_{0i}\transpose\bG_{0in}^{(m)}\bJ & 0 & \ldots & \frac{1}{m}\bz_{0i}\transpose\bG_{0in}^{(m)}\bz_{0i}
\end{bmatrix}
\begin{bmatrix}\bdelta_{\bz_i}\\\bdelta_{\btheta_i}\end{bmatrix}\\
&\qquad + \Optilde\bigg\{\eps_n\frac{(\log n)^\xi}{\sqrt{n}} + \eps_n^3\bigg\}\\
&\quad \leq -\frac{1}{4}\eps_n^2\lambda_{d + m - 1}\left(\bGamma_{in}^{-1}\right)\quad\text{w.h.p.}.
\end{align*}
This implies that $\widetilde{\ell}_{in}(\cdot, \cdot)$ has a local maximizer $(\widehat{\bz}_i^*, \widehat{\btheta}_i)$ inside the interior of $\calQ_n$ satisfying the likelihood equation
\[
    \frac{\partial\widetilde{\ell}_{in}}{\partial\bz_i^*}(\bz_i^*, \btheta_i)\mathrel{\Bigg|_{\bz_i^* = \widehat{\bz}_i^*, \btheta_i = \widehat{\btheta}_i }} = \zero_{d - 1},\quad
    \frac{\partial\widetilde{\ell}_{in}}{\partial\btheta_i}(\bz_i^*, \btheta_i)\mathrel{\Bigg|_{\bz_i^* = \widehat{\bz}_i^*, \btheta_i = \widehat{\btheta}_i }} = \zero_m. 
\]
By the construction of $\calQ_n$, we obtain
$\|\widehat{\bz}_i^* - \bz_{0i}^*\|_2 = \Optilde\{(\log n)^{2\xi}/\sqrt{n}\}$, $|\widehat{\theta}_i^{(t)} - \theta_{0i}^{(t)}| = \Optilde\{(\log n)^{2\xi}/\sqrt{n}\}$. 
Now by Taylor's theorem again,
\begin{align*}
\begin{bmatrix}
\zero_{d - 1}\\
\zero_m
\end{bmatrix} & = \frac{1}{mn}\sum_{t = 1}^m\sum_{j = 1}^n\begin{bmatrix}
\frac{\partial}{\partial\bz_i^{*}}\log f(A_{ij}^{(t)}; \widehat{\theta}_i^{(t)}\bkappa(\widehat{\bz}_i^{*\mathrm{T}})\widetilde{\by}_j^{(t)})\\
\frac{\partial}{\partial\btheta_i}\log f(A_{ij}^{(t)}; \widehat{\theta}_i^{(t)}\bkappa(\widehat{\bz}_i^{*\mathrm{T}})\widetilde{\by}_j^{(t)})
\end{bmatrix}
\\ 
& = \frac{1}{mn}\sum_{t = 1}^m\sum_{j = 1}^n\begin{bmatrix}
\frac{\partial}{\partial\bz_i^*}\log f(A_{ij}^{(t)}; \theta_{0i}^{(t)}\bkappa(\bz_i^*)\widetilde{\by}_j^{(t)})\\ 
\frac{\partial}{\partial\btheta_i}\log f(A_{ij}^{(t)}; \theta_{0i}^{(t)}\bkappa(\bz_i^*)\widetilde{\by}_j^{(t)})
\end{bmatrix}\mathrel{\Bigg|_{\bz_i^* = \bz_{0i}}}
\\ 
&\quad + \frac{1}{mn}\sum_{t = 1}^m\sum_{j = 1}^n
\begin{bmatrix}
\frac{\partial^2}{\partial\bz_{i}^{*}\partial\bz_i^{*\mathrm{T}}}\log f(A_{ij}^{(t)};\theta_{0i}^{(t)}\bkappa(\bz_i^*)\widetilde{\by}_j^{(t)}) & 
\frac{\partial^2}{\partial\bz_{i}^{*}\partial\btheta_i^{\mathrm{T}}}\log f(A_{ij}^{(t)};\theta_{0i}^{(t)}\bkappa(\bz_i^*)\widetilde{\by}_j^{(t)})\\ 
\frac{\partial^2}{\partial\btheta_{i}\partial\bz_i^{*\mathrm{T}}}\log f(A_{ij}^{(t)};\theta_{0i}^{(t)}\bkappa(\bz_i^*)\widetilde{\by}_j^{(t)})
&\frac{\partial^2}{\partial\btheta_{i}\partial\btheta_i^{\mathrm{T}}}\log f(A_{ij}^{(t)};\theta_{0i}^{(t)}\bkappa(\bz_i^*)\widetilde{\by}_j^{(t)})
\end{bmatrix}\mathrel{\Bigg|_{\bz_i^* = \bz_{0i}}}\\
&\quad\times\begin{bmatrix}
\widehat{\bz}_i^* - \bz_{0i}^*\\ 
\widehat{\btheta}_i - \btheta_{0i}
\end{bmatrix}
 + \begin{bmatrix}\bzeta_{1i}\\ \bzeta_{2i}\end{bmatrix},
\end{align*}
where $\bzeta_{1i} = [\zeta_{1i1},\ldots,\zeta_{1i(d - 1)}]\transpose$ and for each $k\in[d - 1]$ and $\bzeta_{2i} = [\zeta_{2i1},\ldots,\zeta_{2im}]\transpose$, there exists $(\bar{\bz}_i^{*(k)},\bar{\btheta}_i^{(t, k)})$ and $(\bar{\bz}_i^{*(t)},\bar{\btheta}_i^{(t)})$ between $(\widehat{\bz}_i^*, \widehat{\btheta}_i)$ and $(\bz_{0i}, \btheta_{0i})$, such that
\begin{align*}
\zeta_{1ik}& = (\widehat{\bz}_i^{*} - \bz_{0i}^*)\transpose\\
&\quad\times \frac{1}{2mn}\sum_{t = 1}^m\sum_{j = 1}^n\bar{\theta}_i^{2(t, k)}\{\eta'''(\bar{\theta}_i^{(t, k)}\bar{\bz}_i^{*(k)\mathrm{T}}\widetilde{\by}_j^{(t)})(A_{ij}^{(t)} - \bar{\theta}_i^{(t, k)}\bar{\bz}_i^{*(k)\mathrm{T}}\widetilde{\by}_j^{(t)}) - 2\eta''(\bar{\theta}_i^{(t, k)}\bar{\bz}_i^{*(k)\mathrm{T}}\widetilde{\by}_j^{(t)})\}\\ 
&\quad\times\bJ\transpose\widetilde{\by}_j^{(t)}\widetilde{\by}_j^{(t)\mathrm{T}}\bJ (\be_k\transpose\bJ\transpose\widetilde{\by}_j^{(t)})
(\widehat{\bz}_i^* - \bz_{0i}^*)\\ 
&\quad + \frac{1}{2mn}\sum_{t = 1}^m\sum_{j = 1}^n\{\eta'''(\bar{\theta}_i^{(t, k)}\bar{\bz}_i^{*(k)\mathrm{T}}\widetilde{\by}_j^{(t)})(A_{ij}^{(t)} - \bar{\theta}_i^{(t, k)}\bar{\bz}_i^{*(k)\mathrm{T}}\widetilde{\by}_j^{(t)}) - 2\eta''(\bar{\theta}_i^{(t, k)}\bar{\bz}_i^{*(k)\mathrm{T}}\widetilde{\by}_j^{(t)})\}\\ 
&\quad\times (\bkappa(\bar{\bz}_i^{*(k)})\transpose\widetilde{\by}_j^{(t)})^2
(\be_k\transpose\bJ\transpose\widetilde{\by}_j^{(t)})(\widehat{\theta}_i^{(t)} - \theta_{0i}^{(t)})^2\\
&\quad + \frac{1}{mn}\sum_{t = 1}^m\sum_{j = 1}^n\bar{\theta}_i^{(t, k)}\{\eta'''(\bar{\theta}_i^{(t, k)}\bar{\bz}_i^{*(k)\mathrm{T}}\widetilde{\by}_j^{(t)})(A_{ij}^{(t)} - \bar{\theta}_i^{(t, k)}\bar{\bz}_i^{*(k)\mathrm{T}}\widetilde{\by}_j^{(t)}) - 2\eta''(\bar{\theta}_i^{(t, k)}\bar{\bz}_i^{*(k)\mathrm{T}}\widetilde{\by}_j^{(t)})\}\\ 
&\quad\times\bkappa(\bar{\bz}_i^{*(k)})\transpose\widetilde{\by}_j^{(t)}\widetilde{\by}_j^{(t)\mathrm{T}}\bJ(\be_k\transpose\bJ\transpose\widetilde{\by}_j^{(t)})(\widehat{\bz}_i^* - \bz_{0i}^*)
\end{align*}
and
\begin{align*}
\zeta_{2it}& = \frac{1}{2mn}\sum_{j = 1}^n\{\eta'''(\bar{\theta}_i^{(t)}\bar{\bz}_i^{*(t)\mathrm{T}}\widetilde{\by}_j^{(t)})(A_{ij}^{(t)} - \bar{\theta}_i^{(t)}\bar{\bz}_i^{*(t)\mathrm{T}}\widetilde{\by}_j^{(t)}) - 2\eta''(\bar{\theta}_i^{(t)}\bar{\bz}_i^{*(t)\mathrm{T}}\widetilde{\by}_j^{(t)})\}\\ 
&\quad\times (\bkappa(\bz_i^*)\transpose\widetilde{\by}_j^{(t)})^3(\widehat{\theta}_i^{(t)} - \theta_{0i}^{(t)})^2\\ 
&\quad + \frac{1}{2mn}\sum_{j = 1}^n\bar{\theta}_i^{2(t)}\{\eta'''(\bar{\theta}_i^{(t)}\bar{\bz}_i^{*(t)\mathrm{T}}\widetilde{\by}_j^{(t)})(A_{ij}^{(t)} - \bar{\theta}_i^{(t)}\bar{\bz}_i^{*(t)\mathrm{T}}\widetilde{\by}_j^{(t)}) - 2\eta''(\bar{\theta}_i^{(t)}\bar{\bz}_i^{*(t)\mathrm{T}}\widetilde{\by}_j^{(t)})\}\\ 
&\qquad\times (\widehat{\bz}_i^* - \bz_{0i}^*)\transpose\bJ\transpose\widetilde{\by}_j^{(t)}\widetilde{\by}_j^{(t)\mathrm{T}}\bJ(\bkappa(\bar{\bz}_i^{*(t)})\transpose\widetilde{\by}_j^{(t)})(\widehat{\bz}_i^* - \bz_{0i}^*)\\ 
&\quad + \frac{1}{2mn}\sum_{j = 1}^n2\bar{\theta}_i^{(t)}\{\eta''(\bar{\theta}_i^{(t)}\bar{\bz}_i^{*(t)\mathrm{T}}\widetilde{\by}_j^{(t)})(A_{ij}^{(t)} - \bar{\theta}_i^{(t)}\bar{\bz}_i^{*(t)\mathrm{T}}\widetilde{\by}_j^{(t)}) - \eta'(\bar{\theta}_i^{(t)}\bar{\bz}_i^{*(t)\mathrm{T}}\widetilde{\by}_j^{(t)})\}\\ 
&\qquad\times (\widehat{\bz}_i^* - \bz_{0i}^*)\transpose\bJ\transpose\widetilde{\by}_j^{(t)}\widetilde{\by}_j^{(t)\mathrm{T}}\bJ(\widehat{\bz}_i^* - \bz_{0i}^*)\\ 
&\quad + \frac{1}{mn}\sum_{j = 1}^n\bar{\theta}_i^{(t)}\{\eta'''(\bar{\theta}_i^{(t)}\bar{\bz}_i^{*(t)\mathrm{T}}\widetilde{\by}_j^{(t)})(A_{ij}^{(t)} - \bar{\theta}_i^{(t)}\bar{\bz}_i^{*(t)\mathrm{T}}\widetilde{\by}_j^{(t)}) - 2\eta''(\bar{\theta}_i^{(t)}\bar{\bz}_i^{*(t)\mathrm{T}}\widetilde{\by}_j^{(t)})\}\\
&\qquad\times(\bkappa(\bar{\bz}_i^{*(t)})\transpose\widetilde{\by}_j)^2(\widehat{\theta}_i^{(t)} - \theta_{0i}^{(t)})\widetilde{\by}_j^{(t)\mathrm{T}}\bJ(\widehat{\bz}_i^* - \bz_{0i}^*)\\ 
&\quad + \frac{1}{mn}\sum_{j = 1}^n\{\eta''(\bar{\theta}_i^{(t)}\bar{\bz}_i^{*(t)\mathrm{T}}\widetilde{\by}_j^{(t)})(A_{ij}^{(t)} - \bar{\theta}_i^{(t)}\bar{\bz}_i^{*(t)\mathrm{T}}\widetilde{\by}_j^{(t)}) - \eta'(\bar{\theta}_i^{(t)}\bar{\bz}_i^{*(t)\mathrm{T}}\widetilde{\by}_j^{(t)})\}(\bkappa(\bar{\bz}_i^{*(t)})\transpose\widetilde{\by}_j)\\
&\qquad\times(\widehat{\theta}_i^{(t)} - \theta_{0i}^{(t)})\widetilde{\by}_j^{(t)\mathrm{T}}\bJ(\widehat{\bz}_i^* - \bz_{0i}^*). 
\end{align*}
Clearly, $\bar{\theta}_i^{(t, k)}\bar{\bz}_i^{*(k)\mathrm{T}}\widetilde{\by}_j^{(t)}$ and $\bar{\theta}_i^{(t)}\bar{\bz}_i^{*(t)\mathrm{T}}\widetilde{\by}_j^{(t)}$ are in a neighborhood of $P_{0ij}^{(t)}$ w.h.p., so that 
\begin{align*}
\|\bzeta_{1i}\|_2 = \Optilde\bigg\{\frac{(\log n)^{3\xi}}{n}\bigg\},\quad \max_{t\in[m]}|\zeta_{2it}| = \Optilde\bigg\{\frac{(\log n)^{3\xi}}{n}\bigg\}
\end{align*}
by \eqref{eqn:noise_absolute_bound} and Lemma \ref{lemma:Y_expansion}. 
Together with Lemma \ref{lemma:LLN} and Lemma \ref{lemma:CLT}, we obtain
\begin{align*}
&\begin{bmatrix}
-\bJ\transpose\frac{1}{mn}\sum_{t = 1}^m\sum_{j = 1}^nE_{ij}^{(t)}\theta_{0i}^{(t)}\eta'(P_{0ij})\by_{0j}^{(t)} + \Optilde\left\{\frac{(\log n)^{2\xi}}{n}\right\}\\ 
-\frac{1}{mn}\sum_{j = 1}^nE_{ij}^{(1)}\eta'(P_{0ij}^{(1)})\bz_{0i}\transpose\by_{0j}^{(1)} + \Optilde\left\{\frac{(\log n)^{2\xi}}{mn}\right\}\\
\vdots\\ 
-\frac{1}{mn}\sum_{j = 1}^nE_{ij}^{(m)}\eta'(P_{0ij}^{(m)})\bz_{0i}\transpose\by_{0j}^{(m)} + \Optilde\left\{\frac{(\log n)^{2\xi}}{mn}\right\}
\end{bmatrix}\\
&\quad = -\begin{bmatrix}
\bJ\transpose\frac{1}{m}\sum_{t = 1}^m\theta_{0i}^{2(t)}\bG_{0in}^{(t)}\bJ
& \frac{1}{m}\bJ\transpose\theta_{0i}^{(1)}\bG_{0in}^{(1)}\bz_{0i}
& \ldots 
& \frac{1}{m}\bJ\transpose\theta_{0i}^{(m)}\bG_{0in}^{(m)}\bz_{0i}
\\ 
\frac{1}{m}\theta_{0i}^{(1)}\bz_{0i}\transpose\bG_{0in}^{(1)}\bJ& \frac{1}{m}\bz_{0i}\transpose\bG_{0in}^{(1)}\bz_{0i} & \ldots & 0\\
\vdots & \vdots & \ddots & \vdots\\ 
\frac{1}{m}\theta_{0i}^{(1)}\bz_{0i}\transpose\bG_{0in}^{(m)}\bJ
& 0 & \ldots & \frac{1}{m}\bz_{0i}\transpose\bG_{0in}^{(m)}\bz_{0i}
\end{bmatrix}
\begin{bmatrix}
\widehat{\bz}_i^* - \bz_{0i}\\ 
\widehat{\btheta}_i - \btheta_{0i}\\ 
\end{bmatrix}\\
&\qquad + 
\begin{bmatrix}
\Optilde\left\{\frac{(\log n)^\xi}{\sqrt{n}}\right\} & \Optilde\left\{\frac{(\log n)^\xi}{m\sqrt{n}}\right\} & \ldots & \Optilde\left\{\frac{(\log n)^\xi}{m\sqrt{n}}\right\}\\ 
\Optilde\left\{\frac{(\log n)^\xi}{m\sqrt{n}}\right\} & \Optilde\left\{\frac{(\log n)^\xi}{m\sqrt{n}}\right\} & \ldots & 0\\ 
\vdots & \vdots & \ddots & \vdots \\ 
\Optilde\left\{\frac{(\log n)^\xi}{m\sqrt{n}}\right\} & 0 & \ldots & \Optilde\left\{\frac{(\log n)^\xi}{m\sqrt{n}}\right\}
\end{bmatrix}
\begin{bmatrix}
\Optilde\left\{\frac{(\log n)^{2\xi}}{\sqrt{n}}\right\}\\ 
\vdots \\ 
\Optilde\left\{\frac{(\log n)^{2\xi}}{\sqrt{n}}\right\}
\end{bmatrix}
 + \Optilde\bigg\{\frac{(\log n)^{2\xi}}{n}\bigg\}\\ 
&\quad = -\bGamma_{in}
\begin{bmatrix}
\widehat{\bz}_i^* - \bz_{0i}^*\\ 
\widehat{\btheta}_i - \btheta_{0i}\\ 
\end{bmatrix} + \begin{bmatrix}
\Optilde\left\{\frac{(\log n)^{3\xi}}{n}\right\}\\ 
\Optilde\left\{\frac{(\log n)^{3\xi}}{n}\right\}\\
\vdots\\ 
\Optilde\left\{\frac{(\log n)^{3\xi}}{n}\right\}
\end{bmatrix},
\end{align*}
which further implies that
\begin{align*}
&\begin{bmatrix}
\widehat{\bz}_i^* - \bz_{0i}^*\\ 
\widehat{\btheta}_i - \btheta_{0i}\\ 
\end{bmatrix}\\
&\quad = \begin{bmatrix}
(\frac{1}{m}\sum_{t = 1}^m\theta_{0i}^{2(t)}\bJ\transpose\bDelta_{in}^{(t)}\bJ)^{-1} & -
(\frac{1}{m}\sum_{t = 1}^m\theta_{0i}^{2(t)}\bJ\transpose\bDelta_{in}^{(t)}\bJ)^{-1}\bJ\transpose\bM_{in}\transpose \\
-\bM_{in}\bJ(\frac{1}{m}\sum_{t = 1}^m\theta_{0i}^{2(t)}\bJ\transpose\bDelta_{in}^{(t)}\bJ)^{-1}
&
\bL_{in}^{-1} + \bM_{in}\bJ(\frac{1}{m}\sum_{t = 1}^m\theta_{0i}^{2(t)}\bJ\transpose\bDelta_{in}^{(t)}\bJ)^{-1}
\bM_{in}\transpose\bJ\transpose
\end{bmatrix}\\ 
&\qquad\times 
\begin{bmatrix}
\frac{1}{mn}\sum_{t = 1}^m\sum_{j = 1}^nE_{ij}^{(t)}\theta_{0i}^{(t)}\eta'(P_{0ij})\bJ\transpose\by_{0j}^{(t)} + \Optilde\left\{\frac{(\log n)^{2\xi}}{n}\right\}\\ 
\frac{1}{mn}\sum_{j = 1}^nE_{ij}^{(1)}\eta'(P_{0ij}^{(1)})\bz_{0i}\transpose\by_{0j}^{(1)} + \Optilde\left\{\frac{(\log n)^{2\xi}}{mn}\right\}\\
\vdots\\ 
\frac{1}{mn}\sum_{j = 1}^nE_{ij}^{(m)}\eta'(P_{0ij}^{(m)})\bz_{0i}\transpose\by_{0j}^{(m)} + \Optilde\left\{\frac{(\log n)^{2\xi}}{mn}\right\}
\end{bmatrix} + \begin{bmatrix}
\Optilde\left\{\frac{(\log n)^{3\xi}}{n}\right\}\\ 
\Optilde\left\{\frac{(\log n)^{3\xi}}{n}\right\}\\
\vdots\\ 
\Optilde\left\{\frac{(\log n)^{3\xi}}{n}\right\}
\end{bmatrix}\\
    &\quad = \frac{1}{mn}\bGamma_{in}
    ^{-1}\begin{bmatrix}
    \sum_{t = 1}^m\sum_{j = 1}^nE_{ij}^{(t)}\theta_{0i}^{2(t)}\eta'(P_{0ij}^{(t)})\by_{0j}^{(t)}\\
    \sum_{j = 1}^nE_{ij}^{(1)}\eta'(P_{0ij}^{(1)})\by_{0j}^{(1)}\\
    \vdots \\
    \sum_{j = 1}^nE_{ij}^{(m)}\eta'(P_{0ij}^{(m)})\by_{0j}^{(m)}
    \end{bmatrix} + \begin{bmatrix}\Optilde\left\{\frac{(\log n)^{3\xi}}{n}\right\}\\ 
    \Optilde\left\{\frac{(\log n)^{3\xi}}{n}\right\}\\
    \vdots \\
    \Optilde\left\{\frac{(\log n)^{3\xi}}{n}\right\}
    \end{bmatrix}.
\end{align*}
In particular, we obtain
\begin{align*}
\widehat{\bz}_i^* - \bz_{0i}^*
& = \bigg(\frac{1}{m}\sum_{t = 1}^m\theta_{0i}^{2(t)}\bJ\transpose\bDelta_{in}^{(t)}\bJ\bigg)^{-1}\frac{1}{mn}\sum_{t = 1}^m\sum_{j = 1}^nE_{ij}^{(t)}\theta_{0i}^{(t)}\eta'(P_{0ij}^{(t)})\bJ\transpose\by_{0j}^{(t)}\\ 
&\quad - \bigg(\frac{1}{m}\sum_{t = 1}^m\theta_{0i}^{2(t)}\bJ\transpose\bDelta_{in}^{(t)}\bJ\bigg)^{-1}\frac{1}{mn}\sum_{t = 1}^m\sum_{j = 1}^n\frac{E_{ij}^{(t)}\eta'(P_{0ij}^{(t)})\bJ\transpose\bG_{0in}^{(t)}\bz_{0i}\bz_{0i}\transpose\by_{0j}^{(t)}}{\bz_{0i}\transpose\bG_{0in}^{(t)}\bz_{0i}}
 + \Optilde\bigg\{\frac{(\log n)^{3\xi}}{n}\bigg\}\\ 
& = \bigg(\frac{1}{m}\sum_{t = 1}^m\theta_{0i}^{2(t)}\bJ\transpose\bDelta_{in}^{(t)}\bJ\bigg)^{-1}\frac{1}{mn}\sum_{t = 1}^m\sum_{j = 1}^nE_{ij}^{(t)}\theta_{0i}^{(t)}\eta'(P_{0ij}^{(t)})\bJ\transpose
\bigg(\eye_d - \frac{\bG_{0in}^{(t)}\bz_{0i}\bz_{0i}\transpose}{\bz_{0i}\transpose\bG_{0in}^{(t)}\bz_{0i}}\bigg)\by_{0j}\\ 
&\quad + \Optilde\bigg\{\frac{(\log n)^{3\xi}}{n}\bigg\}.
\end{align*}
and
\begin{align*}
\widehat{\theta}_i^{(t)} - \theta_{0i}^{(t)}
& = \frac{1}{mn}\sum_{j = 1}^n E_{ij}^{(t)}\theta_{0i}^{(t)}\eta'(P_{0ij}^{(t)})\by_{0j}^{(t)}\\ &\quad - \frac{\theta_{0i}^{(t)}\bG_{0in}^{(t)}}{\bz_{0i}\transpose\bG_{0in}^{(t)}\bz_{0i}}\bigg(\frac{1}{m}\sum_{s = 1}^m\theta_{0i}^{2(s)}\bG_{0in}^{(s)}\bigg)^{-1}\frac{1}{mn}\sum_{s = 1}^m\sum_{j = 1}^nE_{ij}^{(s)}\theta_{0i}^{(s)}\eta'(P_{0ij}^{(s)})\bigg(\eye_d - \frac{\bG_{0in}^{(s)}\bz_{0i}\bz_{0i}\transpose}{\bz_{0i}\transpose\bG_{0in}^{(s)}\bz_{0i}}\bigg)\by_{0j}^{(s)}\\ 
&\quad + \Optilde\bigg\{\frac{(\log n)^{3\xi}}{n}\bigg\}\\
& = \Optilde\bigg\{\frac{(\log n)^\xi}{\sqrt{mn}} + \frac{(\log n)^{3\xi}}{n}\bigg\}.
\end{align*}
The proof is therefore completed.
\end{proof}

\section{Proof of Theorem \ref{thm:BvM}}
\label{sec:proof_of_theorem_ref_bvm}

We now provide the proof of Theorem \ref{thm:BvM}, namely, the BvM theorem of the exact posterior distribution. Because a BvM theorem typically requires a uniform local asymptotic normality result on the likelihood and is stronger than the conditions necessary for the asymptotic normality of the maximum likelihood estimator, we need to first establish the following uniform law of large numbers result. 

\begin{lemma}[Uniform Law of Large Numbers]
\label{lemma:ULLN}
Suppose Assumptions \ref{assumption:identifiability}--\ref{assumption:likelihood} hold. Further assume that $\psi(\cdot)$ is a continuous function defined on $\calI$. Then, 
\begin{align*}
&\sup_{\substack{\bz_i^*\in\calS^{d - 1}\\\theta_i^{(t)}\in\calK}}
\bigg\|\frac{1}{n}\sum_{j = 1}^n\psi(\theta_i^{(t)}\bkappa(\bz_i^*)\transpose\widetilde{\by}_j^{(t)})(A_{ij}^{(t)} - \theta_i^{(t)}\bkappa(\bz_i^*)\transpose\widetilde{\by}_j)\widetilde{\by}_j^{(t)}\\
&\qquad\qquad - \frac{1}{n}\sum_{j = 1}^n\psi(\theta_i^{(t)}\bkappa(\bz_i^*)\transpose\by_{0j}^{(t)})(P_{0ij}^{(t)} - \theta_i^{(t)}\bkappa(\bz_i^*)\transpose\by_{0j})\by_{0j}^{(t)}\bigg\|_2 = \Optilde\bigg\{\frac{(\log n)^\xi}{\sqrt{n}}\bigg\},\\
&\sup_{\substack{\bz_i^*\in\calS^{d - 1}\\ \theta_i^{(t)}\in\calK}}
\bigg\|\frac{1}{n}\sum_{j = 1}^n\{\psi(\theta_i^{(t)}\bkappa(\bz_i^*)\transpose\widetilde{\by}_j^{(t)})(A_{ij}^{(t)} - \theta_i^{(t)}\bkappa(\bz_i^*)\transpose\widetilde{\by}_j) - \eta'(\theta_i^{(t)}\bkappa(\bz_i^*)\transpose\widetilde{\by}_j^{(t)})\}\widetilde{\by}_j^{(t)}\widetilde{\by}_j^{(t)\mathrm{T}}\\
&\qquad\qquad - \frac{1}{n}\sum_{j = 1}^n\{\psi(\theta_i^{(t)}\bkappa(\bz_i^*)\transpose\by_{0j})(P_{0ij}^{(t)} - \theta_i^{(t)}\bkappa(\bz_i^*)\transpose\by_{0j}^{(t)}) - \eta'(\theta_i^{(t)}\bkappa(\bz_i^*)\transpose\by_{0j}^{(t)})\}\by_{0j}^{(t)}\by_{0j}^{(t)\mathrm{T}}\bigg\|_2\\ 
&\qquad\qquad = \Optilde\bigg\{\frac{(\log n)^\xi}{\sqrt{n}}\bigg\},\\
&\sup_{\substack{\bz_i^*\in\calS^{d - 1}\\ \theta_i^{(t)}\in\calK}}
\bigg\|\frac{1}{n}\sum_{j = 1}^n\{\eta(\theta_i^{(t)}\bkappa(\bz_i^*)\transpose\widetilde{\by}_j^{(t)})A_{ij}^{(t)} - B(\theta_i^{(t)}\bkappa(\bz_i^*)\transpose\widetilde{\by}_j^{(t)})\}
\\&\qquad\qquad
 - \frac{1}{n}\sum_{j = 1}^n\{\eta(\theta_i^{(t)}\bkappa(\bz_i^*)\transpose\by_{0j})P_{0ij}^{(t)} -B(\theta_i^{(t)}\bkappa(\bz_i^*)\transpose\by_{0j}^{(t)})\}\bigg\|_2
 = \Optilde\bigg\{\frac{(\log n)^\xi}{\sqrt{n}}\bigg\}
\end{align*}
\end{lemma}

\begin{proof}[\bf Proof]
The proof closely resembles that of Lemma S2 in \cite{XieWu2024}, albeit a more straightforward application of a maximal inequality for stochastic processes. By definition, we have
\begin{align}
&\frac{1}{n}\sum_{j = 1}^n\psi(\theta_i^{(t)}\bkappa(\bz_i^*)\transpose\widetilde{\by}_j^{(t)})(A_{ij}^{(t)} - \theta_i^{(t)}\bkappa(\bz_i^*)\transpose\widetilde{\by}_j)\widetilde{\by}_j^{(t)}\nonumber\\
&\quad - \frac{1}{n}\sum_{j = 1}^n\psi(\theta_i^{(t)}\bkappa(\bz_i^*)\transpose\by_{0j}^{(t)})(P_{0ij}^{(t)} - \theta_i^{(t)}\bkappa(\bz_i^*)\transpose\by_{0j}^{(t)})\by_{0j}^{(t)}\nonumber\\
\label{eqn:ULLN1_term1}
&\quad = \frac{1}{n}\sum_{j = 1}^nE_{ij}^{(t)}\psi(\theta_i^{(t)}\bkappa(\bz_i^*)\transpose\by_{0j}^{(t)})\by_{0j}^{(t)}\\
\label{eqn:ULLN1_term2}
&\qquad + \frac{1}{n}\sum_{j = 1}^n(E_{ij}^{(t)} + P_{0ij}^{(t)})\{\psi(\theta_i^{(t)}\bkappa(\bz_i^*)\transpose\widetilde{\by}_{j}^{(t)}) - \psi(\theta_i^{(t)}\bkappa(\bz_i^*)\transpose\by_{0j}^{(t)})\}\widetilde{\by}_j^{(t)}\\
\label{eqn:ULLN1_term3}
&\qquad + \frac{1}{n}\sum_{j = 1}^n(E_{ij}^{(t)} + P_{0ij}^{(t)})\psi(\theta_i^{(t)}\bkappa(\bz_i^*)\transpose\by_{0j}^{(t)})(\widetilde{\by}_j^{(t)} - \by_{0j}^{(t)})\\ 
\label{eqn:ULLN1_term4}
&\qquad + \frac{1}{n}\sum_{j = 1}^n\{\psi(\theta_i^{(t)}\bkappa(\bz_i^*)\transpose\widetilde{\by}_{j}^{(t)})\theta_i^{(t)}\bkappa(\bz_i^*)\transpose\widetilde{\by}_{j}^{(t)} - \psi(\theta_i^{(t)}\bkappa(\bz_i^*)\transpose\by_{0j}^{(t)})\theta_i^{(t)}\bkappa(\bz_i^*)\transpose\by_{0j}^{(t)}\}\widetilde{\by}_j^{(t)}\\ 
\label{eqn:ULLN1_term5}
&\qquad + \frac{1}{n}\sum_{j = 1}^n\psi(\theta_i^{(t)}\bkappa(\bz_i^*)\transpose\by_{0j}^{(t)})\theta_i^{(t)}\bkappa(\bz_i^*)\transpose\by_{0j}^{(t)}(\widetilde{\by}_j^{(t)} - \by_{0j}^{(t)}).
\end{align}
By Lemma \ref{lemma:Y_expansion}, we know that
\begin{align*}
|\theta_i^{(t)}\bkappa(\bz_i^*)\widetilde{\by}_j^{(t)} - \theta_i^{(t)}\bkappa(\bz_i^*)\by_{0j}^{(t)}|\leq \theta_i^{(t)}\|\bkappa(\bz_i^*)\|_2\|\widetilde{\by}_j^{(t)} - \by_{0j}^{(t)}\|_2 = \Optilde\bigg\{\frac{(\log n)^\xi}{\sqrt{n}}\bigg\}.
\end{align*}
By Assumption \ref{assumption:likelihood}, 
\begin{equation}
\begin{aligned}
\max_{j\in[n]}\sup_{\substack{\theta_i^{(t)}\in\calI \\ \bz_i^*\in\calS^{d - 1}}}\theta_i^{(t)}\bkappa(\bz_i^*)\by_{0j}^{(t)}
& = \max_{j\in[n]}\sup_{\substack{\theta_i^{(t)}\in\calI \\ \bz_i\in\Delta^{d - 1}}}\theta_i^{(t)}\bz_i\transpose\bB_0^{(t)}\bz_{0j}\theta_{0j}^{(t)}\leq \sup_{\theta_i^{(t)}\in\calI}\theta_i^{(t)}\theta_{0j}b < b,\\
\min_{j\in[n]}\inf_{\substack{\theta_i^{(t)}\in\calI \\ \bz_i^*\in\calS^{d - 1}}}\theta_i^{(t)}\bkappa(\bz_i^*)\by_{0j}^{(t)}
& \geq \inf_{\theta_i^{(t)}\in\calI}\theta_i^{(t)}\theta_{0j}a > c_1a\inf\calI > 0,\\ 
\max_{j\in[n]}\sup_{\substack{\theta_i^{(t)}\in\calI \\ \bz_i^*\in\calS^{d - 1}}}\theta_i^{(t)}\bkappa(\bz_i^*)\widetilde{\by}_{j}^{(t)}
& \leq \max_{j\in[n]}\sup_{\substack{\theta_i^{(t)}\in\calI \\ \bz_i^*\in\calS^{d - 1}}}\theta_i^{(t)}\bkappa(\bz_i^*)\by_{0j}^{(t)} + \Optilde\bigg\{\frac{(\log n)^\xi}{\sqrt{n}}\bigg\} < b,\\ 
\min_{j\in[n]}\inf_{\substack{\theta_i^{(t)}\in\calI \\ \bz_i^*\in\calS^{d - 1}}}\theta_i^{(t)}\bkappa(\bz_i^*)\widetilde{\by}_{j}^{(t)}
&\geq \min_{j\in[n]}\inf_{\substack{\theta_i^{(t)}\in\calI \\ \bz_i^*\in\calS^{d - 1}}}\theta_i^{(t)}\bkappa(\bz_i^*)\by_{0j}^{(t)} - \Optilde\bigg\{\frac{(\log n)^\xi}{\sqrt{n}}\bigg\} \geq \frac{c_1a\inf\calI}{2}.
\end{aligned}
\end{equation}
This implies that
\begin{align*}
&\max_{j\in[n]}\sup_{\substack{\theta_i^{(t)}\in\calI \\ \bz_i^*\in\calS^{d - 1}}}\psi(\theta_i^{(t)}\bkappa(\bz_i^*)\transpose\by_{0j}^{(t)})\lesssim 1,\quad
\max_{j\in[n]}\sup_{\substack{\theta_i^{(t)}\in\calI \\ \bz_i^*\in\calS^{d - 1}}}\psi(\theta_i^{(t)}\bkappa(\bz_i^*)\transpose\widetilde{\by}_{j}^{(t)}) = \Optilde(1),\\ 
&\max_{j\in[n]}\sup_{\substack{\theta_i^{(t)}\in\calI \\ \bz_i^*\in\calS^{d - 1}}}|\psi(\theta_i^{(t)}\bkappa(\bz_i^*)\transpose\widetilde{\by}_{j}^{(t)}) - \psi(\theta_i^{(t)}\bkappa(\bz_i^*)\transpose\by_{0j}^{(t)})|\\ 
&\quad = \max_{j\in[n]}\sup_{\substack{\theta_i^{(t)}\in\calI \\ \bz_i^*\in\calS^{d - 1}}}|\psi'(\theta_i^{(t)}\bkappa(\bz_i^*)\transpose\by_{j}^{(t)})|\|\widetilde{\by}_j^{(t)} - \by_{0j}^{(t)}\|_2 = \Optilde\bigg\{\frac{(\log n)^\xi}{\sqrt{n}}\bigg\}.
\end{align*}
Therefore, by Lemma \ref{lemma:Y_expansion} and \eqref{eqn:psi_concentration}, the terms on line \eqref{eqn:ULLN1_term2} through \eqref{eqn:ULLN1_term5} are $\Optilde\{(\log n)^\xi/\sqrt{n}\}$. It remains to work on the term on line \eqref{eqn:ULLN1_term1}. It is sufficient to show that for any $k\in [d]$, 
\begin{align*}
\sup_{\substack{\bz_i^*\in\calS^{d - 1} \\ \theta_i^{(t)}\in\calK}}|J_{ink}^{(t)}(\bz_i^*, \theta_i^{(t)})| = \Optilde\bigg\{\frac{(\log n)^\xi}{\sqrt{n}}\bigg\},\quad\text{where}\quad J_{ink}^{(t)}(\bz^*, \theta) = \frac{1}{n}\sum_{j = 1}^nE_{ij}^{(t)}\psi(\theta\bkappa(\bz^*)\transpose\by_{0j}^{(t)})\by_{0j}^{(t)}. 
\end{align*}
For any $\bz_{1}^*,\bz_{2}^*\in\calS^{d - 1}$, $\theta_1,\theta_2\in\calK$, there exists some $\omega\in[0, 1]$, such that
\begin{align*}
&|\psi(\theta_1\bkappa(\bz^*_1)\transpose\by_{0j}^{(t)}) - \psi(\theta_2\bkappa(\bz^*_2)\transpose\by_{0j}^{(t)})|\\ 
&\quad = |\psi'((\omega\theta_1\bkappa(\bz_1^*) + (1 - \omega)\theta_2\bkappa(\bz_2^*))\transpose\by_{0j}^{(t)})(\theta_1\bkappa(\bz_1^*) - \theta_2\bkappa(\bz_2^*))\transpose\by_{0j}|\\
&\quad \leq C(|\theta_1 - \theta_2| + \|\bz_1^* - \bz_2^*\|_2).
\end{align*}
Then, by Proposition 5.16, for any $\tau > 0$, 
\begin{align*}
&\prob_0(|J_{ink}^{(t)}(\bz_1^*, \theta_1) - J_{ink}^{(t)}(\bz_2^*, \theta_2)| \geq \tau)\\
&\quad\leq 2\exp\bigg\{-C\min\bigg(\frac{n\tau^2}{|\theta_1 - \theta_2|^2 + \|\bz_1^* - \bz_2^*\|_2^2}, \frac{n\tau}{(|\theta_1 - \theta_2|^2 + \|\bz_1^* - \bz_2^*\|_2^2)^{1/2}}\bigg)\bigg\}.
\end{align*}
It follows that
\begin{align*}
&\prob_0\left[|J_{ink}^{(t)}(\bz_1^*, \theta_1) - J_{ink}^{(t)}(\bz_2^*, \theta_2)|\mathbbm{1}\left\{|J_{ink}^{(t)}(\bz_1^*, \theta_1) - J_{ink}^{(t)}(\bz_2^*, \theta_2)|\leq (|\theta_1 - \theta_2|^2 + \|\bz_1^* - \bz_2^*\|_2^2)^{1/2}\right\} \geq \tau\right]\\
&\quad\leq 2\exp\bigg(-\frac{Cn\tau^2}{|\theta_1 - \theta_2|^2 + \|\bz_1^* - \bz_2^*\|_2^2}\bigg)\mathbbm{1}\left\{\tau\leq (|\theta_1 - \theta_2|^2 + \|\bz_1^* - \bz_2^*\|_2^2)^{1/2}\right\}\\
&\quad\leq 2\exp\bigg(-\frac{Cn\tau^2}{|\theta_1 - \theta_2|^2 + \|\bz_1^* - \bz_2^*\|_2^2}\bigg),\\
&\prob_0\left[|J_{ink}^{(t)}(\bz_1^*, \theta_1) - J_{ink}^{(t)}(\bz_2^*, \theta_2)|\mathbbm{1}\left\{|J_{ink}^{(t)}(\bz_1^*, \theta_1) - J_{ink}^{(t)}(\bz_2^*, \theta_2)| > (|\theta_1 - \theta_2|^2 + \|\bz_1^* - \bz_2^*\|_2^2)^{1/2}\right\} \geq \tau\right]\\
&\quad = \prob_0\bigg[|J_{ink}^{(t)}(\bz_1^*, \theta_1) - J_{ink}^{(t)}(\bz_2^*, \theta_2)|\\
&\qquad\qquad \times\mathbbm{1}\left\{|J_{ink}^{(t)}(\bz_1^*, \theta_1) - J_{ink}^{(t)}(\bz_2^*, \theta_2)| > (|\theta_1 - \theta_2|^2 + \|\bz_1^* - \bz_2^*\|_2^2)^{1/2}\right\} \geq \tau\bigg]\\ 
&\qquad\times \mathbbm{1}\left\{\tau\leq (|\theta_1 - \theta_2|^2 + \|\bz_1^* - \bz_2^*\|_2^2)^{1/2}\right\}\\
&\quad + \prob_0\bigg[|J_{ink}^{(t)}(\bz_1^*, \theta_1) - J_{ink}^{(t)}(\bz_2^*, \theta_2)|\\
&\qquad\qquad \times\mathbbm{1}\left\{|J_{ink}^{(t)}(\bz_1^*, \theta_1) - J_{ink}^{(t)}(\bz_2^*, \theta_2)| > (|\theta_1 - \theta_2|^2 + \|\bz_1^* - \bz_2^*\|_2^2)^{1/2}\right\} \geq \tau\bigg]\\ 
&\qquad\times \mathbbm{1}\left\{\tau > (|\theta_1 - \theta_2|^2 + \|\bz_1^* - \bz_2^*\|_2^2)^{1/2}\right\}\\
&\quad = \prob_0\left\{|J_{ink}^{(t)}(\bz_1^*, \theta_1) - J_{ink}^{(t)}(\bz_2^*, \theta_2)| \geq (|\theta_1 - \theta_2|^2 + \|\bz_1^* - \bz_2^*\|_2^2)^{1/2}\right\}\\
&\qquad\times \mathbbm{1}\left\{\tau\leq (|\theta_1 - \theta_2|^2 + \|\bz_1^* - \bz_2^*\|_2^2)^{1/2}\right\}\\
&\qquad + \prob_0\left\{|J_{ink}^{(t)}(\bz_1^*, \theta_1) - J_{ink}^{(t)}(\bz_2^*, \theta_2)| \geq \tau\right\}\mathbbm{1}\left\{\tau > (|\theta_1 - \theta_2|^2 + \|\bz_1^* - \bz_2^*\|_2^2)^{1/2}\right\}\\
&\quad\leq 2e^{-Cn}\mathbbm{1}\left\{\tau\leq (|\theta_1 - \theta_2|^2 + \|\bz_1^* - \bz_2^*\|_2^2)^{1/2}\right\}\\ 
&\qquad + 2\exp\bigg\{-\frac{Cn\tau}{(|\theta_1 - \theta_2|^2 + \|\bz_1^* - \bz_2^*\|_2^2)^{1/2}}\bigg\}\mathbbm{1}\left\{\tau > (|\theta_1 - \theta_2|^2 + \|\bz_1^* - \bz_2^*\|_2^2)^{1/2}\right\}\\
&\quad\leq 2\exp\bigg\{-\frac{Cn\tau}{(|\theta_1 - \theta_2|^2 + \|\bz_1^* - \bz_2^*\|_2^2)^{1/2}}\bigg\}.
\end{align*}
By the fact that $\|X\|_{\psi_1}\leq (\log 2)^{-1/2}\|X\|_{\psi_2}$ (where $\|X\|_{\psi_2}:= \sup_{p\geq 1}p^{-1/2}(\expect|X|^p)^{1/p}$) due to Problem 2.2.5 and Lemma 2.2.1 in \cite{van2023weak}, 
\[
\|J_{ink}^{(t)}(\bz_1^*, \theta_1) - J_{ink}^{(t)}(\bz_2^*, \theta_2)\|_{\psi_1}\leq \frac{C}{\sqrt{n}}(|\theta_1 - \theta_2|^2 + \|\bz_1^* - \bz_2^*\|_2^2)^{1/2}. 
\]
Similarly, we also have
\begin{align*}
\|J_{ink}^{(t)}(\bz_{0i}^*, \theta_{0i}^{(t)})\|_{\psi_1}\lesssim \frac{1}{\sqrt{n}}.
\end{align*}
Consider the metric introduced by $d_J((\bz_1^*, \theta_1), (\bz_2^*, \theta_2)) := C(|\theta_1 - \theta_2|^2 + \|\bz_1^* - \bz_2^*\|_2^2)^{1/2}/\sqrt{n}$ over $\calS^{d - 1}\times \calI$. Denote by $D(\eps, d_J)$ the $\eps$-packing number defined by the maximum number of possible points in $\calS^{d - 1}\times\calI$ that are at least $\eps$-away from each other under the $d_J$ metric, and it is clear that $\log D(\eps, d_J) \lesssim -d\log (\sqrt{n}\eps)$. Under the $d_J$ metric, the diameter $T$ of $\calS^{d - 1}\times \calI$ is upper bounded by a constant multiple of $1/\sqrt{n}$. Then, by Corollary 2.2.5 of \cite{van2023weak}, we have
\begin{align*}
\bigg\|\sup_{\substack{\bz_1^*,\bz_2^*\in\calS^{d - 1} \\
 \theta_1,\theta_2\in\calI}}|J_{ink}^{(t)}(\bz_1^*, \theta_1) - J_{ink}^{(t)}(\bz_2^*, \theta_2)|\bigg\|_{\psi_1}
\lesssim \int_0^T \log(1 + D(\eps, d_J))\mathrm{d}\eps\lesssim \frac{1}{\sqrt{n}}.
\end{align*}
It follows that
\begin{align*}
\bigg\|\sup_{\substack{\bz\in\calS^{d - 1}\\ \theta\in\calI}}|J_{ink}^{(t)}(\bz^*, \theta)|\bigg\|_{\psi_1}
&\leq \bigg\|\sup_{\substack{\bz^*\in\calS^{d - 1}\\ \theta\in\calI}}|J_{ink}^{(t)}(\bz^*, \theta) - J_{ink}^{(t)}(\bz_{0i}^*, \theta_{0i}^{(t)})|\bigg\|_{\psi_1} + \|J_{ink}^{(t)}(\bz_{0i}^*, \theta_{0i}^{(t)})\|_{\psi_1}\lesssim \frac{1}{\sqrt{n}}.
\end{align*}
In particular, this implies that
\begin{align*}
\expect_0\bigg|\sup_{\substack{\bz\in\calS^{d - 1}\\ \theta\in\calI}}|J_{ink}^{(t)}(\bz^*, \theta)|\bigg| \lesssim \frac{1}{\sqrt{n}}.
\end{align*}
Hence, we obtain
\begin{align*}
\bigg\|\sup_{\substack{\bz\in\calS^{d - 1} \\ \theta\in\calI}}|J_{ink}^{(t)}(\bz^*, \theta)| - \expect_0\bigg|\sup_{\substack{\bz\in\calS^{d - 1} \\ \theta\in\calI}}|J_{ink}^{(t)}(\bz^*, \theta)|\bigg|\bigg\|_{\psi_1}\lesssim \frac{1}{\sqrt{n}},
\end{align*}
implying that
\[
  \sup_{\substack{\bz\in\calS^{d - 1}\\ \theta\in\calI}}|J_{ink}^{(t)}(\bz^*, \theta)| = \Optilde\bigg(\frac{\log n}{\sqrt{n}}\bigg).
\]
The proof of the first assertion is therefore completed. The proofs of the second and the third assertions proceed almost exactly the same. 
\end{proof}

\begin{lemma}[Local Convergence of Hessian]
\label{lemma:Hessian_local_convergence}
Suppose the conditions of Theorem \ref{thm:BvM} hold. Further assume $m\lesssim n^\alpha$ for some $\alpha > 0$. Then, 
for any $i\in[n]$, we have
\begin{align*}
\sup_{\substack{\|\bkappa(\bz_i^*) - \bz_{0i}\|_2\leq \eps_n\\ \|\btheta_i - \btheta_{0i}\|_\infty \leq \eps_n}}
&
\bigg\|\frac{1}{mn}
\sum_{t = 1}^m\sum_{j = 1}^n\begin{bmatrix}
\frac{\partial^2}{\partial\bz_i^*\partial\bz_i^{*\mathrm{T}}}\log f(A_{ij}^{(t)}; \theta_{i}^{(t)}\bkappa(\bz_i^*)\transpose\widetilde{\by}_j^{(t)})
& \frac{\partial^2}{\partial\bz_i^*\partial\btheta_i^{\mathrm{T}}}\log f(A_{ij}^{(t)}; \theta_{i}^{(t)}\bkappa(\bz_i^*)\transpose\widetilde{\by}_j^{(t)})\\
\frac{\partial^2}{\partial\btheta_i\partial\bz_i^{*\mathrm{T}}}\log f(A_{ij}^{(t)}; \theta_{i}^{(t)}\bkappa(\bz_i^*)\transpose\widetilde{\by}_j^{(t)})
&
\frac{\partial^2}{\partial\btheta_i\partial\btheta_i^{\mathrm{T}}}\log f(A_{ij}^{(t)}; \theta_{i}^{(t)}\bkappa(\bz_i^*)\transpose\widetilde{\by}_j^{(t)})
\end{bmatrix}\\
& + \begin{bmatrix}\bJ\transpose & \\ & \eye_m\end{bmatrix}\bF_{in}\begin{bmatrix}\bJ & \\ & \eye_m\end{bmatrix}\bigg\|_2 = \Optilde\bigg\{\frac{(\log n)^\xi}{\sqrt{n}} + \eps_n\bigg\},
\end{align*}
where $(\eps_n)_{n = 1}^\infty$ is a positive sequence converging to $0$ as $n\to\infty$. 
\end{lemma}

\begin{proof}[\bf Proof]
By definition, we have
\begin{align*}
&\frac{\partial^2}{\partial\bz_i^*\partial\bz_i^{*\mathrm{T}}}\sum_{t = 1}^m\sum_{j = 1}^n\log f(A_{ij}^{(t)}, \theta_i^{(t)}\bkappa(\bz_i^*)\transpose\widetilde{\by}_j^{(t)})\\
&\quad = \sum_{t = 1}^m\sum_{j = 1}^n\theta_i^{2(t)}\{\eta''(\theta_i^{(t)}\bkappa(\bz_i^*)\transpose\widetilde{\by}_j^{(t)})(A_{ij}^{(t)} - \theta_i^{(t)}\bkappa(\bz_i^*)\transpose\widetilde{\by}_j^{(t)}) - \eta'(\theta_i^{(t)}\bkappa(\bz_i^*)\transpose\widetilde{\by}_j^{(t)})\}\bJ\transpose\widetilde{\by}_j^{(t)}\widetilde{\by}_j^{(t)\mathrm{T}}\bJ,\\
&\frac{\partial^2}{\partial\theta_i^{2(t)}}\sum_{t = 1}^m\sum_{j = 1}^n\log f(A_{ij}^{(t)}, \theta_i^{(t)}\bkappa(\bz_i^*)\transpose\widetilde{\by}_j^{(t)})\\
&\quad =\sum_{j = 1}^n \{\eta''(\theta_i^{(t)}\bkappa(\bz_i^*)\transpose\widetilde{\by}_j^{(t)})(A_{ij}^{(t)} - \theta_i^{(t)}\bkappa(\bz_i^*)\transpose\widetilde{\by}_j^{(t)}) - \eta'(\theta_i^{(t)}\bkappa(\bz_i^*)\transpose\widetilde{\by}_j^{(t)})\}(\bkappa(\bz_i^*)\transpose\widetilde{\by}_j^{(t)})^2,\\
&\frac{\partial^2}{\partial\bz_i^*\partial\theta_i^{(t)\mathrm{T}}}\sum_{t = 1}^m\sum_{j = 1}^n\log f(A_{ij}^{(t)}, \theta_i^{(t)}\bkappa(\bz_i^*)\transpose\widetilde{\by}_j^{(t)})\\
&\quad = \sum_{j = 1}^n \theta_i^{(t)}\{\eta''(\theta_i^{(t)}\bkappa(\bz_i^*)\transpose\widetilde{\by}_j^{(t)})(A_{ij}^{(t)} - \theta_i^{(t)}\bkappa(\bz_i^*)\transpose\widetilde{\by}_j^{(t)}) - \eta'(\theta_i^{(t)}\bkappa(\bz_i^*)\transpose\widetilde{\by}_j^{(t)})\}\bJ\transpose\widetilde{\by}_j^{(t)}\widetilde{\by}_j^{(t)\mathrm{T}}\bkappa(\bz_i^*)\\ 
&\qquad\quad - \sum_{j = 1}^n\eta'(\theta_i^{(t)}\bkappa(\bz_i^*)\transpose\widetilde{\by}_j)(A_{ij}^{(t)} - \theta_i^{(t)}\bkappa(\bz_i^*)\transpose\widetilde{\by}_j^{(t)})\bJ\transpose\widetilde{\by}_j^{(t)}.
\end{align*}
It is sufficient to show that, for any continuously differentiable function $\psi(\cdot)$ over $\calI$ and any $t\in[m]$,
\begin{align*}
\sup_{\substack{\|\bkappa(\bz_i^*) - \bz_{0i}\|_2 \leq \eps_n\\
                \|\btheta - \btheta_i\|_\infty\leq \eps_n}}
\bigg\|\frac{1}{n}\sum_{j = 1}^n(A_{ij}^{(t)} - \theta_i^{(t)}\bkappa(\bz_i^*)\transpose\widetilde{\by}_j^{(t)})\psi(\theta_i^{(t)}\bkappa(\bz_i^*)\transpose\widetilde{\by}_j^{(t)})\widetilde{\by}_j^{(t)}\widetilde{\by}_j^{(t)\mathrm{T}}\bigg\|_2
& = \Optilde\bigg\{\frac{(\log n)^\xi}{\sqrt{n}} + \eps_n\bigg\},\\
\sup_{\substack{\|\bkappa(\bz_i^*) - \bz_{0i}\|_2 \leq \eps_n\\
                \|\btheta - \btheta_i\|_\infty\leq \eps_n }}
\bigg\|\frac{1}{n}\sum_{j = 1}^n\eta'(\theta_i^{(t)}\bkappa(\bz_i^*)\transpose\widetilde{\by}_j^{(t)})\widetilde{\by}_j^{(t)}\widetilde{\by}_j^{(t)\mathrm{T}} - \bG_{0in}^{(t)}\bigg\|_2
& = \Optilde\bigg\{\frac{(\log n)^\xi}{\sqrt{n}} + \eps_n\bigg\}.
\end{align*}
Observe that
\begin{align*}
&\sup_{\substack{\|\bkappa(\bz_i^*) - \bz_{0i}\|_2 \leq \eps_n\\
                \|\btheta - \btheta_i\|_\infty\leq \eps_n}}
                |\theta_i^{(t)}\bkappa(\bz_i^*)\transpose\widetilde{\by}_j^{(t)} - \theta_{0i}^{(t)}\bz_{0i}\transpose\by_{0j}|\\
&\quad\leq 
\sup_{\substack{\|\bkappa(\bz_i^*) - \bz_{0i}\|_2 \leq \eps_n\\
                \|\btheta - \btheta_i\|_\infty\leq \eps_n }}
                \left\{
                |\theta_i^{(t)} - \theta_{0i}|\bkappa(\bz_i^*)\transpose\widetilde{\by}_j^{(t)} + \theta_{0i}^{(t)}|\{\bkappa(\bz_i^*)\transpose - \bz_{0i}\transpose\}\widetilde{\by}_j^{(t)}| + 
               \theta_{0i}^{(t)}|\bz_{0i}\transpose(\widetilde{\by}_j^{(t)} - \by_{0j})|
               \right\}\\
&\quad = \Optilde\bigg\{\frac{(\log n)^\xi}{\sqrt{n}} + \eps_n\bigg\},
\end{align*}
so that $|\psi(\theta_i^{(t)}\bkappa(\bz_i^*)\transpose\widetilde{\by}_j^{(t)})| = \Optilde(1)$ and 
\begin{align*}
\max_{j\in[n]}\sup_{\substack{\|\bkappa(\bz_i^*) - \bz_{0i}\|_2 \leq \eps_n\\
                \|\btheta - \btheta_i\|_\infty \leq \eps_n }}
|\psi(\theta_i^{(t)}\bkappa(\bz_i^*)\transpose\widetilde{\by}_j^{(t)}) - \psi(P_{0ij}^{(t)})|& = \Optilde\bigg\{\frac{(\log n)^\xi}{\sqrt{n}} + \eps_n\bigg\},\\
\max_{j\in[n]}\sup_{\substack{\|\bkappa(\bz_i^*) - \bz_{0i}\|_2 \leq \eps_n\\
                \|\btheta - \btheta_i\|_\infty \leq \eps_n}}
|\psi(\theta_i^{(t)}\bkappa(\bz_i^*)\transpose\widetilde{\by}_j^{(t)}) - \psi(P_{0ij}^{(t)})|& = \Optilde\bigg\{\frac{(\log n)^\xi}{\sqrt{n}} + \eps_n\bigg\}.
\end{align*}
by mean-value theorem.
By definition, Lemma 3.7 in \cite{10.1214/11-AOP734}, and Lemma \ref{lemma:Y_expansion}, we have
\begin{align*}
&\sup_{\substack{\|\bkappa(\bz_i^*) - \bz_{0i}\|_2 \leq \eps_n\\
                \|\btheta - \btheta_i\|_\infty\leq \eps_n }}\bigg\|\frac{1}{n}\sum_{j = 1}^n(A_{ij}^{(t)} - \theta_i^{(t)}\bkappa(\bz_i^*)\transpose\widetilde{\by}_j^{(t)})\psi(\theta_i^{(t)}\bkappa(\bz_i^*)\transpose\widetilde{\by}_j^{(t)})\widetilde{\by}_j^{(t)}\widetilde{\by}_j^{(t)\mathrm{T}}\bigg\|_2\\
&\quad \leq \bigg\|\frac{1}{n}\sum_{j = 1}^nE_{ij}^{(t)}\psi(P_{0ij}^{(t)})\by_{0j}^{(t)}\by_{0j}^{(t)\mathrm{T}}\bigg\|_2\\
&\qquad + \sup_{\substack{\|\bkappa(\bz_i^*) - \bz_{0i}\|_2 \leq \eps_n\\
                \|\btheta - \btheta_i\|_\infty\leq \eps_n }}\bigg\|\frac{1}{n}\sum_{j = 1}^nE_{ij}^{(t)}\{
\psi(\theta_i^{(t)}\bkappa(\bz_i^*)\transpose\widetilde{\by}_j^{(t)}) - \psi(P_{0ij}^{(t)})\}\widetilde{\by}_j^{(t)}\widetilde{\by}_j^{(t)\mathrm{T}}\bigg\|_2\\
&\qquad + \sup_{\substack{\|\bkappa(\bz_i^*) - \bz_{0i}\|_2 \leq \eps_n\\
                \|\btheta - \btheta_i\|_\infty\leq \eps_n }}\bigg\|\frac{1}{n}\sum_{j = 1}^nE_{ij}^{(t)}
\psi(P_{0ij})(\widetilde{\by}_j^{(t)} - \by_{0j}^{(t)})\widetilde{\by}_j^{(t)\mathrm{T}}\bigg\|_2\\
&\qquad + \sup_{\substack{\|\bkappa(\bz_i^*) - \bz_{0i}\|_2 \leq \eps_n\\
                \|\btheta - \btheta_i\|_\infty\leq \eps_n }}\bigg\|\frac{1}{n}\sum_{j = 1}^nE_{ij}^{(t)}
\psi(P_{0ij})\by_{0j}^{(t)}(\widetilde{\by}_j^{(t)} - \by_{0j}^{(t)})\transpose\bigg\|_2\\
&\qquad + \sup_{\substack{\|\bkappa(\bz_i^*) - \bz_{0i}\|_2 \leq \eps_n\\
                \|\btheta - \btheta_i\|_\infty\leq \eps_n }}\bigg\|\frac{1}{n}\sum_{j = 1}^n(\theta_{0i}^{(t)}\bz_{0i}\transpose\by_{0j}^{(t)} - \theta_i^{(t)}\bkappa(\bz_i^*)\transpose\widetilde{\by}_j^{(t)})\psi(\theta_i^{(t)}\bkappa(\bz_i^*)\transpose\widetilde{\by}_j^{(t)})\widetilde{\by}_j^{(t)}\widetilde{\by}_j^{(t)\mathrm{T}}\bigg\|_2\\
&\quad = \Optilde\bigg\{\frac{(\log n)^\xi}{\sqrt{n}} + \eps_n\bigg\}
\end{align*}
and
\begin{align*}
&\sup_{\substack{\|\bkappa(\bz_i^*) - \bz_{0i}\|_2 \leq \eps_n\\
                \|\btheta - \btheta_i\|_\infty\leq \eps_n 
                }}
\bigg\|\frac{1}{n}\sum_{j = 1}^n\eta'(\theta_i^{(t)}\bkappa(\bz_i^*)\transpose\widetilde{\by}_j^{(t)})\widetilde{\by}_j^{(t)}\widetilde{\by}_j^{(t)\mathrm{T}} - \bG_{0in}^{(t)}\bigg\|_2\\
&\quad\leq \sup_{\substack{\|\bkappa(\bz_i^*) - \bz_{0i}\|_2 \leq \eps_n\\
                \|\btheta - \btheta_i\|_\infty\leq \eps_n }}
\bigg\|\frac{1}{n}\sum_{j = 1}^n\{\eta'(\theta_i^{(t)}\bkappa(\bz_i^*)\transpose\widetilde{\by}_j^{(t)}) - \eta'(P_{0ij}^{(t)})\}\widetilde{\by}_j^{(t)}\widetilde{\by}_j^{(t)\mathrm{T}} - \bG_{0in}^{(t)}\bigg\|_2\\
&\qquad + \sup_{\substack{\|\bkappa(\bz_i^*) - \bz_{0i}\|_2 \leq \eps_n\\
                \|\btheta - \btheta_i\|_\infty \leq \eps_n }}
\bigg\|\frac{1}{n}\sum_{j = 1}^n\eta'(P_{0ij}^{(t)})(\widetilde{\by}_j^{(t)} - \by_{0j}^{(t)})\widetilde{\by}_j^{(t)\mathrm{T}}\bigg\|_2
\\&\qquad
 + \sup_{\substack{\|\bkappa(\bz_i^*) - \bz_{0i}\|_2 \leq \eps_n\\
                \|\btheta - \btheta_i\|_\infty\leq \eps_n }}
\bigg\|\frac{1}{n}\sum_{j = 1}^n\eta'(P_{0ij}^{(t)})\by_{0j}^{(t)}(\widetilde{\by}_j^{(t)} - \by_{0j}^{(t)})\bigg\|_2\\
&\quad = \Optilde\bigg\{\frac{(\log n)^\xi}{\sqrt{n}} + \eps_n\bigg\}.
\end{align*}
The proof is thus completed. 
\end{proof}

We are now in a position to prove Theorem \ref{thm:BvM}. 
\begin{proof}[\bf Proof of Theorem \ref{thm:BvM}]
Denote by $\delta_n = (\log n)^{\xi/2}(\log\log n)/(mn)^{1/4}$ and $\eps_n = (\log n)^{\xi}/\sqrt{mn}$ for convenience. We first show that
\begin{align}\label{eqn:identifiability}
\inf_{\|\bkappa(\bz_i^*) - \bz_{0i}\|_2 + \|\btheta_i - \btheta_{0i}\|_\infty > \delta_n}\{\widetilde{\ell}_{in}(\widehat{\bz}_i^*, \widehat{\btheta}_i) - \widetilde{\ell}_{in}(\bz_i^*, \btheta_i)\}\geq 2(1 + \alpha)dmn^{1/2}(\log mn)\quad\text{w.h.p..}
\end{align}
To see this, first observe that
\begin{align*}
&\|\theta_i^{(t)}\bkappa(\bz_i^*) - \theta_{0i}^{(t)}\bz_{0i}\|_2\\
&\quad = \|\theta_i^{(t)}\bkappa(\bz_i^*) - \theta_{0i}^{(t)}\bz_{0i}\|_2\mathbbm{1}\bigg(
|\theta_i^{(t)} - \theta_{0i}^{(t)}| \geq\frac{\theta_{0i}^{(t)}}{4}\|\bkappa(\bz_i^*) - \bz_{0i}\|_2
\bigg)
\\ 
&\qquad + \|\theta_i^{(t)}\bkappa(\bz_i^*) - \theta_{0i}^{(t)}\bz_{0i}\|_2\mathbbm{1}\bigg(
|\theta_i^{(t)} - \theta_{0i}^{(t)}| < \frac{\theta_{0i}^{(t)}}{4}\|\bkappa(\bz_i^*) - \bz_{0i}\|_2
\bigg)\\ 
&\quad\geq \frac{\|\theta_i^{(t)}\bkappa(\bz_i^*) - \theta_{0i}^{(t)}\bz_{0i}\|_1}{\sqrt{d - 1}}\mathbbm{1}\bigg(
|\theta_i^{(t)} - \theta_{0i}^{(t)}| \geq\frac{\theta_{0i}^{(t)}}{4}\|\bkappa(\bz_i^*) - \bz_{0i}\|_2
\bigg)\\ 
&\qquad + \left(\theta_{0i}^{(t)}\|\bkappa(\bz_i^*) - \bz_{0i}\|_2 - |\theta_i^{(t)} - \theta_{0i}^{(t)}|\|\bkappa(\bz_{i}^*)\|_2\right)\mathbbm{1}\bigg(
|\theta_i^{(t)} - \theta_{0i}^{(t)}| < \frac{\theta_{0i}^{(t)}}{4}\|\bkappa(\bz_i^*) - \bz_{0i}\|_2
\bigg)\\ 
&\quad\geq \frac{|\|\theta_i^{(t)}\bkappa(\bz_i^*)\|_1 - \|\theta_{0i}^{(t)}\bz_{0i}\|_1|}{\sqrt{d - 1}}\mathbbm{1}\bigg(
|\theta_i^{(t)} - \theta_{0i}^{(t)}| \geq\frac{\theta_{0i}^{(t)}}{4}\|\bkappa(\bz_i^*) - \bz_{0i}\|_2
\bigg)\\
&\qquad + \left(\frac{3}{4}\theta_{0i}^{(t)}\|\bkappa(\bz_i^*) - \bz_{0i}\|_2\right)\mathbbm{1}\bigg(
|\theta_i^{(t)} - \theta_{0i}^{(t)}| < \frac{\theta_{0i}^{(t)}}{4}\|\bkappa(\bz_i^*) - \bz_{0i}\|_2
\bigg)\\ 
&\quad\geq \left(\frac{|\theta_i^{(t)} - \theta_{0i}^{(t)}|}{2\sqrt{d - 1}} + \frac{\theta_{0i}^{(t)}\|\bkappa(\bz_i^*) - \bz_{0i}\|_2}{8\sqrt{d - 1}}\right)\mathbbm{1}\bigg(
|\theta_i^{(t)} - \theta_{0i}^{(t)}| \geq\frac{\theta_{0i}^{(t)}}{4}\|\bkappa(\bz_i^*) - \bz_{0i}\|_2
\bigg)\\ 
&\qquad + \left(\frac{1}{2}\theta_{0i}^{(t)}\|\bkappa(\bz_i^*) - \bz_{0i}\|_2
 + |\theta_{0i}^{(t)} - \theta_i^{(t)}|\right)\mathbbm{1}\bigg(
|\theta_i^{(t)} - \theta_{0i}^{(t)}| < \frac{\theta_{0i}^{(t)}}{4}\|\bkappa(\bz_i^*) - \bz_{0i}\|_2\bigg)\\ 
&\quad\geq c\left(|\theta_i^{(t)} - \theta_{0i}^{(t)}|^2 + \|\bkappa(\bz_i^*) - \bz_{0i}\|_2^2\right)^{1/2}.
\end{align*}
Then, by the property of exponential family and mean-value theorem, for any $\delta > 0$,
\begin{align*}
&\inf_{\|\bkappa(\bz_i^*) - \bz_{0i}\|_2 + \|\btheta_i - \btheta_{0i}\|_\infty > \delta}\sum_{t = 1}^m\sum_{j = 1}^n\left\{\expect_0\log f(A_{ij}^{(t)}, P_{0ij}^{(t)}) - \expect_0\log f(A_{ij}^{(t)}; \theta_i^{(t)}\bkappa(\bz_i^*)\transpose\by_{0j}^{(t)})\right\}\\ 
&\quad = \inf_{\|\bkappa(\bz_i^*) - \bz_{0i}\|_2 + \|\btheta_i - \btheta_{0i}\|_\infty > \delta}\sum_{t = 1}^m\sum_{j = 1}^n\int_{\theta_i^{(t)}\bkappa(\bz_i^*)\transpose\by_{0j}}^{P_{0ij}^{(t)}}\eta'(\theta)(P_{0ij}^{(t)} - \theta)\mathrm{d}\theta\\ 
&\quad\gtrsim \inf_{\|\bkappa(\bz_i^*) - \bz_{0i}\|_2 + \|\btheta_i - \btheta_{0i}\|_\infty > \delta}\sum_{t = 1}^m\sum_{j = 1}^n\int_{\theta_i^{(t)}\bkappa(\bz_i^*)\transpose\by_{0j}}^{P_{0ij}^{(t)}}(P_{0ij}^{(t)} - \theta)\mathrm{d}\theta\\ 
&\quad = \inf_{\|\bkappa(\bz_i^*) - \bz_{0i}\|_2 + \|\btheta_i - \btheta_{0i}\|_\infty > \delta}\sum_{t = 1}^m\sum_{j = 1}^n\{P_{0ij}^{(t)} - \theta_i^{(t)}\bkappa(\bz_i^*)\transpose\by_{0j}^{(t)}\}^2\\ 
&\quad = \inf_{\|\bkappa(\bz_i^*) - \bz_{0i}\|_2 + \|\btheta_i - \btheta_{0i}\|_\infty > \delta}\sum_{t = 1}^m(\theta_i^{(t)}\bkappa(\bz_i^*) - \theta_{0i}^{(t)}\bz_{0i})\transpose\sum_{j = 1}^n\by_{0j}^{(t)}\by_{0j}^{(t)\mathrm{T}}(\theta_i^{(t)}\bkappa(\bz_i^*) - \theta_{0i}^{(t)}\bz_{0i})\\ 
&\quad\gtrsim n\inf_{\|\bkappa(\bz_i^*) - \bz_{0i}\|_2 + \|\btheta_i - \btheta_{0i}\|_\infty > \delta}\sum_{t = 1}^m\|\theta_i^{(t)}\bkappa(\bz_i^*) - \theta_{0i}^{(t)}\bz_{0i}\|_2^2\\ 
&\quad\gtrsim n\inf_{\|\bkappa(\bz_i^*) - \bz_{0i}\|_2 + \|\btheta_i - \btheta_{0i}\|_\infty > \delta}\sum_{t = 1}^m\{(\theta_i^{(t)} - \theta_{0i}^{(t)})^2 + \|\bkappa(\bz_i^*) - \bz_{0i}\|_2^2\} \gtrsim  mn\delta^2.
\end{align*}
Denote by 
\[
  \bM_{in}(\bz_i^*, \btheta_i)
  = \frac{1}{mn}\sum_{t = 1}^m\sum_{j = 1}^n\expect\{\log f(A_{ij}^{(t)}; \theta_i^{(t)}\bkappa(\bz_i^*)\by_{0j}^{(t)})\}.
\]
Then, following a similar reasoning, we have
\begin{align*}
|\bM_{in}(\bz_{0i}^*, \btheta_{0i}) - \bM_{in}(\widehat{\bz}_i^*, \widehat{\btheta}_i)|
&\lesssim \frac{1}{m}\sum_{t = 1}^m\{(\widehat{\theta}_i^{(t)} - \theta_{0i}^{(t)})^2 + \|\bkappa(\widehat{\bz}_i^*) - \bz_{0i}\|_2^2\}\\
& = \Optilde\bigg\{\frac{(\log n)^{2\xi}}{mn} + \frac{(\log n)^{6\xi}}{n^2}\bigg\}\leq \Optilde\bigg\{\frac{(\log n)^{2\xi}}{n}\bigg\}.
\end{align*}
Together with Lemma \ref{lemma:ULLN}, we obtain
\begin{align*}
&\inf_{\|\bkappa(\bz_i^*) - \bz_{0i}\|_2 + \|\btheta_i - \btheta_{0i}\|_\infty > \delta_n}\frac{1}{mn}\left\{\widetilde{\ell}_{in}(\widehat{\bz}_i^*, \widehat{\btheta}_i) - \widetilde{\ell}_{in}(\bz_i^*, \btheta_i)\right\}\\ 
&\quad\geq \left\{\frac{1}{mn}\widetilde{\ell}_{in}(\widehat{\bz}_i^*, \widehat{\btheta}_i) - M_{in}(\bz_{0i}, \widehat{\btheta}_i)\right\} + \{M_{in}(\bz_{0i}, \widehat{\btheta}_i) - M_{in}(\bz_{0i}^*, \btheta_{0i})\}\\ 
&\qquad + \inf_{\|\bkappa(\bz_i^*) - \bz_{0i}\|_2 + \|\btheta_i - \btheta_{0i}\|_\infty > \delta_n}\{M_{in}(\bz_{0i}^*, \btheta_{0i}) - M_{in}(\bz_i^*, \btheta_i)\}\\ 
&\qquad + \inf_{\|\bkappa(\bz_i^*) - \bz_{0i}\|_2 + \|\btheta_i - \btheta_{0i}\|_\infty > \delta_n}\bigg\{M_{in}(\bz_i^*, \btheta_i) - \frac{1}{mn}\widetilde{\ell}_{in}(\bz_i^*, \btheta_i)\bigg\}\\ 
&\quad\geq \inf_{\|\bkappa(\bz_i^*) - \bz_{0i}\|_2 + \|\btheta_i - \btheta_{0i}\|_\infty > \delta_n}\{M_{in}(\bz_{0i}^*, \btheta_{0i}) - M_{in}(\bz_i^*, \btheta_i)\}\\
&\qquad - 2\sup_{\bz_i^*\in\calS^{d - 1},\btheta_i\in\calK^m}\bigg|\frac{1}{mn}\widetilde{\ell}_{in}(\bz_i^*, \btheta_i) - M_{in}(\bz_i^*, \btheta_i)\bigg|
 - |\bM_{in}(\bz_{0i}^*, \btheta_{0i}) - \bM_{in}(\widehat{\bz}_i^*, \widehat{\btheta}_i)|\\ 
&\quad \geq c_1\delta_n^2 - \frac{C_2(\log n)^\xi}{\sqrt{n}} - \frac{C_3(\log n)^{2\xi}}{n}\geq \frac{c_1}{2}\delta_n^2\\
&\quad = \frac{c_1}{2}\frac{(\log n)^{\xi}(\log\log n)^2}{\sqrt{mn}}\geq 2(1 + \alpha)d\frac{(\log mn)}{\sqrt{mn}}\quad\text{w.h.p.},
\end{align*}
thereby establishing \eqref{eqn:identifiability}. By Lemma \ref{lemma:Hessian_local_convergence}, it is clear that
\begin{align}
\label{eqn:Hessian_locally_bounded}
&\inf_{\|\bkappa(\bz_i^*) - \bz_{0i}\|_2 + \|\btheta_i - \btheta_{0i}\|_\infty \leq 5d\delta_n}\lambda_{\min}\bigg\{-\frac{1}{mn}\mathscr{H}_{\widetilde{\ell}_{in}}(\bz_i^*, \btheta_i)\bigg\} \gtrsim 1\quad\text{w.h.p.},\\
\label{eqn:Hessian_local_concentration}
&\sup_{\|\bkappa(\bz_i^*) - \bz_{0i}\|_2 + \|\btheta_i - \btheta_{0i}\|_\infty \leq \eps_n}\left\|-\frac{1}{mn}\mathscr{H}_{\widetilde{\ell}_{in}}(\bz_i^*, \btheta_i) + \bGamma_{in}\right\|_2\leq\frac{C(\log n)^{\xi}}{\sqrt{n}}\quad\text{w.h.p.},
\end{align}
where we denote by $\mathscr{H}_{\widetilde{\ell}_{in}}$ the Hessian matrix of $\widetilde{\ell}_{in}(\bz_i^*, \btheta_i)$. 

Now, we are in a position to prove the theorem. By definition, $(\bz_i^*, \btheta_i) = (\widehat{\bz}_i^*, \widehat{\btheta}_i) + \bt/\sqrt{mn}$. Define
\begin{align*}
\calD_{in} = \iint_{\mathbb{R}^{d + m - 1}}&\exp\bigg\{\widetilde{\ell}_{in}\bigg(\widehat{\bz}_i^* + \frac{\bt_\bz}{\sqrt{mn}}, \widehat{\btheta}_i + \frac{\bt_\btheta}{\sqrt{mn}}\bigg)\bigg\}\\
&\times\pi_{\bz^*}\bigg(\widehat{\bz}_i^* + \frac{\bt_\bz}{\sqrt{mn}}\bigg)\prod_{t = 1}^m\pi_\theta\bigg(\widehat{\theta}_i + \frac{t_{\theta_i}}{\sqrt{mn}}\bigg)\mathbbm{1}(\bt\in\widehat{\Theta}_i)\mathrm{d}\bt_\bz\mathrm{d}\bt_\btheta,
\end{align*}
where $\widehat{\Theta}_i = \{\bt = (\bt_\bz, \bt_\btheta):\widehat{\bz}_i^* + \bt_\bz/\sqrt{mn}\in\calS^{d - 1}, \widehat{\btheta}_i + \bt_\btheta/\sqrt{mn}\in\calK^m\}$. Note that $\sup_{\bt\in\widehat{\Theta}_i}(\|\bt_\bz\|_2 + \|\bt_\btheta\|_\infty) \lesssim \sqrt{mn}$. For convenience, denote by $\pi_{0i} = \pi_{\bz^*}(\bz_i^*)\prod_{t = 1}^m\pi_\theta(\theta_{0i}^{(t)})$. By the change of variable formula, it is clear that
\begin{align*}
\pi_{(\bt_\bz, \bt_\btheta)}(\bt_\bz, \bt_\btheta\mid\mathbb{A}) & = \frac{1}{\calD_{in}}\exp\bigg\{\widetilde{\ell}_{in}\bigg(\widehat{\bz}_i^* + \frac{\bt_\bz}{\sqrt{mn}}, \widehat{\btheta}_i + \frac{\bt_\btheta}{\sqrt{mn}}\bigg)\bigg\}\\
&\quad\times\pi_{\bz^*}\bigg(\widehat{\bz}_i^* + \frac{\bt_\bz}{\sqrt{mn}}\bigg)\prod_{t = 1}^m\pi_\theta\bigg(\widehat{\theta}_i + \frac{t_{\theta_i}}{\sqrt{mn}}\bigg)\mathbbm{1}(\bt\in\widehat{\Theta}_i).
\end{align*}
Also, observe that
\begin{align*}
&\frac{1}{\calD_{in}}\iint_{\mathbb{R}^{d + m - 1}}(1 + \|\bt_\bz\|_2^\alpha)\bigg|\exp\bigg\{\widetilde{\ell}_{in}\bigg(\widehat{\bz}_i^* + \frac{\bt_\bz}{\sqrt{mn}}, \widehat{\btheta}_i + \frac{\bt_\btheta}{\sqrt{mn}}\bigg) - \widetilde{\ell}_{in}(\widehat{\bz}_i^*, \widehat{\btheta}_i^*)\bigg\}\pi_{\bz^*}\bigg(\widehat{\bz}_i^* + \frac{\bt_\bz}{\sqrt{mn}}\bigg)\\
&\qquad\qquad\qquad\qquad\qquad\quad \times\prod_{t = 1}^m\pi_\theta\bigg(\widehat{\theta}_i + \frac{t_{\theta_i}}{\sqrt{mn}}\bigg)\mathbbm{1}(\bt\in\widehat{\Theta}_i)
 - \frac{\calD_{in}e^{-\bt\transpose\bGamma_{in}\bt/2}}{\det(2\pi\bGamma_{in}^{-1})^{1/2}}\bigg|\mathrm{d}\bt_\bz\mathrm{d}\bt_\btheta\\
&\quad\leq \frac{1}{\calD_{in}}\iint_{\mathbb{R}^{d + m - 1}}(1 + \|\bt_\bz\|_2^\alpha)\bigg|\exp\bigg\{\widetilde{\ell}_{in}\bigg(\widehat{\bz}_i^* + \frac{\bt_\bz}{\sqrt{mn}}, \widehat{\btheta}_i + \frac{\bt_\btheta}{\sqrt{mn}}\bigg) - \widetilde{\ell}_{in}(\widehat{\bz}_i^*, \widehat{\btheta}_i^*)\bigg\}\pi_{\bz^*}\bigg(\widehat{\bz}_i^* + \frac{\bt_\bz}{\sqrt{mn}}\bigg)\\
&\qquad\qquad\qquad\qquad\qquad\qquad\qquad \times \prod_{t = 1}^m\pi_\theta\bigg(\widehat{\theta}_i + \frac{t_{\theta_i}}{\sqrt{mn}}\bigg)\mathbbm{1}(\bt\in\widehat{\Theta}_i)
 - e^{-\bt\transpose\bGamma_{in}\bt/2}
 \pi_{0i}
 \bigg|\mathrm{d}\bt_\bz\mathrm{d}\bt_\btheta\\ 
&\qquad + \frac{1}{\calD_{in}}\bigg|
\pi_{0i}
 - \frac{\calD_{in}}{\det(2\pi\bGamma_{in}^{-1})^{1/2}}\bigg|\iint_{\mathbb{R}^{d + m - 1}}(1 + \|\bt_\bz\|_2^\alpha)e^{-\bt\transpose\bGamma_{in}\bt/2}\mathrm{d}\bt.
\end{align*}
It is sufficient to show that
\begin{align*}
&\iint_{\mathbb{R}^{d + m - 1}}(1 + \|\bt_\bz\|_2^\alpha)\bigg|\exp\bigg\{\widetilde{\ell}_{in}\bigg(\widehat{\bz}_i^* + \frac{\bt_\bz}{\sqrt{mn}}, \widehat{\btheta}_i + \frac{\bt_\btheta}{\sqrt{mn}}\bigg) - \widetilde{\ell}_{in}(\widehat{\bz}_i^*, \widehat{\btheta}_i^*)\bigg\}\pi_{\bz^*}\bigg(\widehat{\bz}_i^* + \frac{\bt_\bz}{\sqrt{mn}}\bigg)
\\
&\qquad\qquad\qquad\qquad\qquad 
\times\prod_{t = 1}^m\pi_\theta\bigg(\widehat{\theta}_i + \frac{t_{\theta_i}}{\sqrt{mn}}\bigg)\mathbbm{1}(\bt\in\widehat{\Theta}_i) - e^{-\bt\transpose\bGamma_{in}\bt/2}
\pi_{0i}
\bigg|\mathrm{d}\bt_\bz\mathrm{d}\bt_\btheta\times\frac{1}{\det(2\pi\bGamma_{in}^{-1})^{1/2}\pi_{0i}}\\ 
&\quad = \Optilde\bigg\{\frac{\sqrt{m}(\log n)^{3\xi}}{\sqrt{n}}\bigg\}.
\end{align*}
We decompose the integration domain into the following blocks
\begin{align*}
\calA_1 &= \{(\bt_\bz, \bt_\btheta)\in\widehat{\Theta}_i:\|\bt_\bz\|_2 + \|\bt_\btheta\|_\infty < \sqrt{mn}\eps_n/2\},\\ 
\calA_2 &= \{(\bt_\bz, \bt_\btheta)\in\widehat{\Theta}_i:\sqrt{mn}\eps_n/2\leq \|\bt_\bz\|_2 + \|\bt_\btheta\|_\infty\leq 2\sqrt{2mn}\delta_n\},\\
\calA_3 &= \{(\bt_\bz, \bt_\btheta)\in\widehat{\Theta}_i:\|\bt_\bz\|_2 + \|\bt_\btheta\|_\infty > 2\sqrt{2mn}\delta_n\},\\ 
\calA_4 &= \mathbb{R}^{d + m - 1}\backslash\widehat{\Theta}_i.
\end{align*}
$\blacksquare$ \underline{Integral over $\calA_4$.}
Since $(\bz_{0i}^*, \btheta_{0i})$ is inside the interior of $\calS^{d - 1}\times\calK^m$ with respect to the $\|\cdot\|_2\times\|\cdot\|_\infty$ metric, then there exists a constant $\delta_0 > 0$, such that $\{(\bz_i^*, \btheta_{i})\in\mathbb{R}^{d - 1}\times\mathbb{R}^m:\|\bz_i^* - \bz_{0i}\|_2 + \|\btheta - \btheta_{0i}\|_\infty < \delta_0\}\subset\calS^{d - 1}\times\calK^m$, and therefore, for any $(\bt_\bz, \bt_\btheta)\in\calA_4 = \mathbb{R}^{d + m - 1}\backslash\widehat{\Theta}_i$, 
\begin{align*}
\delta_0 \leq \bigg\|\widehat{\bz}_i^* + \frac{\bt_\bz}{\sqrt{mn}} - \bz_{0i}^*\bigg\|_2 + \bigg\|\widehat{\btheta}_i^* + \frac{\bt_\btheta}{\sqrt{mn}} - \btheta_{0i}\bigg\|_\infty\leq \|\widehat{\bz}_i^* - \bz_{0i}\|_2 + \|\widehat{\btheta}_i - \btheta_{0i}\|_\infty + \frac{\|\bt_\bz\|_2 + \|\bt_\btheta\|_\infty}{\sqrt{mn}}.
\end{align*}
By Lemma \ref{lemma:aggregated_MLE_theory}, we know that 
\[
\|\widehat{\bz}_i^* - \bz_{0i}\|_2 + \|\widehat{\btheta}_i - \btheta_{0i}\|_\infty = \Optilde\bigg\{\frac{(\log n)^\xi}{\sqrt{mn}} + \frac{(\log n)^{2\xi + \eps}}{n}\bigg\},
\]
so that $\|\bt_\bz\|_2 + \|\bt_\btheta\|_2 \geq \sqrt{mn}\delta_0/2$ w.h.p. if $t\in\calA_4$, implying that
\begin{align*}
&\iint_{\calA_4}(1 + \|\bt_\bz\|_2^\alpha)\bigg|\exp\bigg\{\widetilde{\ell}_{in}\bigg(\widehat{\bz}_i^* + \frac{\bt_\bz}{\sqrt{mn}}, \widehat{\btheta}_i + \frac{\bt_\btheta}{\sqrt{mn}}\bigg) - \widetilde{\ell}_{in}(\widehat{\bz}_i^*, \widehat{\btheta}_i^*)\bigg\}\pi_{\bz^*}\bigg(\widehat{\bz}_i^* + \frac{\bt_\bz}{\sqrt{mn}}\bigg)\\
&\qquad\qquad\qquad\qquad 
\times\prod_{t = 1}^m\pi_\theta\bigg(\widehat{\theta}_i + \frac{t_{\theta_i}}{\sqrt{mn}}\bigg)\mathbbm{1}(\bt\in\widehat{\Theta}_i) - e^{-\bt\transpose\bGamma_{in}\bt/2}
\pi_{0i}
\bigg|\mathrm{d}\bt_\bz\mathrm{d}\bt_\btheta\times\frac{1}{\det(2\pi\bGamma_{in}^{-1})^{1/2}\pi_{0i}}\\ 
&\quad\leq \frac{1}{\det(2\pi\bGamma_{in}^{-1})^{1/2}}\iint_{\calA_4}(1 + \|\bt_\bz\|_2^\alpha)e^{-\bt\transpose\bGamma_{in}\bt/2}\mathrm{d}\bt_\bz\mathrm{d}\bt_\btheta\\ 
&\quad\leq \frac{1}{\det(2\pi\bGamma_{in}^{-1})^{1/2}}\iint_{\{\bt\in\mathbb{R}^{d + m - 1}:\|\bt\|_2 > \sqrt{mn}\delta_0/2\sqrt{2}\}}(1 + \|\bt_\bz\|_2^\alpha)e^{-\bt\transpose\bGamma_{in}\bt/2}\mathrm{d}\bt_\bz\mathrm{d}\bt_\btheta\\ 
&\quad\lesssim c^m\int_{\{\bt\in\mathbb{R}^{d + m - 1}:\|\bt\|_2 > c\sqrt{mn}\}}(1 + \|\bt\|_2^\alpha)e^{-\|\bt\|_2^2/2}\mathrm{d}\bt\\ 
&\quad\leq c^me^{-c^2mn/8}\iint_{\mathbb{R}^{d + m - 1}}e^{-\|\bt\|_2^2/4}\mathrm{d}\bt \lesssim e^{-cmn}\quad\text{w.h.p..}
\end{align*}
$\blacksquare$ \underline{Integral over $\calA_3$.} For any $\bt\in\calA_3$, by Lemma \ref{lemma:aggregated_MLE_theory}, we have
\begin{align*}
&
\bigg\|\bkappa\bigg(\widehat{\bz}_i^* + \frac{\bt_\bz}{\sqrt{mn}}\bigg) - \bkappa(\bz_{0i})\bigg\|_2 + \bigg\|\widehat{\btheta}_i + \frac{\bt_\btheta}{\sqrt{mn}} - \btheta_{0i}\bigg\|_\infty\\
&\quad \geq \bigg\|\bJ\bigg(\widehat{\bz}_i^* + \frac{\bt_\bz}{\sqrt{mn}} - \bz_{0i}\bigg)\bigg\|_2 + \bigg\|\widehat{\btheta}_i + \frac{\bt_\btheta}{\sqrt{mn}} - \btheta_{0i}\bigg\|_\infty\\
&\quad\geq \frac{\|\bt_\bz\|_2 + \|\bt_\btheta\|_\infty}{\sqrt{mn}} - \|\bJ(\widehat{\bz}_i^* - \bz_{0i})\|_2 - \|\widehat{\btheta}_i - \btheta_{0i}\|_\infty\\ 
&\quad \geq 2{\delta_n} - \bigg|\Optilde\bigg\{\frac{(\log n)^\xi}{\sqrt{mn}} + \frac{(\log n)^{2\xi + \eps}}{n}\bigg\}\bigg|\geq {\delta_n}\quad\text{w.h.p..}
\end{align*}
Therefore, by \eqref{eqn:identifiability}, we obtain
\begin{align*}
&\iint_{\calA_3}(1 + \|\bt_\bz\|_2^\alpha)\bigg|\exp\bigg\{\widetilde{\ell}_{in}\bigg(\widehat{\bz}_i^* + \frac{\bt_\bz}{\sqrt{mn}}, \widehat{\btheta}_i + \frac{\bt_\btheta}{\sqrt{mn}}\bigg) - \widetilde{\ell}_{in}(\widehat{\bz}_i^*, \widehat{\btheta}_i^*)\bigg\}\pi_{\bz^*}\bigg(\widehat{\bz}_i^* + \frac{\bt_\bz}{\sqrt{mn}}\bigg)\\
&\qquad\qquad\qquad\qquad 
\times\prod_{t = 1}^m\pi_\theta\bigg(\widehat{\theta}_i + \frac{t_{\theta_i}}{\sqrt{mn}}\bigg)\mathbbm{1}(\bt\in\widehat{\Theta}_i) - e^{-\bt\transpose\bGamma_{in}\bt/2}
\pi_{0i}
\bigg|\mathrm{d}\bt_\bz\mathrm{d}\bt_\btheta\times\frac{1}{\det(2\pi\bGamma_{in}^{-1})^{1/2}\pi_{0i}}\\ 
&\quad\leq \iint_{\calA_3}(1 + \|\bt_\bz\|_2^\alpha)\exp\bigg\{\sup_{\|\bkappa(\bz_i^*) - \bz_{0i}\|_2 + \|\btheta_i - \btheta_{0i}\|_\infty > \delta_n}\widetilde{\ell}(\bz_i^*, \btheta_i) - \widetilde{\ell}_{in}(\widehat{\bz}_i^*, \widehat{\btheta}_i^*)\bigg\}\pi_{\bz^*}\bigg(\widehat{\bz}_i^* + \frac{\bt_\bz}{\sqrt{mn}}\bigg)\\ 
&\qquad\qquad\qquad\qquad\qquad
\times\prod_{t = 1}^m\pi_\theta\bigg(\widehat{\theta}_i + \frac{t_{\theta_i}}{\sqrt{mn}}\bigg)\mathrm{d}\bt_\bz\mathrm{d}\bt_\btheta\times\frac{1}{\det(2\pi\bGamma_{in}^{-1})^{1/2}\pi_{0i}}\\
&\qquad + \iint_{\calA_3}(1 + \|\bt_\bz\|_2^\alpha)\frac{e^{-\bt\transpose\bGamma_{in}\bt/2}\mathrm{d}\bt_\bz\mathrm{d}\bt_\btheta}{\det(2\pi\bGamma_{in}^{-1})^{1/2}}\\
&\quad\leq C_1C_2^m\iint_{\calA_3}(1 + \|\bt_\bz\|_2^\alpha)\bigg(\frac{1}{mn}\bigg)^{2(1 + \alpha)d(mn)^{1/2}}\mathrm{d}\bt_\bz\mathrm{d}\bt_\btheta
 + C_3\int_{\calA_3}\frac{(1 + \|\bt\|_2^\alpha)e^{-\|\bt\|_2^2}}{\sqrt{(2\pi)^{d + m - 1}}}\mathrm{d}\bt\\
&\quad\leq C_1C_2^m\int_{\{\bt\in\mathbb{R}^{d + m - 1}:\|\bt_\bz\|_2 + \|\bt_\btheta\|_\infty\leq c\sqrt{mn}\}}(1 + \|\bt_\bz\|_2^\alpha)\bigg(\frac{1}{mn}\bigg)^{2(1 + \alpha)d(mn)^{1/2}}\mathrm{d}\bt\\
&\qquad + C_3\int_{\{\bt\in\mathbb{R}^{d + m - 1}:\|\bt\|_2\geq c\sqrt{mn}\delta_n\}}\frac{(1 + \|\bt\|_2^\alpha)e^{-\|\bt\|_2^2}}{\sqrt{(2\pi)^{d + m - 1}}}\mathrm{d}\bt\\
&\quad\leq \frac{(C\sqrt{mn})^{d + m - 1}\{1 + (c\sqrt{mn})^\alpha\}}{(mn)^{2(1 + \alpha)d(mn)^{1/2}}} + C^me^{-c^2mn\delta_n^2/8}\leq \frac{C}{\sqrt{mn}}\quad\text{w.h.p..}
\end{align*}
$\blacksquare$ \underline{Integral over $\calA_2$.} By Taylor's theorem, there exists some $\omega_i(\bt)\in[0, 1]$, such that 
\[
\bar{\bz}_i^* = \widehat{\bz}_i^* + \frac{\omega_i(\bt)\bt_\bz}{\sqrt{mn}},\quad
\bar{\btheta}_i = \widehat{\btheta}_i + \frac{\omega_i(\bt)\bt_\btheta}{\sqrt{mn}},
\]
 and
\begin{align*}
\widetilde{\ell}_{in}\bigg(\widehat{\bz}_i^* + \frac{\bt_\bz}{\sqrt{mn}}, \widehat{\btheta}_i + \frac{\bt_\btheta}{\sqrt{mn}}\bigg) - \widetilde{\ell}_{in}(\widehat{\bz}_i^*, \widehat{\btheta}_i^*)
& = \frac{1}{2}\bt\transpose\mathscr{H}_{\widetilde{\ell}_{in}}(\bar{\bz}_i^*, \bar{\btheta}_i)\bt.
\end{align*}
By Lemma \ref{lemma:aggregated_MLE_theory}, for all $\bt\in\calA_2$, 
\begin{align*}
\|\bkappa(\bar{\bz}_i^*) - \bz_{0i}\|_2 + \|\bar{\btheta}_i - \btheta_{0i}\|_\infty
&\leq \|\bJ\|_2\bigg\|\widehat{\bz}_i^* - \bz_{0i}^* + \frac{\omega_i(\bt)\bt_\bz}{\sqrt{mn}}\bigg\|_2 + \bigg\|\widehat{\btheta}_i - \btheta_{0i} + \frac{\omega_i(\bt)\bt_\btheta}{\sqrt{mn}}\bigg\|_\infty\\
&\leq \|\bJ\|_2\|\widehat{\bz}_i^* - \bz_{0i}\|_2 + \|\widehat{\btheta}_i - \btheta_{0i}\|_\infty + \frac{1}{\sqrt{mn}}(\|\bJ\|_2\|\bt_\bz\|_2 + \|\bt_\btheta\|_\infty)\\
&\leq \Optilde\bigg\{\frac{(\log n)^\xi}{\sqrt{mn}} + \frac{(\log n)^{2\xi + \eps}}{n}\bigg\} + 4d\delta_n\leq 5d\delta_n\quad\text{w.h.p..}
\end{align*}
By \eqref{eqn:Hessian_locally_bounded}, this implies that, for any $\bt\in\calA_2$,
\begin{align*}
&\widetilde{\ell}_{in}\bigg(\widehat{\bz}_i^* + \frac{\bt_\bz}{\sqrt{mn}}, \widehat{\btheta}_i + \frac{\bt_\btheta}{\sqrt{mn}}\bigg) - \widetilde{\ell}_{in}(\widehat{\bz}_i^*, \widehat{\btheta}_i^*)\\
&\quad = -\frac{1}{2}\bt\transpose\bigg\{-\mathscr{H}_{\widetilde{\ell}_{in}}\bigg(\widehat{\bz}_i^* + \frac{\omega_i(\bt)\bt_\bz}{\sqrt{mn}}, \widehat{\btheta}_i + \frac{\omega_i(\bt)\bt_\btheta}{\sqrt{mn}}\bigg)\bigg\}\bt\\
&\quad\leq -\frac{1}{2}\|\bt\|_2^2\lambda_{\min}\bigg\{-\mathscr{H}_{\widetilde{\ell}_{in}}\bigg(\widehat{\bz}_i^* + \frac{\omega_i(\bt)\bt_\bz}{\sqrt{mn}}, \widehat{\btheta}_i + \frac{\omega_i(\bt)\bt_\btheta}{\sqrt{mn}}\bigg)\bigg\}\\
&\quad\leq -\frac{1}{2}\|\bt\|_2^2\inf_{\|\bkappa(\bz_i^*) - \bz_{0i}\|_2 + \|\btheta_i - \btheta_{0i}\|_\infty\leq 5d\delta_n}\lambda_{\min}\left\{-\calH_{\widetilde{\ell}_{in}}(\bz_i^*, \btheta_i)\right\}\\
&\quad\leq -c\|\bt\|_2^2\quad\text{w.h.p.}.
\end{align*}
Therefore, 
\begin{align*}
&\iint_{\calA_2}(1 + \|\bt_\bz\|_2^\alpha)\bigg|\exp\bigg\{\widetilde{\ell}_{in}\bigg(\widehat{\bz}_i^* + \frac{\bt_\bz}{\sqrt{mn}}, \widehat{\btheta}_i + \frac{\bt_\btheta}{\sqrt{mn}}\bigg) - \widetilde{\ell}_{in}(\widehat{\bz}_i^*, \widehat{\btheta}_i^*)\bigg\}\pi_{\bz^*}\bigg(\widehat{\bz}_i^* + \frac{\bt_\bz}{\sqrt{mn}}\bigg)\\
&\qquad\qquad\qquad\qquad 
\times\prod_{t = 1}^m\pi_\theta\bigg(\widehat{\theta}_i + \frac{t_{\theta_i}}{\sqrt{mn}}\bigg)\mathbbm{1}(\bt\in\widehat{\Theta}_i) - e^{-\bt\transpose\bGamma_{in}\bt/2}
\pi_{0i}
\bigg|\mathrm{d}\bt_\bz\mathrm{d}\bt_\btheta\times\frac{1}{\det(2\pi\bGamma_{in}^{-1})^{1/2}\pi_{0i}}\\ 
&\quad\leq \iint_{\calA_2}(1 + \|\bt_\bz\|_2^\alpha)e^{-c\|\bt\|_2^2}\pi_{\bz^*}\bigg(\widehat{\bz}_i^* + \frac{\bt_\bz}{\sqrt{mn}}\bigg)\prod_{t = 1}^m\pi_\theta\bigg(\widehat{\theta}_i + \frac{t_{\theta_i}}{\sqrt{mn}}\bigg)\mathrm{d}\bt\times\frac{1}{\det(2\pi\bGamma_{in}^{-1})^{1/2}\pi_{0i}}\\ 
&\qquad + \frac{1}{\det(2\pi\bGamma_{in}^{-1})^{1/2}}\int_{\calA_2}(1 + \|\bt_\bz\|_2^\alpha)e^{-\bt\transpose\bGamma_{in}\bt/2}\mathrm{d}\bt\\
&\quad\leq C^m\int_{\{\bt\in\mathbb{R}^{d + m - 1}:\|\bt\|_2\geq c\sqrt{mn}\eps_n\}}(1 + \|\bt\|_2^\alpha)e^{-c\|\bt\|_2^2}\mathrm{d}\bt\leq C^me^{-cmn\eps_n^2/8}\leq 2e^{-c_1mn\eps_n^2}.
\end{align*}
$\blacksquare$ \underline{Integral over $\calA_1$.} For all $\bt = (\bt_\bz, \bt_\btheta)\in\calA_1$, by Lemma \ref{lemma:aggregated_MLE_theory}, we have
\begin{align*}
\bigg\|\bkappa(\widehat{\bz}_i^*) + \frac{\omega\bt_\bz}{\sqrt{mn}} - \bz_{0i}\bigg\|_2 + \bigg\|\widehat{\btheta}_i + \frac{\omega\bt_\btheta}{\sqrt{mn}} - \btheta_{0i}\bigg\|_\infty
&\leq \Optilde\bigg\{\frac{(\log n)^\xi}{\sqrt{mn}} + \frac{(\log n)^{2\xi + \eps}}{n}\bigg\} + \frac{\|\bt_\bz\|_2 + \|\bt_\btheta\|_\infty}{\sqrt{mn}}\\
&\leq \frac{C(\log n)^{\xi}}{\sqrt{n}}\quad\text{w.h.p.}
\end{align*}
for any $\omega\in[0, 1]$. Now denote by
\[
\bXi_{in}(\bz_i^*, \btheta_i) = \frac{1}{mn}\mathscr{H}_{\widetilde{\ell}_{in}}(\bz_i^*, \btheta_i) + \bGamma_{in}.   
\]
Then by \eqref{eqn:Hessian_local_concentration}, we obtain
\[
\sup_{\bt\in\calA_1,\omega\in[0, 1]}\bigg\|\bXi_{in}\bigg(\widehat{\bz}_i^* + \frac{\omega\bt_\bz}{\sqrt{mn}}, \widehat{\btheta}_i + \frac{\omega\bt_\btheta}{\sqrt{mn}}\bigg)\bigg\|_2 = \Optilde\bigg\{\frac{(\log n)^{\xi}}{\sqrt{n}}\bigg\},
\]
and therefore, there exists some $\omega_i(\bt)\in[0, 1]$, such that
\begin{align*}
\widetilde{\ell}_{in}\bigg(\widehat{\bz}_i^* + \frac{\bt_\bz}{\sqrt{mn}}, \widehat{\btheta}_i + \frac{\bt_\btheta}{\sqrt{mn}}\bigg) - \widetilde{\ell}_{in}(\widehat{\bz}_i^*, \widehat{\btheta}_i^*)
& = -\frac{1}{2}\bt\transpose\bGamma_{in}\bt + \frac{1}{2}\bt\transpose\bXi_{in}\bigg(\widehat{\bz}_i^* + \frac{\omega_i(\bt)\bt_\bz}{\sqrt{mn}}, \widehat{\btheta}_i + \frac{\omega_i(\bt)\bt_\btheta}{\sqrt{mn}}\bigg)\bt.
\end{align*}
Also, by assumption on $\pi_{\bz^*}$ and $\pi_\theta$, 
\begin{align*}
\sup_{\|\bt_\bz\|_2\leq \sqrt{mn}\eps_n/2}\left|\frac{\pi_{\bz_i^*}\left(\widehat{\bz}^*_i + \frac{\bt_\bz}{\sqrt{mn}}\right)}{\pi_{\bz^*}(\bz_{0i}^*)} - 1\right|
&\leq C\sup_{\bz_i^*\in B(\bz_{0i}^*, C\eps_n)}|\pi_{\bz_i^*}(\bz_i^*) - \pi_{\bz_i^*}(\bz_{0i}^*)|\leq C\eps_n\quad\text{w.h.p.},\\ 
\sup_{\|\bt_\btheta\|_\infty\leq \sqrt{mn}\eps_n/2}\left|\frac{\pi_{\theta}\left(\widehat{\theta}^{(t)}_i + \frac{[\bt_\btheta]_t}{\sqrt{mn}}\right)}{\pi_{\theta}(\theta_{0i}^{(t)})} - 1\right|
&\leq C\sup_{\theta_i^{(t)}\in B(\theta_{0i}^{(t)}, C\eps_n)}|\pi_{\theta}(\theta_i^{(t)}) - \pi_{\theta}(\theta_{0i}^{(t)})|\leq C\eps_n\quad\text{w.h.p.},
\end{align*}
and together with the inequality $|e^x - 1|\leq e^{|x| - 1} - 1 + |x|$ for all $x\in\mathbb{R}$, implying that
\begin{align*}
&\sup_{\|\bt_\bz\|_2 + \|\bt_\btheta\|_\infty\leq \sqrt{mn}\eps_n/2}\left|e^{Cmn\eps_n^2(\log n)^{\xi}/\sqrt{n}}\frac{\pi_{\bz_i^*}\left(\widehat{\bz}^*_i + \frac{\bt_\bz}{\sqrt{mn}}\right)}{\pi_{\bz^*}(\bz_{0i}^*)}\prod_{t = 1}^m\frac{\pi_{\theta}\left(\widehat{\theta}^{(t)}_i + \frac{[\bt_\btheta]_t}{\sqrt{mn}}\right)}{\pi_{\theta}(\theta_{0i}^{(t)})} - 1\right|\\
&\quad = \sup_{\|\bt_\bz\|_2 + \|\bt_\btheta\|_\infty\leq \sqrt{mn}\eps_n/2}\left|e^{Cmn\eps_n^2(\log n)^{\xi}/\sqrt{n}}\left\{1 + \frac{\pi_{\bz_i^*}\left(\widehat{\bz}^*_i + \frac{\bt_\bz}{\sqrt{mn}}\right) - \pi_{\bz^*}(\bz_{0i}^*)}{\pi_{\bz^*}(\bz_{0i}^*)}\right\}\right.\\
&\qquad\qquad\qquad\qquad\qquad\qquad\times\left.\prod_{t = 1}^m\left\{1 + \frac{\pi_{\theta}\left(\widehat{\theta}^{(t)}_i + \frac{[\bt_\btheta]_t}{\sqrt{mn}}\right) - \pi_\theta(\theta_{0i}^{(t)})}{\pi_{\theta}(\theta_{0i}^{(t)})}\right\} - 1\right|\\
&\quad = \sup_{\|\bt_\bz\|_2 + \|\bt_\btheta\|_\infty\leq\sqrt{mn}\eps_n/2}\left|\exp\left[\frac{Cmn\eps_n^2(\log n)^{\xi}}{\sqrt{n}} + \log\bigg\{1 + 
\frac{\pi_{\bz_i^*}\left(\widehat{\bz}^*_i + \frac{\bt_\bz}{\sqrt{mn}}\right) - \pi_{\bz^*}(\bz_{0i}^*)}{\pi_{\bz^*}(\bz_{0i}^*)}
\bigg\}\right]\right.\\
&\qquad\qquad\qquad\qquad\qquad\qquad\times\left.\exp\left[\sum_{t = 1}^m\log\bigg\{1 + \frac{\pi_{\theta}\left(\widehat{\theta}^{(t)}_i + \frac{[\bt_\btheta]_t}{\sqrt{mn}}\right) - \pi_\theta(\theta_{0i}^{(t)})}{\pi_{\theta}(\theta_{0i}^{(t)})}\bigg\}\right] - 1\right|\\ 
&\quad \leq \sup_{\|\bt_\bz\|_2 + \|\bt_\btheta\|_\infty\leq\sqrt{mn}\eps_n/2}\exp\left[\frac{Cmn\eps_n^2(\log n)^{\xi}}{\sqrt{n}} + 
\left|\log\bigg\{1 + 
\frac{\pi_{\bz_i^*}\left(\widehat{\bz}^*_i + \frac{\bt_\bz}{\sqrt{mn}}\right) - \pi_{\bz^*}(\bz_{0i}^*)}{\pi_{\bz^*}(\bz_{0i}^*)}
\bigg\}\right|\right]\\ 
&\qquad\qquad\qquad\qquad\qquad\qquad\times\exp\left[\sum_{t = 1}^m\left|\log\bigg\{1 + \frac{\pi_{\theta}\left(\widehat{\theta}^{(t)}_i + \frac{[\bt_\btheta]_t}{\sqrt{mn}}\right) - \pi_\theta(\theta_{0i}^{(t)})}{\pi_{\theta}(\theta_{0i}^{(t)})}\bigg\}\right|\right] - 1\\ 
&\qquad + \frac{Cmn\eps_n^2(\log n)^{\xi}}{\sqrt{n}} + \sup_{\|\bt_\bz\|_2\leq\sqrt{mn}\eps_n/2}\left|\log\bigg\{1 + 
\frac{\pi_{\bz_i^*}\left(\widehat{\bz}^*_i + \frac{\bt_\bz}{\sqrt{mn}}\right) - \pi_{\bz^*}(\bz_{0i}^*)}{\pi_{\bz^*}(\bz_{0i}^*)}
\bigg\}\right|\\
&\qquad + \sup_{\|\bt_\btheta\|_\infty\leq\sqrt{mn}\eps_n/2}\sum_{t = 1}^m\left|\log\bigg\{1 + \frac{\pi_{\theta}\left(\widehat{\theta}^{(t)}_i + \frac{[\bt_\btheta]_t}{\sqrt{mn}}\right) - \pi_\theta(\theta_{0i}^{(t)})}{\pi_{\theta}(\theta_{0i}^{(t)})}\bigg\}\right|\\
&\quad\leq \sup_{\substack{ 
\|\bt_\bz\|_2\leq\sqrt{mn}\eps_n/2\\
\|\bt_\btheta\|_\infty\leq \sqrt{mn}\eps_n/2
 }}\exp\Bigg[
\frac{Cmn\eps_n^2(\log n)^{\xi}}{\sqrt{n}} + 
C
\left|
\frac{\pi_{\bz_i^*}\left(\widehat{\bz}^*_i + \frac{\bt_\bz}{\sqrt{mn}}\right) - \pi_{\bz^*}(\bz_{0i}^*)}{\pi_{\bz^*}(\bz_{0i}^*)}
\right|\\
&\qquad\qquad\qquad\qquad\qquad\quad
 + C\sum_{t = 1}^m\left|
 \frac{\pi_{\theta}\left(\widehat{\theta}^{(t)}_i + \frac{[\bt_\btheta]_t}{\sqrt{mn}}\right) - \pi_\theta(\theta_{0i}^{(t)})}{\pi_{\theta}(\theta_{0i}^{(t)})}
 \right|\Bigg]\\ 
&\qquad - 1 + C\sup_{\|\bt_\bz\|_2\leq\sqrt{mn}\eps_n/2}\left|
\frac{\pi_{\bz_i^*}\left(\widehat{\bz}^*_i + \frac{\bt_\bz}{\sqrt{mn}}\right) - \pi_{\bz^*}(\bz_{0i}^*)}{\pi_{\bz^*}(\bz_{0i}^*)}\right|\\
&\qquad + C\sup_{\|\bt_\btheta\|_\infty\leq\sqrt{mn}\eps_n/2}\sum_{t = 1}^m\left|\frac{\pi_{\theta}\left(\widehat{\theta}^{(t)}_i + \frac{[\bt_\btheta]_t}{\sqrt{mn}}\right) - \pi_\theta(\theta_{0i}^{(t)})}{\pi_{\theta}(\theta_{0i}^{(t)})}\right|
 + \frac{Cmn\eps_n^2(\log n)^{\xi}}{\sqrt{n}}\\
&\quad\leq e^{Cm\eps_n + Cmn\eps_n^2(\log n)^{\xi}/\sqrt{n}} - 1 + Cm\eps_n\leq\frac{C_1\sqrt{m}(\log n)^{3\xi}}{\sqrt{n}},
\end{align*}
where $B(\bx, \eps) := \{\by:\|\bx - \by\|_2 \leq \eps\}$ and $[\bt_\btheta]_t$ denotes the $t$th element of $\bt_\btheta$.
Hence,
\begin{align*}
&\iint_{\calA_1}(1 + \|\bt_\bz\|_2^\alpha)\bigg|\exp\bigg\{\widetilde{\ell}_{in}\bigg(\widehat{\bz}_i^* + \frac{\bt_\bz}{\sqrt{mn}}, \widehat{\btheta}_i + \frac{\bt_\btheta}{\sqrt{mn}}\bigg) - \widetilde{\ell}_{in}(\widehat{\bz}_i^*, \widehat{\btheta}_i^*)\bigg\}\pi_{\bz^*}\bigg(\widehat{\bz}_i^* + \frac{\bt_\bz}{\sqrt{mn}}\bigg)\\
&\qquad\qquad\qquad\qquad 
\times\prod_{t = 1}^m\pi_\theta\bigg(\widehat{\theta}_i + \frac{t_{\theta_i}}{\sqrt{mn}}\bigg)\mathbbm{1}(\bt\in\widehat{\Theta}_i) - e^{-\bt\transpose\bGamma_{in}\bt/2}
\pi_{0i}
\bigg|\mathrm{d}\bt_\bz\mathrm{d}\bt_\btheta\times\frac{1}{\det(2\pi\bGamma_{in}^{-1})^{1/2}\pi_{0i}}\\
&\quad = \iint_{\calA_1}(1 + \|\bt_\bz\|_2^\alpha)\bigg|\exp\bigg\{\frac{1}{2}\bt\transpose\bXi_{in}\bigg(\widehat{\bz}_i^* + \frac{\omega_i(\bt)\bt_\bz}{\sqrt{mn}}, \widehat{\btheta}_i + \frac{\omega_i(\bt)\bt_\btheta}{\sqrt{mn}}\bigg)\bt\bigg\}e^{-\bt\transpose\bGamma_{in}\bt/2}\pi_{\bz^*}\bigg(\widehat{\bz}_i^* + \frac{\bt_\bz}{\sqrt{mn}}\bigg)\\
&\qquad\qquad\qquad\qquad\qquad
\times\prod_{t = 1}^m\pi_\theta\bigg(\widehat{\theta}_i + \frac{t_{\theta_i}}{\sqrt{mn}}\bigg) - e^{-\bt\transpose\bGamma_{in}\bt/2}
\pi_{0i}
\bigg|\mathrm{d}\bt_\bz\mathrm{d}\bt_\btheta\times\frac{1}{\det(2\pi\bGamma_{in}^{-1})^{1/2}\pi_{0i}}\\  
&\quad \leq \iint_{\calA_1}(1 + \|\bt_\bz\|_2^\alpha)\bigg|
e^{mn\eps_n^2(\log n)^{\xi}/\sqrt{n}}
\pi_{\bz^*}\bigg(\widehat{\bz}_i^* + \frac{\bt_\bz}{\sqrt{mn}}\bigg)\prod_{t = 1}^m\pi_\theta\bigg(\widehat{\theta}_i + \frac{t_{\theta_i}}{\sqrt{mn}}\bigg)\frac{1}{\pi_{0i}} - 1\bigg|\\
&\qquad\times\frac{e^{-\bt\transpose\bGamma_{in}\bt/2}}{\det(2\pi\bGamma_{in}^{-1})^{1/2}}\mathrm{d}\bt_\bz\mathrm{d}\bt_\btheta\\
&\quad\leq  \frac{C_1\sqrt{m}(\log n)^{3\xi}}{\sqrt{n}}\iint_{\mathbb{R}^{d + m - 1}}(1 + \|\bt_\bz\|_2^\alpha)\frac{e^{-\bt\transpose\bGamma_{in}\bt/2}}{\det(2\pi\bGamma_{in}^{-1})^{1/2}}\mathrm{d}\bt_\bz\mathrm{d}\bt_\btheta\leq \frac{C_2\sqrt{m}(\log n)^{3\xi}}{\sqrt{n}}
\end{align*}
The proof is completed by combining the above arguments. 
\end{proof}

\section{Proof of Theorem \ref{thm:variational_BvM}}
\label{sec:proof_of_variational_BvM_theorem}
This section is devoted to the proof of the variational BvM theorem (Theorem \ref{thm:variational_BvM}). First observe that, by the construction of the transformations $\calT_\bz$, $\calT_\btheta$, it is clear that
\begin{align*}
&\sup_{\bx\in\mathbb{R}^{d - 1}}\bigg\|\frac{\partial\calT_\bz(\bx)}{\partial\bx\transpose}\bigg\|_2 < \infty,\quad\sup_{\bnu\in\mathbb{R}^{m}}\bigg\|\frac{\partial\calT_\btheta(\bnu)}{\partial\bnu\transpose}\bigg\|_2 < \infty,\\
&\sup_{\bx\in\mathbb{R}^{d - 1}}\bigg\|\frac{\partial^2\calT_\bz(\bx)}{\partial x_{k}\partial\bx\transpose}\bigg\|_2 < \infty,\quad\max_{t\in[m]}\sup_{\nu_t\in\mathbb{R}}\bigg|\frac{\mathrm{d}^2\calT_\btheta^{(t)}(\nu_{t})}{\mathrm{d}\nu_{t}^2}\bigg|_2 < \infty,\\
&\sup_{\bx\in\mathbb{R}^{d - 1}}\bigg\|\frac{\partial^3\calT_\bz}{\partial x_{k}\partial x_{l}\partial\bx\transpose}\bigg\|_2 < \infty,\quad\max_{t\in[m]}\sup_{\nu_t\in\mathbb{R}}\bigg|\frac{\mathrm{d}^3\calT_\btheta^{(t)}(\nu_t)}{\mathrm{d}\nu_{t}^3}\bigg|_2 < \infty
\end{align*} for any $k\in[d - 1]$ and $t\in[m]$. We first establish two technical lemmas that facilitate the proof. The first two lemmas below (Lemmas \ref{lemma:Lipschitz_Hessian} and \ref{lemma:Hessian_local_convergence_transform}) establishes the global Lipschitz continuity and local convergence of the Hessian of the transformed spectral-assisted log-likelihood, and the third lemma (Lemma \ref{lemma:VB_lemma}) shows that the KL-divergence between the variational posterior distribution and the exact posterior distribution converges to $0$ in probability as $n\to\infty$.
\begin{lemma}[Lipschitz Continuity of Hessian]
\label{lemma:Lipschitz_Hessian}
Suppose the conditions of Theorem \ref{thm:BvM} hold and further assume that $0 < \inf\calK < \sup\calK < \infty$. Then, 
\begin{align*}
\sup_{\substack{\bx_i, \bx_i'\in\mathbb{R}^{d - 1}\\ \bnu_i,\bnu_i\in\mathbb{R}^m}}\frac{1}{\|\bx_i - \bx_i'\|_2 + \|\bnu_i - \bnu_i'\|_\infty}
\left\|\frac{\partial^2\calL_{in}}{\partial\bx_i\partial\bx_i\transpose}(\bx_i, \bnu_i)
 - \frac{\partial^2\calL_{in}}{\partial\bx_i\partial\bx_i\transpose}(\bx_i', \bnu_i')\right\|_2& = \Optilde(mn),\\
 \sup_{\substack{\bx_i, \bx_i'\in\mathbb{R}^{d - 1}\\ \bnu_i,\bnu_i\in\mathbb{R}^m}}\frac{1}{\|\bx_i - \bx_i'\|_2 + \|\bnu_i - \bnu_i'\|_\infty}
\left\|\frac{\partial^2\calL_{in}}{\partial\bnu_i\partial\bnu_i\transpose}(\bx_i, \bnu_i)
 - \frac{\partial^2\calL_{in}}{\partial\bnu_i\partial\bnu_i\transpose}(\bx_i', \bnu_i')\right\|_2& = \Optilde(n),\\
 \sup_{\substack{\bx_i, \bx_i'\in\mathbb{R}^{d - 1}\\ \bnu_i,\bnu_i\in\mathbb{R}^m}}\frac{1}{\|\bx_i - \bx_i'\|_2 + \|\bnu_i - \bnu_i'\|_\infty}
\left\|\frac{\partial^2\calL_{in}}{\partial\bx_i\partial\bnu_i\transpose}(\bx_i, \bnu_i)
 - \frac{\partial^2\calL_{in}}{\partial\bx_i\partial\bnu_i\transpose}(\bx_i', \bnu_i')\right\|_2& = \Optilde(\sqrt{m}n).
\end{align*}
\end{lemma}

\begin{proof}[\bf Proof]
By definition, for any $k\in[d - 1]$ and $t\in[m]$, we can compute
\begin{align*}
&\frac{\partial^2}{\partial\bx_i\partial\bx_i\transpose}\frac{\partial\calL_{in}}{\partial x_{ik}}(\bx_i, \bnu_i)\\
&\quad = \sum_{t = 1}^m\sum_{j = 1}^n\bigg[\{A_{ij}^{(t)} - \calT_\btheta^{(t)}(\nu_{it})\bkappa(\calT_\bz(\bx_i))\transpose\widetilde{\by}_j^{(t)}\}\eta''(\calT_\btheta^{(t)}(\nu_{it})\bkappa(\calT_\bz(\bx_i))\transpose\widetilde{\by}_j^{(t)})\\
&\qquad\qquad\qquad - \eta'(\calT_\btheta^{(t)}(\nu_{it})\bkappa(\calT_\bz(\bx_i))\transpose\widetilde{\by}_j^{(t)})\bigg]\\ 
&\qquad\times \bigg\{
\bigg(\frac{\partial^2\calT_\btheta}{\partial x_{k}\partial\bx\transpose}(\bx_i)\bigg)\transpose\bJ\transpose\widetilde{\by}_j^{(t)}\widetilde{\by}_j^{(t)\mathrm{T}}\bJ\bigg(\frac{\partial\calT_\btheta(\bx_i)}{\partial\bx\transpose}\bigg) + 
\bigg(\frac{\partial\calT_\btheta(\bx_i)}{\partial\bx\transpose}\bigg)
\transpose\bJ\transpose\widetilde{\by}_j^{(t)}\widetilde{\by}_j^{(t)\mathrm{T}}\bJ
\bigg(\frac{\partial^2\calT_\btheta}{\partial x_{k}\partial\bx\transpose}(\bx_i)\bigg)
\bigg\}\\ 
&\qquad + \sum_{t = 1}^m\sum_{j = 1}^n\bigg[\{A_{ij}^{(t)} - \calT_\btheta^{(t)}(\nu_{it})\bkappa(\calT_\bz(\bx_i))\transpose\widetilde{\by}_j^{(t)}\}\eta'''(\calT_\btheta^{(t)}(\nu_{it})\bkappa(\calT_\bz(\bx_i))\transpose\widetilde{\by}_j^{(t)})\\
&\qquad\qquad\qquad\quad - 2\eta''(\calT_\btheta^{(t)}(\nu_{it})\bkappa(\calT_\bz(\bx_i))\transpose\widetilde{\by}_j^{(t)})\bigg]\\
&\qquad\times 
\bigg\{\calT_\btheta^{3(t)}(\nu_{it})\be_k\transpose\bigg(\frac{\partial\calT_\bz(\bx_i)}{\partial\bx\transpose}\bigg)\transpose\bJ\transpose\widetilde{\by}_j^{(t)}\bigg\}
\bigg(\frac{\partial\calT_\bz(\bx_i)}{\partial\bx\transpose}\bigg)\transpose\bJ\transpose\widetilde{\by}_j^{(t)}\widetilde{\by}_j^{(t)\mathrm{T}}\bJ\bigg(\frac{\partial\calT_\bz(\bx_i)}{\partial\bx\transpose}\bigg)\\ 
&\qquad + \sum_{t = 1}^m\sum_{j = 1}^n\left[\{A_{ij}^{(t)} - \calT_\btheta^{(t)}(\nu_{it})\bkappa(\calT_\bz(\bx_i))\transpose\widetilde{\by}_j^{(t)}\}\eta'(\calT_\btheta^{(t)}(\nu_{it})\bkappa(\calT_\bz(\bx_i))\transpose\widetilde{\by}_j^{(t)})\right]\\ 
&\qquad\times\calT_\btheta^{(t)}(\nu_{it})\frac{\partial^2}{\partial x_k\partial\bx}\bigg\{\widetilde{\by}_j^{(t)\mathrm{T}}\bJ\transpose\bigg(\frac{\partial\calT_\bz(\bx_i)}{\partial\bx\transpose}\bigg)\bigg\}\\ 
&\qquad + \sum_{t = 1}^m\sum_{j = 1}^n\bigg[\{A_{ij}^{(t)} - \calT_\btheta^{(t)}(\nu_{it})\bkappa(\calT_\bz(\bx_i))\transpose\widetilde{\by}_j^{(t)}\}\eta''(\calT_\btheta^{(t)}(\nu_{it})\bkappa(\calT_\bz(\bx_i))\transpose\widetilde{\by}_j^{(t)})\\
&\qquad\qquad\qquad\quad - \eta'(\calT_\btheta^{(t)}(\nu_{it})\bkappa(\calT_\bz(\bx_i))\transpose\widetilde{\by}_j^{(t)})\bigg]\\ 
&\qquad\times \calT_\btheta^{(t)}(\nu_{it})\frac{\partial}{\partial\bx}\bigg\{\widetilde{\by}_j^{(t)\mathrm{T}}\bJ\transpose\bigg(\frac{\partial\calT_\bz(\bx_i)}{\partial\bx\transpose}\bigg)\bigg\},
\end{align*}
\begin{align*}
&\frac{\partial^2}{\partial\bx_i\partial\bx_i\transpose}\frac{\partial\calL_{in}}{\partial \nu_{it}}(\bx_i, \bnu_i)\\
&\quad = \sum_{j = 1}^n\bigg[\{A_{ij}^{(t)} - \calT_\btheta^{(t)}(\nu_{it})\bkappa(\calT_\bz(\bx_i))\transpose\widetilde{\by}_j^{(t)}\}\eta''(\calT_\btheta^{(t)}(\nu_{it})\bkappa(\calT_\bz(\bx_i))\transpose\widetilde{\by}_j^{(t)})\\
&\qquad\qquad - \eta'(\calT_\btheta^{(t)}(\nu_{it})\bkappa(\calT_\bz(\bx_i))\transpose\widetilde{\by}_j^{(t)})\bigg]\\ 
&\qquad\times \bigg\{2\calT_\btheta^{(t)}(\nu_{it})\frac{\mathrm{d}\calT_\btheta^{(t)}}{\mathrm{d}\nu_{t}}(\nu_{it})
\bigg(\frac{\partial\calT_\btheta}{\partial\bx\transpose}(\bx_i)\bigg)\transpose\bJ\transpose\widetilde{\by}_j^{(t)}\widetilde{\by}_j^{(t)\mathrm{T}}\bJ\bigg(\frac{\partial\calT_\btheta(\bx_i)}{\partial\bx\transpose}\bigg)
\bigg\}\\ 
&\qquad + \sum_{j = 1}^n\bigg[\{A_{ij}^{(t)} - \calT_\btheta^{(t)}(\nu_{it})\bkappa(\calT_\bz(\bx_i))\transpose\widetilde{\by}_j^{(t)}\}\eta'''(\calT_\btheta^{(t)}(\nu_{it})\bkappa(\calT_\bz(\bx_i))\transpose\widetilde{\by}_j^{(t)})\\
&\qquad\qquad\quad - 2\eta''(\calT_\btheta^{(t)}(\nu_{it})\bkappa(\calT_\bz(\bx_i))\transpose\widetilde{\by}_j^{(t)})\bigg]\\
&\qquad\times 
\bigg\{\frac{\mathrm{d}\calT_\btheta^{(t)}}{\mathrm{d}\nu_{t}}(\nu_{it})\bkappa(\calT_\btheta(\bx_i))\transpose\widetilde{\by}_j^{(t)}\bigg\}
\bigg(\frac{\partial\calT_\bz(\bx_i)}{\partial\bx\transpose}\bigg)\transpose\bJ\transpose\widetilde{\by}_j^{(t)}\widetilde{\by}_j^{(t)\mathrm{T}}\bJ\bigg(\frac{\partial\calT_\bz(\bx_i)}{\partial\bx\transpose}\bigg)\\ 
&\qquad + \sum_{j = 1}^n\left[\{A_{ij}^{(t)} - \calT_\btheta^{(t)}(\nu_{it})\bkappa(\calT_\bz(\bx_i))\transpose\widetilde{\by}_j^{(t)}\}\eta'(\calT_\btheta^{(t)}(\nu_{it})\bkappa(\calT_\bz(\bx_i))\transpose\widetilde{\by}_j^{(t)})\right]\\ 
&\qquad\times\frac{\mathrm{d}\calT_\btheta^{(t)}}{\mathrm{d}\nu_{t}}(\nu_{it})\frac{\partial}{\partial\bx}\bigg\{\widetilde{\by}_j^{(t)\mathrm{T}}\bJ\bigg(\frac{\partial\calT_\bz(\bx_i)}{\partial\bx\transpose}\bigg)\bigg\}\\ 
&\qquad + \sum_{j = 1}^n\bigg[\{A_{ij}^{(t)} - \calT_\btheta^{(t)}(\nu_{it})\bkappa(\calT_\bz(\bx_i))\transpose\widetilde{\by}_j^{(t)}\}\eta''(\calT_\btheta^{(t)}(\nu_{it})\bkappa(\calT_\bz(\bx_i))\transpose\widetilde{\by}_j^{(t)})\\
&\qquad\qquad\quad - \eta'(\calT_\btheta^{(t)}(\nu_{it})\bkappa(\calT_\bz(\bx_i))\transpose\widetilde{\by}_j^{(t)})\bigg]\\ 
&\qquad\times \frac{\mathrm{d}\calT_\btheta^{(t)}}{\mathrm{d}\nu_t}(\nu_{it})\bkappa(\calT_\bz(\bx_i))\transpose\widetilde{\by}_j^{(t)}
\bigg\{\calT_\btheta^{2(t)}(\nu_{it})\bigg(\frac{\partial\calT_\bz(\bx_i)}{\partial\bx\transpose}\bigg)\transpose\bJ\transpose\widetilde{\by}_j^{(t)}\widetilde{\by}_j^{(t)\mathrm{T}}\bJ\bigg(\frac{\partial\calT_\bz(\bx_i)}{\partial\bx\transpose}\bigg)\bigg\},
\end{align*}

\begin{align*}
&\frac{\partial^2}{\partial\nu_{it}^2}\frac{\partial\calL_{in}}{\partial \bx_{i}}(\bx_i, \bnu_i)\\
&\quad = \sum_{j = 1}^n\bigg[\{A_{ij}^{(t)} - \calT_\btheta^{(t)}(\nu_{it})\bkappa(\calT_\bz(\bx_i))\transpose\widetilde{\by}_j^{(t)}\}\eta''(\calT_\btheta^{(t)}(\nu_{it})\bkappa(\calT_\bz(\bx_i))\transpose\widetilde{\by}_j^{(t)})\\
&\qquad\qquad\quad - \eta'(\calT_\btheta^{(t)}(\nu_{it})\bkappa(\calT_\bz(\bx_i))\transpose\widetilde{\by}_j^{(t)})\bigg]\\ 
&\qquad\times \bigg\{\bigg(\frac{\mathrm{d}\calT_\btheta^{(t)}}{\mathrm{d}\nu_{t}}(\nu_{it})\bigg)^22\{\bkappa(\calT_\bz(\bx_i))\transpose\widetilde{\by}_j^{(t)}\}
\bigg(\frac{\partial\calT_\btheta}{\partial\bx\transpose}(\bx_i)\bigg)\transpose\bJ\transpose\widetilde{\by}_j^{(t)}
\bigg\}\\ 
&\qquad + \sum_{j = 1}^n\bigg[\{A_{ij}^{(t)} - \calT_\btheta^{(t)}(\nu_{it})\bkappa(\calT_\bz(\bx_i))\transpose\widetilde{\by}_j^{(t)}\}\eta'''(\calT_\btheta^{(t)}(\nu_{it})\bkappa(\calT_\bz(\bx_i))\transpose\widetilde{\by}_j^{(t)})\\
&\qquad\qquad\qquad - 2\eta''(\calT_\btheta^{(t)}(\nu_{it})\bkappa(\calT_\bz(\bx_i))\transpose\widetilde{\by}_j^{(t)})\bigg]\\
&\qquad\times 
\bigg(\frac{\mathrm{d}\calT_\btheta^{(t)}}{\mathrm{d}\nu_{t}}(\nu_{it})\bigg)^2\{\bkappa(\calT_\btheta(\bx_i))\transpose\widetilde{\by}_j^{(t)}\}^2\calT_\btheta^{(t)}(\nu_{it})
\bigg(\frac{\partial\calT_\bz(\bx_i)}{\partial\bx\transpose}\bigg)\transpose\bJ\transpose\widetilde{\by}_j^{(t)}\\ 
&\qquad + \sum_{j = 1}^n\left[\{A_{ij}^{(t)} - \calT_\btheta^{(t)}(\nu_{it})\bkappa(\calT_\bz(\bx_i))\transpose\widetilde{\by}_j^{(t)}\}\eta'(\calT_\btheta^{(t)}(\nu_{it})\bkappa(\calT_\bz(\bx_i))\transpose\widetilde{\by}_j^{(t)})\right]\\ 
&\qquad\times\frac{\mathrm{d}^2\calT_\btheta^{(t)}}{\mathrm{d}\nu_{t}^2}(\nu_{it})\bigg\{\bigg(\frac{\partial\calT_\bz(\bx_i)}{\partial\bx\transpose}\bigg)\transpose\bJ\transpose\widetilde{\by}_j^{(t)}\bigg\}\\ 
&\qquad + \sum_{j = 1}^n\bigg[\{A_{ij}^{(t)} - \calT_\btheta^{(t)}(\nu_{it})\bkappa(\calT_\bz(\bx_i))\transpose\widetilde{\by}_j^{(t)}\}\eta''(\calT_\btheta^{(t)}(\nu_{it})\bkappa(\calT_\bz(\bx_i))\transpose\widetilde{\by}_j^{(t)})\\
&\qquad\qquad\qquad - \eta'(\calT_\btheta^{(t)}(\nu_{it})\bkappa(\calT_\bz(\bx_i))\transpose\widetilde{\by}_j^{(t)})\bigg]\\ 
&\qquad\times \frac{\mathrm{d}^2\calT_\btheta^{(t)}}{\mathrm{d}\nu_t^2}(\nu_{it})\bkappa(\calT_\bz(\bx_i))\transpose\widetilde{\by}_j^{(t)}
\bigg\{\calT_\btheta^{(t)}(\nu_{it})\bigg(\frac{\partial\calT_\bz(\bx_i)}{\partial\bx\transpose}\bigg)\transpose\bJ\transpose\widetilde{\by}_j^{(t)}\bigg\},
\end{align*}

\begin{align*}
&\frac{\partial^3\calL_{in}}{\partial\nu_{it}^3}(\bx_i, \bnu_i)\\
&\quad = \sum_{j = 1}^n\bigg[\{A_{ij}^{(t)} - \calT_\btheta^{(t)}(\nu_{it})\bkappa(\calT_\bz(\bx_i))\transpose\widetilde{\by}_j^{(t)}\}\eta''(\calT_\btheta^{(t)}(\nu_{it})\bkappa(\calT_\bz(\bx_i))\transpose\widetilde{\by}_j^{(t)})\\
&\qquad\qquad\quad - \eta'(\calT_\btheta^{(t)}(\nu_{it})\bkappa(\calT_\bz(\bx_i))\transpose\widetilde{\by}_j^{(t)})\bigg]\\ 
&\qquad\times 2\bigg(\frac{\mathrm{d}\calT_\btheta^{(t)}}{\mathrm{d}\nu_{t}}(\nu_{it})\bigg)\bigg(\frac{\mathrm{d}^2\calT_\btheta^{(t)}}{\mathrm{d}\nu_{t}^2}(\nu_{it})\bigg)\{\bkappa(\calT_\bz(\bx_i))\transpose\widetilde{\by}_j^{(t)}\}^2\\ 
&\qquad + \sum_{j = 1}^n\bigg[\{A_{ij}^{(t)} - \calT_\btheta^{(t)}(\nu_{it})\bkappa(\calT_\bz(\bx_i))\transpose\widetilde{\by}_j^{(t)}\}\eta'''(\calT_\btheta^{(t)}(\nu_{it})\bkappa(\calT_\bz(\bx_i))\transpose\widetilde{\by}_j^{(t)})\\
&\qquad\qquad\qquad - 2\eta''(\calT_\btheta^{(t)}(\nu_{it})\bkappa(\calT_\bz(\bx_i))\transpose\widetilde{\by}_j^{(t)})\bigg]\\
&\qquad\times 
\bigg(\frac{\mathrm{d}\calT_\btheta^{(t)}}{\mathrm{d}\nu_{t}}(\nu_{it})\bigg)^3\{\bkappa(\calT_\btheta(\bx_i))\transpose\widetilde{\by}_j^{(t)}\}^3\\ 
&\qquad + \sum_{j = 1}^n\left[\{A_{ij}^{(t)} - \calT_\btheta^{(t)}(\nu_{it})\bkappa(\calT_\bz(\bx_i))\transpose\widetilde{\by}_j^{(t)}\}\eta'(\calT_\btheta^{(t)}(\nu_{it})\bkappa(\calT_\bz(\bx_i))\transpose\widetilde{\by}_j^{(t)})\right]\\ 
&\qquad\times\frac{\mathrm{d}^3\calT_\btheta^{(t)}}{\mathrm{d}\nu_{t}^3}(\nu_{it})\{\bkappa(\calT_\bz(\bx_i))\transpose\widetilde{\by}_j^{(t)}\}\\ 
&\qquad + \sum_{j = 1}^n\bigg[\{A_{ij}^{(t)} - \calT_\btheta^{(t)}(\nu_{it})\bkappa(\calT_\bz(\bx_i))\transpose\widetilde{\by}_j^{(t)}\}\eta''(\calT_\btheta^{(t)}(\nu_{it})\bkappa(\calT_\bz(\bx_i))\transpose\widetilde{\by}_j^{(t)})\\
&\qquad\qquad\qquad - \eta'(\calT_\btheta^{(t)}(\nu_{it})\bkappa(\calT_\bz(\bx_i))\transpose\widetilde{\by}_j^{(t)})\bigg]\\ 
&\qquad\times \bigg(\frac{\mathrm{d}\calT_\btheta^{(t)}(\nu_{it})}{\mathrm{d}\nu_t}\bigg)\frac{\mathrm{d}^2\calT_\btheta^{(t)}}{\mathrm{d}\nu_t^2}(\nu_{it})\{\bkappa(\calT_\bz(\bx_i))\transpose\widetilde{\by}_j^{(t)}\}^2
\end{align*}
Following the proof of Theorem \ref{thm:BvM}, it is clear that 
\begin{align*}
\sup_{\bx_i\in\mathbb{R}^{d - 1}, \bnu_i\in\mathbb{R}^m}\bigg\|\frac{\partial^3\calL_{in}}{\partial\bx_i\partial\bx_i\transpose\partial x_{ik}}(\bx_i, \bnu_i)\bigg\|_2
& = \Optilde(mn),\\
\max_{t\in[m]}\sup_{\bx_i\in\mathbb{R}^{d - 1}, \bnu_i\in\mathbb{R}^m}\bigg\|\frac{\partial^3\calL_{in}}{\partial\bx_i\partial\bx_i\transpose\partial \nu_{it}}(\bx_i, \bnu_i)\bigg\|_2
& = \Optilde(n),\\
\sup_{\bx_i\in\mathbb{R}^{d - 1}, \bnu_i\in\mathbb{R}^m}\bigg\|\frac{\partial^3\calL_{in}}{\partial\nu_{it}^2\partial\bx_i}(\bx_i, \bnu_i)\bigg\|_2
& = \Optilde(n),\\
\max_{t\in[m]}\sup_{\bx_i\in\mathbb{R}^{d - 1}, \bnu_i\in\mathbb{R}^m}\bigg\|\frac{\partial^3\calL_{in}}{\partial\nu_{it}^3}(\bx_i, \bnu_i)\bigg\|_2
& = \Optilde(n).
\end{align*}
Therefore, by mean-value theorem, we have
\begin{align*}
&\bigg\|\frac{\partial^2\calL_{in}}{\partial\bx_i\partial\bx_i\transpose}(\bx_i, \bnu_i)- \frac{\partial^2\calL_{in}}{\partial\bx_i\partial\bx_i\transpose}(\bx_i', \bnu_i')\bigg\|_2\\
&\quad\leq d\max_{k,l\in[d - 1]}\bigg|\frac{\partial^2\calL_{in}}{\partial x_{ik}\partial x_{il}}(\bx_i, \bnu_i) - \frac{\partial^2\calL_{in}}{\partial x_{ik}\partial x_{il}}(\bx_i', \bnu_i')\bigg|\\
&\quad\leq d\max_{k,l\in[d - 1]}\sup_{\bar{\bx}_i\in\mathbb{R}^{d - 1},\bar{\bnu}_i\in\mathbb{R}^{m}}\bigg\|\frac{\partial^3\calL_{in}}{\partial\bx_i\partial x_{ik}\partial x_{il}}(\bar{\bx}_i, \bar{\bnu}_i)\bigg\|_2\|\bx_i - \bx_i'\|_2 \\ 
&\qquad + d\max_{k,l\in[d - 1]}\sup_{\bar{\bx}_i\in\mathbb{R}^{d - 1},\bar{\bnu}_i\in\mathbb{R}^{m}}\sum_{t = 1}^m\bigg|\frac{\partial^3\calL_{in}}{\partial\nu_{it}\partial x_{ik}\partial x_{il}}(\bar{\bx}_i, \bar{\bnu}_i)\bigg||\nu_{it} - \nu_{it}'|\\
&\quad = \Optilde(mn)\|\bx_i - \bx_i'\|_2 + \Optilde(mn)\|\bnu_i - \bnu_{0i}\|_\infty,\\
&\bigg\|\frac{\partial^2\calL_{in}}{\partial\bx_i\partial\bnu_i\transpose}(\bx_i, \bnu_i)- \frac{\partial^2\calL_{in}}{\partial\bx_i\partial\bnu_i\transpose}(\bx_i', \bnu_i')\bigg\|_2\\
&\quad\leq \sqrt{md}\max_{k\in[d - 1], t\in[m]}\bigg|\frac{\partial^2\calL_{in}}{\partial x_{ik}\partial \nu_{it}}(\bx_i, \bnu_i) - \frac{\partial^2\calL_{in}}{\partial x_{ik}\partial \nu_{it}}(\bx_i', \bnu_i')\bigg|_2\\ 
&\quad\leq \sqrt{md}\max_{k\in[d - 1], t\in[m]}\sup_{\bar{\bx}_i\in\mathbb{R}^{d - 1},\bar{\bnu}_i\in\mathbb{R}^{m}}\bigg\|\frac{\partial^3\calL_{in}}{\partial\bx_i\partial x_{ik}\partial \nu_{it}}(\bar{\bx}_i, \bar{\bnu}_i)\bigg\|_2\|\bx_i - \bx_i'\|_2 \\ 
&\qquad + \sqrt{md}\max_{k\in[d - 1], t\in[m]}\sup_{\bar{\bx}_i\in\mathbb{R}^{d - 1},\bar{\bnu}_i\in\mathbb{R}^{m}}\sum_{s = 1}^m\bigg|\frac{\partial^3\calL_{in}}{\partial\nu_{it}\partial \bx_{i}\partial \nu_{is}}(\bar{\bx}_i, \bar{\bnu}_i)\bigg||\nu_{is} - \nu_{is}'|\\
&\quad = \Optilde(\sqrt{m}n)\|\bx_i - \bx_i'\|_2 + \Optilde(\sqrt{m}n)\|\bnu_i - \bnu_i'\|_\infty,\\
&\bigg\|\frac{\partial^2\calL_{in}}{\partial\bnu_i\partial\bnu_i\transpose}(\bx_i, \bnu_i)- \frac{\partial^2\calL_{in}}{\partial\bnu_i\partial\bnu_i\transpose}(\bx_i', \bnu_i')\bigg\|_2\\
&\quad = \max_{t\in[m]}\bigg|\frac{\partial^2\calL_{in}}{\partial\nu_{it}^2}(\bx_i, \bnu_i)- \frac{\partial^2\calL_{in}}{\partial\nu_{it}^2}(\bx_i', \bnu_i')\bigg|\\
&\quad\leq \max_{t\in[m]}\bigg\|\frac{\partial^3\calL_{in}}{\partial\bx_i\partial^2\nu_{it}}(\bar{\bx}_i, \bar{\bnu}_i)\bigg\|_2\|\bx_i - \bx_i'\|_2
 + \max_{t\in[m]}\sum_{s = 1}^m\bigg\|\frac{\partial^3}{\partial\nu_{is}\partial\nu_{it}^2}(\bar{\bx}_i, \bar{\bnu}_i)\bigg\|_2|\nu_{it} - \nu_{it}'|\\
&\quad = \Optilde(n)(\|\bx_i - \bx_i'\|_2 + \|\bnu_i - \bnu_i'\|_\infty).
\end{align*}
The proof is therefore completed by combining the above Lipschitz error bounds. 
\end{proof}

\begin{lemma}[Local Convergence of Hessian, Transformed]
\label{lemma:Hessian_local_convergence_transform}
Suppose the conditions of Theorem \ref{thm:BvM} hold. Then, for any $i\in[n]$ and any positive sequence $(\eps_n)_{n = 1}^\infty$ converging to $0$, we have
\begin{align*}
&\sup_{\substack{\|\bx_i - \bx_{0i}\|_2\leq \eps_n\\ \|\bnu_i - \bnu_{0i}\|_\infty\leq \eps_n}}
\left\|\frac{1}{mn}
\frac{\partial^2\calL_{in}}{\partial\bx_i\partial\bx_i\transpose}(\bx_i, \bnu_i) + \frac{1}{m}\sum_{t = 1}^m\theta_{0i}^{2(t)}\bigg(\frac{\partial\calT_\bz}{\partial\bx\transpose}(\bx_{0i})\bigg)\transpose\bJ\transpose\bG_{0in}^{(t)}\bJ\bigg(\frac{\partial\calT_\bz}{\partial\bx\transpose}(\bx_{0i})\bigg)\right\|_2\\
&\quad = \Optilde\bigg\{\frac{(\log n)^\xi}{\sqrt{n}} + \eps_n\bigg\},\\
&\sup_{\substack{\|\bx_i - \bx_{0i}\|_2\leq \eps_n\\ \|\bnu_i - \bnu_{0i}\|_\infty\leq \eps_n}}
\left\|\frac{1}{mn}
\frac{\partial^2\calL_{in}}{\partial\bnu_i\partial\bx_i\transpose}(\bx_i, \bnu_i) + 
\frac{1}{m}\bigg(\frac{\partial\calT_\btheta}{\partial\btheta}(\btheta_{0i})\bigg)\transpose
\begin{bmatrix}
\theta_{0i}^{(1)}\bz_{0i}\transpose\bG_{0in}^{(1)}\bJ\\
\vdots\\
\theta_{0i}^{(m)}\bz_{0i}\transpose\bG_{0in}^{(m)}\bJ
\end{bmatrix}\bigg(\frac{\partial\calT_\bz}{\partial\bx\transpose}(\bx_{0i})\bigg)
\right\|_2\\
&\quad = \Optilde\bigg\{\frac{(\log n)^\xi}{\sqrt{mn}} + \frac{\eps_n}{\sqrt{m}}\bigg\},\\
&\sup_{\substack{\|\bx_i - \bx_{0i}\|_2\leq \eps_n\\ \|\bnu_i - \bnu_{0i}\|_\infty\leq \eps_n}}
\left\|\frac{1}{mn}
\frac{\partial^2\calL_{in}}{\partial\bnu_i\partial\bnu_i\transpose}(\bx_i, \bnu_i)
 + \bigg(\frac{\partial\calT_\btheta}{\partial\btheta}(\btheta_{0i})\bigg)\transpose
 \textsf{diag}\bigg[\frac{1}{m}\bz_{0i}\transpose\bG_{0in}^{(t)}\bz_{0i}:t\in[m]\bigg]
 \bigg(\frac{\partial\calT_\btheta}{\partial\btheta}(\btheta_{0i})\bigg)
 \right\|_2\\
&\quad = \Optilde\bigg\{\frac{(\log n)^\xi}{m\sqrt{n}} + \frac{\eps_n}{m}\bigg\}.
\end{align*}
\end{lemma}

\begin{proof}[\bf Proof]
Direct computation shows
\begin{align*}
&\frac{\partial^2\calL_{in}}{\partial\bx_i\partial\bx_i\transpose}(\bx_i, \bnu_i)\\
&\quad = \sum_{t = 1}^m\sum_{j = 1}^n\bigg[\{A_{ij}^{(t)} - \calT_\btheta^{(t)}(\nu_{it})\bkappa(\calT_\bz(\bx_i))\transpose\widetilde{\by}_j^{(t)}\}\eta''(\calT_\btheta^{(t)}(\nu_{it})\bkappa(\calT_\bz(\bx_i))\transpose\widetilde{\by}_j^{(t)})\\
&\qquad\qquad\qquad - \eta'(\calT_\btheta^{(t)}(\nu_{it})\bkappa(\calT_\bz(\bx_i))\transpose\widetilde{\by}_j^{(t)})\bigg]\\
&\qquad\times\calT_\btheta^{2(t)}(\nu_{it})\bigg(\frac{\partial\calT_\bz}{\partial\bx\transpose}(\bx_i)\bigg)\transpose\bJ\transpose\widetilde{\by}_j^{(t)}\widetilde{\by}_j^{(t)\mathrm{T}}\bJ\bigg(\frac{\partial\calT_\bz}{\partial\bx\transpose}(\bx_i)\bigg)\\ 
&\qquad + \sum_{t = 1}^m\sum_{j = 1}^n\{A_{ij}^{(t)} - \calT_\btheta^{(t)}(\nu_{it})\bkappa(\calT_\bz(\bx_i))\transpose\widetilde{\by}_j^{(t)}\}\eta'(\calT_\btheta^{(t)}(\nu_{it})\bkappa(\calT_\bz(\bx_i))\transpose\widetilde{\by}_j^{(t)})\calT_\btheta^{(t)}(\nu_{it})\\
&\qquad\times \frac{\partial}{\partial\bx_i}\bigg\{\widetilde{\by}_j^{(t)\mathrm{T}}\bJ \bigg(\frac{\partial\calT_\bz}{\partial\bx\transpose}(\bx_i)\bigg)\bigg\},\\
&\frac{\partial^2\calL_{in}}{\partial\bnu_i\partial\bnu_i\transpose}(\bx_i, \bnu_i)\\
&\quad = \sum_{j = 1}^n\left[\{A_{ij}^{(t)} - \calT_\btheta^{(t)}(\nu_{it})\bkappa(\calT_\bz(\bx_i))\transpose\widetilde{\by}_j^{(t)}\}\eta''(\calT_\btheta^{(t)}(\nu_{it})\bkappa(\calT_\bz(\bx_i))\transpose\widetilde{\by}_j^{(t)}) - \eta'(\calT_\btheta^{(t)}(\nu_{it})\bkappa(\calT_\bz(\bx_i))\transpose\widetilde{\by}_j^{(t)})\right]\\
&\qquad\times\bigg(\frac{\partial\calT_\btheta^{(t)}}{\partial\nu_t}(\nu_{it})\bigg)^2\{\kappa(\calT_\bz(\bx_i))\transpose\widetilde{\by}_j^{(t)}\}^2\\ 
&\qquad + \sum_{j = 1}^n\{A_{ij}^{(t)} - \calT_\btheta^{(t)}(\nu_{it})\bkappa(\calT_\bz(\bx_i))\transpose\widetilde{\by}_j^{(t)}\}\eta'(\calT_\btheta^{(t)}(\nu_{it})\bkappa(\calT_\bz(\bx_i))\transpose\widetilde{\by}_j^{(t)})
\frac{\partial^2\calT_\btheta^{(t)}}{\partial\nu_t^2}(\nu_{it})\bkappa(\calT_\bz(\bx_i))\transpose\widetilde{\by}_j^{(t)},\\
&\frac{\partial^2\calL_{in}}{\partial\bx_i\partial\nu_{it}}(\bx_i, \bnu_i)\\
&\quad = \sum_{j = 1}^n\left[\{A_{ij}^{(t)} - \calT_\btheta^{(t)}(\nu_{it})\bkappa(\calT_\bz(\bx_i))\transpose\widetilde{\by}_j^{(t)}\}\eta''(\calT_\btheta^{(t)}(\nu_{it})\bkappa(\calT_\bz(\bx_i))\transpose\widetilde{\by}_j^{(t)}) - \eta'(\calT_\btheta^{(t)}(\nu_{it})\bkappa(\calT_\bz(\bx_i))\transpose\widetilde{\by}_j^{(t)})\right]\\
&\qquad\times\calT_\btheta^{(t)}(\nu_{it}) \bigg(\frac{\mathrm{d}\calT_\btheta^{(t)}}{\partial\nu_t}(\nu_{it})\bigg) \bigg(\frac{\partial\calT_\bz}{\partial\bx\transpose}(\bx_i)\bigg)\transpose\bJ\transpose\widetilde{\by}_j^{(t)}\widetilde{\by}_j^{(t)\mathrm{T}}\bkappa(\calT_\btheta(\bx_i))\\ 
&\qquad + \sum_{j = 1}^n\{A_{ij}^{(t)} - \calT_\btheta^{(t)}(\nu_{it})\bkappa(\calT_\bz(\bx_i))\transpose\widetilde{\by}_j^{(t)}\}\eta'(\calT_\btheta^{(t)}(\nu_{it})\bkappa(\calT_\bz(\bx_i))\transpose\widetilde{\by}_j^{(t)})
\frac{\mathrm{d}\calT_\btheta^{(t)}}{\mathrm{d}\nu_t}(\nu_{it})\bigg(\frac{\partial\calT_\bz}{\partial\bx\transpose}(\bx_i)\bigg)\transpose\bJ\widetilde{\by}_j^{(t)}.
\end{align*}
Note that by the global Lipschitz continuity of $\calT_\bz$ and $\calT_\btheta$, there exists a constant $C > 0$, such that
\begin{align*}
&\{(\bx_i, \bnu_i)\in\mathbb{R}^{d - 1}\times\mathbb{R}^m:\|\bx_i - \bx_{0i}\|_2 \leq \eps_n, \|\bnu_i - \bnu_{0i}\|_\infty \leq \eps_n\}\\ 
&\quad\subset
\{(\bx_i, \bnu_i)\in\mathbb{R}^{d - 1}\times\mathbb{R}^m:\|\bkappa(\calT_\bz(\bx_i)) - \bz_{0i}\|_2 \leq C\eps_n, \|\calT_\btheta(\bnu_i) - \btheta_{0i}\|_\infty \leq C\eps_n\}.
\end{align*}
Then, following the proof of Lemma \ref{lemma:Hessian_local_convergence}, we know that
\begin{align*}
&\sup_{\substack{\|\bx_i - \bx_{0i}\|_2 \leq \eps_n\\\|\bnu_i - \bnu_{0i}\|_\infty\leq \eps_n }}\bigg\|\frac{1}{mn}\frac{\partial^2\calL_{in}}{\partial\bx_i\partial\bx_i\transpose}(\bx_i, \bnu_i) + \frac{1}{m}\sum_{t = 1}^m\calT_\btheta^{2(t)}(\nu_{it})\bigg(\frac{\partial\calT_\bz}{\partial\bx\transpose}(\bx_i)\bigg)\transpose\bJ\transpose
  \bG_{0in}^{(t)}\bJ
  \bigg(\frac{\partial\calT_\bz}{\partial\bx\transpose}(\bx_i)\bigg)\bigg\|_2\\ 
&\quad = \Optilde\bigg\{\frac{(\log n)^\xi}{\sqrt{n}} + \eps_n\bigg\},\\
&\sup_{\substack{\|\bx_i - \bx_{0i}\|_2 \leq \eps_n\\\|\bnu_i - \bnu_{0i}\|_\infty\leq \eps_n }}\bigg\|\frac{1}{mn}\frac{\partial^2\calL_{in}}{\partial\nu_{it}^2}(\bx_i, \bnu_i) + \frac{1}{m}\bigg(\frac{\mathrm{d}\calT_\btheta^{(t)}}{\mathrm{d}\nu_t}(\nu_{it})\bigg)^2\bkappa(\calT_\bz(\bx_i))\transpose
  \bG_{0in}^{(t)}\bkappa(\calT_\bz(\bx_i))
\bigg\|_2\\ 
&\quad = \Optilde\bigg\{\frac{(\log n)^\xi}{m\sqrt{n}} + \frac{\eps_n}{m}\bigg\},\\
&\sup_{\substack{\|\bx_i - \bx_{0i}\|_2 \leq \eps_n\\\|\bnu_i - \bnu_{0i}\|_\infty\leq \eps_n }}\bigg\|\frac{1}{mn}\frac{\partial^2\calL_{in}}{\partial\bx_i\partial\nu_{it}}(\bx_i, \bnu_i) + \frac{1}{m}\calT_\btheta^{(t)}(\nu_{it})\bigg(\frac{\mathrm{d}\calT_\btheta^{(t)}}{\mathrm{d}\nu_t}(\nu_{it})\bigg)\bigg(\frac{\partial\calT_\bz}{\partial\bx\transpose}(\bx_i)\bigg)\bJ\transpose
  \bG_{0in}^{(t)}\bkappa(\calT_\bz(\bx_i))
\bigg\|_2\\
&\quad = \Optilde\bigg\{\frac{(\log n)^\xi}{m\sqrt{n}} + \frac{\eps_n}{m}\bigg\}.
\end{align*}
The proof is then completed by the Lipschitz continuity of $\calT_\btheta$, $\partial\calT_\btheta/\partial\bnu\transpose$, and $\partial\calT_\bz/\partial\bx\transpose$. 
\end{proof}

\begin{lemma}\label{lemma:VB_lemma}
Suppose the conditions of Theorem \ref{thm:BvM} and Assumption \ref{assumption:prior_transform} hold. 
Denote by $\widehat{\bx}_i = \calT_\bz^{-1}(\widehat{\bz}_i^*)$ and $\widehat{\bnu}_i = \calT_\btheta^{-1}(\widehat{\btheta}_i)$. Then,
\begin{align*}
D_{\mathrm{KL}}\bigg(\phi\bigg(\begin{bmatrix}\bx_i\\\bnu_i\end{bmatrix}\mathrel{\bigg|}\begin{bmatrix}\widehat{\bx}_i\\\widehat{\bnu}_i\end{bmatrix}, \frac{1}{{mn}}(\bGamma_{in}^*)^{-1}\bigg)\bigg\|\pi_{(\bx_i, \bnu_i)}(\bx_i, \bnu_i\mid\mathbb{A})\bigg) = \Optilde\bigg\{\frac{\sqrt{m}(\log n)^{3\xi}}{\sqrt{n}}\bigg\}.
\end{align*}
\end{lemma}

\begin{proof}[\bf Proof]
Recall that the proof of Theorem \ref{thm:BvM} implies
\begin{align*}
\bigg|\frac{\iint_{\calS^{d - 1}\times\calK^m} \pi_{\bz^*}(\bz_i^*)\prod_{t = 1}^m\pi_\theta(\theta_i^{(t)})\exp\{\widetilde{\ell}_{in}(\bz_i^*, \btheta_i) - \widetilde{\ell}_{in}(\widehat{\bz}_i^*, \widehat{\btheta}_i)\}\mathrm{d}\bz_i^*\mathrm{d}\btheta_i }{\pi_{\bz^*}(\bz_{0i}^*)\prod_{t = 1}^m\pi_\theta(\theta_{0i}^{(t)}) \det\{2\pi(mn\bGamma_{in})^{-1}\} } - 1\bigg| = \Optilde\bigg\{\frac{\sqrt{m}(\log n)^{3\xi}}{\sqrt{n}}\bigg\}.
\end{align*}  
Since the change of variable on $\bz_i^* = \calT_\bz(\bx_i)$, $\btheta_i = \calT_\btheta(\bnu_i)$ leaves the marginal likelihood invariant, this immediately implies that
\begin{align*}
\bigg|\frac{\iint_{\mathbb{R}^{d - 1}\times\mathbb{R}^m} \pi_{\bx}(\bx_i)\prod_{t = 1}^m\pi_\nu(\nu_{it})\exp\{\calL_{in}(\bx_i, \bnu_i) - \calL_{in}(\widehat{\bx}_i, \widehat{\bnu}_i)\}\mathrm{d}\bx_i\mathrm{d}\bnu_i }{\pi_{\bx}(\bx_{0i})\prod_{t = 1}^m\pi_\nu(\nu_{0it}) \det\{2\pi(mn\bGamma_{in}^*)^{-1}\} } - 1\bigg| = \Optilde\bigg\{\frac{\sqrt{m}(\log n)^{3\xi}}{\sqrt{n}}\bigg\}.
\end{align*}  
Now by the definition of the KL-divergence, we have
\begin{align*}
&D_{\mathrm{KL}}\bigg(\phi\bigg(\begin{bmatrix}\bx_i\\\bnu_i\end{bmatrix}\mathrel{\bigg|}\begin{bmatrix}\widehat{\bx}_i\\\widehat{\bnu}_i\end{bmatrix}, \frac{1}{{mn}}\bGamma_{in}^{-1}\bigg)\bigg\|\pi_{(\bx_i, \bnu_i)}(\bx_i, \bnu_i\mid\mathbb{A})\bigg)\\
&\quad = \bigg|\iint_{\mathbb{R}^{d - 1}\times\mathbb{R}^m}
\log\frac{\phi([\bx_i\transpose,\bnu_i\transpose]\transpose\mid[\widehat{\bx}_i\transpose,\widehat{\bnu}_i\transpose]\transpose,(mn\bGamma_{in}^*)^{-1})}
  {\pi_{(\bx_i, \bnu_i)}(\bx_i, \bnu_i)}
  \phi\bigg(\begin{bmatrix}\bx_i\\\bnu_i\end{bmatrix}\mathrel{\bigg|}\begin{bmatrix}\widehat{\bx}_i\\\widehat{\bnu}_i\end{bmatrix}, \frac{1}{{mn}}\bGamma_{in}^{-1}\bigg)
\mathrm{d}\bx_i\mathrm{d}\bnu_i\bigg|.
\end{align*}
By Taylor's theorem, there exists some $\bar{\bx}_i,\bx_i'$ between $\bx_i$ and $\widehat{\bx}_i$, and $\bar{\bnu}_i,\bnu_i'$ between $\bnu_i$ and $\widehat{\bnu}_i$, such that
\begin{align*}
&\log\frac{\phi([\bx_i\transpose,\bnu_i\transpose]\transpose\mid[\widehat{\bx}_i\transpose,\widehat{\bnu}_i\transpose]\transpose,(mn\bGamma_{in}^*)^{-1})}
  {\pi_{(\bx_i, \bnu_i)}(\bx_i, \bnu_i\mid\mathbb{A})}\\
&\quad = \calL_{in}(\widehat{\bx}_i, \widehat{\bnu}_i) - \calL_{in}(\bx_i, \bnu_i) - \frac{mn}{2}\begin{bmatrix}
  \bx_i\transpose - \widehat{\bx}_i\transpose & 
  \bnu_i\transpose - \widehat{\bnu}_i\transpose
\end{bmatrix}\bGamma_{in}^*\begin{bmatrix}
  \bx_i - \widehat{\bx}_i\\
  \bnu_i - \widehat{\bnu}_i
\end{bmatrix}\\
&\qquad + \log\bigg(1 + \Optilde\bigg\{\frac{\sqrt{m}(\log n)^{3\xi}}{\sqrt{n}}\bigg\}\bigg)
 + \log\pi_\bx(\bx_{0i}) + \sum_{t = 1}^m\log\pi_\nu(\nu_{0it}) - \log\pi_\bx(\bx_i)\\
&\qquad - \sum_{t = 1}^m\log\pi_\nu(\nu_{it})\\
&\quad = -\frac{mn}{2}\begin{bmatrix}
  \bx_i\transpose - \widehat{\bx}_i\transpose & 
  \bnu_i\transpose - \widehat{\bnu}_i\transpose
\end{bmatrix}\{\mathscr{H}_{\calL_{in}}(\bar{\bx}_i, \bar{\bnu}_i) + \bGamma_{in}^*\}\begin{bmatrix}
  \bx_i - \widehat{\bx}_i\\
  \bnu_i - \widehat{\bnu}_i
\end{bmatrix} + \log\bigg(1 + \Optilde\bigg\{\frac{\sqrt{m}(\log n)^{3\xi}}{\sqrt{n}}\bigg\}\bigg)\\ 
&\qquad - \frac{\partial}{\partial\bx_i\transpose}\log\pi_\bx(\bx_i')(\bx_i - \bx_{0i}) - \sum_{t = 1}^m\frac{\mathrm{d}}{\mathrm{d}\nu_{it}}\log\pi_\nu(\nu_{it}')(\nu_{it} - \nu_{0it}).
\end{align*}
Now we consider the event
\begin{align*}
\calB_{in} = \{(\bx_i, \bnu_i)\in\mathbb{R}^{d - 1}\times\mathbb{R}^m:\sqrt{mn}(\|\bx_i - \widehat{\bx}_i\|_2 + \|\bnu_i - \widehat{\bnu}_i\|_\infty)\leq(\log n)^\xi\}.
\end{align*}
Clearly, by Lemma \ref{lemma:aggregated_MLE_theory} by the Lipschitz property of $\calT_\bz$ and $\calT_\btheta$ and, for any $(\bx_i, \bnu_i)\in\calB_{in}$, 
\begin{align*}
\|\bx_i - \bx_{0i}\|_2&\leq \|\bx_i - \widehat{\bx}_i\|_2 + \|\widehat{\bx}_i - \bx_{0i}\|_2
\leq \frac{(\log n)^\xi}{\sqrt{mn}} + \sup_{\bx\in\mathbb{R}^{d - 1}}\bigg\|\frac{\partial\calT_\bz(\bx)}{\partial\bx\transpose}\bigg\|_2\|\widehat{\bz}_i^* - \bz_{0i}^*\|_2\\
&\leq \frac{C(\log n)^\xi}{\sqrt{mn}}\quad\text{w.h.p..},\\
\|\bnu_i - \bnu_{0i}\|_\infty
&\leq \|\bnu_i - \widehat{\bnu}_i\|_\infty + \|\widehat{\bnu}_i - \bnu_{0i}\|_\infty
\leq \frac{(\log n)^\xi}{\sqrt{mn}} + \sup_{\bnu\in\mathbb{R}^{m}}\bigg\|\frac{\partial\calT_\btheta(\btheta)}{\partial\btheta\transpose}\bigg\|_\infty\|\widehat{\btheta}_i - \btheta_{0i}\|_\infty\\
&\leq \frac{C(\log n)^\xi}{\sqrt{mn}}\quad\text{w.h.p..}
\end{align*}
Similarly, the same high-probability error bound also holds for $\|\bar{\bx}_i - \bx_{0i}\|_2$, $\|\bx_i' - \bx_{0i}\|_2$, $\|\bar{\bnu}_i - \bnu_{0i}\|_2$, and $\|\bnu_i' - \bnu_{0i}\|_2$. Then by Lemma \ref{lemma:Hessian_local_convergence_transform} and Assumption \ref{assumption:prior_transform}, 
\begin{align*}
&\bigg|\iint_{\calB_{in}}\log\frac{\phi([\bx_i\transpose,\bnu_i\transpose]\transpose\mid[\widehat{\bx}_i\transpose,\widehat{\bnu}_i\transpose]\transpose,(mn\bGamma_{in}^*)^{-1})}
  {\pi_{(\bx_i, \bnu_i)}(\bx_i, \bnu_i\mid\mathbb{A})}
  \phi\bigg(\begin{bmatrix}\bx_i\\\bnu_i\end{bmatrix}\mathrel{\bigg|}\begin{bmatrix}\widehat{\bx}_i\\\widehat{\bnu}_i\end{bmatrix}, \frac{1}{{mn}}\bGamma_{in}^{-1}\bigg)\mathrm{d}\bx_i\mathrm{d}\nu_i
  \bigg|\\
&\quad\leq \sup_{(\bx_i,\bnu_i)\in\calB_{in}}\bigg|\log\frac{\phi([\bx_i\transpose,\bnu_i\transpose]\transpose\mid[\widehat{\bx}_i\transpose,\widehat{\bnu}_i\transpose]\transpose,(mn\bGamma_{in}^*)^{-1})}
  {\pi_{(\bx_i, \bnu_i)}(\bx_i, \bnu_i\mid\mathbb{A})}\bigg|\\
&\quad\leq \Optilde\bigg\{\frac{(\log n)^\xi}{\sqrt{n}}\bigg\}(\|\bx_i - \widehat{\bx}_i\|_2^2 + \|\bnu_i - \widehat{\bnu}_i\|_2^2) + \Optilde\bigg\{\frac{\sqrt{m}(\log n)^{3\xi}}{\sqrt{n}}\bigg\}\\
&\qquad + C\|\bx_i - \bx_{0i}\|_2 + Cm\|\bnu_i - \bnu_{0i}\|_\infty\\ 
&\quad = \Optilde\bigg\{\frac{\sqrt{m}(\log n)^{3\xi}}{\sqrt{n}}\bigg\}.
\end{align*}
Now over $\calB_{in}^c$, we know that
\begin{align*}
\|\bar\bx_i - \widehat{\bx}_i\|_2
&\leq \|\bar\bx_i - \bx_{0i}\|_2 + \|\widehat{\bx}_i - \bx_{0i}\|_2\leq \|\bx_i - \bx_{0i}\|_2 + \|\widehat{\bx}_i - \bx_{0i}\|_2\leq 2\|\bx_i - \widehat{\bx}_i\|_2 + \|\widehat{\bx}_i - \bx_{0i}\|_2,\\
&\|\bar\bnu_i - \widehat{\bnu}_i\|_\infty\\
&\leq \|\bar\bnu_i - \bnu_{0i}\|_\infty + \|\widehat{\bnu}_i - \bnu_{0i}\|_\infty\leq \|\bnu_i - \bnu_{0i}\|_\infty + \|\widehat{\bnu}_i - \bnu_{0i}\|_\infty\leq 2\|\bnu_i - \widehat{\bnu}_i\|_\infty + \|\widehat{\bnu}_i - \bnu_{0i}\|_\infty.
\end{align*}
By Lemma \ref{lemma:Lipschitz_Hessian} and Lemma \ref{lemma:Hessian_local_convergence_transform}, we obtain
\begin{align*}
&\bigg|\frac{mn}{2}\begin{bmatrix}
  \bx_i\transpose - \widehat{\bx}_i\transpose & 
  \bnu_i\transpose - \widehat{\bnu}_i\transpose
\end{bmatrix}\{\mathscr{H}_{\calL_{in}}(\bar{\bx}_i, \bar{\bnu}_i) + \bGamma_{in}^*\}\begin{bmatrix}
  \bx_i - \widehat{\bx}_i\\
  \bnu_i - \widehat{\bnu}_i
\end{bmatrix}\bigg|\\
&\quad\leq \bigg|\frac{mn}{2}\begin{bmatrix}
  \bx_i\transpose - \widehat{\bx}_i\transpose & 
  \bnu_i\transpose - \widehat{\bnu}_i\transpose
\end{bmatrix}\{\mathscr{H}_{\calL_{in}}(\widehat{\bx}_i, \widehat{\bnu}_i) + \bGamma_{in}^*\}\begin{bmatrix}
  \bx_i - \widehat{\bx}_i\\
  \bnu_i - \widehat{\bnu}_i
\end{bmatrix}\bigg|\\
&\qquad + \bigg|\frac{mn}{2}\begin{bmatrix}
  \bx_i\transpose - \widehat{\bx}_i\transpose & 
  \bnu_i\transpose - \widehat{\bnu}_i\transpose
\end{bmatrix}\{\mathscr{H}_{\calL_{in}}(\widehat{\bx}_i, \widehat{\bnu}_i) - \mathscr{H}_{\calL_{in}}(\bar{\bx}_i, \bar{\bnu}_i)\}\begin{bmatrix}
  \bx_i - \widehat{\bx}_i\\
  \bnu_i - \widehat{\bnu}_i
\end{bmatrix}\bigg|\\
&\quad\leq \Optilde\bigg\{\frac{(\log n)^\xi}{\sqrt{n}}\bigg\}mn(\|\bx_i - \widehat{\bx}_i\|_2^2 + \|\bnu_i - \widehat{\bnu}_i\|_\infty^2)\\
&\qquad + \Optilde(mn)(\|\widehat{\bx}_i - \bx_i\|_2^2 + \|\bnu_i - \widehat{\bnu}_i\|_\infty^2)(\|\bar{\bx}_i - \widehat{\bx}_i\|_2 + \|\bar{\bnu}_i - \widehat{\bnu}_i\|_\infty)\\
&\quad\leq \Optilde\bigg\{\frac{(\log n)^\xi}{\sqrt{n}}\bigg\}mn(\|\bx_i - \widehat{\bx}_i\|_2^2 + \|\bnu_i - \widehat{\bnu}_i\|_\infty^2)\\
&\qquad + \Optilde(mn)(\|\widehat{\bx}_i - \bx_i\|_2^2 + \|\bnu_i - \widehat{\bnu}_i\|_\infty^2)\bigg[\|\bx_i - \widehat{\bx}_i\|_2 + \|\bnu_i - \widehat{\bnu}_i\|_\infty + \Optilde\bigg\{\frac{(\log n)^\xi}{\sqrt{mn}}\bigg\}\bigg]\\
&\quad\leq \Optilde\bigg\{\frac{(\log n)^\xi}{\sqrt{n}}\bigg\}mn(\|\bx_i - \widehat{\bx}_i\|_2^2 + \|\bnu_i - \widehat{\bnu}_i\|_\infty^2) + \Optilde(mn)(\|\widehat{\bx}_i - \bx_i\|_2^3 + \|\bnu_i - \widehat{\bnu}_i\|_\infty^3).
\end{align*}
Also, by Assumption \ref{assumption:prior_transform}, it is clear that
\begin{align*}
&\bigg|\frac{\partial}{\partial\bx_i\transpose}\log\pi_\bx(\bx_i')(\bx_i - \bx_{0i})\bigg| + \bigg|\sum_{t = 1}^m\frac{\mathrm{d}}{\mathrm{d}\nu_{it}}\log\pi_\nu(\nu_{it}')(\nu_{it} - \nu_{0it})\bigg|\\
&\quad\leq C(\|\bx_i - \bx_{0i}\|_2 + m\|\bnu_i - \bnu_{0i}\|_\infty)\leq C(\|\bx_i - \widehat{\bx}_{i}\|_2 + m\|\bnu_i - \widehat{\bnu}_{i}\|_\infty) + \Optilde\bigg\{\frac{\sqrt{m}(\log n)^\xi}{\sqrt{n}}\bigg\}. 
\end{align*}
Note that the elments of $\sqrt{mn}(\bx_i - \widehat{\bx}_i)$ and $\sqrt{mn}(\bnu_i - \widehat{\bnu}_i)$ are mean-zero Gaussian random variables with bounded variance under the integrating distribution. Therefore, 
\begin{align*}
&\bigg|\iint_{\calB_{in}^c}\log\frac{\phi([\bx_i\transpose,\bnu_i\transpose]\transpose\mid[\widehat{\bx}_i\transpose,\widehat{\bnu}_i\transpose]\transpose,(mn\bGamma_{in}^*)^{-1})}
  {\pi_{(\bx_i, \bnu_i)}(\bx_i, \bnu_i\mid\mathbb{A})}
  \phi\bigg(\begin{bmatrix}\bx_i\\\bnu_i\end{bmatrix}\mathrel{\bigg|}\begin{bmatrix}\widehat{\bx}_i\\\widehat{\bnu}_i\end{bmatrix}, \frac{1}{{mn}}\bGamma_{in}^{-1}\bigg)
  \mathrm{d}\bx_i\mathrm{d}\nu_i
  \bigg|\\
&\quad\leq \iint_{\calB_{in}^c}
\bigg[C(\|\bx_i - \widehat{\bx}_i\|_2 + m\|\bnu_i - \widehat{\bnu}_i\|_\infty) + \Optilde\bigg\{\frac{(\log n)^\xi}{\sqrt{n}}\bigg\}mn(\|\bx_i - \widehat{\bx}_i\|_2^2 + \|\bnu_i - \widehat{\bnu}_i\|_\infty^2)\\
&\qquad + \Optilde(mn)(\|\bx_i - \widehat{\bx}_i\|_2^3 + \|\bnu_i - \widehat{\bnu}_i\|_\infty^3) + \Optilde\bigg\{\frac{\sqrt{m}(\log n)^{3\xi}}{\sqrt{n}}\bigg\}\bigg]
\phi\bigg(\begin{bmatrix}\bx_i\\\bnu_i\end{bmatrix}\mathrel{\bigg|}\begin{bmatrix}\widehat{\bx}_i\\\widehat{\bnu}_i\end{bmatrix}, \frac{1}{{mn}}\bGamma_{in}^{-1}\bigg)\mathrm{d}\bx_i\mathrm{d}\bnu_i\\
&\quad = \Optilde\bigg\{\frac{\sqrt{m}(\log n)^{3\xi}}{\sqrt{n}}\bigg\}.
\end{align*}
The proof is thus completed.
\end{proof}
We are now in a position to prove Theorem \ref{thm:variational_BvM}. 
\begin{proof}[\bf Proof of Theorem \ref{thm:variational_BvM}]
By Pinsker's inequality, Lemma \ref{lemma:VB_lemma}, and the definition of $(\bmu_i,\bSigma_i)$, we have
\begin{align*}
&\iint_{\mathbb{R}^{d - 1}\times\mathbb{R}^m}\bigg|\pi_{(\bx_i, \bnu_i)}(\bx_i, \bnu_i\mid\mathbb{A}) - 
  \phi\bigg(\begin{bmatrix}\bx_i\\\bnu_i\end{bmatrix}\mathrel{\bigg|}
  \bmu_i^\star, \bSigma_{in}^\star\bigg)
  \bigg|\mathrm{d}\bx_i\mathrm{d}\bnu_i\\
&\quad\leq D_{\mathrm{KL}}\bigg(\phi\bigg(\begin{bmatrix}\bx_i\\\bnu_i\end{bmatrix}\mathrel{\bigg|}
  \bmu_i^\star, \bSigma_{in}^\star\bigg)\bigg\|\pi_{(\bx_i, \bnu_i)}(\bx_i, \bnu_i\mid\mathbb{A})\bigg)^{1/2}\\
&\quad\leq D_{\mathrm{KL}}\bigg(\phi\bigg(\begin{bmatrix}\bx_i\\\bnu_i\end{bmatrix}\mathrel{\bigg|}
  \begin{bmatrix}\widehat{\bx}_i\\\widehat{\bnu}_i\end{bmatrix}, \frac{1}{mn}(\bGamma_{in}^*)^{-1}\bigg)\bigg\|\pi_{(\bx_i, \bnu_i)}(\bx_i, \bnu_i\mid\mathbb{A})\bigg)^{1/2} = \Optilde\bigg\{\frac{m^{1/4}(\log n)^{3\xi/2}}{n^{1/4}}\bigg\}.
\end{align*}
Then, by triangle inequality and Theorem \ref{thm:BvM}, we further have
\begin{align*}
&\iint_{\mathbb{R}^{d - 1}\times\mathbb{R}^m}\bigg|\phi\bigg(\begin{bmatrix}\bx_i\\\bnu_i\end{bmatrix}\mathrel{\bigg|}
  \bmu_i^\star, \bSigma_{in}^\star\bigg)
   - 
  \phi\bigg(\begin{bmatrix}\bx_i\\\bnu_i\end{bmatrix}\mathrel{\bigg|}
  \begin{bmatrix}\widehat{\bx}_i\\\widehat{\bnu}_i\end{bmatrix}, \frac{1}{mn}(\bGamma_{in}^*)^{-1}\bigg)
  \bigg|\mathrm{d}\bx_i\mathrm{d}\bnu_i\\
&\quad
\leq \iint_{\mathbb{R}^{d - 1}\times\mathbb{R}^m}\bigg|\pi_{(\bx_i, \bnu_i)}(\bx_i, \bnu_i\mid\mathbb{A}) - 
  \phi\bigg(\begin{bmatrix}\bx_i\\\bnu_i\end{bmatrix}\mathrel{\bigg|}
  \begin{bmatrix}\widehat{\bx}_i\\\widehat{\bnu}_i\end{bmatrix}, \frac{1}{mn}(\bGamma_{in}^*)^{-1}\bigg)
  \bigg|\mathrm{d}\bx_i\mathrm{d}\bnu_i\\
&\qquad + \iint_{\mathbb{R}^{d - 1}\times\mathbb{R}^m}\bigg|\pi_{(\bx_i, \bnu_i)}(\bx_i, \bnu_i\mid\mathbb{A}) - 
  \phi\bigg(\begin{bmatrix}\bx_i\\\bnu_i\end{bmatrix}\mathrel{\bigg|}
  \bmu_i^\star, \bSigma_{in}^\star\bigg)
  \bigg|\mathrm{d}\bx_i\mathrm{d}\bnu_i\\
&\quad = \Optilde\bigg\{\frac{m^{1/4}(\log n)^{3\xi/2}}{n^{1/4}}\bigg\}.
\end{align*}
For the second assertion, it is sufficient to show that $\sqrt{mn}(\widehat{\bx}_i - \bx_i^*) = o_p(1)$ by delta method. We prove it using the characteristic function. 
Let $\varphi_{in}(\bt)$ denote the characteristic function of the marginal distribution of $\bx_i$ induced from the Bernstein-von Mises limit distribution $[\bx_i\transpose, \bnu_i\transpose]\transpose\sim\mathrm{N}([\widehat{\bx}_i\transpose, \widehat{\bnu}_i\transpose]\transpose, (mn\bGamma_{in}^*)^{-1})$, 
and $\varphi_{in}^*$ denote the characteristic function of the marginal distribution of $\bx_i$ induced from the Gaussian VB distribution $[\bx_i\transpose,\bnu_i\transpose]\transpose\sim\mathrm{N}(\bmu_i^\star, \bSigma_{in}^\star)$. 
Since $|e^{\mathbbm{i}x}| = 1$ for any $x\in\mathbb{R}$, where $\mathbbm{i} = \sqrt{-1}$, it follows that $\sup_{\bt\in\mathbb{R}^{d - 1}}|\varphi_{in}(\bt) - \varphi_{in}^*(\bt)| = \Optilde(\tau_n)$ for some $\tau_n\to 0$ as $n\to\infty$. 
Also, because $\varphi_{in}(\cdot)$ and $\varphi_{in}^*(\cdot)$ are associated with multivariate normal distributions, it is clear that
\begin{align*}
\varphi_{in}(\bt) = \exp\bigg(\mathbbm{i}\bt\transpose\widehat{\bx}_i - \frac{1}{2mn}\bt\transpose\bGamma_{1in}^{-1}\bt\bigg),\quad
\varphi_{in}^*(\bt) = \exp\bigg(\mathbbm{i}\bt\transpose{\bx}_i^* - \frac{1}{2}\bt\transpose\bSigma_{1in}\bt\bigg),
\end{align*}
where $\bGamma_{1in}$ is some $(d - 1)\times (d - 1)$ positive definite matrix whose eigenvalues are bounded away from $0$ and $\infty$, and $\bSigma_{1in}$ is some $(d - 1)\times(d - 1)$ positive definite matrices.
Furthermore, the convergence of the characteristic functions implies
\[
\sup_{\bt\in\mathbb{R}^{d - 1}}\bigg|\exp\bigg(-\frac{1}{2mn}\bt\transpose\bGamma_{1in}^{-1}\bt\bigg) - \exp\bigg(-\frac{1}{2}\bt\transpose\bSigma_{1in}\bt\bigg)\bigg| = \Optilde(\tau_n).
\]
 For any $\bu\in\mathbb{R}^{d - 1}$ with $\|\bu\|_2 = O(1)$, we know that $(1/2)\bu\transpose\bGamma_{in}^{-1}\bu \asymp 1$. Then, 
\begin{align*}
&\exp\bigg(-\frac{1}{2mn}\bt\transpose\bGamma_{1in}^{-1}\bt\bigg)|\exp(\mathbbm{i}\bt\transpose\bx_i^*) - \exp(\mathbbm{i}\bt\transpose\widehat{\bx}_i)|\\
&\quad = \bigg|\exp\bigg(-\frac{1}{2mn}\bt\transpose\bGamma_{1in}^{-1}\bt + \mathbbm{i}\bt\transpose\bx_i^*\bigg) - \exp\bigg(-\frac{1}{2mn}\bt\transpose\bGamma_{1in}^{-1}\bt + \mathbbm{i}\bt\transpose\widehat{\bx}_i\bigg)\bigg|\\
&\quad \leq \bigg|\exp\bigg(-\frac{1}{2mn}\bt\transpose\bGamma_{1in}^{-1}\bt + \mathbbm{i}\bt\transpose\bx_i^*\bigg) - \exp\bigg(-\frac{1}{2}\bt\transpose\bSigma_{1in}\bt + \mathbbm{i}\bt\transpose{\bx}_i^*\bigg)\bigg|\\
&\qquad + \bigg|\exp\bigg(-\frac{1}{2mn}\bt\transpose\bGamma_{1in}^{-1}\bt + \mathbbm{i}\bt\transpose\widehat{\bx}_i\bigg) - \exp\bigg(-\frac{1}{2}\bt\transpose\bSigma_{1in}\bt + \mathbbm{i}\bt\transpose{\bx}_i^*\bigg)\bigg|\\
&\quad\leq \bigg|\exp\bigg(-\frac{1}{2mn}\bt\transpose\bGamma_{1in}^{-1}\bt\bigg) - \exp\bigg(-\frac{1}{2}\bt\transpose\bSigma_{1in}\bt\bigg)\bigg| + |\varphi_{in}(\bt) - \varphi_{in}^*(\bt)|. 
\end{align*}
With $\bt = \sqrt{mn}\bu$, we further obtain
\begin{align*}
|e^{\mathbbm{i}\bu\transpose\sqrt{mn}\bx_i^*} - e^{\mathbbm{i}\bu\transpose\sqrt{mn}\widehat{\bx}_i}| = \Optilde(\tau_n).
\end{align*}
Hence, 
\begin{align*}
\left|\expect_0\left\{\exp(\mathbbm{i}\sqrt{mn}(\bx_i^* - \widehat{\bx}_i))\right\} - 1\right|
& = \left|\expect_0\left\{\exp(\mathbbm{i}\sqrt{mn}\bx_i^*) - \exp(\mathbbm{i}\sqrt{mn}\widehat{\bx}_i)\right\}\exp(-\mathbbm{i}\sqrt{mn}\widehat{\bx}_i)\right|\\
& =  \left|\expect_0\left\{\exp(\mathbbm{i}\sqrt{mn}\bx_i^*) - \exp(\mathbbm{i}\sqrt{mn}\widehat{\bx}_i)\right\}\right| \to 0
\end{align*}
as $n\to\infty$. Namely, we obtain $\sqrt{mn}(\bx_i^* - \widehat{\bx}_i) = o_p(1)$, thereby completing the proof. 
\end{proof}

\end{appendices}

\bibliographystyle{abbrvnat}
\bibliography{reference_inference_MLDCMM,reference_MG}

\end{document}